\pdfoutput=1
\documentclass[sigplan,screen]{acmart}

\usepackage{etex}
\usepackage{epsfig}
\usepackage{latexsym}
\usepackage{graphicx}
\usepackage{amssymb}
\usepackage{amsfonts}
\usepackage{amsthm}
\usepackage[inline]{enumitem}
\usepackage{mathrsfs}

\usepackage{amsmath}
\usepackage{mathtools}
\usepackage{verbatim}
\usepackage{url}
\usepackage{colortbl}
\usepackage{stmaryrd}
\usepackage{epic,eepic}
\usepackage{xspace}
\usepackage{caption}
\usepackage{hyperref}
\usepackage{tcolorbox}
\usepackage{multirow}
\usepackage{multicol}

\usepackage{scalerel}
\usepackage{cleveref}

\usepackage{tikz}
\usetikzlibrary{automata}
\usetikzlibrary{arrows}
\usetikzlibrary{backgrounds}
\usetikzlibrary{positioning}
\usetikzlibrary{fit}
\usetikzlibrary{calc}
\usetikzlibrary{matrix}
\usetikzlibrary{patterns}
\usetikzlibrary{shapes.geometric}
\usetikzlibrary{decorations.pathreplacing}
\usetikzlibrary{decorations.pathmorphing}

\tikzset{
  dot/.style={draw, minimum size=0.7mm, inner sep=0pt, outer sep=0pt, shape=circle, fill=black},
  pto/.style={-latex, line width=0.28mm},
  alin/.style={line width=0.25mm,black!90},
  path/.style={-latex,dashed,line width=0.28mm,dash pattern=on 7pt off 3pt on 3pt off 3pt on 3pt off 3pt on 3pt off 3pt},
  segment/.style={|-|,line width=0.25mm,black!90},
  aux/.style={-latex,line width=0.25mm,black!80}
}

\pgfdeclarelayer{bg0}    
\pgfdeclarelayer{bg1}
\pgfdeclarelayer{bg2}
\pgfsetlayers{bg0,bg1,bg2,main}

\hypersetup{
  linkcolor  = blue,
  citecolor  = blue,
  urlcolor   = blue,
  colorlinks = true,
}

\usepackage{todonotes}
\usepackage{stackengine}

\usepackage{xcolor}

\usepackage{algorithm}
\usepackage[noend]{algpseudocode}

\algnewcommand{\Label}{\State\unskip}

\usepackage{booktabs}   
\usepackage{subcaption} 

\newcommand{\set}[1]{\{ #1 \}}

\newcommand{\pair}[2]{(#1,#2)}
\newcommand{\triple}[3]{(#1,#2,#3)}

\newcommand{\Bool}{\ensuremath{\mathbb{B}}}
\newcommand{\Nat}{\ensuremath{\mathbb{N}}}

\newcommand{\true}{{\top}}
\newcommand{\false}{{\perp}}
\newcommand{\bottom}{\false}
\renewcommand{\implies}{\Rightarrow}
\renewcommand{\iff}{\Leftrightarrow}
\newcommand{\powerset}[1]{\mathcal{P}(#1)}

\DeclarePairedDelimiter\card{\lvert}{\rvert}
\newcommand{\range}[1]{\mbox{ran}(#1)}

\newcommand{\grammardef}{:=}

\newcommand{\aset}{X}
\newcommand{\asetbis}{Y}

\newcommand{\avarprop}{p}
\newcommand{\avarpropbis}{q}
\newcommand{\avarpropter}{r}
\newcommand{\varprop}{{\rm AP}} 

\newcommand{\aformula}{\varphi} 
\newcommand{\aformulabis}{\psi} 
\newcommand{\aformulater}{\chi} 
\newcommand{\asetformulae}{X}
\newcommand{\subf}[1]{{\sf subf}(#1)}

\newcommand{\eqtext}[1]{\mathbin{\raisebox{-0.55pt}{$\stackrel{\text{\raisebox{0.8pt}[0pt][0pt]{\begin{tiny}\makebox(0,0){#1}\end{tiny}}}}{=}$}}}
\newcommand{\egdef}{\mathbin{\raisebox{-0.55pt}{$\stackrel{\text{\raisebox{-1pt}[0pt][0pt]{\begin{tiny}$\mathsf{def}$\end{tiny}}}}{=}$}}} 
\newcommand{\equivdef}{\mathbin{\raisebox{-0.55pt}{$\stackrel{\text{\raisebox{-1pt}[0pt][0pt]{\begin{tiny}$\mathsf{def}$\end{tiny}}}}{\equivaut}$}}} 
\newcommand{\equivaut}{\;\Leftrightarrow\;}

\newcommand{\amap}{\mathfrak{f}}
\newcommand{\amapbis}{\mathfrak{g}}
\newcommand{\amapter}{\mathfrak{t}}

\newcommand {\pspace} {\textsc{PSpace}\xspace}

\newcommand {\polytime} {\textsc{Pol}\xspace}

\newcommand {\nexptime} {\textsc{NExpTime}\xspace}

\newcommand {\tower} {\textsc{Tower}\xspace}
\newcommand{\aexppol}{\textsc{AExp}$_{\text{\polytime}}$\xspace}

\newenvironment {lemma*} {\noindent {\bf Lemma} \em} {\rm}

\newcommand{\aalphabet}{\Sigma}     
\newcommand{\aword}{\mathfrak{w}}

\newcommand{\avariable}{\mathtt{x}}
\newcommand{\avariablebis}{\mathtt{y}}
\newcommand{\avariableter}{\mathtt{z}}

\newcommand{\separate}{\mathbin{\ast}}
\newcommand{\size}{\mathtt{size}}

\newcommand{\amodel}{\mathfrak{M}}

\newcommand{\aCCmodeltranslation}[4]{\mathcal{T}^{#4}_{#1}\pair{#2}{#3}}

\newcommand{\aSALmodeltranslationaux}[4]{\mathcal{T}^{#4}_{#1}\pair{#2}{#3}}

\newcommand{\atranslation}{\tau}

\newcommand{\defstyle}[1]{{\emph{#1}}}

\newcommand{\cut}[1]{\ignorespaces}
\newcommand{\interval}[2]{[#1,#2]}

\mathchardef\mhyphen="2D 

\makeatletter
\newcommand{\@transb}[2]{t_{#1}\!\left(#2\right)}
\newcommand{\@transnob}[1]{t\!\left(#1\right)}
\newcommand{\trans}{\@ifstar{\@transb}{\@transnob}}
\makeatother

 \newcounter{theoreme}[section]

\newcommand{\Diamondminus}{\Diamond^{-1}}

\newcommand{\SabDiamond}{
   {\begin{tikzpicture}[baseline=-3.5pt]
     \node[draw,scale=0.47,diamond,fill=black]{};
   \end{tikzpicture}}
}

\newcommand{\HMDiamond}[1]{\tup{{\rm #1}}}
\newcommand{\HMBox}[1]{[{\rm #1}]}

\newcommand{\alogic}{\mathfrak{L}}
\newcommand{\MSL}[2]{\fullMSL^{#1}(#2)}
\newcommand{\fullMSL}{\ensuremath{\mathsf{MSL}}\xspace}

\newcommand{\ML}{\ensuremath{\mathsf{ML}}\xspace}
\newcommand{\QK}{\ensuremath{\mathsf{QK}}\xspace}
\newcommand{\BI}{\ensuremath{\mathsf{BI}}\xspace}
\newcommand{\RML}{\ensuremath{\mathsf{RML}}\xspace}

\newcommand{\HML}{\ensuremath{\mathsf{HML}}\xspace}

\newcommand{\tup}[1]{\langle #1 \rangle}

\newcommand{\model}{\amodel}

\newcommand{\fsize}[1]{{\sf size(#1)}}
\newcommand{\md}[1]{{\sf md}(#1)}
\newcommand{\gr}[1]{{\sf gr}(#1)}

\newcommand{\abisim}{\mathcal{Z}}
\newcommand{\abisimtwo}{\mathcal{K}}
\newcommand{\bisrel}[2]{\xspace\leftrightarrows_{#1}^{#2}\xspace}

\newcommand{\gamerel}[2]{\approx_{#1}^{#2}}
\newcommand{\neggamerel}[2]{{\not\approx}_{#1}^{#2}}

\newcommand\fin{\text{\normalfont{fin}}}

\newcommand{\acharformula}[3]{\Gamma(#1)_{#2}^{#3}} 
\newcommand{\atypeset}[2]{\mathcal{T}^{#2}{(#1)}}
\newcommand{\atowertypeset}[2]{\mathscr{T}^{#2}{(#1)}}
\newcommand{\atermset}{\mathsf{T}}
\newcommand{\redformula}{\widehat{\varphi}}
\newcommand{\redformulabis}{\widehat{\psi}}

\newcommand{\typeeqclass}[2]{\equiv_{#1}^{#2}}

\algrenewcommand\algorithmicrequire{\textbf{In:}}
\algrenewcommand\algorithmicensure{\textbf{Out:}}

\newcommand{\agarbage}{\mathcal{G}}

\newcommand{\auniverse}{\worlds}
\newcommand{\aaccessrelation}{\arelation}

\newcommand{\apropeval}{\avaluation}

\newcommand{\apropset}{\mathsf{P}}
\newcommand{\apropsetbis}{\mathsf{Q}}

\newcommand{\apropsetN}{\mathsf{N}}

\algblockdefx{Cases}{EndCases}%
   [1]{\textbf{case} #1 \textbf{of}}%
   {\textbf{end case}}
\algcblockx[Cases]{Cases}{Case}{EndCases}%
   [1]{#1:\enspace}%
   {\textbf{end case}}

\DeclareMathOperator{\chopop}{\raisebox{-2pt}{\rule{1.2pt}{2.1ex}}}

\newcommand{\modallogicCC}{{\ensuremath{\ML(\,\chopop\,)}}\xspace}
\newcommand{\modallogicSC}{{\ensuremath{\ML(\separate)}}\xspace}

\newcommand{\msokt}{\ensuremath{\QK^{t}}\xspace}
\newcommand{\cH}{\mathcal{H}}
\newcommand{\cT}{\mathcal{T}}
\newcommand{\cTT}{\mathcal{T}\!\!\!\!\mathcal{T}}

\newcommand{\cV}{\mathcal{V}}

\newcommand{\AP}{\varprop}

\newcommand{\expchess}[2]{\mathtt{Board}}
\newcommand{\tiling}{\mathtt{Tile}}
\newcommand{\multitiling}[3]{\mathtt{MTiling}}

\newcommand{\initmodel}[3]{\mathtt{initmodel}}

\newcommand{\avarleft}{\mathtt{l}}
\newcommand{\avarselect}{\mathtt{s}}
\newcommand{\avarright}{\mathtt{r}}
\newcommand{\avartree}{\mathit{t}}

\newcommand{\treeval}{\mathtt{val}}

\newcommand{\atiling}[2]{\mathtt{tiling}_{#2}(#1)}
\newcommand{\agrid}[2]{\mathtt{grid}_{#2}(#1)}

\newcommand{\Aux}{\mathtt{Aux}}

\newcommand{\aaux}{\mathtt{ax}}
\newcommand{\aauxbis}{\mathtt{bx}}

\newcommand{\Boxbox}{\boxplus}

\newcommand{\complete}[1]{\mathtt{type}(#1)}
\newcommand{\completeplus}[1]{\mathtt{type}_{\mathtt{lsr}}(#1)}
\newcommand{\ainit}{init}
\newcommand{\init}[1]{\mathtt{\ainit}(#1)}
\newcommand{\nominal}[2]{\mathtt{nom}_{#2}(#1)}
\newcommand{\atnom}[2]{\mathtt{@}_{#1}^{#2}}
\newcommand{\zero}{\mathtt{0}}
\newcommand{\one}{\mathtt{1}}
\newcommand{\successor}[3]{[#2 = #1 {+} 1]_{#3}}
\newcommand{\bsuccessor}[4]{[#2\,{\eqtext{$#4$}}\,#1 {+} 1]_{#3}}
\newcommand{\less}[4]{[#1\,{<}\,#2]^{#3}_{#4}}
\newcommand{\equivalent}[4]{[#1\,{=}\,#2]^{#3}_{#4}}
\newcommand{\bequivalent}[4]{[#1\,{\eqtext{$#4$}}\,#2]_{#3}}
\newcommand{\fork}[4]{\mathtt{fork}^{#3}_{#4}(#1,#2)}
\newcommand{\lsrpartition}[1]{\mathtt{lsr}(#1)}
\newcommand{\twonoms}[3]{\mathtt{nom}_{#3}(#1\!\neq\!#2)}

\newcommand{\heapdim}{D}

\newcommand{\selectpred}[4]{\mathtt{S}^{#3}_{#4}(#1, #2)}
\newcommand{\bselectpred}[4]{\mathtt{S[}#4\mathtt{]}_{#3}(#1,#2)}
\newcommand{\leftpred}[4]{\mathtt{L}^{#3}_{#4}(#1, #2)}
\newcommand{\bleftpred}[4]{\mathtt{L[}#4\mathtt{]}_{#3}(#1,#2)}
\newcommand{\selectleftpred}[4]{\mathtt{LS}^{#3}_{#4}(#1, #2)}

\newcommand{\rightpred}[2]{\mathtt{R}(#1,#2)}
\newcommand{\brightpred}[3]{\mathtt{R[}#3\mathtt{]}(#1,#2)}

\newcommand{\apS}{sub}
\newcommand{\apZ}{zero}
\newcommand{\apU}{uniq}
\newcommand{\apC}{compl}
\newcommand{\apA}{aux}
\newcommand{\pS}[1]{\mathtt{\apS}{(#1)}}
\newcommand{\pZ}[1]{\mathtt{\apZ}{(#1)}}
\newcommand{\pU}[1]{\mathtt{\apU}{(#1)}}
\newcommand{\pC}[1]{\mathtt{\apC}{(#1)}}
\newcommand{\pA}{\mathtt{\apA}}

\newcommand{\apone}{one}
\newcommand{\apfirst}{first}
\newcommand{\aphor}{hor}
\newcommand{\apvert}{vert}
\newcommand{\pone}[1]{\texttt{\apone}_{#1}}
\newcommand{\pfirst}[2]{\texttt{\apfirst}_{#2}{(#1)}}
\newcommand{\phor}[2]{\texttt{\aphor}_{#2}{(#1)}}
\newcommand{\pvert}[2]{\texttt{\apvert}_{#2}{(#1)}}

\newcommand{\pZT}[2]{\texttt{\apZ}_{#2}{(#1)}}
\newcommand{\pUT}[2]{\texttt{\apU}_{#2}{(#1)}}
\newcommand{\pCTB}[3]{\texttt{\apC}[#3]_{#2}{(#1)}}
\newcommand{\pCTBB}[2]{\texttt{\apC}_{#2}{(#1)}}

\newcommand{\bigO}{O}

\newcommand{\nb}[1]{\boldsymbol{\mathfrak{n}}(#1)}
\newcommand{\nbexp}[2]{\boldsymbol{\mathfrak{n}}_{#2}(#1)}

\newcommand{\SAL}{\ensuremath{\mathsf{SAL}(\,\ambientchop\,)}\xspace}
\newcommand{\fullSAL}{\ensuremath{\mathsf{SAL}}\xspace}

\DeclareMathOperator{\ambientchop}{\mathbin{\chopop}}

\newcommand{\ambientprop}{\mathtt{ap}}
\newcommand{\ambientchild}{\mathtt{rel}}
\newcommand{\numchildgeq}[1]{\mathtt{[\#\,{\geq}\,#1]}}
\newcommand{\numchildeq}[1]{\mathtt{[\#\,{=}\,#1]}}

\newcommand{\ambienttrees}{\mathbb{T}_{\mathsf{SAL}}\xspace}

\definecolor{Gray}{gray}{0.8}

\newenvironment{nscenter}
 {\parskip=2pt\par\nopagebreak\centering}
 {\par\noindent\ignorespacesafterend}

 \makeatletter
 \def\rulelab#1#2{%
   \ \ #1:%
   \begingroup%
     \def\@currentlabel{#1}%
     \phantomsection\label{#2}%
   \endgroup%
 }
 \makeatother

 \makeatletter
 \def\desclabel#1#2{\begingroup
    \def\@currentlabel{#1}%
    #1\label{#2}\endgroup
 }
 \makeatother

 \makeatletter
 \def\pointlabel#1#2{\begingroup
    \def\@currentlabel{(#1)}%
    #1.\label{#2}\endgroup
 }
 \makeatother

 \makeatletter
 \def\phantomlabel#1#2{\begingroup
    \def\@currentlabel{#1}%
    \label{#2}\endgroup
 }
 \makeatother

\newcommand{\GML}{\ensuremath{\mathsf{GML}}\xspace}
\newcommand{\MSO}{\ensuremath{\mathsf{MSO}}\xspace}
\newcommand{\PA}{\ensuremath{\mathsf{PA}}\xspace}
\newcommand{\Gdiamond}[2]{\Diamond_{#1} \ #2}

\newcommand{\aworld}{w}

\newcommand{\worlds}{W}
\newcommand{\arelation}{R}
\newcommand{\avaluation}{V}
\newcommand{\maxpc}[1]{{\tt max}_{\tt PC}(#1)}
\newcommand{\maxgmod}[1]{{\tt max}_{\tt GM}(#1)}
\newcommand{\aliteral}{L}

\newcommand{\bd}[1]{{\sf bd}(#1)}
\newcommand{\topbd}[1]{{\sf bd}^{\rm top}(#1)}
\newcommand{\toppybd}[1]{{\sf w}(#1)}

\newcommand{\DNF}[1]{{\sf DNF}(#1)}

\newcommand{\boundd}[1]{\mathcal{B}(#1)}

\newcommand{\aformulaerankset}[2]{\modallogicSC[#1,#2]}
\newcommand{\acharformulaSC}[3]{\Pi(#1)_{#2}^{#3}} 

\newif\ifLongVersionOnly\LongVersionOnlyfalse

\newcommand{\satproblem}[1]{{\rm Sat(}#1{\rm)}}

\newcommand{\atile}{\mathtt{c}}

\newcommand{\aname}{\mathtt{n}}
\newcommand{\anamebis}{\mathtt{m}}
\newcommand{\anvarprop}{\avarprop}

\newcommand{\plneg}{\dot{\neg}}
\newcommand{\plvee}{\dot{\vee}}
\newcommand{\plcnot}{\text{\textasciitilde}}
\newcommand{\ateam}{\mathfrak{T}}

\newcommand{\aplvaluation}{\mathfrak{v}}

\newcommand{\tluniq}[1]{\mathtt{uni}(#1)}
\newcommand{\tlcopies}[1]{\mathtt{cp}(#1)}

\Crefname{definition}{Definition}{Defs.}
\Crefname{equation}{Equation}{Eqs.}
\Crefname{figure}{Figure}{Figs.}
\Crefname{proposition}{Proposition}{Props.}
\Crefname{theorem}{Theorem}{Thms.}
\Crefname{example}{Example}{Exs.}
\Crefname{corollary}{Corollary}{Cors.}
\Crefname{enumi}{}{}
\Crefname{lemma}{Lemma}{Lemmata}
\Crefname{section}{Section}{Secs.}
\Crefname{appendix}{Appendix}{Apps.}

\makeatletter
\newcommand{\pushright}[1]{\ifmeasuring@#1\else\omit\hfill$\displaystyle#1$\fi\ignorespaces}
\newcommand{\pushleft}[1]{\ifmeasuring@#1\else\omit$\displaystyle#1$\hfill\fi\ignorespaces}
\makeatother

\newcommand{\veeweight}[1]{{\tt w}_{\plvee}(#1)}
\newcommand{\newbd}[2]{{\tt bd}(#1,#2)}
\newcommand{\maxbd}[1]{{\tt max}_{\tt bd}(#1)}
\newcommand{\submax}[1]{{\tt submax}_{\tt GM}(#1)}
\newcommand{\cd}[1]{{\tt cd}(#1)} 
\newcommand{\dweight}[1]{{\tt w}_{\Diamond}(#1)}

\AtBeginDocument{%
  \providecommand\BibTeX{{%
    \normalfont B\kern-0.5em{\scshape i\kern-0.25em b}\kern-0.8em\TeX}}}

\setcopyright{acmcopyright}
\copyrightyear{2020}
\acmYear{2020}

\acmConference[LICS '20]{Proceedings of the 35th Annual ACM/IEEE Symposium on Logic in Computer Science (LICS)}{July 8--11, 2020}{Saarbr\"ucken, Germany}
\acmBooktitle{Proceedings of the 35th Annual ACM/IEEE Symposium on Logic in Computer Science (LICS '20), July 8--11, 2020, Saarbr\"ucken, Germany}

\acmPrice{15.00}
\acmDOI{10.1145/3373718.3394787}
\acmISBN{978-1-4503-7104-9/20/07}

\newif\ifLongVersionWithAppendix\LongVersionWithAppendixtrue

\begin{document}

\title[Modal Logics with Composition on Finite Forests]{Modal Logics with Composition
on Finite Forests: Expressivity and Complexity (Extra Material)}


\author{Bartosz Bednarczyk}
\affiliation{\institution{TU Dresden \& University of Wroc\l{}aw}}

\author{St\'ephane Demri}
\affiliation{\institution{LSV, CNRS, ENS Paris-Saclay, Universit\'e Paris-Saclay}}

\author{Raul Fervari}
\affiliation{\institution{FAMAF, Universidad Nacional de C\'ordoba \& CONICET}}

\author{Alessio Mansutti}
\affiliation{\institution{LSV, CNRS, ENS Paris-Saclay, Universit\'e Paris-Saclay}}

\renewcommand{\shortauthors}{Bednarczyk, Demri, Fervari \& Mansutti}

\begin{abstract}
We study the expressivity and
complexity of two modal logics interpreted on finite forests and
equipped with standard modalities to reason on submodels. The logic~\modallogicCC extends the modal logic K
with the composition operator~$\chopop$ from ambient logic,
whereas \modallogicSC features the separating conjunction~$\separate$ from separation logic.
Both operators are second-order in nature. We show that~\modallogicCC is as expressive as the graded modal
logic~\GML (on trees)
whereas~\modallogicSC is  strictly less expressive than~\GML. Moreover,
we
establish that the satisfiability problem is \tower-complete for~\modallogicSC, whereas
it is (only)~\aexppol-complete for~\modallogicCC, a result which is surprising
given their relative expressivity.
As by-products, we solve open problems related to sister logics such as static ambient logic and
modal separation logic.

\end{abstract}

\ccsdesc[500]{Theory of computation~Modal and temporal logics}

\keywords{modal logic on trees, separation logic, static ambient logic, graded modal logic, expressive power, complexity}

\maketitle

\section{Introduction}
\label{section-intro}

The ability to quantify over substructures to express properties of a model is often instrumental to perform modular and local reasoning. Two well-known examples are provided by
separation logics~\cite{Ishtiaq&OHearn01,OHearn&Reynolds&Yang01,Reynolds02},
dedicated to reasoning
on pointer programs,
and ambient (or more generally, spatial) logics~\cite{CCG03,CGZ05,Boneva&Talbot&Tison05,Dawar&Gardner&Ghelli07},
dedicated to reasoning on disjoint data structures.
In the realm of modal logics dedicated to knowledge representation,
submodel reasoning remains a key ingredient to express the dynamics of knowledge and belief, as done in the
logics of public announcement~\cite{Plaza89,Lutz06,Balbianietal08},
sabotage modal logics~\cite{AucherBG18}, refinement modal logics~\cite{Bozzelli&vanDitmarsch&Pinchinat15} and
relation-changing logics~\cite{Aucheretal09,arfeho:movi12,AFH15}.
Though the models may be of different nature (e.g.\ memory states for separation logics, epistemic models
for logics of public announcement or finite edge-labelled trees for ambient logics), all those logics feature composition operators that enable
to compose or decompose substructures in a very natural way.

From a technical point of view, reasoning about submodels requires a  global analysis, unlike the
local approach for classical modal and temporal logics (typically based on automata techniques~\cite{Vardi&Wolper86,Vardi&Wolper94}).
This makes the comparison between those formalisms quite challenging and
often limited to a superficial analysis on the different classes of models and  composition operators.
For instance, the composition operator $\chopop$ in ambient logics decomposes
a tree into two disjoint pieces such that
once a node has been assigned to one submodel, all its descendants belong to the same submodel. Instead,
the separating conjunction $\separate$ from separation logic decomposes
the memory states into two disjoint memory states.
Obviously, these and other well-known operators are closely related
but no uniform framework investigates exhaustively their relationships in terms of expressive power.

Most of these logics can be easily encoded in monadic second-order logic \MSO (or in second-order modal logics~\cite{Fine70,Laroussinie&Markey14}).
Complexity-wise, if models are tree-like structures, we can then infer decidability thanks to the celebrated Rabin's theorem~\cite{Rabin69}.
However, most likely, this does not produce the best decision procedures when it comes to solving simple reasoning tasks (e.g.\ the satisfiability problem of \MSO is \tower-complete~\cite{Schmitz16}).
Thus, relying on \MSO as a common umbrella to capture and understand the differences between those logical formalisms is often not satisfactory.

\subsubsection*{Our motivations.}
Our intention in this work is to provide an in-depth comparison between the composition operator $\chopop$
from static ambient logic~\cite{CCG03} and the separating conjunction $\separate$ from
separation logics~\cite{Reynolds02} by identifying a common ground in terms of logical languages and models.
As a consequence, we are able to study the effects of having these operators as far as expressivity and  complexity are concerned.
We aim at defining two logics whose only differences rest on their use of $\chopop$ and $\separate$ syntactically
and semantically (by considering the adequate composition operation). To do so, we pick as our common class of models,
the Kripke-style finite trees (actually finite forests, so that the class is closed under taking submodels),
which provides an ubiquitous class of structures, extremely well-studied in computer science.
For the  underlying logical language (i.e. apart from $\chopop$ or $\separate$),
we advocate the use of the standard modal logic K (i.e. to have Boolean connectives and the standard modality $\Diamond$)
so that the main operations on the models amount  to quantify over submodels or to move along the edges. This framework
is sufficiently fundamental to give us the possibility to take advantage of model theoretical tools from modal logics~\cite{DeRijke00,Blackburn&deRijke&Venema01,FattorosiBarnaba&DeCaro85}.
The benefits of settling a common ground for comparison may lead to further comparisons with other logics and new results.

\subsubsection*{Our contributions.}
We introduce \modallogicCC and \modallogicSC, two logics interpreted on Kripke-style forest models,
equipped with the standard modality $\Diamond$, and respectively with the
composition operator~$\chopop$ from static ambient logic~\cite{CCG03} and with
the separating conjunction~$\separate$ from separation logic~\cite{Reynolds02}.
Both logical formalisms can state non-trivial properties about submodels, but the binary modalities
$\chopop$ and $\separate$ operate differently: whereas~$\separate$
is able to decompose the models at any depth, $\chopop$ is much less permissive as the decomposition
is completely determined by what happens at the level of the children of the current node.
We study their expressive power and complexity, obtaining surprising results.
We show that \modallogicCC is as expressive as the graded modal logic~\GML~\cite{FattorosiBarnaba&DeCaro85,Tobies00}
whereas \modallogicSC is strictly less expressive than \GML.
Interestingly, this latter development partially reuses the result for \modallogicCC,
hence showing how our framework allows us to transpose results between the two logics.
To show that \GML is strictly more expressive than~\modallogicSC,
we define Ehrenfeucht-Fra\"iss\'e games for~\modallogicSC.
In terms of complexity, the satisfiability problem for~\modallogicCC is shown~\aexppol-complete\footnote{
Problems in \aexppol are decidable by an alternating Turing machine working
in exponential-time and using polynomially many alternations~\cite{Bozzelli17}.
}, interestingly the same complexity as for the refinement modal logic
\RML~\cite{Bozzelli&vanDitmarsch&Pinchinat15} handling a quantifier over refinements (generalising  the submodel construction).
The~\aexppol upper bound follows from
 an ex\-po\-nen\-tial-size model property, whereas the lower bound is by reducing
the satisfiability problem for an \aexppol-complete team logic~\cite{HannulaKVV18}.
Much more surprisingly, although \modallogicSC is strictly less expressive
than \modallogicCC, its complexity is much higher (not even elementary).
Precisely, we show that the satisfiability problem for~\modallogicSC
is \tower-complete.
The \tower upper bound is a consequence of~\cite{Rabin69},
whereas hardness is shown by reduction from a \tower-complete tiling problem,
adapting substantially the \tower-hardness proof from~\cite{Bednarczyk&Demri19} for second-order modal logic K on finite trees.
To conclude, we get the best of our results on
\modallogicCC and \modallogicSC to solve several open problems.
We relate~\modallogicCC
with an intensional fragment of static ambient logic~\SAL from~\cite{CCG03} by providing
polynomial-time reductions between their satisfiability problems. Consequently, we
establish~\aexppol-completeness of~\SAL, refuting hints from~\cite[Section 6]{CCG03}. Similarly, we show that
the modal separation logic MSL($\Diamond^{-1},\separate$) from~\cite{DemriF19}
is \tower-complete.

\noindent {\em This document extends~\cite{BDFM20} with a technical appendix including additional 
information and all omitted proofs.}


%
\section{Preliminaries}
\label{section-preliminaries}
In this section, we introduce the logics \modallogicCC and \modallogicSC
interpreted on tree-like structures
equipped with operators to split the structure into disjoint pieces.
Due to the presence of such operators, we are required to consider a class of models that is closed under submodels, which we call
Kripke-style finite  forests (or finite forests for short).

Let $\varprop$ be a countably infinite
set of \defstyle{atomic propositions}.
A \defstyle{(Kripke-style) finite forest} is a triple $\amodel = \triple{\worlds}{\arelation}{\avaluation}$ where
$\worlds$ is a non-empty finite set of \defstyle{worlds},
$\avaluation: \varprop \rightarrow \powerset{\worlds}$ is a \defstyle{valuation} and
$\arelation \subseteq \worlds \times \worlds$ is a binary relation whose inverse $\arelation^{-1}$ is functional and acyclic.
Then, in particular the graph described by $\pair{\worlds}{\arelation}$ is a finite collection of disjoint finite trees (where $\arelation$ encodes the child relation).



We define $\arelation(\aworld) \egdef \set{\aworld' \in \worlds \mid \pair{\aworld}{\aworld'} \in \arelation}$.
Worlds
in~${\arelation(\aworld)}$ are understood as \defstyle{children} of $\aworld$.
We inductively define $\arelation^n$:
$\arelation^0 \egdef \{\pair{\aworld}{\aworld} \mid \aworld \in \worlds \}$;
$\arelation^{n+1} \egdef \{\pair{\aworld}{\aworld''} \mid \exists \aworld'\, {\pair{\aworld}{\aworld'} \in \arelation^n}$
$\text{and } \pair{\aworld'}{\aworld''} \in \arelation\,\}$.
$\arelation^+$ denotes~the transitive closure of~$\arelation$.

We define operators that chop a finite forest. It should be noted that these operators, as well as the resulting logics, can be cast under the umbrella of the logic of bunched implications \BI~\cite{Pym02,GalmicheBI}, with the exception that we do not explicitly require them to have an identity element (as enforced on the multiplicative operators of \BI, see \cite{GalmicheBI}).
Let $\amodel = \triple{\worlds}{\arelation}{\avaluation}$ and $\amodel_i = \triple{\worlds_i}{\arelation_i}{\avaluation_i}$ (for $i \in \{1,2\}$) be three
finite forests.
\subsubsection*{The separation logic composition.}
We introduce the binary operator $+$
that performs the disjoint union at the level of parent-child relation. Formally,
\begin{center}
$\amodel = \amodel_1 + \amodel_2$ 
      $\equivdef$
      $\arelation_1 \uplus \arelation_2 = \arelation$,
      ${\worlds_1 = \worlds_2 = \worlds}$, $\avaluation_1 = \avaluation_2 = \avaluation\!.$
\end{center}
      This is the composition used in~separation logic~\cite{Reynolds02, DemriF19}.
 The figure below depicts possible instances for $\amodel$, $\amodel_1$ and $\amodel_2$.
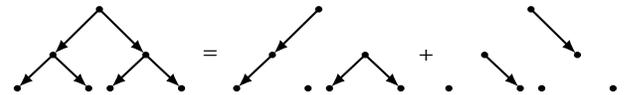
\begin{center}
    \begin{tikzpicture}[thick,yscale=0.81, xscale=0.95, every node/.style={transform shape}]
    \node[dot] (w) []{ };

    \node[dot] (w1) [below left = 0.7cm and 0.6cm of w] {};
    \node[dot] (w2) [below right  = 0.7cm and 0.6cm of w] {};

    \node[dot] (w3) [below left = 0.5cm and 0.45cm of w1] {};
    \node[dot] (w4) [below right  = 0.5cm and 0.45cm of w1] {};

    \node[dot] (w5) [below left = 0.5cm and 0.45cm of w2] {};
    \node[dot] (w6) [below right  = 0.5cm and 0.45cm of w2] {};

    \draw[pto] (w) -- (w1);
    \draw[pto] (w) -- (w2);

    \draw[pto] (w1) -- (w3);
    \draw[pto] (w1) -- (w4);

    \draw[pto] (w2) -- (w5);
    \draw[pto] (w2) -- (w6);

    \node (eq) [right = 0.625cm of w2] {$=$};

    \node[dot] (ww) [right = 3cm of w] { };

    \node[dot] (ww1) [below left = 0.7cm and 0.6cm of ww] {};
    \node[dot] (ww2) [below right  = 0.7cm and 0.6cm of ww] {};

    \node[dot] (ww3) [below left = 0.5cm and 0.45cm of ww1] {};
    \node[dot] (ww4) [below right  = 0.5cm and 0.45cm of ww1] {};

    \node[dot] (ww5) [below left = 0.5cm and 0.45cm of ww2] {};
    \node[dot] (ww6) [below right  = 0.5cm and 0.45cm of ww2] {};

    \draw[pto] (ww) -- (ww1);

    \draw[pto] (ww1) -- (ww3);

    \draw[pto] (ww2) -- (ww5);
    \draw[pto] (ww2) -- (ww6);

    \node (plw) [right = 0.575cm of ww2] {$+$};

    \node[dot] (wv) [right = 2.9cm of ww] {};

    \node[dot] (wv1) [below left = 0.7cm and 0.6cm of wv] {};
    \node[dot] (wv2) [below right  = 0.7cm and 0.6cm of wv] {};

    \node[dot] (wv3) [below left = 0.5cm and 0.45cm of wv1] {};
    \node[dot] (wv4) [below right  = 0.5cm and 0.45cm of wv1] {};

    \node[dot] (wv5) [below left = 0.5cm and 0.45cm of wv2] {};
    \node[dot] (wv6) [below right  = 0.5cm and 0.45cm of wv2] {};

    \draw[pto] (wv) -- (wv2);

    \draw[pto] (wv1) -- (wv4);

    \end{tikzpicture}
  \end{center}

\subsubsection*{The ambient logic composition.}
We introduce the operator $+_{\aworld}$, where $\aworld \in \worlds$, that constraints further~$+$:
\begin{center}
$
\begin{aligned}
\amodel = \amodel_1 +_{\aworld} \amodel_2
      \equivdef
      &
      \amodel = \amodel_1 + \amodel_2
      \text{ and }
      \arelation_i^+(\aworld') = \arelation^+(\aworld')
      \\[-3pt]
      &\text{holds for all }
      i \in \set{1,2}
      \text{ and }
      \aworld' \in \arelation_i(\aworld).
\end{aligned}$
\end{center}
$\amodel$ is 
    a disjoint union between $\amodel_1$ and $\amodel_2$
      except that,
       as soon as $\aworld' \in \arelation_i(\aworld)$,
       the whole subtree of $\aworld'$ in $\arelation$ belongs to~$\amodel_i$,
      like the composition in ambient logic~\cite{CCG03}.
      Below, we illustrate a model decomposed with $+_{\aworld}$.

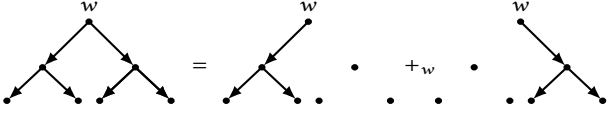
      \begin{center}
          \begin{tikzpicture}[thick, yscale=0.81, xscale=0.95, every node/.style={transform shape}]
          \node[dot] (w) [label=above:$\aworld$]{ };

          \node[dot] (w1) [below left = 0.7cm and 0.6cm of w] {};
          \node[dot] (w2) [below right  = 0.7cm and 0.6cm of w] {};

          \node[dot] (w3) [below left = 0.5cm and 0.45cm of w1] {};
          \node[dot] (w4) [below right  = 0.5cm and 0.45cm of w1] {};

          \node[dot] (w5) [below left = 0.5cm and 0.45cm of w2] {};
          \node[dot] (w6) [below right  = 0.5cm and 0.45cm of w2] {};

          \draw[pto] (w) -- (w1);
          \draw[pto] (w) -- (w2);

          \draw[pto] (w1) -- (w3);
          \draw[pto] (w1) -- (w4);

          \draw[pto] (w2) -- (w5);
          \draw[pto] (w2) -- (w6);

          \node (eq) [right = 0.625cm of w2] {$=$};

          \node[dot] (ww) [right = 3cm of w,label=above:$\aworld$] { };

          \node[dot] (ww1) [below left = 0.7cm and 0.6cm of ww] {};
          \node[dot] (ww2) [below right  = 0.7cm and 0.6cm of ww] {};

          \node[dot] (ww3) [below left = 0.5cm and 0.45cm of ww1] {};
          \node[dot] (ww4) [below right  = 0.5cm and 0.45cm of ww1] {};

          \node[dot] (ww5) [below left = 0.5cm and 0.45cm of ww2] {};
          \node[dot] (ww6) [below right  = 0.5cm and 0.45cm of ww2] {};

          \draw[pto] (ww) -- (ww1);

          \draw[pto] (ww1) -- (ww3);
          \draw[pto] (ww1) -- (ww4);

          \draw[pto] (w2) -- (w5);
          \draw[pto] (w2) -- (w6);

          \node (plw) [right = 0.525cm of ww2] {$+_{\aworld}$};

          \node[dot] (wv) [right = 2.9cm of ww,label=above:$\aworld$] {};

          \node[dot] (wv1) [below left = 0.7cm and 0.6cm of wv] {};
          \node[dot] (wv2) [below right  = 0.7cm and 0.6cm of wv] {};

          \node[dot] (wv3) [below left = 0.5cm and 0.45cm of wv1] {};
          \node[dot] (wv4) [below right  = 0.5cm and 0.45cm of wv1] {};

          \node[dot] (wv5) [below left = 0.5cm and 0.45cm of wv2] {};
          \node[dot] (wv6) [below right  = 0.5cm and 0.45cm of wv2] {};

          \draw[pto] (wv) -- (wv2);

          \draw[pto] (wv2) -- (wv5);
          \draw[pto] (wv2) -- (wv6);

          \end{tikzpicture}
        \end{center}

We say that $\amodel_1$ is a \defstyle{submodel} of $\amodel$, written $\amodel_1 \sqsubseteq \amodel$
if there is $\amodel_2$ such that $\amodel = \amodel_1 + \amodel_2$.

\subsubsection*{Modal logics on trees}
The logic \modallogicCC enriches the modal logic K (a.k.a. \ML) with a binary connective $\chopop$\,, called \defstyle{composition operator},
that admits submodel reasoning via the  operator $+_{\aworld}$.
Similarly,
\modallogicSC enriches \ML with the connective $\separate$, called \defstyle{separating conjunction} (or \defstyle{star}) that admits submodel reasoning via the operator  $+$.
Both connectives $\chopop$ and $\separate$ are understood as binary modalities.
As we show throughout the paper, \modallogicCC and \modallogicSC
are strongly related to the graded modal logic \GML~\cite{DeRijke00}. For conciseness, let us define all these logics by considering formulae that contain all of their ingredients.
These formulae are built from 
\begin{center}
$
\aformula \grammardef\
\top \mid  \avarprop  \mid
\aformula \land \aformula \mid
\lnot\aformula \mid
\Diamond \aformula  \mid
\Gdiamond{\geq k}{\aformula} \mid
\aformula \separate \aformula \mid
\aformula \chopop \aformula,
$
\end{center}
\cut{
the grammar below:
\begin{nscenter}
$
\aformula \grammardef\
\avarprop\ \mid\
\aformula \land \aformula\ \mid\
\lnot\aformula\ \mid\
\Diamond\ \aformula \ \mid \
\Gdiamond{\geq k}{\aformula} \ \mid \
\aformula \separate \aformula\ \mid\
\aformula \chopop \aformula,
$
\end{nscenter}
}
where $\avarprop \in \varprop$ and $k \in \Nat$ (encoded in binary). 
A \defstyle{pointed forest}  $\pair{\model}{\aworld}$ is
a finite forest $\amodel = \triple{\worlds}{\arelation}{\avaluation}$ together with
a world $\aworld \in \worlds$.
The satisfaction relation $\models$ is defined as follows
(standard clauses for $\land$, $\lnot$ and $\top$ are omitted):
\begin{center}
$\begin{array}{l@{\,}c@{\,}l}
    \amodel, \aworld \models \avarprop & \iff &\aworld \in \avaluation(\avarprop); \\[2pt]
\amodel, \aworld \models \Diamond \aformula & \iff & \text{there is } \aworld' \in \arelation(\aworld) \text{ s.t.}\ \amodel, \aworld'\,{\models}\, \aformula; \\[2pt]
\amodel, \aworld \models \Gdiamond{\geq k}{\aformula} & \iff
 & \card{\set{\aworld' \in \arelation(\aworld) \mid \amodel, \aworld' \models \aformula}}\,{\geq}\,k;\\[2pt]
\amodel, \aworld \models \aformula_1 \separate \aformula_2 & \iff &
\text{there are $\amodel_1$, $\amodel_2$ s.t. } \amodel = \amodel_1 + \amodel_2, \\
& & \amodel_1, \aworld \models \aformula_1 \text{ and }\amodel_2, \aworld \models \aformula_2; \\[2pt]
\amodel, \aworld \models \aformula_1 \chopop \aformula_2 & \iff &
\text{there are $\amodel_1$, $\amodel_2$ s.t. $\amodel = \amodel_1 +_{\aworld} \amodel_2$,} \\
& & \amodel_1, \aworld \models \aformula_1 \text{ and }\amodel_2, \aworld \models \aformula_2.
\end{array}
$
\end{center}
The formulae $\aformula \implies \aformulabis$, $\aformula \lor \aformulabis$ and
$\false$ are defined as usual. 
We use
the following standard abbreviations:
$\Box \aformula \egdef \neg \Diamond \neg \aformula$,
$\Gdiamond{\leq k}{\aformula}  \egdef \neg \Gdiamond{\geq k+1}{\aformula}$ and
$\Gdiamond{= k}{\aformula}  \egdef \Gdiamond{\geq k}{\aformula}  \wedge \Gdiamond{\leq k}{\aformula}$.
We write  $\fsize{\aformula}$ to denote the \defstyle{size} of $\aformula$ with a tree representation
of formulae and with a reasonably succinct encoding of atomic formulae.
Besides, we write $\md{\aformula}$ to denote  the \defstyle{modal degree} of $\aformula$ understood
as the maximal number of nested unary modalities (i.e. $\Diamond$ or $\Gdiamond{\geq k}$) in $\aformula$.
Similarly, the \defstyle{graded rank} $\gr{\aformula}$ of $\aformula$ is defined as
$\max (\set{k \mid \Gdiamond{\geq k}{\aformulabis} \in \subf{\aformula}}\,{\cup}\, \set{0})$, where $\subf{\aformula}$ is the set of all the subformulae of $\aformula$.

Given the formulae $\aformula$ and $\aformulabis$, $\aformula \equiv \aformulabis$ denotes that
$\aformula$ and $\aformulabis$ are \defstyle{logically equivalent}; i.e.,
for every pointed forest $\pair{\amodel}{\aworld}$, \   
$\amodel, \aworld \models \aformula$ iff  $\amodel, \aworld \models \aformulabis$.
For instance ($k \geq 1$ and $\avarprop \in \varprop$): 
\begin{nscenter}
\begin{tabular}{l}
    \pointlabel{1}{logical-equiv-one}\ $\Diamond \aformula \equiv \Gdiamond{\geq 1}{\aformula}$;
    \hspace{0.75cm}
    \pointlabel{2}{logical-equiv-two}\
    $
    (\Box \Box {\perp} \chopop \Box \Box {\perp})
    \not \equiv
    (\Box \Box {\perp} \separate \Box \Box {\perp});
    $
    \\[1pt]
    \pointlabel{3}{logical-equiv-three}\
    $\Gdiamond{\geq k}{\avarprop} \equiv
    \underbrace{\Diamond \avarprop \separate \cdots \separate \Diamond \avarprop}_{k \ {\rm times}};$
    \hspace{0.65cm}
    \pointlabel{4}{logical-equiv-four}\
    $\Gdiamond{\geq k}{\aformula} \equiv
    \underbrace{\Diamond \aformula \chopop \cdots \chopop \Diamond \aformula}_{k \ {\rm times}}.$
\end{tabular}
\end{nscenter}
The modal logic \ML is the logic restricted to formulae with the unique modality $\Diamond$~\cite{Blackburn&deRijke&Venema01}.
Similarly, the graded modal logic \GML is restricted to the \defstyle{graded modalities}~$\Gdiamond{\geq k}$~\cite{DeRijke00}.
We introduce the modal logics \modallogicCC and \modallogicSC, which are restricted to the suites of modalities $\pair{\Diamond}{\chopop\,}$ and $\pair{\Diamond}{\separate}$, respectively.
The two equivalences~\ref{logical-equiv-three} and~\ref{logical-equiv-four} already shed some light on \modallogicCC and \modallogicSC:
the two logics are similar when it comes to their formulae of modal degree one.

\begin{lemma}\label{lemma:SC-CC-equiv-modal-rank-1}
Let $\aformula$ be a formula in \modallogicCC with $\md{\aformula} \leq 1$. Then, $\aformula \equiv \aformula[\,\chopop \leftarrow \separate]$ where
$ \aformula[\,\chopop \leftarrow \separate]$ is the formula in \modallogicSC obtained from $\aformula$ by replacing every occurrence of $\,\chopop$ by $\separate$.
\end{lemma}
However, as shown by the non-equivalence~\ref{logical-equiv-two}, it is unclear how the two logics compare when it comes to formulae of modal degree greater than one. Indeed, since $\amodel = \amodel_1 +_{\aworld} \amodel_2$ implies $\amodel = \amodel_1 + \amodel_2$,
but not vice-versa, the  separating conjunction~$\separate$ is more permissive than
the operator~$\chopop$.
However, further connections between the two operators can be easily established.
Let us introduce the auxiliary operator $\SabDiamond$ defined as
$\SabDiamond \aformula \egdef \aformula \separate \Box \bottom$. Formally,
\begin{center}
$
\begin{array}{l@{\,}c@{\,}l}
\triple{\worlds}{\arelation}{\avaluation},\aworld \models \SabDiamond \aformula
& \iff &
\text{there is }\arelation' \subseteq \arelation\ \text{s.t.}\ \arelation'(\aworld) = \arelation(\aworld)\\
&& \text{and } \triple{\worlds}{\arelation'}{\avaluation},\aworld \models \aformula.
\end{array}
$
\end{center}
Similar operators are studied
in~\cite{AFH15,AucherBG18,Bozzelli&vanDitmarsch&Pinchinat15}.
We show that $\SabDiamond$ and $\chopop$ are
sufficient to capture $\separate$ (essential property for~\Cref{section-expressivity-SC}).

\begin{lemma}\label{lemma:SC-CC-tautology}
Let $\aformula,\aformulabis \in\GML$. We have
 $\aformula \separate \aformulabis \equiv \SabDiamond ( \aformula \chopop \aformulabis)$.
\end{lemma}

Unlike $\chopop$\,, when $\separate$ splits a finite forest~$\amodel$ into~$\amodel_1$ and $\amodel_2$,
it may disconnect in both submodels worlds that are otherwise reachable, from the current world, in~$\amodel$.
Applying $\SabDiamond$ before~$\chopop$ allows us to imitate this behaviour.
Indeed, even though~$\chopop$ preserves reachability in either~$\amodel_1$ or $\amodel_2$,
$\SabDiamond$ deletes part of~$\amodel$, making some world inaccessible.
This way of expressing the
separating conjunction allows us to reuse some methods developed
for $\modallogicCC$ in order to study $\modallogicSC$.

\subsubsection*{The logic \msokt.}
Both \modallogicCC and \modallogicSC can be seen as fragments of the logic~\msokt,
which in turn is known to be a fragment of monadic second-order logic on trees~\cite{Bednarczyk&Demri19}.
The logic~\msokt extends \ML with second-order quantification and is interpreted on finite trees.
Its formulae are defined according to the following grammar:
\begin{center}
$
\aformula \grammardef\
\avarprop \ \mid \
\Diamond \aformula \ \mid \
\aformula \wedge \aformula \ \mid \
\lnot \aformula \ \mid \
\exists \avarprop\, \aformula .
$
\end{center}
Given $\amodel = \triple{\worlds}{\arelation}{\avaluation}$ and $\aworld \in \worlds$,
the satisfaction relation $\models$ of \ML is extended as follows:
\begin{center}
  $\amodel,\aworld \models \exists \avarprop\, \aformula$ iff
  $\exists \worlds' \subseteq \worlds$ s.t.\ $\triple{\worlds}{\arelation}{\avaluation[\avarprop \gets \worlds']},\aworld \models \aformula$.
\end{center}
One can show logspace reductions from \modallogicCC and \modallogicSC to
\msokt, by simply reinterpreting the operators $\separate$ and $\chopop$ as restrictive forms of second-order quantification,
and by
relativising~$\Diamond$ to appropriate propositional symbols in order
to capture the notion of submodel (details are omitted).

\subsubsection*{Satisfiability problem.}
The \defstyle{satisfiability problem} for a logic $\alogic$, written \satproblem{$\alogic$},
takes as input a formula $\aformula$ in $\alogic$ and checks whether
there is a pointed forest~$\pair{\model}{\aworld}$ such that
$\amodel, \aworld \models
\aformula$.

Note that any $\alogic$ among \ML, \GML, \modallogicCC or \modallogicSC
has the tree model property, i.e. any satisfiable formula is also satisfied in some tree structure.
The problems \satproblem{\ML} and \satproblem{\GML} are known to be
\pspace-complete, see e.g.~\cite{Ladner77,Blackburn&deRijke&Venema01,Tobies00,Schroder&Pattinson06},
and therefore \satproblem{\modallogicCC} and \satproblem{\modallogicSC} are \pspace-hard.
As an upper bound,
by Rabin's theorem~\cite{Rabin69}, the satisfiability problem for~\msokt is decidable in \tower,
which transfers directly to \satproblem{\modallogicCC} and \satproblem{\modallogicSC}.
\subsubsection*{Expressive power.}
Given two logics $\alogic_1$ and $\alogic_2$, we say that
$\alogic_2$ is \defstyle{at least as expressive as} $\alogic_1$ (written $\alogic_1 \preceq \alogic_2$) whenever
for every formula $\aformula$ of $\alogic_1$, there is a formula $\aformulabis$ of $\alogic_2$
such that $\aformula \equiv \aformulabis$.
$\alogic_1 \approx \alogic_2$ denotes that
$\alogic_1$ and $\alogic_2$ are \defstyle{equally expressive}, i.e.
$\alogic_1 \preceq \alogic_2$ and $\alogic_2 \preceq \alogic_1$.
Lastly, $\alogic_1 \prec\alogic_2$ denotes that $\alogic_2$ is \defstyle{strictly more expressive} than $\alogic_1$, i.e. $\alogic_1 \preceq \alogic_2$ and $\alogic_1 \not \approx \alogic_2$.
The equivalence~\ref{logical-equiv-one} recalls us that \ML $\prec$ \GML~\cite{DeRijke00}.
From the equivalence~\ref{logical-equiv-four}, we get \GML $\preceq$ \modallogicCC.
%



%
\section{\modallogicCC: Expressiveness and Complexity}
\label{section-CC}
In this section,  we study the expressive power  of \modallogicCC and the complexity
of \satproblem{\modallogicCC}. We show constructively that $\modallogicCC\preceq\GML$,
hence proving
$\modallogicCC \approx \GML$.
Next, we show that $\satproblem{\modallogicCC}$ is \aexppol-complete. The upper bound is achieved
by proving an exponential-size model property. The lower bound is by reduction from
the satisfiability problem for propositional team logic~\cite[Thm. 4.9]{HannulaKVV18}.

\subsection{\modallogicCC is not more expressive than \GML}
\label{section-cc-less-gml}

Establishing \modallogicCC $\preceq$ \GML amounts to show
that
given $\aformula_1$, $\aformula_2$  in \GML, one can construct $\aformulabis$ in \GML such that
$\aformula_1 \chopop \aformula_2 \equiv \aformulabis$.
For instance, a simple case analysis yields the equivalence
$
(\avarprop \vee \Gdiamond{\geq 3} \avarpropter) \chopop \
(\avarpropbis \vee \Gdiamond{\leq 5} \avarpropbis)
\equiv
(\avarprop \vee \Gdiamond{\geq 3} \avarpropter)$.
With this property, the general algorithm consists in iteratively replacing innermost subformulae
of the form $\aformula_1 \chopop \aformula_2$ by a counterpart in \GML, allowing us to eliminate all the occurrences of $\chopop$
and obtain an equivalent formula in \GML.
The base case involves subformulae $\aformula_1$ and $\aformula_2$
in \ML (a fragment of \GML).



Let us provide a few definitions.
Let $\aformula$ be a formula in \GML.
We write $\maxpc{\aformula}$ to denote the set of subformulae
$\aformulabis$ of $\aformula$ that are maximal and modality-free, i.e.
\begin{enumerate}
  \item $\aformulabis$ is modality-free: it does not contain modalities~$\Gdiamond{\geq k}$
  and one of its occurrences is not in the scope of~$\Gdiamond{\geq k}$\!;
  \item $\aformulabis$ is maximal: one of its occurrences does not belong to a larger modality-free subformula of $\aformula$.
\end{enumerate}
For instance, $\maxpc{(\avarprop \vee \Gdiamond{\geq 3} \avarpropter) \wedge (\avarpropbis \vee \avarprop)}
= \set{\avarprop, \avarpropbis \vee \avarprop}$.
Similarly, 
$\maxgmod{\aformula}$ denotes the set of subformulae
$\aformulabis$ of $\aformula$ such that $\aformulabis$ is of the form $\Gdiamond{\geq k} \aformulabis'$
and one of its occurrences in $\aformula$ is not in the scope of graded modalities~$\Gdiamond{\geq k}$.
For instance,
\begin{nscenter}
$\maxgmod{(\avarprop \vee \Gdiamond{\geq 3} \avarpropter) \wedge (\avarpropbis \vee
\Gdiamond{\geq 5}{\Gdiamond{\geq 2}{\avarpropbis}})}
= \set{\Gdiamond{\geq 3}{\avarpropter}, \Gdiamond{\geq 5}{\Gdiamond{\geq 2}{\avarpropbis}}}$.
\end{nscenter}
Every formula $\aformula$ in \GML is a Boolean combination of formulae from
 $\maxpc{\aformula} \cup \maxgmod{\aformula}$.
Lastly,
$\aformula$
is in \defstyle{good shape} if the properties~\ref{good-shape-1} and~\ref{good-shape-2} below hold:
\begin{enumerate}
\item[\pointlabel{1}{good-shape-1}] $\maxpc{\aformula} \subseteq \set{\perp,\top}$.
      Consequently, every propositional variable in $\aformula$ occurs
      in the scope of a graded modality;
\item[\pointlabel{2}{good-shape-2}] For all  $\Gdiamond{\geq k} \aformulabis,\Gdiamond{\geq k'} \aformulabis'$ in $\maxgmod{\aformula}$ with $\aformulabis \neq \aformulabis'$,
the conjunction  $\aformulabis \wedge \aformulabis'$ is unsatisfiable.
\end{enumerate}

Let $\aformula_1$ and $\aformula_2$ be \GML formulae.
First, we show that when $\aformula_1 \wedge \aformula_2$ is in good shape, there is a \GML formula $\aformulabis$
such that $\aformula_1 \chopop \aformula_2 \equiv \aformulabis$.
To do so, we take a slight
detour through Presburger arithmetic (\PA), see e.g.~\cite{Presburger29,Haase18}. Given two formulae $\aformula_1,\aformula_2$ in \GML,
we will characterise the formula $\aformula_1 \chopop \aformula_2$ by
using arithmetical constraints for the number of successors. Then, we will take
advantage of basic properties of \PA in order to eliminate quantifiers, and
obtain a \GML formula.
Below, the variables $\avariable, \avariablebis, \avariableter, \ldots$, possibly decorated and occurring in formulae, are from \PA and therefore
they are
interpreted by natural
numbers.

Let $\aformula$ be in \GML s.t.\
$\maxpc{\aformula} \subseteq \set{\true,\false}$ and
$\set{\aformulabis_1, \ldots, \aformulabis_n}$  contains
the set $\set{\aformulabis \mid \Gdiamond{\geq k} \aformulabis \in \maxgmod{\aformula}}$.
We define formulae in \PA that state constraints about the number of children satisfying
a formula $\aformulabis_j$. The variable $\avariable_j$ is intended to be interpreted as the number of children satisfying
$\aformulabis_j$.
We write $\aformula^{\PA}(\avariable_1, \ldots, \avariable_n)$ to denote the arithmetical formula obtained
from $\aformula$ by replacing with $\avariable_j \geq k$ every occurrence of $\Gdiamond{\geq k}
\aformulabis_j$ that it is not in the scope of a graded modality.
For instance, assuming that $\aformula = \Gdiamond{\geq 5}{(\avarprop \wedge \avarpropbis)}
\vee \neg \Gdiamond{\geq 4}{\neg \avarprop}$, the expression
$\aformula^{\PA}(\avariable_1,\avariable_2)$ denotes
the formula~$\avariable_1 \geq 5 \vee \neg (\avariable_2 \geq 4)$.

Let $\aformula_1, \aformula_2$ be \GML formulae such that
$\aformula_1 \land \aformula_2$ is in good shape and
$\set{\aformulabis_1, \ldots, \aformulabis_n} = {\set{\aformulabis \mid \Gdiamond{\geq k} \aformulabis \in \maxgmod{\aformula_1 \wedge \aformula_2}}}$.
We consider the formula~$[\aformula_1,\aformula_2]^{\PA}$ in~\PA defined below:
\begin{center}
$
\begin{aligned}
[\aformula_1,\aformula_2]^{\PA} \, \egdef \, & \exists \ \avariablebis_1^1, \avariablebis_1^2, \ldots, \avariablebis_n^1, \avariablebis_n^2
 \ (\textstyle\bigwedge_{j=1}^{n} \avariable_j = \avariablebis_j^1 + \avariablebis_j^2) \wedge \\[-1pt]
 & \ \ \ \ \aformula^{\PA}_1(\avariablebis_1^1, \ldots,\avariablebis_n^1) \wedge
\aformula^{\PA}_2(\avariablebis_1^2, \ldots,\avariablebis_n^2).
\end{aligned}
$
\end{center}
The formula $[\aformula_1,\aformula_2]^{\PA}$  states that there is a way to divide the children in two distinct sets
and each set allows to satisfy $\aformula^{\PA}_1$ or $\aformula^{\PA}_2$, respectively.
As \PA admits quantifier elimination~\cite{Presburger29,Cooper72,Reddy&Loveland78},
there is a quantifier-free formula $\aformulater$
equivalent to $[\aformula_1,\aformula_2]^{\PA}$ and its free variables
are among $\avariable_1, \ldots, \avariable_n$.
A priori, the atomic formulae of $\aformulater$ may not be of the simple form
$\avariable_j \geq k$ (e.g. `modulo constraints' or constraints of the form $\sum a_i \avariable_j \geq k$
may be involved).
However, if the atomic formulae of $\aformulater$ are restricted to expressions of the form
$\avariable_j \geq k$, then we write $\aformulater^{\GML}$ to denote the \GML formula
obtained from $\aformulater$ by replacing every occurrence of $\avariable_j \geq k$
by $\Gdiamond{\geq k}{\aformulabis_j}$.

\begin{lemma}
\label{lemma:elimination-three-two}
Let  $\aformula_1$, $\aformula_2$ be  in \GML such that
$\aformula_1 \wedge \aformula_2$ is in good shape.
$[\aformula_1,\aformula_2]^{\PA}$ is equivalent to a quantifier-free \PA formula $\aformulater$ whose
atomic formulae are only of the form $\avariable_j \geq k$.
Moreover, $\aformula_1 \chopop \aformula_2 \equiv
\aformulater^{\GML}$ and $\gr{\aformulater^{\GML}} \leq \gr{\aformula_1} + \gr{\aformula_2}$.
\end{lemma}
The bound on $\gr{\aformulater^{\GML}}$ stated in this key lemma is essential to obtain an exponential bound on the smallest model satisfying a formula in
\modallogicCC (see~\Cref{section-aexppol}).
Thanks to Lemma~\ref{lemma:elimination-three-two}, we can show that \GML is closed under the operator $\chopop$ by reducing the occurrences of this operator to formulae in good shape. In particular, we show that given two arbitrary formulae $\aformula_1$ and $\aformula_2$ in \GML,
$\aformula_1 \chopop \aformula_2$ is equivalent to a disjunction of formulae of the form $(\aformulabis_1 \chopop \aformulabis_2) \land \aformulater$,
where $\aformulater$ is a Boolean combination of atomic propositions and $\aformulabis_1 \land \aformulabis_2$ is in good shape (hence $\aformulabis_1 \chopop \aformulabis_2$ is equivalent to a formula in \GML by \Cref{lemma:elimination-three-two}).
This is shown syntactically: atomic propositions are dealt with
by propositional reasoning,
whereas to produce $\aformulabis_1$ and $\aformulabis_2$ we use axioms from \GML~\cite{FattorosiBarnaba&DeCaro85} and rely on the following equivalences:
\begin{description}[align=right,leftmargin=*,labelindent=1.3cm]
  \item[\desclabel{(guess)}{axiom:guess}]
    \scalebox{0.85}{$\Gdiamond{\geq k}{\aformula}$ $\equiv$ $\Gdiamond{\geq k}{\big((\aformula \land \aformulabis) \lor (\aformula \land \lnot \aformulabis)\big)}$};
  \item[\desclabel{($\scriptstyle\Gdiamond{\!\geq k}{}$\!dist)}{axiom:gdistr}]
    \scalebox{0.85}{if  $\aformula\,{\land}\,\aformulabis$ unsat.,
   $\Gdiamond{\geq k}{\!(\aformula {\lor} \aformulabis)}$ $\equiv$ $ \bigvee_{\!k {=} k_1 {+} k_2}(\Gdiamond{\geq k_1}{\!\aformula} \,{\land}\, \Gdiamond{\geq k_2}{\!\aformulabis})$};
  \item[\desclabel{(\,$\chopop$\,dist)}{axiom:chopdist}] \scalebox{0.85}{$(\aformula \lor \aformulabis) \chopop \aformulater$
  $\equiv$ $(\aformula \chopop \aformulater) \lor (\aformulabis \chopop \aformulater)$}.
\end{description}
Notice that the conjunction of $\aformula \land \aformulabis$ and $\aformula \land \lnot \aformulabis$ from~\ref{axiom:guess} is trivially unsatisfiable,
allowing us to use \ref{axiom:gdistr}.
As \GML is shown to be closed under the operator $\chopop$\,, we conclude.
\begin{theorem}
\label{theorem:clean-cut}
\modallogicCC $\preceq$ \GML. Therefore, $\modallogicCC \approx \GML$.
\end{theorem}
%
%
To prove \modallogicCC $\preceq$ \GML, we iteratively put subformulae in good shape and apply \Cref{lemma:elimination-three-two}. This is done several times, potentially causing an exponential blow-up each time a formula is transformed.
To provide an optimal complexity upper bound, we need to tame this combinatorial explosion.

\subsection{\aexppol-completeness}
\label{section-aexppol}

In order to show that \satproblem{\modallogicCC} is in \aexppol,
the main ingredient is to show that given
$\aformula$ in \modallogicCC, we build
$\aformula'$ in \GML
 such that $\aformula' \equiv \aformula$ and  the models for $\aformula'$ (if any) do not require a number of children per node
more than exponential in  $\fsize{\aformula}$.
The proof
of Theorem~\ref{theorem:clean-cut} needs to be
refined to improve the way $\aformula'$ is computed.
In particular, this requires a strategy for the application of the equivalences used to put a formula in good shape.
We need to introduce a few more simple notions.
Let $\aformula$ be a \GML formula with
$\maxgmod{\aformula} = \set{\Gdiamond{\geq k_1} \aformulabis_1, \ldots,\Gdiamond{\geq k_n} \aformulabis_n}$.
We define $\newbd{0}{\aformula} \egdef k_1 + \cdots + k_n$.
For all
$m \geq 0$, we define $\newbd{m+1}{\aformula} \egdef \max \set{\newbd{m}{\aformulabis} \mid  \Gdiamond{\geq k}{\aformulabis}
\in \maxgmod{\aformula}}$. Hence, $\newbd{m}{\aformula}$ can be understood as the maximal
$\newbd{0}{\aformulabis}$ for some subformula $\aformulabis$ occurring at the modal depth $m$ within $\aformula$.
We write $\maxbd{\aformula}$ for the value $\max \set{\newbd{m}{\aformula}
\mid  m \in \interval{0}{\md{\aformula}}}$.
If~$\aformula$ is satisfiable,
we can use~$\maxbd{\aformula}$ to obtain a bound on the smallest model satisfying it,
as stated in~\Cref{lemma:small-model-newbd} below.

\begin{lemma}
\label{lemma:small-model-newbd}
Every satisfiable $\aformula$ in \GML is satisfied by a pointed forest
with at most $\maxbd{\aformula}^{\md{\aformula}+1}$ worlds.
\end{lemma}

To show that \modallogicCC has the exponential-size model property, we
 establish that given $\aformula$ in \modallogicCC, there is
$\aformula'$ in \GML such that $\aformula' \equiv \aformula$,  $\md{\aformula'} \leq \md{\aformula}$ and
$\maxbd{\aformula'}$ is exponential in $\fsize{\aformula}$.
%
%
%
First, we consider the fragment $\sf{F}$ of \modallogicCC:
$\aformula::= \Gdiamond{\geq k} \aformulabis \ \mid \ \avarprop \ \mid \
            \aformula \chopop \aformula  \ \mid \
            \aformula \wedge \aformula  \ \mid \
            \neg \aformula,\
$
where $\avarprop \in \varprop$ and $\Gdiamond{\geq k} \aformulabis$ is a formula in \GML (abusively
assumed in \modallogicCC but we
know $\GML \preceq \modallogicCC$).
Given
$\aformula$ in \modallogicCC or in $\sf{F}$, we write $\cd{\aformula}$ to denote
its \defstyle{composition degree}, i.e. the maximal number of imbrications of $\chopop$ in $\aformula$.
We extend the notion of ${\tt bd}$ to formulae in $\sf{F}$, so that
$\newbd{m}{\aformula} = \newbd{m}{\aformula[\,\chopop \gets\, \land]}$,
where $\aformula[\,\chopop \gets \land]$ is the formula obtained from $\aformula$ by replacing every occurrence of $\chopop$ by $\land$.
Similarly, $\maxgmod{\aformula} \egdef \maxgmod{\aformula[\,\chopop \gets\, \land]}$.

Let $\aformula$ be in $\sf{F}$ such that $\maxgmod{\aformula} = \set{\Gdiamond{\geq k_1}\! \aformulater_1, \ldots,
   \Gdiamond{\geq k_n}\!\aformulater_n}$.
The key step to show the exponential-size model property essentially
manipulates the formulae in $\maxgmod{\aformula}$ in order to produce equivalent formulae $\aformulabis_1,\dots,\aformulabis_n$, so that for all distinct $i$ and $j$, $\aformulabis_i \land \aformulabis_j$ is in good shape.
Moreover, by replacing in $\aformula$ every $\Gdiamond{\geq k_i}\! \aformulater_i$ with the equivalent formula $\aformulabis_i$, we only witness
an exponential blow-up on $\newbd{0}{\aformula}$, whereas for every $m > 1$, $\newbd{m}{\aformula}$ remains polynomially bounded by the $\mathtt{bd}$ of the original formula. With the bound on the graded rank
found in \Cref{lemma:elimination-three-two}, we  derive \Cref{lemma:f-gml}.
\begin{lemma}
\label{lemma:f-gml}
Let $\aformula$ be a formula
of the fragment $\sf{F}$ such that $\maxgmod{\aformula} = \set{\Gdiamond{\geq k_1}\! \aformulater_1, \ldots,
     \!\Gdiamond{\geq k_n}\!\aformulater_n}$ and $\widehat{k} = \max \set{k_1, \ldots,\!k_n}$.
There is a \GML
formula $\aformulabis$ such that $\aformula \equiv \aformulabis$ and,
\begin{nscenter}
\begin{tabular}{ll}
1. $\md{\aformulabis} \leq \md{\aformula}$; &
2. $\newbd{0}{\aformulabis} \leq \widehat{k} \times \, 2^{n+\cd{\aformula}}$;\\
3. $\newbd{1}{\aformulabis} \leq n \times \newbd{1}{\aformula}$; &
4. $\forall m\,{\geq}\,2$, $\newbd{m}{\aformulabis} = \newbd{m}{\aformula}$.
\end{tabular}
\end{nscenter}
\end{lemma}

In the proof of \Cref{lemma:f-gml}, a first step  essentially consists in applying multiple times \ref{axiom:guess} in order to derive, for every $i \in \interval{1}{n}$, an equivalence $\Gdiamond{\geq k_i}\! \aformulater_i$ $\equiv$ $\aformulabis_i'$ where
\begin{center}
  $\aformulabis_i' \egdef \Gdiamond{\geq k_i}\!\!\!\bigvee_{\amap : \interval{1}{n} \to \{\true,\false\}}\big(\aformulater_i \land [\aformulater_1]^{\amap(1)} \land \dots \land [\aformulater_n]^{\amap(n)}\big)$.
\end{center}
Here, $[\aformulater_j]^{\true} \egdef \aformulater_j$ and $[\aformulater_j]^{\false} \egdef \lnot \aformulater_j$.
Roughly speaking, in this step, we expand $\aformulater_i$ by considering all the possible truth values for the formulae  $\aformulater_1,\dots,\aformulater_n$ (the disjuncts where $\aformulater_i$ is negated can be simply discharged from the disjunction, as they are unsatisfiable). Substituting every $\Gdiamond{\geq k_i}\! \aformulater_i$ by $\aformulabis_i'$ in $\aformula$ leads to a formula $\aformula'$ such that $\newbd{1}{\aformula'} \leq n \times \newbd{1}{\aformula}$ (as in \Cref{lemma:f-gml}) and
for every $m \neq 1$, $\newbd{m}{\aformula'} = \newbd{m}{\aformula}$.
Afterwards, we repeatedly apply \ref{axiom:gdistr} to $\aformulabis_i'$ and obtain the formula $\aformulabis_i$ satisfying the aforementioned property, i.e. for
all distinct $i$ and $j$, $\aformulabis_i \land \aformulabis_j$ is in good shape.
With \ref{axiom:chopdist}, this allows us to apply \Cref{lemma:elimination-three-two} until all the operators $\chopop$ are removed.
Besides, replacing every $\aformulabis_i'$ by $\aformulabis_i$ in $\aformula'$ leads to a formula having the same $\mathtt{bd}$ as the formula $\aformulabis$ in \Cref{lemma:f-gml}.

Applying adequately the transformation from Lemma~\ref{lemma:f-gml} to a formula
in \modallogicCC, i.e. by considering maximal subformulae of the fragment $\sf{F}$, allows us to get a logically equivalent \GML formula having small models.
\begin{lemma}
\label{lemma:small-gml-formula}
Every satisfiable $\aformula$ in \modallogicCC is satisfied by a pointed forest of size at most exponential in $\size(\aformula)$.
\end{lemma}
The proof of~\Cref{lemma:small-gml-formula}
(relying on
\Cref{lemma:f-gml})
consists in showing that
for all
$\aformula$ in  \modallogicCC, there is
$\aformula'$ in \GML such that $\aformula' \equiv \aformula$ and  $\maxbd{\aformula'}$ is exponential in $\fsize{\aformula}$,
which is sufficient by Lemma~\ref{lemma:small-model-newbd} to get the exponential-size model property,
whence the upper bound \aexppol.

\begin{theorem}
\label{theorem:complexity-clean-cut}
\satproblem{\modallogicCC} is  in \aexppol.
\end{theorem}

The (standard) proof consists in observing that to check the satisfiability
status of $\aformula$ in \modallogicCC, first guess a pointed forest of exponential-size
(thanks to Lemma~\ref{lemma:small-gml-formula}) and check whether it satisfies $\aformula$.
This can be done in exponential-time using an
alternating Turing machine
with a linear amount
of alternations
(between universal states and existential states)
by viewing \modallogicCC as a fragment of \MSO.

It remains to establish \aexppol-hardness. We provide a logspace reduction from
the satisfiability problem for the team logic PL[\plcnot] shown \aexppol-complete in~\cite[Thm. 4.9]{HannulaKVV18}.

PL[\plcnot] formulae are defined by the following grammar:
\begin{center}
$
\aformula \grammardef\
\avarprop \ \mid \
\plneg \avarprop \ \mid \
\aformula \wedge \aformula \ \mid \
\plcnot \aformula \ \mid \
\aformula \plvee \aformula\,
$,
\end{center}
where $\avarprop \in \varprop$
and the connectives $\plneg$ and $\plvee$ are
dotted to avoid confusion with those of \modallogicCC.
\cut{
Whereas $\wedge$ and \plcnot \ behave respectively as classical conjunction and classical negation,
the connective $\plneg$ has a negation flavour without being equivalent to
\plcnot and $\plvee$ is actually
a separating connective whose semantics follows. The small dot at the top of $\plneg$ or
$\plvee$ is not originally present in~\cite{HannulaKVV18}: this
is added herein to avoid any confusion with the connectives in \modallogicCC.
}
PL[\plcnot] is interpreted on sets of (Boolean) propositional valuations over a finite subset of~$\varprop$.
They are called \defstyle{teams} and are denoted by $\ateam,\ateam_1,\dots$.
A model for
$\aformula$ is a team $\ateam$ over a set of propositional variables including
those occurring in $\aformula$ and such that $\ateam \models \aformula$ with:
\begin{center}
\scalebox{0.99}{$
\begin{array}{l@{\,}c@{\,}l}
\ateam \models \avarprop & \iff & \mbox{for all $\aplvaluation \in \ateam$, we have } \aplvaluation(\avarprop) = \true; \\[1pt]
\ateam \models \plneg \avarprop & \iff & \mbox{for all $\aplvaluation \in \ateam$, we have }\aplvaluation(\avarprop) = \false; \\[1pt]
\ateam \models \aformula_1 \plvee \aformula_2 & \iff & \exists\, \ateam_1,\ateam_2 \mbox{ s.t.}\ \ateam\,{=}\,\ateam_1\,{\cup}\,\ateam_2,\ \ateam_1\!\models\!\aformula_1,\ \ateam_2\!\models\!\aformula_2.
\end{array}
$}
\end{center}
The connectives $\plcnot$ and $\land$ are interpreted as the classical negation and conjunction, respectively.
Notice that, in the clause for
$\plvee$, the teams $\ateam_1$ and $\ateam_2$ are not necessarily disjoint.

Let us discuss the reduction from \satproblem{PL[\plcnot]} to \satproblem{\modallogicCC}.
A direct encoding of a team $\ateam$ into a pointed forest $\pair{\amodel}{\aworld}$ consists in having a correspondence
between the propositional valuations in $\ateam$ and the propositional valuations of the children of
$\aworld$.
This would work fine if there were no mismatch between the semantics
for $\chopop$ (disjointness of the children)
and the one
for~$\plvee$ (disjointness not required).
To handle this, when checking the satisfaction of $\aformula$ in PL[\plcnot] with~$n$ occurrences
of $\plvee$, we impose that if a propositional valuation occurs among
the children of $\aworld$, then it occurs in least $n+1$ children.
This property must  be maintained after applying $\plvee$ several times, always with respect to the number of occurrences of~$\plvee$ in the subformula of~$\aformula$ that is evaluated.
Non-disjointness of the teams
is encoded by carefully separating the children of $\aworld$ having identical valuations.

We now formalise the reduction.
Assume that we wish to translate $\aformula$ from PL[\plcnot],
written with atomic propositions in
$\apropset = \set{\avarprop_1, \dots, \avarprop_m}$
and containing at most $n$ occurrences of the  operator $\plvee$.
We introduce a set
$\apropsetbis = \set{ \avarpropbis_1, \dots, \avarpropbis_{n+1}}$ of
auxiliary propositions disjoint from $\apropset$.
The elements of~$\apropsetbis$ are used to distinguish
different copies of the same propositional valuation of a team.
Thus, with respect to a pointed forest~$\pair{\amodel}{\aworld}$,
we require each child of $\aworld$ to satisfy exactly one
element of $\apropsetbis$.
This can be done  with the formula
\begin{center}
$\tluniq{\apropsetbis} \egdef
  \Box (\bigwedge_{i \neq i' \in \interval{1}{n+1}} \neg (\avarpropbis_i \wedge \avarpropbis_{i'}) \wedge \bigvee_{i \in \interval{1}{n+1}} \avarpropbis_i)$.
\end{center}
\cut{
This can be done easily, with the following formula:
\begin{nscenter}
  $\tluniq{\apropsetbis} \egdef
  \Box (\bigwedge_{i \neq i' \in \interval{1}{n+1}} \neg (\avarpropbis_i \wedge \avarpropbis_{i'}) \wedge \bigvee_{i \in \interval{1}{n+1}} \avarpropbis_i).$
\end{nscenter}
}
We require that if a child of $\aworld$ satisfies a propositional valuation over (elements in)~$\apropset$,
then
there are $n+1$ children satisfying that valuation over $\apropset$, each of them satisfying
a distinct symbol in $\apropsetbis$.
So, every valuation over $\apropset$ occurring in some child of $\aworld$,
occurs at least in $n+1$ children of $\aworld$.
However, as the translation of the operator  $\plvee$ modifies the set of copies of a propositional valuation, this property must be extended to arbitrary subsets of $\apropsetbis$. Given $\emptyset \neq \aset \subseteq  \interval{1}{n+1}$,
we require that for all $k \neq k' \in \aset$,
if a children of $\aworld$ satisfies $\avarpropbis_{k}$, then there is a child satisfying $\avarpropbis_{k'}$ with the same valuation over $\apropset$.
The formula $\tlcopies{\aset}$ below does the job:
\begin{center}
      \scalebox{0.9}{$
      \displaystyle\bigwedge_{\mathclap{\ \ k \neq k' \in \aset}}\!\!\neg \big(
      \Box \avarpropbis_{k} \chopop
      (
      \Gdiamond{= 1} \avarpropbis_{k} \wedge
      \neg
     (\true \chopop \Gdiamond{= 1} \avarpropbis_{k} \wedge \Gdiamond{= 1} \avarpropbis_{k'}  \wedge
      \!\bigwedge_{\mathclap{j \in \interval{1}{m}}}\!\!\Diamond \avarprop_j \Rightarrow \Box \avarprop_j
     )
      )
       \big).
       $}
\end{center}
Lastly,
before defining the translation map $\atranslation$, we
describe how different copies of the same propositional valuation are split.
We introduce two auxiliary choice functions $\mathfrak{c}_1$ and $\mathfrak{c}_2$
that take as arguments $\aset \subseteq \interval{1}{n+1}$, and $n_1, n_2 \in \Nat$ with $\card{\aset} \geq n_1 + n_2$ such that
for each $i \in \set{1,2}$, we have $\mathfrak{c}_i(\aset,n_1,n_2) \subseteq \aset$,
$\card{\mathfrak{c}_i(\aset,n_1,n_2)} \geq n_i$. Moreover $\mathfrak{c}_1(\aset,n_1,n_2) \uplus \mathfrak{c}_2(\aset,n_1,n_2) = \aset$.
The maps $\mathfrak{c}_1$ and $\mathfrak{c}_2$ are instrumental to decide how to split
$\aset$ into two disjoint subsets respecting basic cardinality constraints.
The translation map $\atranslation$ is  designed as follows ($\emptyset \neq \aset \subseteq \interval{1}{n+1}$):
\begin{center}
$
\begin{aligned}
\atranslation(\avarprop,\aset) & \egdef \Box ((\textstyle\bigvee_{\!j \in \aset} \avarpropbis_j) \Rightarrow \avarprop);\\
\atranslation(\plneg \avarprop,\aset) & \egdef \Box ((\textstyle\bigvee_{\!j \in \aset} \avarpropbis_j) \Rightarrow \neg \avarprop);\\
\atranslation(\aformula_1 \wedge \aformula_2,\aset) & \egdef  \atranslation(\aformula_1,\aset) \wedge \atranslation(\aformula_2,\aset);
\quad
\atranslation(\plcnot \aformula,\aset) \egdef \neg  \atranslation(\aformula,\aset);\\
\atranslation(\aformula_1 \plvee \aformula_2,\aset) & \egdef
\displaystyle
(\atranslation(\aformula_1,\aset_1) \land
\tlcopies{\aset_1}) \chopop\, (\atranslation(\aformula_2,\aset_2)
\land
 \tlcopies{\aset_2}),
\end{aligned}
 $
\end{center}
where
{\bfseries \itshape{(i)}}
$\card{\aset}$ is greater or equal to the number of occurrences of $\plvee$ in $\aformula_1 \plvee \aformula_2$ plus one;
{\bfseries \itshape{(ii)}}
given $n_1, n_2$ such that
$n_1$ (resp. $n_2$) is the number of occurrences of $\plvee$ in $\aformula_1$ (resp. $\aformula_2$) plus one,
for each $i \in \set{1,2}$ we have $\mathfrak{c}_i(\aset,n_1,n_2) = \aset_i$.

\Cref{lemma:correctness-PL} below
guarantees that starting with a linear number of children with the same propositional valuation
is sufficient to encode $\plvee$ within \modallogicCC.

\begin{lemma}
\label{lemma:correctness-PL}
Let $\aformula$ be in  PL[\plcnot] with $n$ occurrences of  $\plvee$
and built upon $\avarprop_1$, \ldots, $\avarprop_m$. Then, $\aformula$ is satisfiable iff so is
\begin{nscenter}
$
\tluniq{\avarpropbis_1, \ldots,  \avarpropbis_{n+1}} \wedge \tlcopies{\interval{1}{n+1}}
\wedge \atranslation(\aformula, \interval{1}{n+1}).
$
\end{nscenter}
\end{lemma}

The $\modallogicCC$ formula involved in \Cref{lemma:correctness-PL} has modal depth one. By~\Cref{theorem:complexity-clean-cut},
$\satproblem{\modallogicCC}$ is \aexppol-complete even restricted to formulae of modal depth at most one.

\begin{corollary}
      \label{cor:aexpol}
\satproblem{\modallogicCC} is \aexppol-complete.
\end{corollary}

As we show in the next section,
the complexity of~$\modallogicSC$ does not collapse to modal depth one:
\satproblem{\modallogicSC} restricted to formulae of modal depth~$k$
is exponentially easier than \satproblem{\modallogicSC} restricted to formulae of modal depth~$k+1$.


%
\section{\modallogicSC is \tower-complete}
\label{section-tower-SC}
We show that \satproblem{\modallogicSC} is \tower-complete, i.e. complete for the class of all
problems of time complexity bounded by a tower of exponentials whose height is an elementary function~\cite{Schmitz16}.
Given
$k,n \geq 0$, we inductively define the tetration function $\amapter$ as $\amapter(0,n) \egdef n$ and
$\amapter(k+1,n) = 2^{\amapter(k,n)}$. Intuitively, $\amapter(k,n)$ defines a tower of exponentials of height $k$.
By $k$-\nexptime, we denote the class of all problems decidable with a nondeterministic Turing machine (NTM) of working time
$\bigO(\amapter(k,p(n)))$ for some polynomial $p(.)$, on each input of length $n$.
To show \tower-hardness, we design  a uniform
elementary reduction  allowing us to get
$k$-\nexptime-hardness for all $k$ greater than a certain (fixed) integer.
In our case, we achieve an exponential-space reduction from the $k$-\nexptime variant of the tiling problem, for all $k \geq 2$.

The tiling  problem $\tiling_k$ takes as input a triple $\cTT = \triple{\cT}{\cH}{\cV}$ where
$\cT$ is a finite set of tile types,
$\cH \subseteq \cT \times \cT$ (resp. $\cV \subseteq \cT \times \cT$) represents the horizontal
(resp. vertical) matching relation,
and an initial tile type $\atile \in \cT$.
A solution for the instance $\pair{\cTT}{c}$
is a mapping $\tau: \interval{0}{\amapter(k,n)-1} \times \interval{0}{\amapter(k,n)-1} \to \cT$ such that \textbf{\desclabel{(first)}{tiling_c:1}} $\tau(0,0) = \atile$, and
\begin{description}
\item[\desclabel{(hor\&vert)}{tiling_c:2}] for all $i\in \interval{0}{\amapter(k,n)-1}$ and $j \in \interval{0}{\amapter(k,n)-2}$, $\pair{\tau(j,i)}{\tau(j+1,i)} \in \cH\ $ and $\ \pair{\tau(i,j)}{\tau(i,j+1)} \in \cV$.
\end{description}
The problem of checking whether an instance of $\tiling_k$ has a solution is known to be $k$-\nexptime-complete (see \cite{CC-Papadimitriou}).

The reduction below from $\tiling_k$ to \satproblem{\modallogicSC}
recycles  ideas from~\cite{Bednarczyk&Demri19}
to reduce  $\tiling_k$ to
\satproblem{\msokt}.
To provide the adequate adaptation for \modallogicSC, we need to solve two major issues.
First,  \msokt
admits second-order quantification, whereas in \modallogicSC, the second-order features are limited to the separating conjunction
$\separate$.
Second, the second-order quantification of
\msokt
essentially colours the nodes in
Kripke-style structures without changing the frame $\pair{\worlds}{\arelation}$.
By contrast, the  operator $\separate$ modifies the accessibility relation, possibly making worlds that were reachable from the current world, 
 unreachable in submodels.
The \tower-hardness proof for \satproblem{\modallogicSC} becomes then much more challenging: we would like to characterise the position on the grid encoded
by a world $\aworld$ by exploiting
properties of its descendants
(as done for
\msokt),
but at the same time, we need to be careful and only consider submodels where $\aworld$ keeps encoding the same position.
In a sense, our encoding is robust: when the  operator $\separate$ is used to reason on submodels, we can enforce that no world
changes the position of the grid that it encodes.


\subsection{Enforcing $\amapter(j,n)$ children.} \label{subsec:nonelemchildren}
Let $\amodel = \triple{\worlds}{\arelation}{\avaluation}$ be a finite forest.
We consider two disjoint sets of atomic propositions $\apropset=\{\avarprop_1,\dots,\avarprop_n,\treeval\}$ and $\Aux = \{\avariable, \avariablebis, \avarleft, \avarselect, \avarright\}$ (whose respective role is later defined).
Elements from~$\Aux$ are understood as \emph{auxiliary} propositions.
We call \defstyle{$\aaux$-node} (resp. \defstyle{$\Aux$-node}) a world satisfying the proposition $\aaux \in \Aux$
(resp. satisfying some proposition in $\Aux$). We call 
\defstyle{$\avartree$-node} a world that satisfies the formula $\avartree \egdef \bigwedge_{\aaux \in \Aux}  \lnot \aaux$.
Every world of $\amodel$ is either a $\avartree$-node or an $\Aux$-node.
We say that $\aworld'$ is a $\avartree$-child of $\aworld\,{\in}\,\worlds$ if $\aworld'\,{\in}\,\arelation(\aworld)$ and $\aworld'$ is a $\avartree$-node. We define the concepts of $\Aux$-child and $\aaux$-child similarly.

The key development of our reduction is given by the definition of a formula,
of exponential size in $j \geq 1$ and polynomial size in $n \geq 1$, that when satisfied by $\pair{\amodel}{\aworld}$
forces every $\avartree$-node in $\arelation^i(\aworld)$, where $0 \leq i < j$, to have exactly $\amapter(j-i,n)$ $\avartree$-children, each of them encoding a different number in $\interval{0}{\amapter(j-i,n)-1}$.
As we impose that $\aworld$ is a $\avartree$-node, it must have $\amapter(j,n)$ $\avartree$-children.
We assume $n$ to be fixed throughout the section and denote this formula by $\complete{j}$.
From the property above,
if $\amodel,\aworld \models \complete{j}$ then
for all $i \in \interval{1}{j{-}1}$ and all $\avartree$-nodes $\aworld' \in \arelation^i(\aworld)$ we have $\amodel,\aworld' \models \complete{j{-}i}$.

First, let us informally describe how numbers are encoded in the model $\pair{\amodel}{\aworld}$ satisfying $\complete{j}$. Let $i \in \interval{1}{j}$.
Given a $\avartree$-node $\aworld' \in \arelation^i(\aworld)$, $\nbexp{\aworld'}{i}$ denotes the number encoded by~$\aworld'$.
We omit the subscript $i$ when it is clear from the context.
When $i = j$,
we represent $\nb{\aworld'}$ by using the truth values
of the atomic propositions $\avarprop_1, \dots, \avarprop_n$.
The proposition  $\avarprop_b$ is responsible for the $b$-th bit of the number, with the least significant bit being encoded by $\avarprop_1$.
For example, for $n = 3$, we have $\amodel, \aworld' \models \avarprop_3 \land \avarprop_2 \land \lnot \avarprop_1$ whenever $\nb{\aworld'} = 6$.
The formula $\complete{1}$ forces the parent of $\aworld'$ (i.e.\  is a $\avartree$-node in $\arelation^{j-1}(\aworld)$)
to have exactly $2^n$ $\avartree$-children by requiring one $\avartree$-child for each possible
valuation upon  $\avarprop_1,\dots,\avarprop_n$.
Otherwise, for $i < j$ (and therefore $j \geq 2$), the number $\nbexp{\aworld'}{i}$ is represented by
the binary encoding of the truth values of $\treeval$ on the $\avartree$-children of $\aworld'$ which, since $\pair{\amodel}{\aworld'} \models \complete{j-i}$, are $\amapter(j-i,n)$ children implicitly ordered by the number they, in turn, encode.
The essential property of $\complete{j}$ is therefore the following:
the numbers encoded by the $\avartree$-children of a $\avartree$-node
$\aworld'' \in \arelation^i(\aworld)$, represent positions in the binary representation of the number $\nbexp{\aworld''}{i}$.
Thanks to this property, the formula $\complete{j}$
forces $\aworld$ to have exactly $\amapter(j,n)$
children, all encoding different numbers in $\interval{0}{\amapter(j,n)-1}$.
This is roughly represented in the picture below,
where $``1"$ stands for $\treeval$ being true whereas $``0"$ stands for $\treeval$ being false.
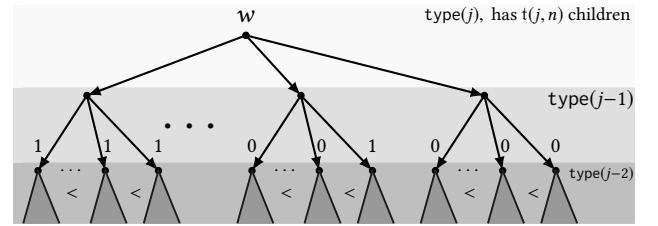
\begin{center}
\begin{tikzpicture}[xscale=0.96]
  \node (upleft) at (0,0) {};

  \path (upleft.center) -- ++(0.38\linewidth,-0.4) node[dot,label={$\aworld$}] (root) {.};

  \draw[pto] (root.center) -- ++ (-0.26\linewidth,-0.8) node[dot] (l1) {};
  \draw[pto] (root.center) -- ++ (0.09\linewidth,-0.8) node[dot] (c1) {};
  \draw[pto] (root.center) -- ++ (0.39\linewidth,-0.8) node[dot] (r1) {};

  \path (root.center) -- ++ (-0.086\linewidth,-1.2) node (dots) {$\scaleobj{1.8}{\dots}$};

  \foreach \x in {l1,c1,r1} {
    \draw[pto] (\x.center) -- ++ (-0.08\linewidth,-1) node[dot] (l\x) {};
    \draw[pto] (\x.center) -- ++ (0.03\linewidth,-1) node[dot] (c\x) {};
    \draw[pto] (\x.center) -- ++ (0.117\linewidth,-1) node[dot] (r\x) {};

      \path (\x.center) -- ++ (-0.024\linewidth,-1) node (dots) {$\scaleobj{0.7}{\dots}$};

      \path (\x.center) -- ++ (0.08\linewidth,-1.3) node (min2\x) {$\scaleobj{0.7}{<}$};

      \path (\x.center) -- ++ (-0.024\linewidth,-1.3) node (min\x) {$\scaleobj{0.7}{<}$};

    \foreach \y in {l,c,r}{
      \path (\y\x.center) -- ++ (-0.024\linewidth,-0.7) node (e1\y\x) {};
      \path (\y\x.center) -- ++ (0.035\linewidth,-0.7) node (e2\y\x) {};

      \draw[gray!80,fill=gray!80] (\y\x.center) -- (e1\y\x.center) -- (e2\y\x.center) -- (\y\x.center);
      \draw[alin] (\y\x.center) -- (e1\y\x.center);
      \draw[alin] (\y\x.center) -- (e2\y\x.center);
    }
  }

  \node[label={$\scaleobj{0.8}{1}$}] at (ll1) {};
  \node[label={[xshift=1pt]{$\scaleobj{0.8}{1}$}}] at (cl1) {};
  \node[label={$\scaleobj{0.8}{1}$}] at (rl1) {};

  \node[label={$\scaleobj{0.8}{0}$}] at (lc1) {};
  \node[label={[xshift=1pt]{$\scaleobj{0.8}{0}$}}] at (cc1) {};
  \node[label={$\scaleobj{0.8}{1}$}] at (rc1) {};

  \node[label={$\scaleobj{0.8}{0}$}] at (lr1) {};
  \node[label={[xshift=1pt]{$\scaleobj{0.8}{0}$}}] at (cr1) {};
  \node[label={$\scaleobj{0.8}{0}$}] at (rr1) {};

  \path (upleft) -- ++(\linewidth,0) node[label={[xshift=-37pt,yshift=-14pt]$\scaleobj{0.7}{\complete{j}, \text{ has } \amapter(j,n) \text{ children}}$}] (upright) {};
  \path let \p1 = (e1ll1) in (upleft) -- ++ (0,\y1) node (bottomleft) {};
  \path let \p1 = (e1ll1) in (upright) -- ++ (0,\y1) node (bottomright) {};

  \path (upleft) -- ++ (0,-1.1) node (midleft) {};
  \path (upright) -- ++ (0,-1.1) node[label={[xshift=-13pt,yshift=-15pt]$\scaleobj{0.8}{\complete{j{-}1}}$}] (midright) {};

  \path (midleft) -- ++ (0,-1) node (downleft) {};
  \path (midright) -- ++ (0,-1) node[label={[xshift=-9pt,yshift=-14pt]$\scaleobj{0.6}{\complete{j{-}2}}$}] (downright) {};

  \begin{pgfonlayer}{bg0}

    \draw[gray!5,fill=gray!5] (upleft.center) -- (bottomleft.center) -- (bottomright.east) -- (upright.east) -- (upleft.east);

    \draw[gray!25,fill=gray!25] (midleft.center) -- (bottomleft.center) -- (bottomright.east) -- (midright.east) -- (midleft.east);

    \draw[gray!50,fill=gray!50] (downleft.center) -- (bottomleft.center) -- (bottomright.east) -- (downright.east) -- (downleft.east);

  \end{pgfonlayer}

\end{tikzpicture}
\end{center}

To characterise these trees in \modallogicSC,  we simulate second-order quantification by using $\Aux$-nodes.
Informally, we require a pointed forest $\pair{\amodel}{\aworld}$ satisfying $\complete{j}$ to be
such that
{\bfseries \itshape{(i)}} every $\avartree$-node $\aworld' \in \arelation(\aworld)$ has exactly one $\avariable$-child, and one (different)
$\avariablebis$-child. These nodes do not satisfy any other auxiliary proposition;
{\bfseries \itshape{(ii)}} for every $i \geq 2$, every $\avartree$-node $\aworld' \in \arelation^i(\aworld)$ has exactly five $\Aux$-children, one for each $\aaux\in\Aux$.
We can simulate second-order existential quantification on $\avartree$-nodes with respect to the symbol $\aaux \in \Aux$
by using the  operator $\separate$ in order to remove edges leading to $\aaux$-nodes. Then, we  evaluate whether a property holds
on the resulting model where a $\avartree$-node ``satisfies'' $\aaux \in \Aux$ if it has a child satisfying $\aaux$.
To better emphasise the need to move along $\avartree$-nodes,
given a formula $\aformula$, we write
$\HMDiamond{\avartree} \aformula$ for the formula $\Diamond (\avartree \land \aformula)$.
Dually, $\HMBox{\avartree} \aformula \egdef \Box (\avartree \implies \aformula)$. $\HMDiamond{\avartree}^i$
and $\HMBox{\avartree}^i$ are also defined, as expected.

Let us start to formalise this encoding. Let $j \geq 1$.
First, we restrict ourselves to models where
every $\avartree$-node reachable in at most $j$ steps does not have two $\Aux$-children satisfying the same proposition. Moreover, these $\Aux$-nodes have no children and only satisfy exactly one $\aaux\in\Aux$.
We express this condition with the formula $\init{j}$ below:
\begin{nscenter}
$
\Boxbox^{j} {\displaystyle\bigwedge_{\mathclap{\aaux \in \Aux}}}
\Big(
  \big(
    \avartree \implies \lnot (\Diamond \aaux \separate \Diamond \aaux)
  \big)
\land
  \Box\big(
    \aaux \implies \Box \bottom \land\!{\displaystyle\bigwedge_{\mathclap{\aauxbis \in \Aux \setminus\{\aaux\}}}} \lnot \aauxbis
  \big)
 \Big),
$
\end{nscenter}
where $\Boxbox^0\aformula \egdef \aformula$ and $\Boxbox^{m+1}\aformula \egdef \aformula \land \Box\Boxbox^{m}(\aformula)$.
Notice that if $\amodel,\aworld \models \init{j}$ and $\amodel' \sqsubseteq \amodel$, then $\amodel',\aworld \models \init{j}$.

Among the models $\pair{\triple{\worlds}{\arelation}{\avaluation}}{\aworld}$ satisfying $\init{j}$,
we define the ones satisfying $\complete{j}$ described below
(see similar conditions in~\cite[Section IV]{Bednarczyk&Demri19}):
\begin{description}
  \item[\desclabel{(\apS$_j$)}{prop:apS}\phantomlabel{(\apS$_1$)}{prop:apS_1}\phantomlabel{(\apS$_k$)}{prop:apS_k}]
    every $\avartree$-node in $\arelation(\aworld)$ satisfies $\complete{j-1}$;
  \item[\desclabel{(\apZ$_j$)}{prop:apZ}\phantomlabel{(\apZ$_1$)}{prop:apZ_1}]
    there is a $\avartree$-node $\tilde\aworld \in \arelation(\aworld)$ such that $\nb{\tilde\aworld} = 0$;
  \item[\desclabel{(\apU$_j$)}{prop:apU}\phantomlabel{(\apU$_1$)}{prop:apU_1}] distinct $\avartree$-nodes in $\arelation(\aworld)$ encode different numbers;
  \item[\desclabel{(\apC$_j$)}{prop:apC}\phantomlabel{(\apC$_1$)}{prop:apC_1}] for every $\avartree$-node $\aworld_1 \in \arelation(\aworld)$, if $\nb{\aworld_1} < \amapter(j,n)-1$ then $\nb{\aworld_2} = \nb{\aworld_1}+1$ for some $\avartree$-node $\aworld_2 \in \arelation(\aworld)$;
  \item[\desclabel{(\apA)}{prop:apA}] $\aworld$ is a $\avartree$-node, every $\avartree$-node in $\arelation(\aworld)$ has one $\avariable$-child and one $\avariablebis$-child, and every $\avartree$-node in $\arelation^2(\aworld)$ has three children satisfying $\avarleft$, $\avarright$ and $\avarselect$, respectively.
\end{description}
We define $\complete{0}\egdef\true$, and
for $j \geq 1$, $\complete{j}$ is defined as
\begin{nscenter}
$
\complete{j} \egdef \pS{j} \land \pZ{j} \land \pU{j} \land \pC{j} \land \pA,
$
\end{nscenter}
where each conjunct expresses its homonymous property.
The formulae for $\pS{j}$, $\pA$ and $\pZ{j}$ can be defined as
\begin{center}
$\begin{aligned}
\pS{j} \ \egdef  &\ \HMBox{\avartree}\complete{j-1};\\[-2pt]
\pA \ \egdef &\ \avartree \land \HMBox{\avartree}(\Diamond \avariable \separate \Diamond \avariablebis) \land \HMBox{\avartree}^2(\Diamond \avarleft \separate \Diamond \avarselect \separate \Diamond \avarright);\\[-2pt]
\pZ{1} \ \egdef  &\ \HMDiamond{\avartree}\!\textstyle\bigwedge_{{b \in \interval{1}{n}}}\! \lnot \avarprop_b;\\
\pZ{j\,{+}\,1} \ \egdef&\ \HMDiamond{\avartree} \HMBox{\avartree} \lnot \treeval .
\end{aligned}
$
\end{center}

The challenge is therefore how to express $\pU{j}$ and $\pC{j}$,
to guarantee that the numbers of children of $\aworld$ span all over $\interval{0}{\amapter(j,n)-1}$.
The structural properties expressed by $\complete{j}$ lead to strong constraints, which permits to control the effects
of $\separate$ when submodels are constructed. This is a key point in designing  $\complete{j}$ as it helps us
to control which edges are lost when considering a submodel.
\subsubsection*{Nominals, forks and number comparisons.}
In order to define
$\pU{j}$ and $\pC{j}$ (completing the definition of $\complete{j}$), we introduce auxiliary formulae,
characterising classes of models that emerge naturally when trying to capture the semantics of \ref{prop:apU} and \ref{prop:apC}.

Let us consider a finite forest $\amodel = \triple{\worlds}{\arelation}{\avaluation}$ and $\aworld \in \worlds$.
A first ingredient is given by the concept of \defstyle{local nominals}, borrowed from~\cite{Bednarczyk&Demri19}.
We say that $\aaux \in \Aux$ is a (local) nominal for the depth $i \geq 1$ if there is exactly one $\avartree$-node $\aworld' \in \arelation^i(\aworld)$ having an $\aaux$-child.
In this case, $\aworld'$ is said to be the world that corresponds to the local nominal~$\aaux$.
The following formula states that $\aaux$ is a local nominal for the depth $i$:
\begin{center}
$
\begin{aligned}
\nominal{\aaux}{i} \egdef
  \HMDiamond{\avartree}^i\Diamond\aaux \land
  {\bigwedge_{\mathclap{k \in \interval{0}{i-1}}}} \HMBox{\avartree}^k \lnot \big(\HMDiamond{\avartree}^{i-k}\Diamond\aaux \separate \HMDiamond{\avartree}^{i-k}\Diamond \aaux\big).
\end{aligned}
$
\end{center}
We define the formula $\atnom{\aaux}{i} \aformula \egdef \HMDiamond{\avartree}^i(\Diamond \aaux \land \aformula)$ which, under the hypothesis
that $\aaux$ is a local nominal for the depth $i$, states that $\aformula$ holds on the $\avartree$-node that corresponds to~$\aaux$.
Moreover, we define 
$\twonoms{\aaux}{\aauxbis}{i} \egdef \nominal{\aaux}{i} \land \nominal{\aauxbis}{i} \land \lnot \atnom{\aaux}{i} \Diamond \aauxbis$,
which states that $\aaux$ and $\aauxbis$ are two nominals for the depth~$i$ with respect to two distinct $\avartree$-nodes.

As a second ingredient, we introduce the notion of \emph{fork} that is a specific type
of models naturally emerging
when trying to compare the numbers $\nb{\aworld_1}$ and $\nb{\aworld_2}$ of two worlds $\aworld_1,\aworld_2 \in \arelation^i(\aworld)$ (e.g. when checking whether $\nb{\aworld_1} = \nb{\aworld_2}$ or $\nb{\aworld_2} = \nb{\aworld_1}+1$ holds).
Given $j \geq i \geq 1$
we introduce the formula $\fork{\aaux}{\aauxbis}{i}{j}$
that is satisfied by $\pair{\amodel}{\aworld}$ iff:
\begin{itemize}
  \item $\aaux$ and $\aauxbis$ are nominals for the depth~$i$.
 \item
  $\aworld$ has exactly two $\avartree$-children, say $\aworld_U$ and $\aworld_D$.
 \item For every~$k \in \interval{1}{i-1}$, both $\arelation^k(\aworld_U)$ and $\arelation^k(\aworld_D)$
  contain exactly one $\avartree$-child.
 \item The only $\avartree$-node in $\arelation^{i-1}(\aworld_U)$, say $\aworld_{\aaux}$,
 corresponds to the nominal $\aaux$.
 The only $\avartree$-node in $\arelation^{i-1}(\aworld_D)$, say $\aworld_{\aauxbis}$,
 corresponds to the nominal $\aauxbis$.
 \item  If $i < j$, then $\pair{\amodel}{\aworld_{\aaux}}$ and $\pair{\amodel}{\aworld_{\aauxbis}}$ satisfy
  \begin{nscenter}
  $\completeplus{j - i} \egdef \complete{j - i} \land \HMBox{\avartree}(\Diamond \avarleft \land \Diamond \avarselect \land \Diamond \avarright)$.
  \end{nscenter}
\end{itemize}
It should be noted that, whenever $\pair{\amodel}{\aworld}$ satisfies the formula~$\fork{\aaux}{\aauxbis}{i}{j}$,
we witness two paths of length~$i$, both starting at~$\aworld$
and leading to~$\aworld_{\aaux}$ and $\aworld_{\aauxbis}$, respectively.
Worlds in this path may have~$\Aux$-children.
Below, we schematise a model satisfying $\fork{\aaux}{\aauxbis}{i}{j}$:
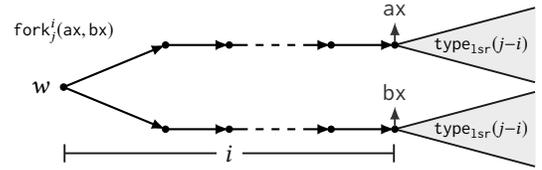
\begin{center}
\begin{tikzpicture}[yscale=.8]
  \node[dot,label={[yshift=12pt]$\scaleobj{0.7}{\fork{\aaux}{\aauxbis}{i}{j}}$}] (root) at (0,0) {.};
  \path
   (root.center) -- ++(-0.3,0) node (lab) {$\aworld$};

  \draw[pto] (root.center) -- ++(0.16\linewidth,0.7) node[dot] (up) {};
  \draw[pto] (root.center) -- ++(0.16\linewidth,-0.7) node[dot] (down) {};

  \foreach \x in {up,down} {
    \draw[pto] (\x.center) -- ++(0.1\linewidth,0) node[dot] (o\x) {};

    \draw[path] (o\x.center) -- ++(0.16\linewidth,0) node[dot] (t\x) {};

    \draw[pto] (t\x.center) -- ++(0.1\linewidth,0) node[dot,label={right:$\quad\scaleobj{0.7}{\completeplus{j{-}i}}$}] (f\x) {};

    \path (f\x.center) -- ++(0.22\linewidth,0.62) node (up\x) {};
    \path (f\x.center) -- ++(0.22\linewidth,-0.62) node (down\x) {};

    \begin{pgfonlayer}{bg0}
      \draw[gray!15,fill=gray!15] (f\x.center) -- (up\x.center) -- (down\x.center) -- (f\x.center);
      \draw[alin] (f\x.center) -- (up\x.center);
      \draw[alin] (f\x.center) -- (down\x.center);
    \end{pgfonlayer}
  }

  \draw[aux] (fup) -- ++(0,0.38) node[label={[yshift=-5pt]$\scaleobj{0.9}{\aaux}$}] (aaux) {};
  \draw[aux] (fdown) -- ++(0,0.38) node[label={[yshift=-5pt]$\scaleobj{0.9}{\aauxbis}$}] (aauxbis) {};

  \path (root.center) -- ++(0,-1.1) node (bottomline) {};

  \path let \p1 = (fdown) in (bottomline) -- ++ (\x1,0) node (end) {};

  \draw[segment] (bottomline.center) -- (end.center) node [midway, fill=white] {$i$};
\end{tikzpicture}
\end{center}
\noindent Since the definition of $\fork{\aaux}{\aauxbis}{i}{j}$ is recursive on $i$ and $j$ (due to $\complete{j-i}$), we postpone its
formal definition
to the next two sections where we
treat the base cases for $i = j$ and the inductive case for $j > i$ separately.

The last  auxiliary formulae are
$\less{\aaux}{\aauxbis}{i}{j}$\,and $\successor{\!\aaux}{\aauxbis\!}{j}$.
Under the hypothesis that $\pair{\amodel}{\aworld}$
satisfies $\fork{\aaux}{\aauxbis}{i}{j}$, the formula $\less{\aaux}{\aauxbis}{i}{j}$
is satisfied whenever the two (distinct) worlds $\aworld_{\aaux},\aworld_{\aauxbis} \in \arelation^i(\aworld)$ corresponding to the nominals $\aaux$ and $\aauxbis$ are such that $\nb{\aworld_{\aaux}} < \nb{\aworld_{\aauxbis}}$.
Similarly, under the hypothesis that $\pair{\amodel}{\aworld}$
satisfies $\fork{\aaux}{\aauxbis}{1}{j}$, the formula $\successor{\aaux}{\aauxbis}{j}$
is satisfied whenever $\nb{\aworld_{\aauxbis}} = \nb{\aworld_{\aaux}}+1$ holds.
Both  formulae are recursively defined, with base cases for $i = j$ and $j = 1$, respectively.

For the base case, we define the formulae  $\fork{\aaux}{\aauxbis}{j}{j}$ and $\less{\aaux}{\aauxbis}{j}{j}$ (for arbitrary $j$), as well as $\successor{\aaux}{\aauxbis}{1}$. From
these formulae, we are then able to define $\pU{1}$ and $\pC{1}$, which completes the characterisation of $\complete{1}$ and $\completeplus{1}$.
Afterwards, we consider the case $1 \leq i < j$ and $j \geq 2$, and define $\fork{\aaux}{\aauxbis}{i}{j}$, $\less{\aaux}{\aauxbis}{i}{j}$, ${\successor{\aaux}{\aauxbis}{j}}$, as well as $\pU{j}$ and $\pC{j}$,
by only relying on formulae that are already defined (by inductive reasoning).
\subsubsection*{Base cases: $i = j$ or $j = 1$.}
In what follows, we consider a finite forest  $\amodel = \triple{\worlds}{\arelation}{\avaluation}$ and a world $\aworld$.
Following its informal description, we have
\begin{center}
$
\begin{aligned}
    \fork{\aaux}{\aauxbis}{j}{j} \egdef \Diamond_{=2} \avartree\,
    {\land}\,\HMBox{\avartree}\Boxbox^{j-2}\!(\avartree{\implies}\Diamond_{=1} \avartree)\,
    {\land}\,\twonoms{\aaux}{\aauxbis}{j},\,
\end{aligned}
$
\end{center}
where
$\Boxbox^j \aformula \egdef \true$ for $j < 0$.
As previously explained, in the base case, the number $\nb{\aworld'}$ encoded by a $\avartree$-node $\aworld' \in \arelation^{j}(\aworld)$ is
represented by the truth values of
$\avarprop_1,\dots,\avarprop_n$.
Then, the formula $\less{\aaux}{\aauxbis}{j}{j}$ is defined as

\begin{nscenter}
$
\begin{aligned}
  \less{\aaux}{\aauxbis}{j}{j} \egdef
    {\bigvee_{\mathclap{u \in \interval{1}{n}}}}\!
      \big(&
        \atnom{\aaux}{j}\lnot \avarprop_u
        \land \atnom{\aauxbis}{j} \, \avarprop_u \land\!\!{\bigwedge_{\mathclap{v \in \interval{u+1}{n}}}}\!
        (\atnom{\aaux}{j} \, \avarprop_v \,{\iff}\,  \atnom{\aauxbis}{j} \, \avarprop_v)
      \big).
\end{aligned}
$
\end{nscenter}
The satisfaction of
$\pair{\amodel}{\aworld} \models \fork{\aaux}{\aauxbis}{j}{j}$ enforces that
the distinct $\avartree$-nodes $\aworld_{\aaux},\aworld_{\aauxbis} \in \arelation^j(\aworld)$ corresponding to $\aaux$ and $\aauxbis$
satisfy $\nb{\aworld_{\aaux}} < \nb{\aworld_{\aauxbis}}$, which can be shown by using standard properties about bit vectors.

The formula $\successor{\aaux}{\aauxbis}{1}$ is similarly defined:
\begin{center}
    \scalebox{0.96}{$\displaystyle{\bigvee_{\mathclap{\quad \, u \in \interval{1}{n}}}}
      \!\!\big(
        \atnom{\aaux}{1}(\lnot \avarprop_u {\land}\!\!\bigwedge_{\mathclap{v \in \interval{1}{u-1}}}\!\avarprop_v)
        \land
        \atnom{\aauxbis}{1} (\avarprop_u {\land}\!\!\bigwedge_{\mathclap{v \in \interval{1}{u-1}}}\!\!\lnot \avarprop_v)
        {\land}\!
        {\bigwedge_{\mathclap{\quad \quad v \in \interval{u+1}{n}}}}\!
        (\atnom{\aaux}{1} \avarprop_v {\iff} \atnom{\aauxbis}{1} \avarprop_v)
      \big).$}
\end{center}
Assuming
$\pair{\amodel}{\aworld} \models \fork{\aaux}{\aauxbis}{1}{1}$, this formula
states that the two distinct $\avartree$-nodes $\aworld_{\aaux},\aworld_{\aauxbis} \in \arelation(\aworld)$ corresponding to $\aaux$ and
$\aauxbis$ are
such that
$\nb{\aworld_{\aauxbis}} = \nb{\aworld_{\aaux}}+1$.
Again, correctness is guaranteed by standard analysis on bit vectors.
\cut{
As done for $\less{\aaux}{\aauxbis}{j}{j}$, this formula states that there must be a bit (encoded by $\avarprop_u$) which is set to $0$ in the binary encoding of $\nb{\aworld_{\aaux}}$ but is set to $1$ in the binary encoding of $\nb{\aworld_{\aauxbis}}$, and that
every the successive bit (encoded by $\avarprop_v$ with $v > u$)
is set to $1$ in $\nb{\aworld_{\aaux}}$
iff
it is set to $1$ also in $\nb{\aworld_{\aauxbis}}$.
However, differently from $\less{\aaux}{\aauxbis}{j}{j}$, this formula also requires that every previous bit of (encoded by $\avarprop_v$ with $v < u$) is set to $1$ in the binary encoding of $\nb{\aworld_{\aaux}}$ but is set to $0$ in the binary encoding of $\nb{\aworld_{\aauxbis}}$.
}

To define $\pU{1}$, we
recall that a model satisfying $\complete{1}$ satisfies the formula $\pA$ and hence every $\avartree$-node in $\arelation(\aworld)$ has two auxiliary children, one $\avariable$-node and one $\avariablebis$-node.
The idea is  to use these two $\Aux$-children and rely 
on~$\separate$ to state that it is not possible to find a submodel of $\amodel$ such that $\aworld$
 has only two distinct children $\aworld_{\avariable}$ and $\aworld_{\avariablebis}$ corresponding to the nominals $\avariable$ and $\avariablebis$, respectively,
and
such that ${\nb{\aworld_\avariable} = \nb{\aworld_\avariablebis}}$.
In a sense, the operator $\separate$
simulates a second-order quantification on $\avariable$ and $\avariablebis$.
Let $\equivalent{\avariable}{\avariablebis}{1}{1} \egdef \lnot (\less{\avariable}{\avariablebis}{1}{1} \lor \less{\avariablebis}{\avariable}{1}{1})$.
We define
$
\pU{1} \egdef \lnot \big(\true \separate (\fork{\avariable}{\avariablebis}{1}{1} \land \equivalent{\avariable}{\avariablebis}{1}{1})\big)
$.
\cut{
The corresponding formula is
$
\pU{1} \egdef \lnot \big(\true \separate (\fork{\avariable}{\avariablebis}{1}{1} \land \equivalent{\avariable}{\avariablebis}{1}{1})\big)
$.
}

To capture $\pC{1}$
we state that it is not possible to find a submodel of $\amodel$ that looses $\avariable$-nodes from
$\arelation^2(\aworld)$, keeps all $\avariablebis$-nodes, and is such that
{\bfseries \itshape{(i)}} $\avariable$ is a local nominal for the depth $1$, corresponding to a world $\aworld_{\avariable}$ encoding $\nb{\aworld_{\avariable}} < 2^n-1$;
{\bfseries \itshape{(ii)}}  there is no submodel where $\aworld$ has two $\avartree$-children, $\aworld_{\avariable}$ and a second
world $\aworld_{\avariablebis}$, such that $\aworld_{\avariablebis}$ corresponds to the nominal $\avariablebis$
and $\nb{\aworld_{\avariablebis}} = \nb{\aworld_{\avariable}}{+}1$.
Thus, $\pC{1}$ is defined as: 
\begin{nscenter}
$
\lnot
  \big(\Box \bottom \separate
    \big(
      \HMBox{\avartree}\Diamond \avariablebis \land \atnom{\avariable}{1} \lnot \one_{1} \land \lnot ( \true \separate (\fork{\avariable}{\avariablebis}{1}{1} \land \successor{\avariable}{\avariablebis}{1}))
    \big)
  \big).
$
\end{nscenter}
The  subscript ``$1$'' in the formula $\one_{1}$ refers to the fact that we are treating the base case of $\pC{j}$ with $j = 1$.
We have $\one_{1} \egdef \bigwedge_{i \in \interval{1}{n}}\avarprop_i$, reflecting the encoding of $2^n-1$.

This concludes the definition of $\complete{1}$ (and $\completeplus{1}$), which is established correct with respect to its specification.

\begin{lemma}\label{lemma:tower-hardness-base-case}
Let $\amodel,\aworld \models \init{1}$. We have $\amodel,\aworld \models \complete{1}$ iff $\pair{\amodel}{\aworld}$ satisfies \ref{prop:apS_1}, \ref{prop:apZ_1},
\ref{prop:apU_1}, \ref{prop:apC_1} and \ref{prop:apA}.
\end{lemma}

\subsubsection*{Inductive case: $1 \leq i < j$.}
As an implicit inductive hypothesis used to prove that the formulae are well-defined, we assume that
$\successor{\aaux}{\aauxbis}{j'}$ and $\complete{j'}$ are already defined for every $j' < j$,
whereas $\fork{\aaux}{\aauxbis}{i'}{j'},$ and $\less{\aaux}{\aauxbis}{i'}{j'}$ are already defined for every $1 \leq i' \leq j'$ such that $j' - i' < j - i$.
Therefore, we define:
 \begin{center}
  $
     \fork{\aaux}{\aauxbis}{i}{j} \egdef \fork{\aaux}{\aauxbis}{i}{i} \land \HMBox{\avartree}^i\completeplus{j-i}.
  $
\end{center}
It is easy to see that this formula is well-defined: $\fork{\aaux}{\aauxbis}{i}{i}$ is
from the base case, whereas $\completeplus{j {-} i}$ is defined by inductive hypothesis, since we have $j-i < j$.

Consider now $\less{\aaux}{\aauxbis}{i}{j}$. Assuming $\amodel,\!\aworld\,{\models}\, \fork{\aaux}{\aauxbis}{i}{j}$, 
we wish to express
 $\nb{\aworld_{\aaux}}\,{<}\,\nb{\aworld_{\aauxbis}}$ for the two distinct worlds $\aworld_{\aaux},\aworld_{\aauxbis}\,{\in}\,\arelation^i(\aworld)$
corresponding to the nominals $\aaux$ and $\aauxbis$, respectively.
 As $i\,{<}\,j$, $\nb{\aworld_{\aaux}}$ (resp. $\nb{\aworld_{\aauxbis}}$) is encoded using the truth value of $\treeval$ on the $\avartree$-children
 of $\aworld_{\aaux}$ (resp. $\aworld_{\aauxbis}$). To rely on
arithmetical properties of binary numbers used to define $\less{\aaux}{\aauxbis}{j}{j}$,
we need to find two partitions
 $P_\aaux\,{=}\,\{L_\aaux,S_\aaux,R_\aaux\}$ and $P_\aauxbis\,{=}\, \{L_\aauxbis,S_\aauxbis,R_\aauxbis\}$,
one for the $\avartree$-children of $\aworld_{\aaux}$ and another one for those of $\aworld_{\aauxbis}$ s.t.:
 \begin{description}
   \item[\desclabel{(LSR)}{desc:LSR}:] Given $b\,{\in}\,\{\aaux,\aauxbis\}$, $P_b$ splits the $\avartree$-children as follows:
 \begin{itemize}
   \item there is a $\avartree$-child $s_b$ of $\aworld_{b}$ such that $S_b = \{s_b\}$;
   \item $\nb{r} < \nb{s_b} < \nb{l}$, for every $r \in R_b$ and $l \in L_b$.
 \end{itemize}
 \item[\desclabel{(LESS)}{desc:LESS}:] $P_{\aaux}$ and $P_{\aauxbis}$ have constraints to satisfy $<$:
  \begin{itemize}
    \item $\nb{s_\aaux} = \nb{s_\aauxbis}$, $\amodel,s_\aaux \models \lnot \treeval$ and $\amodel,s_\aauxbis \models \treeval$;
    \item
      for every $l_\aaux \in L_\aaux$ and $l_\aauxbis \in L_\aauxbis$, if $\nb{l_\aaux} = \nb{l_\aauxbis}$
      then $\amodel,l_\aaux \models \treeval$ iff $\amodel,l_\aauxbis \models \treeval$.
  \end{itemize}
\end{description}
It is important to notice that these conditions essentially revolve around the numbers encoded by $\avartree$-children, which will be compared using the already defined (by inductive reasoning) formulae $\less{\aaux}{\aauxbis}{i'}{j'}$, where $j'-i' < j - i$.
Since the semantics of $\less{\aaux}{\aauxbis}{i}{j}$ is given under the hypothesis that $\amodel,\aworld \models \fork{\aaux}{\aauxbis}{i}{j}$,
we can assume that every child of $\aworld_{\aaux}$ and $\aworld_{\aauxbis}$ has  all the possible $\Aux$-children.
Then, we rely on the auxiliary propositions in $\set{\avarleft, \avarselect, \avarright}$ in order to mimic the reasoning done in~\ref{desc:LSR} and~\ref{desc:LESS}.

We start by considering the constraints involved in \ref{desc:LSR} and express them with the formula $\lsrpartition{j}$, which is satisfied by a pointed
forest $\pair{\amodel = (\worlds,\arelation,\avaluation)}{\aworld}$ whenever:
\begin{itemize}
  \item $\pair{\amodel}{\aworld}$ satisfies $\complete{j}$.
  \item Every $\avartree$-child of $\aworld$ has exactly one $\{\avarleft,\avarselect,\avarright\}$-child, and
  only one of these $\avartree$-children (say $\aworld'$) has an $\avarselect$-child.
  \item Every $\avartree$-child of $\aworld$ that has an $\avarleft$-child (resp. $\avarright$-child) encodes a number greater (resp. smaller) than $\nb{\aworld'}$.
\end{itemize}
Despite this formula being defined in terms of $\complete{j}$, we only rely on $\lsrpartition{j-i}$ (which is defined by inductive reasoning) in order to define $\less{\aaux}{\aauxbis}{i}{j}$.
The picture below
schematises a model satisfying $\lsrpartition{j}$.
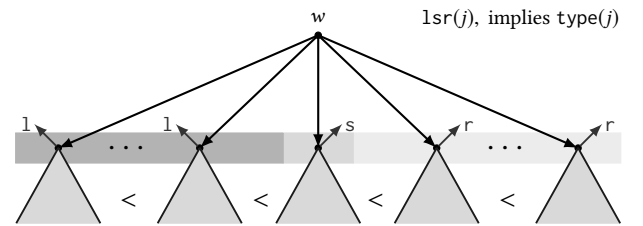
\begin{center}
 \begin{tikzpicture}
  \node (upleft) at (0,0) {};

  \path (upleft.center) -- ++(0.475\linewidth,-0.4) node[dot,label={$\scaleobj{0.8}{\aworld}$}] (root) {.};

  \draw[pto] (root.center) -- ++ (-0.186\linewidth,-1.5) node[dot] (l1) {};
  \draw[pto] (root.center) -- ++ (-0.408\linewidth,-1.5) node[dot] (l2) {};
  \draw[pto] (root.center) -- ++ (0,-1.5) node[dot] (c1) {};
  \draw[pto] (root.center) -- ++ (0.186\linewidth,-1.5) node[dot] (r1) {};
  \draw[pto] (root.center) -- ++ (0.408\linewidth,-1.5) node[dot] (r2) {};

  \path (root.center) -- ++ (-0.2975\linewidth,-1.5) node (dots) {$\scaleobj{1.2}{\dots}$};
  \path (root.center) -- ++ (0.2975\linewidth,-1.5) node (dotsbis) {$\scaleobj{1.2}{\dots}$};

  \path (root.center) -- ++ (0.2975\linewidth,-2.2) node (min1) {$\scaleobj{1}{<}$};
  \path (root.center) -- ++ (0.09\linewidth,-2.2) node (min2) {$\scaleobj{1}{<}$};
  \path (root.center) -- ++ (-0.2975\linewidth,-2.2) node (min3) {$\scaleobj{1}{<}$};
  \path (root.center) -- ++ (-0.09\linewidth,-2.2) node (min4) {$\scaleobj{1}{<}$};

  \foreach \x in {l1,l2,c1,r1,r2} {
    \path (\x.center) -- ++ (-0.066\linewidth,-1) node (l\x) {};
    \path (\x.center) -- ++ (0.066\linewidth,-1) node (r\x) {};

      \draw[gray!30,fill=gray!30] (\x.center) -- (l\x.center) -- (r\x.center) -- (\x.center);
      \draw[alin] (\x.center) -- (l\x.center);
      \draw[alin] (\x.center) -- (r\x.center);
  }

  \path (upleft) -- ++(0.95\linewidth,0) node[label={[xshift=-38pt,yshift=-16pt]$\scaleobj{0.8}{\lsrpartition{j}, \text{ implies } \complete{j}}$}] (upright) {};

  \path (upleft) -- ++ (0,-1.7) node (midleft) {};
  \path (upright) -- ++ (0,-1.7) node (midright) {};

  \path (midleft) -- ++ (0.42\linewidth,0) node (selectleft) {};
  \path (midright) -- ++ (-0.42\linewidth,0) node (selectright) {};

  \path (midleft) -- ++ (0,-0.4) node (downleft) {};
  \path (midright) -- ++ (0,-0.4) node (downright) {};

  \path (downleft) -- ++ (0.42\linewidth,0) node (selectleftbottom) {};
  \path (downright) -- ++ (-0.42\linewidth,0) node (selectrightbottom) {};

  \begin{pgfonlayer}{bg0}

    \draw[gray!15,fill=gray!15] (selectright.center) -- (selectrightbottom.center) -- (downright.center) -- (midright.center) -- (selectright.center);

    \draw[gray!30,fill=gray!30] (selectleft.center) -- (selectleftbottom.center) -- (selectrightbottom.center) -- (selectright.center) -- (selectleft.center);

    \draw[gray!60,fill=gray!60] (selectleft.center) -- (selectleftbottom.center) -- (downleft.center) -- (midleft.center) -- (selectleft.center);

  \end{pgfonlayer}

  \draw[aux] (l1) -- ++(-0.32,0.32) node[label={[xshift=-3pt,yshift=-10pt]$\scaleobj{0.8}{\avarleft}$}] (labl1) {};
  \draw[aux] (l2) -- ++(-0.32,0.32) node[label={[xshift=-3pt,yshift=-10pt]$\scaleobj{0.8}{\avarleft}$}] (labl2) {};

  \draw[aux] (c1) -- ++(0.32,0.32) node[label={[xshift=3pt,yshift=-10pt]$\scaleobj{0.8}{\avarselect}$}] (labc1) {};

  \draw[aux] (r1) -- ++(0.32,0.32) node[label={[xshift=3pt,yshift=-10pt]$\scaleobj{0.8}{\avarright}$}] (labr1) {};
  \draw[aux] (r2) -- ++(0.32,0.32) node[label={[xshift=3pt,yshift=-10pt]$\scaleobj{0.8}{\avarright}$}] (labr2) {};

\end{tikzpicture}
\end{center}

\noindent The definition of $\lsrpartition{j}$ follows closely its specification:
\begin{center}
$
\begin{aligned}
\lsrpartition{j} \egdef \complete{j}
\,{\land}\,\nominal{\avarselect}{1} \land\!
 \lnot (\true\!\separate\!(\fork{\avarselect}{\avarleft}{1}{j}\!\land\!\lnot \less{\avarselect}{\avarleft}{1}{j}))\\[-2pt]
 \land
\lnot (\true\!\separate\!(\fork{\avarselect}{\avarright}{1}{j}\!\land\!\lnot \less{\avarright}{\avarselect}{1}{j}))
\land \HMBox{\avartree}\Diamond_{=1}(\avarleft\,{\lor}\,\avarselect\,{\lor}\, \avarright).
\end{aligned}
$
\end{center}

We define the formula $\less{\aaux}{\aauxbis}{i}{j}$ as follows:
  \begin{center}
  $
  \true \separate \big(
    \twonoms{\aaux}{\aauxbis}{i} \land \HMBox{\avartree}^i\lsrpartition{j-i} \land
  \selectpred{\aaux}{\aauxbis}{i}{j} \land \leftpred{\aaux}{\aauxbis}{i}{j}
  \big),
$
\end{center}
where $\selectpred{\aaux}{\aauxbis}{i}{j}$ and $\leftpred{\aaux}{\aauxbis}{i}{j}$ check the first and second condition in~\ref{desc:LESS}, respectively. In particular, by defining
$\equivalent{\aaux}{\aauxbis}{i}{j} \egdef \lnot (\less{\aaux}{\aauxbis}{i}{j} \lor \less{\aauxbis}{\aaux}{i}{j})$, we have
\begin{center}
    $
    \begin{aligned}
      & \selectpred{\aaux}{\aauxbis}{i}{j} \egdef
      \true \separate \big( \fork{\avariable}{\avariablebis}{i+1}{j} \land
      \atnom{\aaux}{i}\HMDiamond{\avartree}(\Diamond \avarselect \land \Diamond \avariable)\, \land\\[-3pt]
      &\pushright{\atnom{\aauxbis}{i}\HMDiamond{\avartree}(\Diamond \avarselect \land \Diamond \avariablebis) \land \equivalent{\avariable}{\avariablebis}{i+1}{j} \land \atnom{\avariable}{i+1}\lnot \treeval \land \atnom{\avariablebis}{i+1}\treeval
      \big)}\\[1pt]
      &\leftpred{\aaux}{\aauxbis}{i}{j} \egdef \lnot \big( \true \separate \big(
      \fork{\avariable}{\avariablebis}{i+1}{j} \land
      \atnom{\aaux}{i}\HMDiamond{\avartree}(\Diamond \avarleft \land \Diamond \avariable)\, \land\\[-3pt]
      &\pushright{\atnom{\aauxbis}{i}\!\HMDiamond{\avartree}(\Diamond \avarleft \land \Diamond \avariablebis)
      \land  \equivalent{\avariable}{\avariablebis}{i+1}{j}{\land}
      \lnot (\atnom{\avariable}{i+1}\treeval \iff \atnom{\avariablebis}{i+1}\treeval)
      \big)\big).}
    \end{aligned}
    $
\end{center}
Both $\fork{\avariable}{\avariablebis}{i+1}{j}$ and $\equivalent{\avariable}{\avariablebis}{i+1}{j}$ used in these formulae are
defined recursively.
The formula $\selectpred{\aaux}{\aauxbis}{i}{j}$ states that there is a submodel $\amodel' \sqsubseteq \amodel$ such that
\begin{enumerate}[label=\Roman*.,align = left]
  \item\  $\amodel', \aworld \models \fork{\avariable}{\avariablebis}{i+1}{j}$;
  \item\  $s_\aaux$ corresponds to the nominal $\avariable$ at depth $i+1$;
  \item\  $s_\aauxbis$ corresponds to the nominal $\avariablebis$ at depth $i+1$;
  \item[IV-VI.] $\nb{s_\aaux} = \nb{s_\aauxbis}$, $\amodel,s_\aaux \not\models \treeval$ and $\amodel,s_\aauxbis \models \treeval$.
\end{enumerate}
(The enumeration I-VI refers to the conjuncts in the formula)

$\selectpred{\aaux}{\aauxbis}{i}{j}$ correctly models the first condition of~\ref{desc:LESS}.
Regarding $\leftpred{\aaux}{\aauxbis}{i}{j}$ and~\ref{desc:LESS}, a similar analysis can be performed.
We define $\selectleftpred{\aaux}{\aauxbis}{i}{j}\egdef\leftpred{\aaux}{\aauxbis}{i}{j} \land \selectpred{\aaux}{\aauxbis}{i}{j}$.

Let us consider $\successor{\aaux}{\aauxbis}{j}$.
Under the hypothesis that $\amodel,\aworld \models \fork{\aaux}{\aauxbis}{i}{j}$, this formula must express
 $\nb{\aworld_{\aauxbis}} = \nb{\aworld_{\aaux}} + 1$ for the two (distinct) worlds $\aworld_{\aaux},\aworld_{\aauxbis} \in \arelation^i(\aworld)$.
 Then, as done for defining $\less{\aaux}{\aauxbis}{i}{j}$, we take advantage of arithmetical properties on
binary numbers and we search for two partitions
 $P_\aaux = \{L_\aaux,S_\aaux,R_\aaux\}$ and $P_\aauxbis = \{L_\aauxbis,S_\aauxbis,R_\aauxbis\}$ of the $\avartree$-children of $\aworld_{\aaux}$ and $\aworld_{\aauxbis}$, respectively, such that $P_\aaux$ and $P_\aauxbis$ satisfy \ref{desc:LSR} as well as the condition below:
 \begin{description}
  \item[\desclabel{(PLUS)}{desc:PLUS}:] $P_{\aaux}$ and $P_{\aauxbis}$ have the arithmetical properties of $+1$ :
  \begin{itemize}
    \item $P_{\aaux}$ and $P_{\aauxbis}$ satisfy \ref{desc:LESS};
    \item
      for every $r_\aaux \in R_\aaux$, we have $\amodel,r_\aaux \models \treeval$;
    \item  for every $r_\aauxbis \in R_\aauxbis$, we have $\amodel,r_\aaux \not\models \treeval$,
  \end{itemize}
\end{description}
where $S_\aaux = \{s_\aaux\}$ and $S_{\aauxbis} = \{s_\aauxbis\}$, as required by \ref{desc:LSR}.

The definition of $\successor{\aaux}{\aauxbis}{j}$ is similar to $\less{\aaux}{\aauxbis}{i}{j}$:
 \begin{nscenter}
  $
  \true  {\separate} \big(
    \twonoms{\aaux}{\aauxbis}{1} \land \HMBox{\avartree}\lsrpartition{j-1} \land
    \selectleftpred{\aaux}{\aauxbis}{1}{j}
    \land
    \rightpred{\aaux}{\aauxbis}
  \big),
  $
 \end{nscenter}
where $\rightpred{\aaux}{\aauxbis} \egdef \atnom{\aaux}{1}\HMBox{\avartree}(\Diamond \avarright \implies \treeval)
\land \atnom{\aauxbis}{1}\HMBox{\avartree}(\Diamond \avarright \implies \lnot \treeval) $
captures the last two conditions of~\ref{desc:PLUS}.

To define
$\pU{j}$ and $\pC{j}$, we rely on $\fork{\aaux}{\aauxbis}{i}{j}$, $\less{\aaux}{\aauxbis}{i}{j}$ and $\successor{\aaux}{\aauxbis}{j}$.
%
\begin{center}
    $
    \begin{aligned}
      \pU{j} \ \egdef  & \ \lnot \big(\true \separate (\fork{\avariable}{\avariablebis}{1}{j} \land \equivalent{\avariable}{\avariablebis}{1}{j})\big)\\[1pt]
    \pC{j} \ \egdef & \ \lnot
    \Big( \Box \false \separate
      \Big(\HMBox{\avartree} (\completeplus{j-1} \land \Diamond \avariablebis) \land
      \nominal{\avariable}{1}
    \land \\[-3pt]
    & \atnom{\avariable}{1} \lnot \one_{j} \land \lnot \big( \true \separate (\fork{\avariable}{\avariablebis}{1}{j} \land
   \successor{\avariable}{\avariablebis}{j})\big)\Big)\Big),
  \end{aligned}
  $
\end{center}
where $\one_j \egdef \HMBox{\avartree} \treeval$ reflects the encoding of $\amapter(j,n)-1$ for $j > 1$.
The main difference between $\pC{1}$ and $\pC{j}$ ($j > 1$) is that the conjunct $\HMBox{\avartree}\Diamond \avariablebis$ of $\pC{1}$ is replaced by $\HMBox{\avartree} (\completeplus{j-1} \land \Diamond \avariablebis)$
in $\pC{j}$, as needed to correctly evaluate $\fork{\avariable}{\avariablebis}{1}{j}$.
Indeed, the difference between $\fork{\avariable}{\avariablebis}{1}{1}$ and $\fork{\avariable}{\avariablebis}{1}{j}$ is precisely that the latter requires $\HMBox{\avartree}\completeplus{j-1}$.
The definition of $\complete{j}$ is now complete. We can state  its correctness.

\begin{lemma}\label{lemma:tower-hardness-inductive}
Let $\amodel,\aworld \models \init{j}$. We have $\amodel,\aworld \models \complete{j}$ iff $\pair{\amodel}{\aworld}$ satisfies \ref{prop:apS}, \ref{prop:apZ},
\ref{prop:apU}, \ref{prop:apC} and \ref{prop:apA}.
\end{lemma}

The size of $\complete{j}$ is exponential in $j > 1$ and polynomial in $n \geq 1$.
As its size is elementary, we can use this formula as a starting point to reduce $\tiling_k$.


\subsection{Tiling a grid  $\interval{0}{\amapter(k,n)-1}  \times \interval{0}{\amapter(k,n)-1}$}\label{subsection:tiling-grid}

Below, we briefly explain how to use previous developments
to define a uniform reduction from $\tiling_k$, for every ${k \geq 2}$. Several adaptations
are needed
to encode smoothly the grid but the hardest part was the design of $\complete{j}$.
Let $k \! \geq \! 2$ and  $\pair{\cTT}{\atile}$ be an instance of
$\tiling_k$.
We can construct a formula $\atiling{k}{\cTT,\atile}$ that is satisfiable if and only if
$\pair{\cTT}{\atile}$ as a solution.
To represent $\interval{0}{\amapter(k,n)-1}^2$ in
some pointed forest $\pair{\amodel}{\aworld}$, where $\amodel = \triple{\worlds}{\arelation}{\avaluation}$,
we recycle the ideas for defining $\complete{k}$.
From \Cref{lemma:tower-hardness-inductive}, we know that if $\amodel,\aworld \models$ $\init{k} \land \complete{k}$ then
the $\avartree$-children of $\aworld$ encode the interval $\interval{0}{\amapter(k,n)-1}$.
A position in the grid is however a pair of numbers, hence
the {\em crux  of our encoding} rests on the fact that  each $\aworld' \in \arelation(\aworld)$ encodes {\em two} numbers
$\nbexp{\aworld'}{\cH}$ and $\nbexp{\aworld'}{\cV}$.
Similarly to $\complete{k}$, these numbers are represented by the truth values on the $\avartree$-children of $\aworld'$, with the help of {\em new propositions} $\treeval_\cH$ and $\treeval_\cV$.
We are in luck: since both numbers are from $\interval{0}{\amapter(k,n)-1}$, $\aworld'$ just needs as many children as when encoding a single number, and therefore if $\amodel,\aworld \,{\models}\, \atiling{k}{\cTT,\atile}$ then $\amodel,\aworld'\,{\models}\,\complete{k{-}1}$.
In fact, the portion of $\atiling{k}{\cTT,\atile}$ that encodes the grid can be described quite naturally by slightly updating the characterisation of $\complete{k}$. For example, \ref{prop:apU} becomes
\begin{description}
  \item[(\apU$_{\cTT,k}$)]
  $
  \begin{array}[t]{l}
  \text{for all distinct } \avartree\text{-nodes } \aworld_1,\aworld_2  \in \arelation(\aworld)
  \\[-1pt]
  \nbexp{\aworld_1}{\cH} \neq \nbexp{\aworld_2}{\cH} \text{ or }
  \nbexp{\aworld_1}{\cV} \neq \nbexp{\aworld_2}{\cV}.
  \end{array}
  $
\end{description}
The formula $\pU{k}$ has  to be updated accordingly, but without major differences or complications.
Of course, more is required as $\atiling{k}{\cTT,\atile}$ must also
encode the tiling conditions~\ref{tiling_c:1} and~\ref{tiling_c:2}.
Fortunately, the kit of formulae defined for $\complete{k}$ allows us to have access to
$\nbexp{\aworld'}{\cH}$ and $\nbexp{\aworld'}{\cV}$ in such a way that both  conditions can be expressed rather easily.
For example, to express vertical constraints, we design a formula stating that
 for all ${\avartree\text{-nodes}}$ $\aworld_1,\aworld_2\!\in\!\arelation(\aworld)$,
    if $\nbexp{\aworld_2}{\cV} = \nbexp{\aworld_1}{\cV}{+}1$ and $\nbexp{\aworld_2}{\cH} = \nbexp{\aworld_1}{\cH}$ then there is $\pair{\atile_1}{\atile_2} \in \cV$ such that
    $\aworld_1 \in \avaluation(\atile_1)$ and $\aworld_2 \in \avaluation(\atile_2)$.
Further details are omitted by lack of space.

\begin{theorem}\label{theorem:tower-completeness-SC}
  $\satproblem{\modallogicSC}$ is \tower-complete.
\end{theorem}

\section{\modallogicSC  Strictly Less Expressive Than \GML}
\label{section-expressivity-SC}

Below, we focus on the expressivity of \modallogicSC. We first show \modallogicSC $\preceq$ \GML and
then we prove the strictness of the inclusion.
The former result takes advantage of the notion of g-bisimulation,
i.e. the underlying structural indistinguishability relation of \GML, studied in~\cite{DeRijke00}.
To show ${\modallogicSC \prec \GML}$,
we define an ad hoc notion of Ehrenfeucht-Fra\"iss\'e games for \modallogicSC,
see e.g. classical definitions in~\cite{Libkin04} and
similar approaches in~\cite{DawarGG04,Calcagno10}.
Then, we design a simple formula in \GML that cannot be expressed in~\modallogicSC.

\subsection{\modallogicSC is not more expressive than \GML}
\label{section-sc-into-gml}

To establish that \modallogicSC $\preceq$ \GML, we proceed as in \Cref{section-cc-less-gml}. In fact, by \Cref{lemma:SC-CC-tautology},
given $\aformula_1$, $\aformula_2$ in \GML, the formula $\aformula_1 \separate \aformula_2$ is equivalent to $\SabDiamond(\aformula_1 \chopop \aformula_2)$. Moreover, we know that given
$\aformula_1$, $\aformula_2$ in \GML,
$\aformula_1 \chopop \aformula_2$ is equivalent to some formula in \GML, as shown in~\Cref{section-CC}.
So,
to prove that \modallogicSC $\preceq$ \GML by applying the proof schema of \Cref{theorem:clean-cut}, it is sufficient to show that
given $\aformula$ in \GML,
 there is
$\aformulabis$ in \GML such that $\SabDiamond \aformula \equiv \aformulabis$.
To do so, we rely on the indistinguishability relation of \GML, called g-bisimulation~\cite{DeRijke00}.


A g-bisimulation is a refinement of the classical back-and-forth conditions
of a bisimulation
(see e.g.~\cite{Blackburn&deRijke&Venema01}), tailored towards capturing graded modalities.
It relates models with similar structural properties, but up to  parameters~${m,k\in \Nat}$
responsible for the modal degree and the graded rank,
respectively.
The following invariance result holds: g-bisimilar models are
modally equivalent in \GML (up to formulae of modal degree $m$ and graded rank at most $k$).
For simplicity, we present the construction of the above-mentioned formula~$\aformulabis$ by directly using
the notion of model equivalence, without going explicitly through g-bisimulations.

Given $m, k \in \Nat$ and $\apropset \subseteq_{\fin} \varprop$, we write $\GML[m,k,\apropset]$ to denote
the set of \GML formulae $\aformulabis$ having $\md{\aformulabis} \leq m$, $\gr{\aformulabis} \leq k$ and propositional variables from $\apropset$.
$\GML[m,k,\apropset]$
is finite up to logical equivalence~\cite{DeRijke00}.
Given pointed
forests
$\pair{\amodel}{\aworld}$ and $\pair{\amodel'}{\aworld'}$,
we write
$\pair{\amodel}{\aworld} \typeeqclass{m,k}{\apropset}\! \pair{\amodel'}{\aworld'}$
whenever
$\pair{\amodel}{\aworld}$ and $\pair{\amodel'}{\aworld'}$ are \defstyle{$\GML[m,k,\apropset]$-indistinguishable},
i.e.\
for every $\aformulabis$ in $\GML[m,k,\apropset]$, 
$\amodel, \aworld \models \aformulabis$
iff  $\amodel', \aworld' \models \aformulabis$.
We write $\atypeset{m,k}{\apropset}$ to denote the quotient set induced by the equivalence relation $\typeeqclass{m,k}{\apropset}$.
As~$\GML[m,k,\apropset]$ is finite up to logical equivalence,
we get that $\atypeset{m,k}{\apropset}$ is finite.

To establish that \GML is closed under $\SabDiamond$, we
show that there is a function $\amap : \Nat^2 \to \Nat$ such that for all
$m, k \in \Nat$ and $\apropset \subseteq_{\fin} \varprop$, if two models
are in the same equivalence class of $\typeeqclass{m,\amap(m,k)}{\apropset}$,
then they satisfy the same formulae of the form $\SabDiamond \aformula$, where $\aformula$ is in $\GML[m,k,\apropset]$.
By standard arguments and using the fact that $\GML[m,\amap(m,k),\apropset]$ is finite up to logical equivalence, we then conclude that $\SabDiamond \aformula$ is equivalent to a formula in $\GML[m,\amap(m,k),\apropset]$.
Similar approaches are followed in~\cite{Mansutti18,Demri&Fervari&Mansutti19,Echenim&Iosif&Peltier19}.
As we are not interested in the size of the equivalent formula, we can simply use the cardinality of $\atypeset{m,k}{\apropset}$ in order to inductively define a suitable function:
\begin{nscenter}
\hfill$\amap(0,k) \egdef k$,\hfill $\amap(m+1,k) \egdef k \times (\card{\atypeset{m,\amap(m,k)}{\apropset}}+1)$.\hfill\,
\end{nscenter}
In conformity with the results in~\Cref{section-tower-SC}, the map $\amap$ can be shown to be a non-elementary function.
To prove that $\amap$ satisfies the required properties, we start by showing a technical lemma which essentially formalises a simulation argument
on the relation $\typeeqclass{m,\amap(m,k)}{\apropset}$
with respect to the submodel relation.
By taking submodels as with the $\SabDiamond$
operator, equivalence in $\GML$ is preserved.

\begin{lemma}\label{lemma:sabotage-elimination}
Let
$\pair{\amodel}{\aworld}\,{\typeeqclass{m,\amap(m,k)}{\apropset}}\, \pair{\amodel'}{\aworld'}$
where
$m,k\in \Nat$, $\apropset \subseteq_\fin \varprop$,
$\amodel = \triple{\worlds}{\arelation}{\avaluation}$ and $\amodel' = \triple{\worlds'}{\arelation'}{\avaluation'}$.
Let $\arelation_1 \subseteq \arelation$.
There is $\arelation_1' \subseteq \arelation'$ s.t.\
$\pair{\triple{\worlds}{\arelation_1}{\avaluation}}{\aworld}
\,{\typeeqclass{m,k}{\apropset}}\,
\pair{\triple{\worlds'}{\arelation_1'}{\avaluation'}}{\aworld'}$
and
if $\arelation_1(\aworld) = \arelation(\aworld)$, then
$\arelation_1'(\aworld') = \arelation'(\aworld')$.
\end{lemma}

The proof of Lemma~\ref{lemma:sabotage-elimination} is by induction on $m$.
The last condition about~$\arelation_1(\aworld) = \arelation(\aworld)$
will serve in the proof of Lemma~\ref{lemma:sabotage-equivalent-formula},
as it
allows us to capture the semantics of~$\SabDiamond$, by preserving
the children of the world~$\aworld'$.
In the proof,
we
rely on the properties of g-bisimulations~\cite{DeRijke00}
to
define a binary relation $\leftrightarrow$
between worlds of $\arelation(\aworld)$ and~$\arelation'(\aworld')$.
Every $\aworld_1 \leftrightarrow \aworld_1'$ is such that
$\pair{\amodel}{\aworld_1}
\,{\typeeqclass{m-1,\amap(m-1,k)}{\apropset}}\,
\pair{\amodel'}{\aworld_1'}$.
The operator~$\SabDiamond$ does not necessarily preserve the children of $\aworld_1$ and $\aworld_1'$,
so that the induction hypothesis, naturally defined from the statement of~\Cref{lemma:sabotage-elimination},
is applied on models where the condition $\arelation_1(\aworld_1) = \arelation(\aworld_1)$ may not hold.
We show that for all~$\arelation_1\subseteq \arelation$,
it is possible to construct~$\arelation_1' \subseteq \arelation'$ such that, for all~$\aworld_1 \leftrightarrow \aworld_1'$, $\pair{\triple{\worlds}{\arelation_1}{\avaluation}}{\aworld_1}$
$\typeeqclass{m-1,k}{\apropset}$
$\pair{\triple{\worlds'}{\arelation_1'}{\avaluation'}}{\aworld'_1}$.
The result is then lifted to $\pair{\triple{\worlds}{\arelation_1}{\avaluation}}{\aworld}$
$\typeeqclass{m,k}{\apropset}$\!
$\pair{\triple{\worlds'}{\arelation_1'}{\avaluation'}}{\aworld'}$, again thanks to the properties of the g-bisimulation.

Intuitively, Lemma~\ref{lemma:sabotage-elimination} states that given two models satisfying
the same formulae up to the parameters $m$ and $\amap(m,k)$, we can extract submodels
satisfying the same formulae up to $m$ and $k$ (reduced  graded rank).
This allows us to conclude
that  if $\aformula$ is in \GML, there is some \GML formula equivalent to
$\SabDiamond \aformula$ (Lemma~\ref{lemma:sabotage-equivalent-formula}).
In other words, the  operator $\SabDiamond$ can be eliminated to obtain a $\GML$ formula.
This, together with
\Cref{lemma:SC-CC-tautology} and \Cref{theorem:clean-cut} entail
\modallogicSC $\preceq$ \GML.
\begin{lemma}\label{lemma:sabotage-equivalent-formula}
For every $\aformula \in \GML[m,k,\apropset]$
there is a formula $\aformulabis \in \GML[m,\amap(m,k),\apropset]$
such that \
$\SabDiamond \aformula \equiv \aformulabis$.
\end{lemma}

\subsection{Showing  \modallogicSC $\prec$ \GML with EF games for \modallogicSC}
\label{section-sc-less-gml}
We tackle the problem of showing that
\modallogicSC is strictly less expressive than \GML.
To do so, we adapt the notion of Ehrenfeucht-Fra\"iss\'e games (EF games, in short)~\cite{Libkin04} to \modallogicSC, and use it to design
a \GML formula  that is not expressible in \modallogicSC.
We write $\aformulaerankset{m,s}{\apropset}$ for the set of formulae~$\aformula$ of \modallogicSC having $\md{\aformula} \leq m$, at most $s$ nested
$\separate$, and atomic propositions from ${\apropset \subseteq_{\fin} \varprop}$.
It is easy to see that $\aformulaerankset{m,s}{\apropset}$ is finite up to logical equivalence.

%
%

%
%
We introduce the EF games for \modallogicSC.
A game is played between two players: the \defstyle{spoiler} and the \defstyle{duplicator}.
A game state is a triple made of two 
pointed forests $\pair{\amodel}{\aworld}$ and $\pair{\amodel'}{\aworld'}$ and a rank $\triple{m}{s}{\apropset}$, where $m,s \in \Nat$ and $\apropset \subseteq_{\fin} \varprop$.
The goal of the spoiler is to show that the two models are different. The goal of the duplicator is to counter the spoiler and to show that the two models are similar.
Two models are different
whenever
there is
$\aformula \in \aformulaerankset{m,s}{\apropset}$ that is satisfied by only one of the two models.
The EF games for \modallogicSC are formally defined in ~\Cref{figure:EF-games}.
The exact correspondence between the game and the logic is formalised in \Cref{lemma:EF-game-sound-complete}.

\begin{figure}
  \footnotesize
\begin{algorithmic}
  \itemsep 0cm
\State \textbf{Game on} {[$\pair{\amodel_1 {=} \triple{\worlds_1}{\arelation_1}{\avaluation_1}}{\aworld_1}$, $\pair{\amodel_2 {=} \triple{\worlds_2}{\arelation_2}{\avaluation_2}}{\aworld_2}$, $\triple{m}{s}{\apropset}$]\textbf{.}}
\vspace{2pt}
\hrule
\vspace{4pt}
\State \textbf{if} there is $\avarprop \in \apropset$ s.t. $\aworld_1 \in \avaluation_1(\avarprop)$ iff $\aworld_2 \notin \avaluation_2(\avarprop)$ \text{then} the spoiler wins.
\State \textbf{else} the spoiler chooses $i\, {\in}\, \{1,2\}$ and plays on $\amodel_i$.
The duplicator replies on $\amodel_j$ where $j \neq i$.
The spoiler must choose one of the following moves, otherwise the duplicator wins:
  \State {\textbf{modal move}}: if $m \geq 1$ and $\arelation_i(\aworld_i) \neq \emptyset$ then the spoiler \textbf{can} choose to play a modal move by selecting an element $\aworld_i' \in \arelation_i(\aworld_i)$. Then,
      \State \hskip0.5em \labelitemi\ the duplicator must reply with a $\aworld_j' \in \arelation_j(\aworld_j)$ (else, the spoiler wins);
      \State \hskip0.5em \labelitemi\ the game continues on [$\pair{\amodel_1}{\aworld_1'}$, $\pair{\amodel_2}{\aworld_2'}$, $\triple{m-1}{s}{\apropset}$].
  \State {\textbf{spatial move}}: if $s \geq 1$ then the spoiler \textbf{can} choose to play a spatial move by selecting two finite forests $\amodel_i^1$ and $\amodel_i^2$ s.t.\ $\amodel_i^1 + \amodel_i^2 = \amodel_i$. Then,
      \State \hskip0.5em \labelitemi\ the duplicator replies with two forests $\amodel_j^1$ and $\amodel_j^2$ s.t.\ $\amodel_j^1 + \amodel_j^2 = \amodel_j$;
      \State \hskip0.5em \labelitemi\ The game continues on [$\pair{\amodel_1^k}{\aworld_1}$, $\pair{\amodel_2^k}{\aworld_2}$, $\triple{m}{s-1}{\apropset}$], where
      \State \hskip1.4em $k \in \{1,2\}$ is chosen by the spoiler.
\end{algorithmic}
\hrule
\caption{Ehrenfeucht-Fra\"iss\'e games for \modallogicSC}\label{figure:EF-games}
\end{figure}
Using the standard definitions in~\cite{Libkin04},
the duplicator has a \defstyle{winning strategy} for the game
$\triple{\pair{\amodel}{\aworld}}{\pair{\amodel'}{\aworld'}}{\triple{m}{s}{\apropset}}$
if she can play in a way that guarantees her to win regardless how the spoiler plays.
When this is the case, we write $\pair{\amodel}{\aworld}\,{\gamerel{m,s}{\apropset}}\, \pair{\amodel'}{\aworld'}$.
Similarly, the spoiler has a \defstyle{winning strategy}, written
$\pair{\amodel}{\aworld} {\neggamerel{m,s}{\apropset}} \pair{\amodel'}{\aworld'}$,
if he can play in a way that guarantees him to win, regardless how the duplicator plays.
%
%
\Cref{lemma:EF-game-sound-complete} guarantees that the games are well-defined.

\begin{lemma}\label{lemma:EF-game-sound-complete}
$\pair{\amodel}{\aworld} {\neggamerel{m,s}{\apropset}} \pair{\amodel'}{\aworld'}$ iff there is a formula $\aformula$ in $\aformulaerankset{m,s}{\apropset}$ such that
$\amodel,\aworld \models \aformula$ and \ $\amodel',\aworld' \not \models \aformula$.
\end{lemma}

\Cref{lemma:EF-game-sound-complete} is proven with standard arguments from~\cite{Libkin04},
for instance the left-to-right direction, i.e.\ the \emph{completeness of the game},
is by induction on the rank $\triple{m}{s}{\apropset}$.
%
\cut{
\Cref{lemma:EF-game-sound-complete} is proven with standard arguments from~\cite{Libkin04}. The right-to-left direction, i.e.\ the \emph{soundness of the game}, is by structural induction on $\aformula$. The other direction, i.e.\ the \emph{completeness of the game}, is by induction on the rank $\triple{m}{s}{\apropset}$.
%
%
}
Thanks to the EF games,  we are able  to find a \GML formula $\aformula$ that is not expressible in
\modallogicSC. By~\Cref{lemma:SC-CC-equiv-modal-rank-1} and as ${\modallogicCC \approx \GML}$, such a formula is necessarily of
modal degree at least 2. Happily, $\aformula = \Gdiamond{=2}\Gdiamond{=1}\true$ does the job and cannot be expressed in \modallogicSC.
For the proof, we show that for every rank $\triple{m}{s}{\apropset}$, there are two structures $\pair{\amodel}{\aworld}$ and $\pair{\amodel'}{\aworld'}$ such that
$\pair{\amodel}{\aworld} \gamerel{m,s}{\apropset} \pair{\amodel'}{\aworld'}$,
$\amodel,\aworld \models \aformula$ and ${\amodel',\aworld' \not\models \aformula}$.
The inexpressibility of~$\aformula$ then stems from~\Cref{lemma:EF-game-sound-complete}.
The two structures are represented below ($\pair{\amodel}{\aworld}$ on the left).
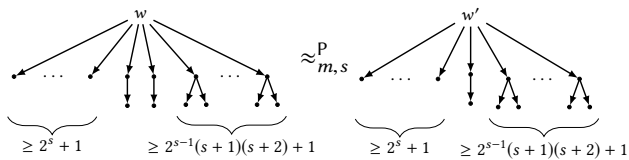
\begin{center}
\scalebox{0.68}{
\begin{tikzpicture}[baseline=-1cm]
\node (w) {$\aworld$};
\node[dot] (w1) [below left = 1cm and 2.2cm of w] {};
\node (h1) [right = 0.4cm of w1] {$\dots$};
\node[dot] (w2) [right = 0.3cm of h1] {};

\node (h2) [below = 0.9cm of w] {};
\node[dot] (w3) [left = 0.1cm of h2] {};
\node[dot] (w4) [right = 0.1cm of h2] {};

\node[dot] (w5) [below right = 1cm and 2.2cm of w] {};
\node (h3) [left = 0.3cm of w5] {$\dots$};
\node[dot] (w6) [left = 0.3cm of h3] {};

\node[dot] (ww1) [below = 0.5cm of w3] {};
\node[dot] (ww2) [below = 0.5cm of w4] {};

\node[dot] (ww3) [below left = 0.5cm and 0.15cm of w5] {};
\node[dot] (ww4) [below right = 0.5cm and 0.15cm of w5] {};

\node[dot] (ww5) [below left = 0.5cm and 0.15cm of w6] {};
\node[dot] (ww6) [below right = 0.5cm and 0.15cm of w6] {};

\draw[pto] (w) -- (w1);
\draw[pto] (w) -- (w2);

\draw[pto] (w) -- (w3);
\draw[pto] (w) -- (w4);

\draw[pto] (w) -- (w5);
\draw[pto] (w) -- (w6);

\draw[pto] (w3) -- (ww1);
\draw[pto] (w4) -- (ww2);

\draw[pto] (w5) -- (ww3);
\draw[pto] (w5) -- (ww4);

\draw[pto] (w6) -- (ww5);
\draw[pto] (w6) -- (ww6);

\node (h4) [below = 0.5cm  of w1] {};
\node (h5) [below = 0.5cm of w2] {};

\node (h8) [below left = 0.5cm and 0.1cm of w6] {};
\node (h9) [below right = 0.5cm and 0.1cm of w5] {};

\draw [decorate,decoration={brace,amplitude=10pt}]
(h5.south east) -- node[below = 0.3cm]
{{$\geq 2^s + 1$}} (h4.south west);

\draw [decorate,decoration={brace,amplitude=10pt}]
(h9.south east) -- node[below = 0.3cm]
{{$\geq 2^{s-1}(s+1)(s+2)+1$}} (h8.south west);
\end{tikzpicture}
}%
\!\!\!\!\!\!$\gamerel{m,s}{\apropset}$%
\scalebox{0.68}{
\begin{tikzpicture}[baseline=-1cm]
\node (w) {$\aworld'$};
\node[dot] (w1) [below left = 1cm and 1.8cm of w] {};
\node (h1) [right = 0.4cm of w1] {$\dots$};
\node[dot] (w2) [right = 0.3cm of h1] {};

\node[dot] (w3) [below = 0.9cm of w] {};

\node[dot] (w5) [below right = 1cm and 1.8cm of w] {};
\node (h3) [left = 0.3cm of w5] {$\dots$};
\node[dot] (w6) [left = 0.3cm of h3] {};

\node[dot] (ww1) [below = 0.5cm of w3] {};

\node[dot] (ww3) [below left = 0.5cm and 0.15cm of w5] {};
\node[dot] (ww4) [below right = 0.5cm and 0.15cm of w5] {};

\node[dot] (ww5) [below left = 0.5cm and 0.15cm of w6] {};
\node[dot] (ww6) [below right = 0.5cm and 0.15cm of w6] {};

\draw[pto] (w) -- (w1);
\draw[pto] (w) -- (w2);

\draw[pto] (w) -- (w3);

\draw[pto] (w) -- (w5);
\draw[pto] (w) -- (w6);

\draw[pto] (w3) -- (ww1);

\draw[pto] (w5) -- (ww3);
\draw[pto] (w5) -- (ww4);

\draw[pto] (w6) -- (ww5);
\draw[pto] (w6) -- (ww6);

\node (h4) [below = 0.5cm of w1] {};
\node (h5) [below = 0.5cm of w2] {};

\node (h8) [below left = 0.5cm and 0.1cm of w6] {};
\node (h9) [below right = 0.5cm and 0.1cm of w5] {};

\draw [decorate,decoration={brace,amplitude=10pt}]
(h5.south east) -- node[below = 0.3cm]
{{$\geq 2^s + 1$}} (h4.south west);

\draw [decorate,decoration={brace,amplitude=10pt}]
(h9.south east) -- node[below = 0.3cm]
{{$\geq 2^{s-1}(s+1)(s+2)+1$}} (h8.south west);
\end{tikzpicture}
}
\end{center}

In the following, we say that a world has \emph{type} $i$ if it has $i$ children.
As one can see in the figure above,
children of the current worlds~$\aworld$ and~$\aworld'$ are of three types:
$0$, $1$ or $2$.
When the spoiler performs a spatial move in the game, a world of type $i$
can take, in the submodels, a type between $0$ and $i$.
That is, the number of children of a world weakly monotonically decreases
when taking submodels.
This monotonicity, together with the finiteness of the game,
lead to bounds on the number of
children of each type, over which the duplicator is guaranteed to win.
For instance, the bound for worlds of type $2$ is given by the value $2^s(s+1)(s+2)$,
where~$s$ is the number of spatial moves in the game.
In the two presented pointed forests, one child of type~$0$ and one of type~$2$ are added
with respect to these bounds, so that the duplicator can make up for the different
numbers of children of type~$1$.


\begin{lemma}\label{lemma:GML-more-SC}
\modallogicSC cannot characterise the class of models satisfying the \GML formula $\Gdiamond{=2}\Gdiamond{=1}\true$.
\end{lemma}
%
\noindent
\begin{minipage}{0.75\linewidth}
Notice that \modallogicSC is more expressive than \ML.
Indeed, the formula
$\Diamond \true \separate \Diamond \true$ distinguishes the two models on the right, which are bisimilar and hence indistinguishable in \ML \cite{VanBenthemThesis}.
\end{minipage}%
\begin{minipage}{0.25\linewidth}
\quad\hfill
\begin{tikzpicture}[baseline]
\node[dot] (o) {};
\node[dot] (d) [below = 0.6cm of o] {};

\node (nn) [below left = 0.1cm and 0.13cm of o] {$\not\approx$};

\node[dot] (a) [left = 1cm of o] {};
\node[dot] (b) [below left = 0.6cm and 0.3cm of a] {};
\node[dot] (c) [below right = 0.6cm and 0.3cm of a] {};

\node (h2) [right = 0cm of o] {};
\draw[pto] (o) -- (d);
\draw[pto] (a) -- (b);
\draw[pto] (a) -- (c);
\end{tikzpicture}\hfill\,
\end{minipage}

\noindent
By $\modallogicSC\preceq\GML$, \Cref{lemma:GML-more-SC} and \Cref{theorem:clean-cut}, we conclude.

\begin{theorem}
\label{theorem:expressive-SC}
$\ML \prec \modallogicSC \prec \GML \approx \modallogicCC$.
\end{theorem}

%
\section{\modallogicCC, \modallogicSC and Sister  Logics}
\label{section-other-logics}

Below, we show how our new results on \modallogicCC and \modallogicSC allow us to
make substantial contributions for
sister logics.

\subsection{Static ambient logic}
\label{section-SAL}
Static ambient logic (\fullSAL) is a formalism proposed to reason about spatial properties of
concurrent processes specified in the
ambient calculus~\cite{Cardelli01}.
In~\cite{CCG03}, the satisfiability and validity problems for a very expressive fragment of \fullSAL are shown to be decidable
and conjectured to be in \pspace (see~\cite[Section 6]{CCG03}).
We invalidate this conjecture by showing that the intensional fragment of \fullSAL (see~\cite{Lozes04}), herein denoted \SAL, is already \aexppol-complete.
More precisely, we
design semantically faithful reductions between
\satproblem{\modallogicCC} and \satproblem{\SAL} (in both directions),
leading to the above-mentioned result by \Cref{cor:aexpol}.
\SAL formulae are from
\begin{center}
$
  \aformula \grammardef\ \true \ \mid\ \zero\ \mid\ \aname[\aformula]\ \mid\ \aformula \land \aformula\ \mid \lnot\aformula\ \mid\ \aformula \ambientchop \aformula
$,
\end{center}
where $\aname \in \AP$ is an \emph{ambient name}.
Historically, the semantics of \fullSAL is given on a class of syntactically defined finite trees. However, 
this class of models is isomorphic to the class of finite
trees $\amodel = \triple{\worlds}{\arelation}{\avaluation}$,
such that
each world in $\worlds$ satisfies exactly one atomic proposition (its ambient name).
Then, the satisfaction relation $\models$ for \SAL is standard for $\true$ and Boolean connectives,  $\aformula_1 \chopop \aformula_2$ is as in \modallogicCC, and otherwise
\begin{center}
$\begin{array}{l@{\, }c@{\ }l}
\amodel, \aworld \models \zero & \iff & \arelation(\aworld) = \emptyset;\\
\amodel, \aworld \models \aname[\aformula] & \iff & \text{there is } \aworld' \in \worlds \text{ such that}\ \arelation(\aworld) = \{\aworld'\},\\
& & \aworld' \in \avaluation(\aname) \text{ and } \amodel,\aworld' \models \aformula.
\end{array}
$
\end{center}
With such a presentation, \SAL is a fragment of \modallogicCC,
where $\zero$ and $\aname[\aformula]$ correspond to
$\Box \false$ and 
$\Gdiamond{=1}{\!\true} \! \land \! \Diamond( \aname \!\land\! \aformula)$, respectively.
However, to reduce $\satproblem{\SAL}$ to $\satproblem{\modallogicCC}$, we must deal with the constraint on
$\avaluation$ (uniqueness of the ambient name).
Let  $\aformula$ be in \SAL written with the ambient names in $\apropsetN \!=\! \{\aname_1,\ldots,\aname_m\}$. It is known (see~\cite[Lemma 8]{CCG03}) that if $\aformula$ is satisfiable, then it can be satisfied by a tree having ambient names from $\apropsetN \cup \{\overline{\aname}\}$,
where $\overline{\aname}$ is a fresh name.
Thus, we can show that $\aformula$ is satisfiable
iff
so is the~\modallogicCC~formula
\begin{center}
$\aformula \land \Boxbox^{\md{\aformula}}(\bigvee_{\aname \in \apropsetN \cup \{\overline{\aname}\}} (\aname \land \bigwedge_{\aname' \in (\apropsetN \cup \{\overline{\aname}\}) \setminus \{\aname\}} \lnot \aname')),$
\end{center}
where the right conjunct
states that
$\avaluation$, restricted to the propositions in $\apropsetN \cup \{\overline{\aname}\}$, forms a partition of the worlds reachable
from the current one in at most $\md{\aformula}$ steps.

Reducing $\satproblem{\modallogicCC}$ to $\satproblem{\SAL}$ requires a bit more work.
Let $\amodel = \triple{\auniverse}{\aaccessrelation}{\apropeval}$ be a finite forest and $\aworld \in \auniverse$.
Assume we want to check the satisfiability status of $\aformula$ in \modallogicCC having atomic propositions from $\apropset = \{\avarprop_1,\dots,\avarprop_m\}$
and with $n$ occurrences of
$\chopop$\,.
We encode $\pair{\amodel}{\aworld}$ into a model $\pair{\amodel' = \triple{\worlds'}{\arelation'}{\avaluation'}}{\aworld}$ of \SAL as follows.
Let $\ambientchild$ and $\ambientprop$ be two ambient names not in $\apropset$.
The ambient name $\ambientchild$ encodes the relation $\aaccessrelation$ whereas $\ambientprop$ can be seen as a
\emph{container} for propositional variables holding on the current world.
{\bfseries \itshape{(i)}} We require $\worlds \subseteq \worlds'$, $\arelation \subseteq \arelation'$ and $\bigcup_{i \in \interval{0}{\md{\aformula}}} \arelation^i(\aworld) \subseteq \avaluation'(\ambientchild)$, 
i.e.,
 every world reachable from $\aworld$ in at most $\md{\aformula}$ steps has the ambient name $\ambientchild$.
Let $\aworld'$ be one of these worlds and suppose that
 $\set{\avarprop  \mid  \aworld' \in \avaluation(\avarprop)}  \cap \apropset = \{\avarpropbis_1 ,\dots, \avarpropbis_l \}$.
{\bfseries \itshape{(ii)}} We require $\worlds'$ to contain
$n+1$ worlds $\aworld_1',\dots,\aworld_{n+1}' \in \arelation'(\aworld') \setminus \arelation(\aworld')$, all having ambient name $\ambientprop$. These worlds encode copies of $\aworld'$'s valuation,
similarly to what is done in \Cref{section-aexppol} 
to encode teams from PL[\plcnot].
{\bfseries \itshape{(iii)}} 
For all $j \in \interval{1}{n+1}$, $\arelation'(\aworld_j')$
contains $l$ worlds, all satisfying $\zero$ and a distinct ambient name from  $\{\avarpropbis_1 ,\dots, \avarpropbis_l \}$.
Below we schematise the encoding (w.r.t. $\aworld'$).

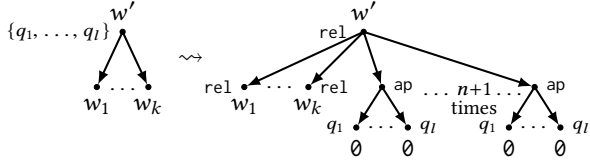
 \begin{center}
   \scalebox{0.98}{
    \begin{tikzpicture}
    \node[dot] (w) [ label=above:$\aworld'$, label=left:\scriptsize{$\{\avarpropbis_1,\ldots,\avarpropbis_l\}$}]{ };

    \node[dot] (w1) [below left = 0.7cm and 0.3cm of w,label=below:$\aworld_1$] {};
    \node [below   =  0.6cm of w] {\scriptsize{$\ldots$}};
    \node[dot] (w3) [below right = 0.7cm and 0.3cm of w,label=below:$\aworld_k$] {};

    \draw[pto] (w) -- (w1);
    \draw[pto] (w) -- (w3);

    \node (arr) [below right = 0.15 and 0.6cm of w] {$\rightsquigarrow$};

    \node[dot] (r) [right= 3.2cm of w, label={above}:$\aworld'$, label=left:{\scriptsize{$\ambientchild$\ \ }}]{ };
    \node[dot] (p3) [below left = 0.7cm and 0.7cm of r,label=below:$\aworld_k$, label=right:{\scriptsize{$\ambientchild$}}] {};

    \node[dot,label={right:{\scriptsize{$\ambientprop$}}}] (pL) [below right = 0.7cm and 0.2cm of r] {};
    \node[dot,label={left:{\scriptsize{$\avarpropbis_1$}}}] (oL) [below left = 0.5cm and 0.3 cm of pL, label=below:$\zero$] {};
    \node[dot,label={right:{\scriptsize{$\avarpropbis_l$}}}] (oLL) [below right = 0.5cm and 0.3 cm of pL, label=below:$\zero$] {};

    \node[dot,label={right:{\scriptsize{$\ambientprop$}}}] (p1) [right = 2cm of pL] {};
    \node[dot,label={left:{\scriptsize{$\avarpropbis_1$}}}] (o1) [below left = 0.5cm and 0.3 cm of p1, label=below:$\zero$] {};
    \node[dot,label={right:{\scriptsize{$\avarpropbis_l$}}}] (o11) [below right = 0.5cm and 0.3 cm of p1, label=below:$\zero$] {};

    \node (aa) [below right = -0.22cm and 0.4cm of pL] {\scriptsize{$\dots$}\footnotesize{\,$n{+}1$\,}\scriptsize{$\dots$}};
    \node (aabis) [below = -5pt of aa] {\footnotesize{times}};

    \node[dot] (p4) [left = 0.8cm of p3, label=below:{$\aworld_1$}, label=left:{\scriptsize{$\ambientchild$}}] {};

    \node [left = 0.02cm of o11] {\scriptsize{$\ldots$}};
    \node [left = 0.02cm of oLL] {\scriptsize{$\ldots$}};
    \node [left = 0.05cm of p3] {\scriptsize{$\ldots$}};

    \node (g1) [below = 0.15cm of o11] {};
    \node (gL) [below = 0.15cm of oL] {};

     \draw[pto] (r) -- (p1);
     \draw[pto] (r) -- (pL);
     \draw[pto] (p1) -- (o1);
     \draw[pto] (pL) -- (oL);
     \draw[pto] (p1) -- (o11);
     \draw[pto] (pL) -- (oLL);
     \draw[pto] (r) -- (p3);
     \draw[pto] (r) -- (p4);

    \end{tikzpicture}}
\end{center}
Let $\aname \in \AP$. We define the modality
$\HMDiamond{\aname} \aformula \egdef \aname[\aformula] \chopop \true$ and its dual
$\HMBox{\aname} \aformula \egdef \lnot \HMDiamond{\aname} \lnot \aformula$.
We write $\forall[\aname]$ for $\lnot((\lnot \zero \land \lnot \aname[\true]) \chopop \true)$, so that $\pair{\amodel}{\aworld} \models \forall[\aname]$ whenever every child of $\aworld$ has the ambient name $\aname$. Moreover,
$\numchildgeq{0} \egdef \true$ and $\numchildgeq{\beta{+}1} \egdef \lnot \zero \chopop \numchildgeq{\beta}$, so
that $\pair{\amodel}{\aworld} \models \numchildgeq{\beta}$ whenever $\aworld$ has at least $\beta$ children. Lastly, $\numchildeq{\beta} \egdef \numchildgeq{\beta} \land \lnot \numchildgeq{\beta{+}1}$.
The models of \SAL encoding models of \modallogicCC are characterised by
\begin{nscenter}
  \scalebox{0.98}{
    $
    \begin{aligned}
    C_{\aformula}\,{\egdef}\!\displaystyle \bigwedge_{\mathclap{\quad j \in \interval{0}{\md{\aformula}}}} \HMBox{\ambientchild}^j
    \Big(
    \forall[\ambientchild]
    \,{\chopop}\,
    \big(\forall[\ambientprop]\,{\land}\,\numchildeq{\mbox{$n$}{+}1}
    \,{\land}\,
    \HMBox{\ambientprop}\big((\avarprop_1[\zero] {\lor} \zero)\chopop \\[-6pt]
    \dots \chopop\, (\avarprop_m[\zero] {\lor} \zero)\big)
    \land
    \textstyle
    \bigwedge_{i \in \interval{1}{m}}
    (
    \HMDiamond{\ambientprop}\HMDiamond{\avarprop_i}\true \implies \HMBox{\ambientprop}\HMDiamond{\avarprop_i} \true)
    \big)
     \Big).
    \end{aligned}
    $
  }
\end{nscenter}
Lastly, we define the translation of
$\aformula$, written $\atranslation(\aformula)$,
into \SAL. It is homomorphic for Boolean connectives and $\true$, $\atranslation(\avarprop)$ $\egdef$ $\HMDiamond{\ambientprop}\HMDiamond{\avarprop}\true$ and otherwise it is inductively defined:
\begin{center}
  $
  \begin{aligned}
  \atranslation(\Diamond \aformulabis) & \egdef \HMDiamond{\ambientchild}\atranslation(\aformulabis);\\[-1pt]
  \atranslation(\aformulabis_1 \chopop \aformulabis_2) & \egdef \big(\atranslation(\aformulabis_1) \land \HMDiamond{\ambientprop}_{\geq j} \true \big)\ \ambientchop \
\big( \atranslation(\aformulabis_2) \land \HMDiamond{\ambientprop}_{\geq {k}} \true \big),
\end{aligned}
$
\end{center}
where in
$\atranslation(\aformulabis_1 \chopop \aformulabis_2)$, $j$ (resp. $k$) is the number of occurrences of~$\chopop$ in $\aformulabis_1$ (resp. $\aformulabis_2$)
plus one and  $\HMDiamond{\ambientprop}_{\geq {\alpha}} \true \egdef (\forall[\ambientprop]\,{\land}\,\numchildeq{\alpha}) \chopop \true$.

We show that $\aformula$ is satisfiable in \modallogicCC iff
$C_{\aformula} \land \atranslation(\aformula)$ is satisfiable in \SAL,
leading to the following results about the complexity of static ambient logics.

\begin{corollary}
\satproblem{\SAL} is \aexppol-complete. \satproblem{\fullSAL} with \fullSAL from~\cite{CCG03} is \aexppol-hard.
\end{corollary}


%

\subsection{Modal separation logic}\label{section:MSL}
The family of \defstyle{modal
separation logics} (\fullMSL), combining
separating and modal connectives, has been recently introduced in~\cite{DemriF19}. Its
models, inspired from the memory states used in separation logic (see also~\cite{Courtault&Galmiche18}),
are Kripke-style structures $\amodel=\triple{\worlds}{\arelation}{\avaluation}$, where
$\worlds=\Nat$ and
$\arelation \subseteq \worlds \times \worlds$ is finite and functional.
Hence, unlike finite forests, $\amodel$ may have loops.

Among the fragments studied in~\cite{DemriF19},
the modal separation logic $\MSL{}{\separate,\Diamondminus}$ was left with a huge complexity gap (between \pspace and \tower).
Its formulae are defined from
\begin{center}
$
\aformula \grammardef\
\avarprop \ \mid \
\Diamond^{-1} \aformula \ \mid \
\aformula \wedge \aformula \ \mid \
\lnot \aformula \ \mid \
\aformula \separate \aformula \
$.
\end{center}
The satisfaction relation
is as in \modallogicSC for $\avarprop\in\AP$, Boolean connectives and  $\aformula_1 \separate \aformula_2$, otherwise
\begin{center}
$\begin{array}{l@{\,}c@{\ }l}
\amodel, \aworld \models \Diamondminus \aformula & \iff & \exists \aworld' \text{ s.t.}\ \pair{\aworld'}{\aworld}\in\arelation \text{ and } \amodel, \aworld' \models \aformula.
\end{array}
$
\end{center}
Since  $\MSL{}{\separate,\Diamondminus}$ is interpreted over a finite and functional relation,
$\Diamondminus$ effectively works as the $\Diamond$ modality of \modallogicSC. Then, assume we want to check the satisfiability of $\aformula$ in \modallogicSC by relying on an algorithm for $\satproblem{\MSL{}{\separate,\Diamondminus}}$.
We simply need to consider the formula $\aformula[\Diamond\,{\gets}\,\Diamondminus]$
obtained from $\aformula$ by replacing every occurrence of $\Diamond$ by $\Diamondminus$, and check if it can be satisfied by a \emph{locally acyclic} model $\pair{\amodel}{\aworld}$ of \fullMSL, i.e.\ one where $\aworld$ does not belong to a loop of length $\leq\md{\aformula}$.
Local acyclicity can be enforced by the formula 
\begin{center}
$\mathtt{locacycl} \egdef r \land\!\bigwedge_{i \in \interval{1}{\md{\aformula}}}(\Box^{-1})^i \lnot r$,
\end{center}
where $r \in \varprop$ is fresh.
Then, $\aformula$ in \modallogicSC is satisfiable iff
$\aformula[\Diamond\,{\gets}\,\Diamondminus] \land \mathtt{locacycl}$
in $\MSL{}{\separate,\Diamondminus}$ is satisfiable.
Hence, the results in~\Cref{section-tower-SC} allow us to close the complexity gap.

\begin{corollary}
$\satproblem{\MSL{}{\separate,\Diamondminus}}$ is \tower-complete.
\end{corollary}

\section{Conclusion}
\label{section-conclusion}

We have studied and compared \modallogicCC and \modallogicSC, two modal
logics interpreted on finite forests and featuring composition operators.
We have not only characterised the expressive power and the complexity
for both logics, but also identified remarkable differences and export our results to other logics.
\modallogicCC is shown as expressive as \GML, and its satisfiability
problem is found to be \aexppol-complete.
Besides the obvious similarities between~\modallogicCC and~\modallogicSC, these results are counter-intuitive:
though the logic $\modallogicSC$ is  strictly less expressive than $\GML$ (and consequently, than \modallogicCC),
\satproblem{$\modallogicSC$} is $\tower$-complete.
We also recalled that there are logspace reductions from \modallogicSC and
\modallogicCC to the second-order modal logic \msokt from~\cite{Bednarczyk&Demri19}.

Our proof techniques go beyond what is known in the literature.
For instance, to design
the \tower-hardness proof we needed substantial modifications from the
proof introduced in~\cite{Bednarczyk&Demri19} for \msokt.
On the other hand, to show the expressivity inclusion of \modallogicSC within \GML, we provided
a novel definition of Ehrenfeucht-Fra\"iss\'e games for~\modallogicSC.

Lastly, our framework led to the characterisation of
 the satisfiability problems for two sister logics
.
We proved that the satisfiability problem for
the modal separation logic $\MSL{}{\separate,\Diamondminus}$
is \tower-complete~\cite{DemriF19}. Moreover,
the satisfiability problem
for the static ambient logic \SAL is \aexppol-complete, solving open problems from~\cite{CCG03,DemriF19} and paving the way to study the complexity of the full \fullSAL.

\clearpage
\section*{Acknowledgements}
We would like to thank the anonymous reviewers for their comments and suggestions that helped
us to improve the quality of the document.
B. Bednarczyk is supported by the Polish Ministry of Science and Higher Education program ``Diamentowy Grant'' no. DI2017 006447.
S. Demri and A. Mansutti are supported by the
Centre National de la Recherche Scientifique (CNRS).
R. Fervari is supported by ANPCyT-PICTs-2017-1130 and 2016-0215, and
by the Laboratoire International Associ\'e SINFIN.


\bibliographystyle{ACM-Reference-Format}
\bibliography{bibliography}

\ifLongVersionWithAppendix
  \newpage
  \appendix
  \onecolumn
  \tableofcontents
  \label{section-appendix}

\section{Proofs of \Cref{section-preliminaries}}
\label{appendix-preliminary-proofs}

We start this appendix by showing a classical property of \ML and \GML which carries over to \modallogicSC and \modallogicCC.
Let $\amodel = \triple{\worlds}{\arelation}{\avaluation}$ be a finite forest and  $\aworld \in \worlds$. We introduce the notation
\begin{nscenter}
  $\arelation|_{\aworld}^{\leq n} \egdef \{\pair{\aworld'}{\aworld''} \in \arelation \mid \aworld' \in \arelation^{i}(\aworld) \text{ for some } i \in \interval{0}{n-1} \}$.
\end{nscenter}
Informally, $\arelation|_{\aworld}^{\leq n}$ is the maximal subset of $\arelation$ encoding  exactly
a subtree rooted at $\aworld$ having only paths of length at most $n$. We denote with
$\arelation|_{\aworld}$ the set $\{\pair{\aworld'}{\aworld''} \in \arelation \mid \aworld' \subseteq \arelation^*(\aworld)\}$, i.e.\ the maximal subset of $R$ encoding exactly a subtree rooted at $\aworld$. Alternatively, $\arelation|_{\aworld} = \bigcup_{n \in \Nat} \arelation|_{\aworld}^{\leq n}$.

\begin{lemma}\label{lemma:height-at-most-n}
Let $n \in \Nat$ and  $\aformula$ be a formula of \modallogicCC or \modallogicSC such that $\md{\aformula} \leq n$.
Let $\amodel = \triple{\worlds}{\arelation}{\avaluation}$ be a finite forest and $\aworld \in \worlds$.
$\amodel , \aworld \models \aformula$ if and only if $\triple{\worlds}{\arelation|_{\aworld}^{\leq n}}{\avaluation}, \aworld \models \aformula$.
\end{lemma}

\begin{proof}
The proof is by structural induction on $\aformula$. As this is the first proof by structural induction of the appendix, we depict also the trivial cases for $\land$ and $\lnot$.
In later proofs, these cases will be omitted (when straightforward) in order to shorten the presentation.
Every case but the ones for $\chopop$ and $\separate$ are from the analogous lemma for \ML.
\begin{description}
\item[Base case: $\aformula = \avarprop$.] This formula only depends on $\aworld$ and $\avaluation$, hence the statement of the lemma trivially holds for these formulae.
\item[Induction case: $\aformula = \lnot \aformulabis$.] The statements below are equivalent.
\begin{itemize}
\item $\amodel,\aworld \models \lnot \aformulabis$
\item $\amodel,\aworld \not\models \aformulabis$ (by definition of $\models$)
\item $\triple{\worlds}{\arelation|_{\aworld}^{\leq n}}{\avaluation},\aworld \not \models \aformulabis$ (by the induction hypothesis, as $\md{\aformulabis} = \md{\lnot \aformulabis} \leq n$)
\item $\triple{\worlds}{\arelation|_{\aworld}^{\leq n}}{\avaluation},\aworld \models \lnot\aformulabis$ (by definition of $\models$).
\end{itemize}
\item[Induction case: $\aformula = \aformulabis \land \aformulater$.] The statements below are equivalent.
\begin{itemize}
\item $\amodel, \aworld \models \aformulabis \land \aformulater$
\item $\amodel, \aworld \models \aformulabis$ and $\amodel, \aworld \models \aformulater$ (by definition of $\models$)
\item $\triple{\worlds}{\arelation|_{\aworld}^{\leq n}}{\avaluation},\aworld \models \aformulabis$ and $\triple{\worlds}{\arelation'}{\avaluation},\aworld \models \aformulater$\\
(by the induction hypothesis, as $\max(\md{\aformulabis},\md{\aformulater}) = \md{\aformulabis \land \aformulater} \leq n$)
\item $\triple{\worlds}{\arelation|_{\aworld}^{\leq n}}{\avaluation},\aworld \models \aformulabis \land \aformulater$ (by definition of $\models$).
\end{itemize}
\item[Induction case: $\aformula = \Diamond \aformulabis$.] The statements below are equivalent.
\begin{itemize}
\item $\amodel, \aworld \models \Diamond \aformulabis$
\item  there is $\aworld_1 \in \arelation(\aworld)$ such that $\amodel, \aworld_1 \models \aformulabis$ (by definition of $\models$)
\item  there is $\aworld_1 \in \arelation(\aworld)$ such that $\triple{\worlds}{\arelation|_{\aworld_1}^{\leq n-1}}{\avaluation}, \aworld_1 \models \aformulabis$\\
(by the induction hypothesis, as $\md{\aformulabis} = \md{\Diamond \aformulabis} - 1 \leq n - 1$)
\item there is $\aworld_1 \in \arelation(\aworld)$ such that $\triple{\worlds}{\arelation|_{\aworld_1}^{\leq n-1} \cup \{\pair{\aworld}{\aworld_1}\}}{\avaluation},\aworld \models \Diamond \aformulabis$ (by definition of $\models$ and by recalling that our models are forests).
\item
$\triple{\worlds}{\arelation|_{\aworld}^{\leq n}}{\avaluation},\aworld \models \Diamond \aformulabis$
(since
$\{\pair{\aworld'}{\aworld''} \in \arelation|_{\aworld}^{\leq n} | \aworld' \in \arelation^*(\aworld_1) \} = \arelation|_{\aworld_1}^{\leq n-1} $).
\end{itemize}
\item[Induction case: $\aformula = \aformulabis \chopop \aformulater$.]  The statements below are equivalent.
\begin{itemize}
\item $\amodel, \aworld \models \aformulabis \chopop \aformulater$
\item  there are $\amodel_1 = \triple{\worlds}{\arelation_1}{\avaluation}$ and $\amodel_2 = \triple{\worlds}{\arelation_2}{\avaluation}$ such that $\amodel_1 +_{\aworld} \amodel_2 = \amodel$,
$\amodel_1, \aworld \models \aformulabis$ and $\amodel_2, \aworld \models \aformulater$ (by definition of $\models$)
\item 
there are $\amodel_1 = \triple{\worlds}{\arelation_1}{\avaluation}$ and $\amodel_2 = \triple{\worlds}{\arelation_2}{\avaluation}$ s.t.\ $\amodel_1 +_{\aworld} \amodel_2 = \amodel$,\\
$\triple{\worlds}{\arelation_1|_{\aworld}^{\leq n}}{\avaluation},\aworld \models \aformulabis$ and
$\triple{\worlds}{\arelation_2|_{\aworld}^{\leq n}}{\avaluation},\aworld \models \aformulater$\\
(by the induction hypothesis, as $\max(\md{\aformulabis},\md{\aformulater}) = \md{\aformulabis \chopop \aformulater} \leq n$)
\item  $\triple{\worlds}{\arelation|_{\aworld}^{\leq n}}{\avaluation},\aworld \models \aformulabis \chopop \aformulater$ (by definition of
$\models$ and as $\arelation|_{\aworld}^{\leq n} = \arelation_1|_{\aworld}^{\leq n} \cup \arelation_2|_{\aworld}^{\leq n}$).
\end{itemize}
\item[Induction case: $\aformula = \aformulabis \separate \aformulater$.]  The statements below are equivalent.
\begin{itemize}
\item $\amodel, \aworld \models \aformulabis \separate \aformulater$
\item  there are $\amodel_1 = \triple{\worlds}{\arelation_1}{\avaluation}$ and $\amodel_2 = \triple{\worlds}{\arelation_2}{\avaluation}$ such that $\amodel_1 + \amodel_2 = \amodel$,
$\amodel_1, \aworld \models \aformulabis$ and $\amodel_2, \aworld \models \aformulater$ (by definition of $\models$)
\item 
there are $\amodel_1 = \triple{\worlds}{\arelation_1}{\avaluation}$ and $\amodel_2 = \triple{\worlds}{\arelation_2}{\avaluation}$ s.t.\ $\amodel_1 + \amodel_2 = \amodel$,\\
$\triple{\worlds}{\arelation_1|_{\aworld}^{\leq n}}{\avaluation},\aworld \models \aformulabis$ and
$\triple{\worlds}{\arelation_2|_{\aworld}^{\leq n}}{\avaluation},\aworld \models \aformulater$\\
(by the induction hypothesis, as $\max(\md{\aformulabis},\md{\aformulater}) = \md{\aformulabis \chopop \aformulater} \leq n$)
\item 
there are $\amodel_1 = \triple{\worlds}{\arelation_1}{\avaluation}$ and $\amodel_2 = \triple{\worlds}{\arelation_2}{\avaluation}$ such that $\amodel_1 + \amodel_2 = \amodel$,\\
$\triple{\worlds}{\arelation_1'}{\avaluation},\aworld \models \aformulabis$ and
$\triple{\worlds}{\arelation_2'}{\avaluation},\aworld \models \aformulater$ where for every $j \in \{1,2\}$
$$
\arelation_j' \egdef \arelation_j \cap  \arelation|_{\aworld}^{\leq n}
$$
(again by the induction hypothesis, right to left direction, as $\arelation_j'|_{\aworld}^{\leq n} = \arelation_j|_{\aworld}^{\leq n}$)
\item iff $\triple{\worlds}{\arelation|_{\aworld}^{\leq n}}{\avaluation},\aworld \models \aformulabis \separate \aformulater$ (by definition of
$\models$ and as $\arelation|_{\aworld}^{\leq n} = \arelation_1'|_{\aworld}^{\leq n} \cup \arelation_2'|_{\aworld}^{\leq n}$). \qedhere
\end{itemize}
\end{description}
\end{proof}

\subsection{Proof of \Cref{lemma:SC-CC-equiv-modal-rank-1}}

\begin{proof}
Let $\amodel = \triple{\worlds}{\arelation}{\avaluation}$ be a finite forest and  $\aworld \in \worlds$.
Notice that if $\md{\aformula}$ is at most $1$, by Lemma~\ref{lemma:height-at-most-n}  the satisfaction of $\aformula$ only depends on the set of worlds $\{\aworld\} \cup \arelation(\aworld)$.
More precisely, $\amodel, \aworld \models \aformula$ iff $\triple{\worlds}{\arelation|_{\aworld}^{\leq 1}}{\avaluation} , \aworld \models \aformula$.
The same holds for formulae in \modallogicSC.
Similarly, $\aformulabis \egdef \aformula[\,\chopop \leftarrow \separate]$ 
(as in the statement) has modal degree at most $1$ and again by Lemma~\ref{lemma:height-at-most-n} we have
 $\amodel, \aworld \models \aformulabis$ iff $\triple{\worlds}{\arelation|_{\aworld}^{\leq 1}}{\avaluation} , \aworld \models \aformulabis$.
To conclude the proof it is sufficient then to prove the following:
\begin{nscenter}
$\triple{\worlds}{\arelation|_{\aworld}^{\leq 1}}{\avaluation} , \aworld \models \aformula$ \ \ iff \ \
$\triple{\worlds}{\arelation|_{\aworld}^{\leq 1}}{\avaluation} , \aworld \models \aformulabis$.
\end{nscenter}
Notice that this result already trivially holds for $\md{\aformula} = 0$. Indeed, in this case the satisfaction of $\aformula$ and $\aformulabis$
only depends on the satisfaction of propositional variables on the current world $\aworld$ and 
therefore not at all on the accessibility relation.
Instead, the proof for $\md{\aformula} = 1$  boils down to the proof of the equivalence
\begin{nscenter}
$\triple{\worlds}{\arelation|_{\aworld}^{\leq 1}}{\avaluation} , \aworld \models \aformula_1 \chopop \aformula_2$ \ \ iff \ \
$\triple{\worlds}{\arelation|_{\aworld}^{\leq 1}}{\avaluation} , \aworld \models \aformula_1 \separate \aformula_2$.
\end{nscenter}
depicted as follows. The statements below are equivalent.
\begin{itemize}
\item $\triple{\worlds}{\arelation|_{\aworld}^{\leq 1}}{\avaluation} , \aworld \models \aformulabis \chopop \aformulater$
\item  there are $\amodel_1 = \triple{\worlds}{\arelation_1}{\avaluation}$ and $\amodel_2 = \triple{\worlds}{\arelation_2}{\avaluation}$ s.t.\ $\amodel_1 +_{\aworld} \amodel_2 = \triple{\worlds}{\arelation|_{\aworld}^{\leq 1}}{\avaluation}$,
$\amodel_1, \aworld \models \aformulabis$ and $\amodel_2, \aworld \models \aformulater$ (by definition of $\models$)
\item  there are disjoint $\arelation_1$ and $\arelation_2$ such that $\arelation_1 \cup \arelation_2 = \arelation|_{\aworld}^{\leq 1}$,
$\triple{\worlds}{\arelation_1}{\avaluation}, \aworld \models \aformulabis$ and $\triple{\worlds}{\arelation_2}{\avaluation}, \aworld \models \aformulater$ (by definition of $+_{\aworld}$, as
$\arelation|_{\aworld}^{\leq 1} = \{\aworld\} \times \arelation(\aworld)$)
\item  there are $\amodel_1 = \triple{\worlds}{\arelation_1}{\avaluation}$ and $\amodel_2 = \triple{\worlds}{\arelation_2}{\avaluation}$ s.t.\ $\amodel_1 + \amodel_2 = \triple{\worlds}{\arelation|_{\aworld}^{\leq 1}}{\avaluation}$,
$\amodel_1, \aworld \models \aformulabis$ and $\amodel_2, \aworld \models \aformulater$ (by definition of $+$)
\item  $\triple{\worlds}{\arelation|_{\aworld}^{\leq 1}}{\avaluation},\aworld \models \aformulabis \separate \aformulater$ (by definition of
$\models$). \qedhere
\end{itemize}
\end{proof}

\subsection{Proof of \Cref{lemma:SC-CC-tautology}}

\begin{proof}
Let $\amodel = \triple{\worlds}{\arelation}{\avaluation}$ be a finite forest and $\aworld \in \worlds$.

For the left to right direction, suppose $\amodel,\aworld \models \aformula \separate \aformulabis$.
Then, by definition of $\models$, there are $\amodel_1 =  \triple{\worlds}{\arelation_1}{\avaluation}$ and $\amodel_2 =  \triple{\worlds}{\arelation_2}{\avaluation}$
such that $\amodel_1 + \amodel_2 = \amodel$, $\amodel_1, \aworld \models \aformula$ and $\amodel_2, \aworld \models \aformulabis$.
By Lemma~\ref{lemma:height-at-most-n} we can easily conclude that $\triple{\worlds}{\arelation_1|_{\aworld}}{\avaluation},\aworld \models \aformula$ and $\triple{\worlds}{\arelation_2|_{\aworld}}{\avaluation},\aworld \models \aformulabis$, where $\arelation|_{\aworld} \egdef \{\pair{\aworld'}{\aworld''} \in \arelation \mid \aworld' \in \arelation^*(\aworld) \}$.
Indeed, this holds as by definition, for every $n \in \Nat$, $(\arelation|_{\aworld})|_\aworld^{\leq n} = \arelation|^{\leq n}_{\aworld}$.
Now, consider the model
$\widehat{\amodel} = \triple{\worlds}{\arelation_1|_{\aworld} \cup \arelation_2|_{\aworld}}{\avaluation}$.
It is easy to see that $\triple{\worlds}{\arelation_1|_{\aworld}}{\avaluation}$ and $\triple{\worlds}{\arelation_2|_{\aworld}}{\avaluation}$ are such that
$\triple{\worlds}{\arelation_1|_{\aworld}}{\avaluation} +_{\aworld} \triple{\worlds}{\arelation_2|_{\aworld}}{\avaluation} = \widehat{\amodel}$.
Hence $\widehat{\amodel},\aworld \models \aformula \chopop \aformulabis$.
Moreover by definition $\arelation_1|_{\aworld} \cup \arelation_2|_{\aworld} \subseteq \arelation$ and $(\arelation_1|_{\aworld} \cup \arelation_2|_{\aworld})(\aworld) = \arelation(\aworld)$.
We conclude that $\amodel,\aworld \models \SabDiamond(\aformula \chopop \aformulabis)$.

For the right to left direction, suppose $\amodel,\aworld \models \SabDiamond(\aformula \chopop \aformulabis)$.
Then by definition of $\models$ there is a model $\widehat{\amodel} = \triple{\worlds}{\widehat{\arelation}}{\avaluation}$ such that $\widehat{\arelation} \subseteq \arelation$, $\widehat{\arelation}(\aworld) = \arelation(\aworld)$ and $\widehat{\amodel}, \aworld \models \aformula \chopop \aformulabis$. Again by definition of $\models$, there are
$\amodel_1 =  \triple{\worlds}{\arelation_1}{\avaluation}$ and $\amodel_2 =  \triple{\worlds}{\arelation_2}{\avaluation}$
such that $\amodel_1 +_{\aworld} \amodel_2 = \widehat{\amodel}$ and $\amodel_1, \aworld \models \aformula$ and $\amodel_2, \aworld \models \aformulabis$.
Consider now the set $\overline{\arelation} = \arelation \setminus \widehat{\arelation}$. We define:
\begin{nscenter}
$
\begin{aligned}[t]
\arelation_1' &\egdef \arelation_1 \cup \{ \pair{\aworld'}{\aworld''} \in \overline{\arelation} \mid \aworld' \not \in \arelation_1^*(\aworld) \}\\
\arelation_2' &\egdef \arelation_2 \cup (\overline{\arelation} \setminus \arelation_1')
\end{aligned}
$
\end{nscenter}
By definition, it is easy to see that $\arelation_1'|_{\aworld} = \arelation_1|_{\aworld}$ and $\arelation_2'|_{\aworld} = \arelation_2|_\aworld$. Moreover, $\arelation_1' \cap \arelation_2' = \emptyset$ and $\arelation_1' \cup \arelation_2' = \arelation$.
Hence, again by using Lemma~\ref{lemma:height-at-most-n} we can easily conclude that $\triple{\worlds}{\arelation_1'}{\avaluation}, \aworld \models \aformula$ and $\triple{\worlds}{\arelation_2'}{\avaluation},\aworld \models \aformulabis$.
From the properties of $\arelation_1'$ and $\arelation_2'$ expressed above, we obtain $\amodel, \aworld \models \aformula \separate \aformulabis$.
\end{proof}

\section{Proofs of \Cref{section-CC}}

\cut{
\subsection{\GML formula $\aformula_1 \wedge \aformula_2$ in good shape}
\label{appendix-good-shape}

The goal of this section is to establish the lemma below.

\begin{lemma}
\label{lemma:good-shape}
Let $\aformulater_1$
and  $\aformulater_2$ be  \GML formulae.
$\aformulater_1 \chopop \aformulater_2$ is logically equivalent
to a finite disjunction of formulae of the form $\aformulater' \chopop \aformulater''$
where $\aformulater' \wedge \aformulater''$ is in good shape.
\end{lemma}

Let $\aformula_1, \aformula_2$ be \GML formulae built over
the propositional variables
$\avarprop_1, \ldots, \avarprop_n$
and a valuation $v: \set{\avarprop_1, \ldots, \avarprop_n} \rightarrow \set{\false,\true}$, we write
$v \leadsto (\aformula_1 \chopop \aformula_2)$ to denote the formula
$\aliteral_1 \wedge \cdots \wedge \aliteral_n \wedge (\redformula_1 \chopop \redformula_2)$
where for all $i \in \interval{1}{n}$, $\aliteral_i \egdef \avarprop_i$ if $v(\avarprop_i) = \true$, otherwise
$\aliteral_i \egdef \neg \avarprop_i$. Moreover, $\redformula_j$ is obtained from $\aformula_j$ by replacing
every occurrence of $\avarprop_i$ that is not in the scope of a graded modality by
$\true$ (resp. by $\false$).
The next lemma establishes that in a formula involving the  connective $\chopop$, we can abstract the propositional
part that is not under the scope of a graded modality.

\begin{lemma}
\label{lemma:elimination-one}
Let  $\aformula_1$, $\aformula_2$ be  in \GML.
Then, we have
$$
\aformula_1 \chopop \aformula_2 \equiv
\bigvee_{v: \set{\avarprop_1, \ldots, \avarprop_n} \rightarrow \set{\false,\true}}
v \leadsto (\aformula_1 \chopop \aformula_2).
$$
\end{lemma}


\begin{proof}
Let $\aformula_1, \aformula_2$ be  \GML formulae
built over the propositional variables $\avarprop_1, \ldots, \avarprop_n$.
Below, we show that
\begin{nscenter}
$
\aformula_1 \chopop \aformula_2 \equiv
\bigvee_{v: \set{\avarprop_1, \ldots, \avarprop_n} \rightarrow \set{\false,\true}}
v \leadsto (\aformula_1 \chopop \aformula_2),
$
\end{nscenter}
where $v \leadsto (\aformula_1 \chopop \aformula_2)$ is defined as
$
\aliteral_1 \wedge \cdots \wedge \aliteral_n \wedge (\redformula_1 \chopop \redformula_2)
$.
By definition,
\begin{itemize}
\item for all $i \in \interval{1}{n}$, $\aliteral_i = \avarprop_i$ if $v(\avarprop_i) = \true$, otherwise
$\aliteral_i = \neg \avarprop_i$.
\item $\redformula_j$ is obtained from $\aformula_j$ by replacing
every occurrence of $\avarprop_i$ that is not in the scope of a graded modality by $v(\avarprop_i)$.
\end{itemize}

First, suppose that $\amodel, \aworld \models \aformula_1 \chopop \aformula_2$.
Let $v: \set{\avarprop_1, \ldots, \avarprop_n} \rightarrow \set{\false,\true}$ be the valuation
such that $v(\avarprop_i) \egdef \true$ if $\aworld \in \avaluation(\avarprop_i)$, otherwise
$v(\avarprop_i) \egdef \false$. Let $\aliteral_1 \wedge \cdots \wedge \aliteral_n$ be the conjunction
of literals such that  for all $i \in \interval{1}{n}$, $\aliteral_i \egdef \avarprop_i$ if $v(\avarprop_i) = \true$, otherwise
$\aliteral_i \egdef \neg \avarprop_i$. We have $\amodel, \aworld \models \aliteral_1 \wedge \cdots \wedge \aliteral_n$.
Indeed,   $\amodel, \aworld \models \avarprop_i$ iff $\aworld \in \avaluation(\avarprop_i)$ (by definition of $\models$)
iff  $v(\avarprop_i) \egdef \true$ (by definition of $v$)
iff $\aliteral_i = \avarprop_i$ (by definition of $\aliteral_i$).
By definition of $\models$, there are $\amodel_1$, $\amodel_2$ such that
$\amodel = \amodel_1 +_{\aworld} \amodel_2$,
$\amodel_1, \aworld \models \aformula_1$ and  $\amodel_2, \aworld \models \aformula_2$.
Let $j \in \set{1,2}$. Let us show that $\amodel_j, \aworld \models \redformula_j$.
The proof is by structural induction on $\aformula_j$ by considering  subformulae $\aformulabis$ of $\aformula_j$.
We show that $\amodel_j, \aworld \models \aformulabis$ iff  $\amodel_j, \aworld \models \redformulabis$
at the top level.
\begin{description}
\item[Base case $\aformulabis = \avarprop_i$.]  If $\aworld \in \avaluation(\avarprop_i)$, then
$\redformulabis = \true$ and therefore
 $\amodel_j, \aworld \models \redformulabis$. Otherwise, $\redformulabis =
\false$
and $\amodel_j, \aworld \not \models \aformulabis$ implies  $\amodel_j, \aworld \not \models \redformulabis$.
\item[Induction case: Boolean connectives.] The proof is by an easy verification.
\item[Induction case: $\aformulabis = \Gdiamond{\geq k} \aformulater$.] We have $\redformulabis = \aformulabis$
and therefore the property trivially holds true.
\end{description}
Consequently, we have  $\amodel_1, \aworld \models \redformula_1$ and $\amodel_2, \aworld \models \redformula_2$,
which is equivalent to $\amodel, \aworld \models \redformula_1 \chopop \redformula_2$ by definition of $\models$.
Hence, we have $\amodel, \aworld \models v \leadsto (\aformula_1 \chopop \aformula_2)$, and therefore we get
$\amodel, \aworld \models \bigvee_{v: \set{\avarprop_1, \ldots, \avarprop_n} \rightarrow \set{\false,\true}}
v \leadsto (\aformula_1 \chopop \aformula_2)$.

Conversely, suppose that $\amodel, \aworld \models \bigvee_{v: \set{\avarprop_1, \ldots, \avarprop_n} \rightarrow \set{\false,\true}}
v \leadsto (\aformula_1 \chopop \aformula_2)$. So, there is $v: \set{\avarprop_1, \ldots, \avarprop_n} \rightarrow \set{\false,\true}$
such that $\amodel, \aworld \models v \leadsto (\aformula_1 \chopop \aformula_2)$, i.e.
$\amodel, \aworld \models \aliteral_1 \wedge \cdots \wedge \aliteral_n \wedge (\redformula_1 \chopop \redformula_2)$
with the above conditions satisfied. By definition of
 $\models$, there are $\amodel_1$, $\amodel_2$ such that
$\amodel = \amodel_1 +_{\aworld} \amodel_2$,
$\amodel_1, \aworld \models \redformula_1$ and  $\amodel_2, \aworld \models \redformula_2$.
Using the fact that $\amodel, \aworld \models \aliteral_1 \wedge \cdots \wedge \aliteral_n$, as above,
the proof is by structural induction on $\aformula_j$ ($j \in \set{1,2}$) by considering  subformulae $\aformulabis$ of $\aformula_j$.
We can show that $\amodel_j, \aworld \models \redformulabis$ iff  $\amodel_j, \aworld \models \aformulabis$
at the top level.  We only provide details for the base case $\aformulabis = \avarprop_i$, for which we have two cases:
\begin{itemize}
    \item If $\aliteral_i = \avarprop_i$, then $\aworld \in \avaluation(\avarprop_i)$ and $\redformulabis = \true$. Hence, $\amodel_j, \aworld \models \avarprop_i$.
    \item If $\aliteral_i = \neg \avarprop_i$, then $\aworld \not \in \avaluation(\avarprop_i)$ and $\redformulabis =
 \false$. Then, $\amodel_j, \aworld \models\neg  \avarprop_i$.
\end{itemize}
By omitting the cases for the induction step, we get that
$\amodel_1, \aworld \models \aformula_1$ and  $\amodel_2, \aworld \models \aformula_2$
and by definition of $\models$, we obtain $\amodel, \aworld \models \aformula_1 \chopop \aformula_2$.
\end{proof}

Notice that for each formula of the form $v \leadsto (\aformula_1 \chopop \aformula_2)$,
the occurrences of the propositional variables in the scope of $\chopop$ are also in the scope of
some graded modality. We present another translation so that for  $\Gdiamond{\geq k} \aformulabis$
and $\Gdiamond{\geq k'} \aformulabis'$ in some $\aformula_i$ with $\aformulabis \neq \aformulabis'$,
$\aformulabis \wedge \aformulabis'$ is unsatisfiable. Essentially, this property will be crucial in order to simulate the separation principle given by $\chopop$: $\aformulabis\wedge\aformulabis'$ unsatisfiable means that $\aformulabis$ and $\aformulabis'$ are exact characterisations of two disjoint subtrees.
\begin{lemma}
\label{lemma:elimination-two}
Let  $\aformula_1$ and $\aformula_2$ be formulae in \GML such that $\maxpc{\aformula_1 \wedge \aformula_2} \subseteq \set{\true,\false}$.
There are formulae $\aformula_1'$ and $\aformula_2'$ in \GML such that $\aformula_1 \chopop \aformula_2 \equiv
\aformula_1' \chopop \aformula_2'$ and $\aformula_1' \chopop \aformula_2'$ is in good shape.
\end{lemma}


\begin{proof} Let $\aformula_1$ and $\aformula_2$ be formulae in \GML such that
 $\maxpc{\aformula_1 \wedge \aformula_2} \subseteq \set{\true,\false}$.
Consequently, for all $i \in \set{1,2}$, $\aformula_i$ is a Boolean combination of formulae from 
$\maxgmod{\aformula_i}$. 

Let $\set{\aformulabis_1, \ldots, \aformulabis_n}$ be 
the set $\set{\aformulabis \mid \Gdiamond{\geq k} \aformulabis \in \maxgmod{\aformula_1 \wedge \aformula_2}}$.
A \defstyle{$\pair{k}{i}$-distribution} is a map 
$\amap: \set{\asetbis \in \powerset{\set{\aformulabis_1, \ldots, \aformulabis_n}} \mid \aformulabis_i \in \asetbis} 
\rightarrow \interval{0}{k}$ such that
$$
(\sum_{\aset \in \set{\asetbis \in \powerset{\set{\aformulabis_1, \ldots, \aformulabis_n}} \mid \aformulabis_i \in \asetbis} } 
\amap(\aset)) = k.
$$
Roughly speaking, when $\amap(\asetbis) = m$, the number of children satisfying all the formulae in $\asetbis$ and not satisfying
the formulae in $\set{\aformulabis_1, \ldots, \aformulabis_n} \setminus \asetbis$, is at least $m$. 
Note that if $\amap$ is a $\pair{k}{i}$-distribution, there are at most $k$ distinct sets $\aset$
such that $\amap(\aset)$ is different from zero. 
Let us define the formula $\amap \leadsto  \Gdiamond{\geq k}{\aformulabis_i}$. In the case $k = 0$, then
$\amap \leadsto \Gdiamond{\geq k}{\aformulabis_i}$ is equal to $\Gdiamond{\geq 0}{(\aformulabis_1 \wedge \cdots \wedge \aformulabis_n)}$.
Otherwise (i.e., $k > 0$), $\amap \leadsto  \Gdiamond{\geq k}{\aformulabis_i}$ is equal to
$$
\bigwedge_{\aset \in  \set{\asetbis \in \powerset{\set{\aformulabis_1, \ldots, \aformulabis_n}} \mid \aformulabis_i \in \asetbis},
\amap(\aset) \neq 0} \Gdiamond{\geq \amap(\aset)}{(\aformulater_1 \wedge \cdots \wedge \aformulater_n)},
$$
where for all $j \in \interval{1}{n}$, if $\aformulabis_j \in \aset$, then $\aformulater_j \egdef \aformulabis_j$, otherwise
$\aformulater_j \egdef \neg \aformulabis_j$. So, $\amap \leadsto  \Gdiamond{\geq k}{\aformulabis_i}$ is a conjunction with at most
$k$ conjuncts.

The formula $\aformula_1' \chopop \aformula_2'$ is obtained from $\aformula_1 \chopop \aformula_2$ by replacing
every occurrence of $\Gdiamond{\geq k}{\aformulabis_i}$ that is not in the scope of a graded modality (i.e., it is maximal) by
$$
\bigvee_{\amap: (k,i){\rm \mbox{-}distribution}} \amap \leadsto \Gdiamond{\geq k}{\aformulabis_i}.
$$
It remains to check that $\aformula_1' \chopop \aformula_2'$ satisfies the announced syntactic properties
as well as $\aformula_1' \chopop \aformula_2' \equiv \aformula_1 \chopop \aformula_2$. Concerning the syntactic restriction, 
assume that maximal  subformulae $\Gdiamond{\geq k}{\aformulabis}$ and $\Gdiamond{\geq k'}{\aformulabis'}$ occur in 
$\aformula_1' \chopop \aformula_2'$ with $\aformulabis \neq \aformulabis'$. Necessarily, 
$\aformulabis$ is of the form $\aformulater_1 \wedge \cdots \wedge \aformulater_n$, 
$\aformulabis'$ is of the form $\aformulater_1' \wedge \cdots \wedge \aformulater_n'$, and for some
$j \in \interval{1}{n}$, we have $\set{\aformulater_j,\aformulater_j'} = \set{\aformulabis_j, \neg \aformulabis_j}$.
Obviously $\aformulabis \wedge \aformulabis'$ is unsatisfiable. 

The fact that $\aformula_1' \chopop \aformula_2' \equiv \aformula_1 \chopop \aformula_2$ relies on the following property. 
Given a set of formulae $\aset = \set{\aformula'_1, \ldots, \aformula'_m}$ such that for all $i \neq j$, $\aformula_i' \wedge \aformula_j'$ is unsatisfiable,
we have:
$$
\Gdiamond{\geq k}{(\aformula_1' \vee \cdots \vee \aformula_m')} 
\equiv
\bigvee_{\amap: \aset \rightarrow \interval{0}{k}, \sum_{i} \amap(\aformula'_i) = k}
(\bigwedge_{i=1}^{m} \Gdiamond{\geq \amap(\aformula_i')} \aformula'_i).
$$
\end{proof}

The proof of Lemma~\ref{lemma:good-shape} is obtained by composing the transformations involved in the respective proofs of
Lemma~\ref{lemma:elimination-one} and Lemma~\ref{lemma:elimination-two}.

}

\subsection{Proof of \Cref{lemma:elimination-three-two}}

Before proving Lemma~\ref{lemma:elimination-three-two}, we establish the lemma below.

\begin{lemma}
\label{lemma:elimination-three-one}
Let  $\aformula_1$, $\aformula_2$ be in \GML such that $\aformula_1 \wedge \aformula_2$
is in good shape.
If there is some quantifier-free  $\aformulater$ equivalent to $[\aformula_1,\aformula_2]^{\PA}$
whose  atomic formulae are of the form
$\avariable_j \geq k$, we have $\aformula_1 \chopop \aformula_2 \equiv
\aformulater^{\GML}$.
\end{lemma}

\begin{proof}
Let  $\aformula_1$ and $\aformula_2$ be formulae in \GML such that 
$\maxpc{\aformula_1 \wedge \aformula_2} \subseteq \set{\true,\false}$ and 
for all  $\Gdiamond{\geq k} \aformulabis$
and $\Gdiamond{\geq k'} \aformulabis'$ in $\maxgmod{\aformula_1 \wedge \aformula_2}$  with $\aformulabis \neq \aformulabis'$,
the formula $\aformulabis \wedge \aformulabis'$ is unsatisfiable, i.e. $\aformula_1 \wedge \aformula_2$ is in good shape. 
Let $\set{\aformulabis_1, \ldots, \aformulabis_n}$ be 
the set $\set{\aformulabis \mid \Gdiamond{\geq k} \aformulabis \in \maxgmod{\aformula_1 \wedge \aformula_2}}$.
By assumption, for all $i \neq j$, the formula $\aformulabis_i \wedge \aformulabis_j$ is unsatisfiable. 

In order to grasp the relationship between $\aformula_i$ and its arithmetical counterpart $\aformula^{\PA}_i$, 
let $\amodel_i = \triple{\worlds_i}{\arelation_i}{\avaluation_i}$ be a model,
$\aworld \in \worlds_i$, and for each $j \in \interval{1}{n}$, let 
$\beta_j^i = \card{\set{\aworld' \in \worlds_i \mid \amodel_i, \aworld' \models \aformulabis_j \ {\rm and} \ \pair{\aworld}{\aworld'} \in \arelation_i}}$. 
Moreover, let $v_{\aworld}: \set{\avariable_1, \ldots, \avariable_n} \rightarrow \Nat$ be the arithmetical valuation
such that $v_{\aworld}(\avariable_j) \egdef \beta_j^i$
for all $j \in \interval{1}{n}$. 
We have the following equivalence
$$
(\dag) \ \ 
\amodel_i, \aworld \models \aformula_i \ \ {\rm iff} \ \ 
v_{\aworld} \models_{\PA} \aformula^{\PA}_i,
$$
where  $\models_{\PA}$ is the satisfaction relation in \PA. 
Below, we also use the notation ``$\aformula^{\PA}_i(\beta_1^i, \ldots, \beta_n^i)$'' instead of 
``$v_{\aworld} \models_{\PA} \aformula^{\PA}_i$''. 

Now, let us show that $\aformula_1 \chopop \aformula_2 \equiv 
\aformulater^{\GML}$. We start by showing that  $\aformula_1 \chopop \aformula_2 \Rightarrow \aformulater^{\GML}$ is valid. 
Let $\amodel = \triple{\worlds}{\arelation}{\avaluation}$ be a model,
$\aworld \in \worlds$ such that $\amodel, \aworld \models \aformula_1 \chopop \aformula_2$.
By definition of $\models$, there are $\amodel_1$, $\amodel_2$ such that 
$\amodel = \amodel_1 +_{\aworld} \amodel_2$,  $\amodel_1, \aworld \models \aformula_1$ and  $\amodel_2, \aworld \models \aformula_2$.
Let us keep the definition of the $\beta_j^i$'s from above, and 
for each $j \in \interval{1}{n}$, let 
$\alpha_j = \card{\set{\aworld' \in \worlds \mid \amodel, \aworld' \models \aformulabis_j \ {\rm and} \ 
\pair{\aworld}{\aworld'} \in \arelation}}$. 
By $(\dag)$ and as $\amodel = \amodel_1 + \amodel_2$ holds too, 
we have the  following relationships:
\begin{nscenter}
$
(j \in \interval{1}{n}) \ \alpha_j = \beta_j^1 + \beta_j^2 \ \ \ \ \ \ 
\aformula^{\PA}_1(\beta_1^1, \ldots,\beta_n^1)  \ \ \ \ \ \ 
\aformula^{\PA}_2(\beta_1^2, \ldots,\beta_n^2).
$
\end{nscenter}
We recall the definition of the arithmetical formula $[\aformula_1,\aformula_2]^{\PA}$: 
$$
[\aformula_1,\aformula_2]^{\PA} \egdef  \exists \ \avariablebis_1^1, \avariablebis_1^2, \ldots, \avariablebis_n^1, \avariablebis_n^2 
 \ (\bigwedge_{j=1}^{n} \avariable_j = \avariablebis_j^1 + \avariablebis_j^2) \wedge
 \aformula^{\PA}_1(\avariablebis_1^1, \ldots,\avariablebis_n^1) \wedge
\aformula^{\PA}_2(\avariablebis_1^2, \ldots,\avariablebis_n^2).
$$
\cut{
\begin{nscenter}
$
\begin{array}{lcl}
[\aformula_1,\aformula_2]^{\PA} & \egdef & \exists \ \avariablebis_1^1, \avariablebis_1^2, \ldots, \avariablebis_n^1, \avariablebis_n^2 
 \ (\bigwedge_{j=1}^{n} \avariable_j = \avariablebis_j^1 + \avariablebis_j^2) \wedge \\ 
 & & \ \ \ \ \aformula^{\PA}_1(\avariablebis_1^1, \ldots,\avariablebis_n^1) \wedge
\aformula^{\PA}_2(\avariablebis_1^2, \ldots,\avariablebis_n^2).
\end{array}
$
\end{nscenter}
}

By assumption, there is a quantifier-free formula $\aformulater$ with free variables among $\avariable_1, \ldots, \avariable_n$
such that $\aformulater$ is logically equivalent to $[\aformula_1,\aformula_2]^{\PA}$ and 
its atomic formulae are of the form $\avariable_j \geq k$.
The formula $\aformulater^{\GML}$ is defined as the \GML formula
obtained from $\aformulater$ by replacing every occurrence of $\avariable_j \geq k$
by $\Gdiamond{\geq k}{\aformulabis_j}$. 
Let  $v_{\aworld}: \set{\avariable_1, \ldots, \avariable_n} \rightarrow \Nat$ be the arithmetical valuation
such that $v_{\aworld}(\avariable_j) \egdef \alpha_j$ for all $j$. Obviously $v_{\aworld} \models_{\PA} \aformulabis^{\PA}$, which is
equivalent to $v_{\aworld} \models_{\PA} \aformulater$. Similarly to $(\dag)$, we can get $\amodel, \aworld
\models \aformulater^{\GML}$. 

Now, we show that $\aformulater^{\GML} \Rightarrow \aformula_1 \chopop \aformula_2$ is valid. 
Let $\amodel = \triple{\worlds}{\arelation}{\avaluation}$ be a model,
$\aworld \in \worlds$ such that $\amodel, \aworld \models \aformulater^{\GML}$. 
As above, for each $j \in \interval{1}{n}$, let 
$\alpha_j = \card{\set{\aworld' \in \worlds \mid \amodel, \aworld' \models \aformulabis_j \ {\rm and} \ 
\pair{\aworld}{\aworld'} \in \arelation}}$. 
Let  $v_{\aworld}: \set{\avariable_1, \ldots, \avariable_n} \rightarrow \Nat$ be the arithmetical valuation
such that $v_{\aworld}(\avariable_j) \egdef \alpha_j$ for all $j$. 
Similarly to $(\dag)$,  we can get $v_{\aworld} \models_{\PA} \aformulater$ and equivalently 
$v_{\aworld} \models_{\PA} [\aformula_1,\aformula_2]^{\PA}$. So, by the semantics of the arithmetical formula $[\aformula_1,\aformula_2]^{\PA}$,
there are natural numbers $\beta_1^1, \beta_1^2, \ldots, \beta_n^1,\beta_n^2$ such that
$$
(j \in \interval{1}{n}) \ \alpha_j = \beta_j^1 + \beta_j^2 \ \ \ \ \ \ 
\aformula^{\PA}_1(\beta_1^1, \ldots,\beta_n^1)  \ \ \ \ \ \ 
\aformula^{\PA}_2(\beta_1^2, \ldots,\beta_n^2).
$$
For each $i \in \set{1,2}$ let us build $\amodel_i$ such that for all $j \in \interval{1}{n}$, 
$\aworld$ has $\beta_j^i$ children in $\amodel_i$, and by construction for each such a child, its whole subtree
in $\pair{\worlds}{\arelation}$ is present in $\pair{\worlds}{\arelation_i}$ too. Such a division is possible
because if a child of $\aworld$ contributes to the value $\alpha_j$ in $\amodel$ (and therefore it satisfies $\aformulabis_j$),
it cannot contribute to any value $\alpha_{j'}$ with $j' \neq j$ (as by assumption $\aformulabis_j \wedge \aformulabis_{j'}$
is unsatisfiable). Hence, by construction $\amodel = \amodel_1 +_{\aworld} \amodel_2$. 
Moreover, for any child $\aworld'$ of $\aworld$ in $\amodel_i$, we have
$\amodel_i, \aworld' \models \aformulabis_j$ iff $\amodel, \aworld' \models \aformulabis_j$ (for all $j \in \interval{1}{n}$)
as the whole subtree of $\aworld'$ in $\amodel$ is present in $\amodel_i$. 
For each $i \in \set{1,2}$, let $v_{\aworld}^i$ be the arithmetical valuation such that
for all $j \in \interval{1}{n}$, we have  $v_{\aworld}^i(\avariable_j) \egdef \beta_j^i$. 
So, obviously, $v_{\aworld}^i \models_{\PA} \aformula^{\PA}_i(\beta_1^i, \ldots,\beta_n^i)$ and therefore
by $(\dag)$, we have $\amodel_i, \aworld \models \aformula_i$. Consequently, we get $\amodel, \aworld \models 
\aformula_1 \chopop \aformula_2$. 
\end{proof}

Condition 2. in the definition of $\aformula_1 \wedge \aformula_2$
in good shape is essential here to obtain  $\aformula_1 \chopop \aformula_2 \equiv
\aformulater^{\GML}$.
Here is a simple counter-example.
The formula $[\aformula_1,\aformula_2]^{\PA}$  obtained from
$\Gdiamond{\geq 1}{\avarprop} \chopop \Gdiamond{\geq 1}{\avarpropbis}$
is  defined as
$
\ \exists \ \avariablebis_1^1, \avariablebis_1^2,  \avariablebis_2^1, \avariablebis_2^2 \
(\avariable_1 = \avariablebis_1^1 + \avariablebis_1^2) \wedge (\avariable_2 = \avariablebis_2^1 + \avariablebis_2^2)
\wedge (\avariablebis_1^1 \geq 1) \wedge (\avariablebis_2^2 \geq 1)
$.
Obviously, $[\aformula_1,\aformula_2]^{\PA}$  is
arithmetically equivalent to $(\avariable_1 \geq 1) \wedge (\avariable_2 \geq 1)$
but
$
\Gdiamond{\geq 1}{\avarprop} \chopop \Gdiamond{\geq 1}{\avarpropbis} \not \equiv \Gdiamond{\geq 1}{\avarprop} \wedge
\Gdiamond{\geq 1}{\avarpropbis}
$.
Indeed, when $\amodel, \aworld \models \Gdiamond{\geq 1}{\avarprop} \wedge \Gdiamond{\geq 1}{\avarpropbis}$
and $\aworld$ has a unique child satisfying $\avarprop \wedge \avarpropbis$, there is no way for $\aworld$ to satisfy
$\Gdiamond{\geq 1}{\avarprop} \chopop \Gdiamond{\geq 1}{\avarpropbis}$.
So the aforementioned assumption is crucial in order to simulate the appropriate partitioning of subtrees.

To prove the result in full generality, we need to establish that such a quantifier-free formula  $\aformulater$ always exists.
Here is the proof of  Lemma~\ref{lemma:elimination-three-two}.

\begin{proof} 
For each $i \in \set{1,2}$, let $\aformula_i'$ be an arithmetical formula
logically equivalent to $\aformula_i^{\PA}$ such that:
\begin{itemize}
\item $\aformula_i'$ is in disjunctive normal form (DNF),
\item each disjunct of $\aformula_i'$ is a conjunction such that for each $j \in \interval{1}{n}$,
the variable $\avariablebis_j^i$ is in at most two literals with the following three options:
     \begin{itemize}
     \item $\avariablebis_j^i$ occurs in a unique literal of the form $\avariablebis_j^i \geq k$,
     \item $\avariablebis_j^i$ occurs in a unique (negative) literal of the form $\neg (\avariablebis_j^i \geq k)$,
     \item $\avariablebis_j^i$ occurs in two  literals whose conjunction is $\avariablebis_j^i \geq k_1 \wedge \neg (\avariablebis_j^i \geq k_2)$
     and $k_2 > k_1$.
     \end{itemize}
\end{itemize}
In the case such a formula  $\aformula_i'$  does not exist, typically when  $\aformula_i'$ is inconsistent, 
$\aformulater$ can simply take the value $\perp$. In the sequel, we assume that both $\aformula_1'$ and $\aformula_2'$ exist.
Using propositional reasoning and the fact that disjunction distributes over
existential first-order quantification, the formula $[\aformula_1,\aformula_2]^{\PA}$ is therefore logically equivalent to a formula
of the form
$$
\bigvee_{\alpha,\beta} 
\exists \ \avariablebis_1^1, \avariablebis_1^2, \ldots, \avariablebis_n^1, \avariablebis_n^2 \
(\bigwedge_{j=1}^{n} \avariable_j = \avariablebis_j^1 + \avariablebis_j^2) \wedge
C_{\alpha}^1 \wedge C_{\beta}^2
$$
where $C_{\alpha}^1$ (resp. $ C_{\beta}^2$) is a conjunction from $\aformula_1'$ (resp. from $\aformula_2'$).  
In order to build $\aformulater$ from $[\aformula_1,\aformula_2]^{\PA}$, 
we take advantage of quantifier elimination in \PA and we explain below how this can be done.
It is sufficient to explain how to eliminate quantifiers for subformulae of the form 
$$
\Psi = 
\exists \ \avariablebis_1^1, \avariablebis_1^2, \ldots, \avariablebis_n^1, \avariablebis_n^2 \
(\bigwedge_{j=1}^{n} \avariable_j = \avariablebis_j^1 + \avariablebis_j^2) \wedge
C_{\alpha}^1 \wedge C_{\beta}^2.
$$
Let $j \in \interval{1}{n}$ and suppose that by performing quantifier elimination on 
$\exists \ \avariablebis_{j+1}^1,\avariablebis_{j+1}^2, \ldots,
\avariablebis_{n}^1,\avariablebis_{n}^2$,  the formula $\Psi$ is equivalent to
$$ 
 \exists \ \avariablebis_1^1, \avariablebis_1^2, \ldots, \avariablebis_j^1, \avariablebis_j^2 \
\Psi_{j+1}.
$$
with $\Psi_{n+1} = (\bigwedge_{j=1}^{n} \avariable_j = \avariablebis_j^1 + \avariablebis_j^2) \wedge
C_{\alpha}^1 \wedge C_{\beta}^2$, 
and,
\begin{enumerate}
\item $\Psi_{j+1}$ is quantifier-free with no occurrences of the variables
      $\avariablebis_{j+1}^1,\avariablebis_{j+1}^2, \ldots,\avariablebis_{n}^1,\avariablebis_{n}^2$,
\item $\Psi_{j+1}$ is of the form 
      $$
      (\bigwedge_{a=1}^{j} \avariable_a = \avariablebis_a^1 + \avariablebis_a^2)
      \wedge D \wedge C'_1 \wedge C'_2
      $$
      where 
      \begin{enumerate}
      \item $D$ is a conjunction of literals built from constraints of the form $\avariable_{j'} \geq k$  with $j' \in \interval{j}{n}$,
      \item for each $i \in \set{1,2}$, $C'_i$ a conjunction such that for each $j' \in \interval{1}{j}$,
       $\avariablebis_{j'}^i$ is in at most two literals with the following three options:
     \begin{itemize}
     \item $\avariablebis_{j'}^i$ occurs in a unique literal of the form $\avariablebis_{j'}^i \geq k$,
     \item $\avariablebis_{j'}^i$ occurs in a unique (negative) literal of the form $\neg (\avariablebis_{j'}^i \geq k)$,
     \item $\avariablebis_{j'}^i$ occurs in two  literals whose conjunction is 
           $\avariablebis_{j'}^i \geq k_1 \wedge \neg (\avariablebis_{j'}^i \geq k_2)$  and $k_2 > k_1$.
     \end{itemize}
      \end{enumerate}
\end{enumerate}
Now, let us show how to perform quantifier elimination of $\exists \ \avariablebis_j^1 \ \exists \ \avariablebis_j^2 \ \Psi_{j+1}$
to preserve the property for $j-1$. First note that  $\exists \ \avariablebis_j^1 \ \exists \ \avariablebis_j^2 \ \Psi_{j+1}$
 is logically equivalent to
$$
(\bigwedge_{a=1}^{j-1} \avariable_a = \avariablebis_a^1 + \avariablebis_a^2) \wedge D 
\wedge C_1'' \wedge C_2'' \wedge 
\  \exists \ \avariablebis_j^1 \ \exists \ \avariablebis_j^2 \ 
(\avariable_j = \avariablebis_j^1 + \avariablebis_j^2) \wedge
D_1 \wedge D_2,
$$
\cut{
$$
\begin{array}{l}
(\bigwedge_{a=1}^{j-1} \avariable_a = \avariablebis_a^1 + \avariablebis_a^2) \wedge D 
\wedge C_1'' \wedge C_2'' \wedge \\ 
\ \ \ \ \exists \ \avariablebis_j^1 \ \exists \ \avariablebis_j^2 \ 
(\avariable_j = \avariablebis_j^1 + \avariablebis_j^2) \wedge
D_1 \wedge D_2,
\end{array}
$$
}
where $C_1' = C_1'' \wedge D_1$ (assuming abusively that $A \wedge \top = A$), $C_2' = C_2'' \wedge D_2$ and each variable $\avariablebis_j^i$ does not occur in 
$C_i''$, and each $D_i$ is either $\true$, or contains at most 2 literals involving the variable $\avariablebis_j^i$. It is then
easy to eliminate quantifiers in $\exists \ \avariablebis_j^1 \ \exists \ \avariablebis_j^2 \ 
(\avariable_j = \avariablebis_j^1 + \avariablebis_j^2)  \wedge 
D_1 \wedge D_2$ and below we treat all the cases depending on the value for $ D_1 \wedge D_2$ leading to
the formula $D_{12}$ 
(we omit the symmetrical cases):
\begin{itemize}
\itemsep 0 cm 
\item $\true \wedge \true$:  $D_{12} \egdef \true$,
\item $(\avariablebis_j^1 \geq k) \wedge \true$:
      $D_{12} \egdef (\avariable_j \geq k)$,
\item $\neg (\avariablebis_j^1 \geq k) \wedge \true$: 
      $D_{12} \egdef \true$,
\item $(\avariablebis_j^1 \geq k) \wedge \neg (\avariablebis_j^1 \geq k') 
      \wedge \true$: 
      $D_{12} \egdef (\avariable_j \geq k)$,
\item $(\avariablebis_j^1 \geq k) \wedge (\avariablebis_j^2 \geq k'')$: 
      $D_{12} \egdef (\avariable_j \geq k+k'')$,
\item $\neg (\avariablebis_j^1 \geq k) \wedge (\avariablebis_j^2 \geq k'')$:  
      $D_{12} \egdef (\avariable_j \geq k'')$,
\item $(\avariablebis_j^1 \geq k) \wedge \neg (\avariablebis_j^1 \geq k') 
      \wedge (\avariablebis_j^2 \geq k'')$: 
      $D_{12} \egdef (\avariable_j \geq k+k'')$,
\item $(\avariablebis_j^1 \geq k) \wedge \neg (\avariablebis_j^1 \geq k') 
      \wedge (\avariablebis_j^2 \geq k'') \wedge \neg (\avariablebis_j^2 \geq k''')$: 
      $D_{12} \egdef (\avariable_j \geq k+k'') \wedge \neg (\avariable_j \geq k'+k''')$.
\end{itemize}
It is now easy to check that the
formula
$$
 \exists \ \avariablebis_1^1, \avariablebis_1^2, \ldots, \avariablebis_{j-1}^1, \avariablebis_{j-1}^2 \
(\bigwedge_{a=1}^{j-1} \avariable_a = \avariablebis_a^1 + \avariablebis_a^2) \wedge (D \wedge D_{12}) 
\wedge C_1'' \wedge C_2'',
$$
satisfies the conditions for $\Psi_{j}$. 
By iterating the process of  quantifier elimination, we get the desired formula $\aformulater$. 
By Lemma~\ref{lemma:elimination-three-one}, we conclude that  $\aformula_1 \chopop \aformula_2 \equiv
\aformulater^{\GML}$.
\end{proof}

\subsection{Proof of  \Cref{theorem:clean-cut}}
\begin{proof} Let $\aformula$ be a formula in \modallogicCC. As $\Diamond \aformulabis \equiv \Gdiamond{\geq 1}{\aformulabis}$, we can assume
that the only modalities in $\aformula$ are of the form $\Gdiamond{\geq 1}{}$ or $\chopop$. If $\aformula$ has no occurrence of $\chopop$, we are done.
Otherwise, let $\aformulabis$ be a subformula of $\aformula$ whose outermost connective is $\chopop$ and the arguments are in \GML,
say $\aformulabis = \aformula_1 \chopop \aformula_2$. 
By Lemma~\ref{lemma:f-one},
there is a formula $\aformulabis'$ in \GML such that $\aformula_1 \chopop \aformula_2 \equiv \aformulabis'$.
One can show that $\aformula \equiv \aformula[\aformulabis \leftarrow \aformulabis']$, where  $\aformula[\aformulabis \leftarrow \aformulabis']$
is obtained from $\aformula$ by replacing every occurrence of $\aformulabis$ by $\aformulabis'$. Note that the number of occurrences of $\chopop$ in 
$\aformula[\aformulabis \leftarrow \aformulabis']$ is strictly less than  the number of occurrences of $\chopop$ in 
$\aformula$. By repeating such a type of replacement, eventually we obtain a formula $\aformula'$ in \GML such that $\aformula \equiv \aformula'$. 
\end{proof}

\subsection{$\GML$ is closed under the operator $\chopop$}\label{proof:lemma-f-one}

Given $\aformula \in \GML$, we write $\submax{\aformula}$ to denote
the set $\set{\aformulater \mid \Gdiamond{\geq k} \aformulater \in \maxgmod{\aformula}}$.

\begin{lemma}
\label{lemma:f-one}
Let $\aformulabis_1$ and $\aformulabis_2$ be two formulae in \GML with
$\maxgmod{\aformulabis_1} \cup \maxgmod{\aformulabis_2} = \set{\Gdiamond{\geq k_1} \aformulater_1, \ldots,
     \Gdiamond{\geq k_n} \aformulater_n}$ and $\widehat{k} = \max \set{k_1, \ldots, k_n}$.
There is a \GML
formula $\aformulabis$ such that
$\aformulabis\,{\equiv}\,\aformulabis_1 \chopop \aformulabis_2$,
$\newbd{0}{\aformulabis}\,{\leq}\, \widehat{k} \, 2^{n+1}$ and
$\newbd{1}{\aformulabis}\,{\leq}\, n \, \newbd{1}{\aformulabis_1 {\land} \aformulabis_2}$.
\end{lemma}

\begin{proof} Without loss of generality, we assume that  $\submax{\aformulabis_1}  = \submax{\aformulabis_2}$.
Otherwise,  if $\aformulater \in \submax{\aformulabis_j} \setminus \submax{\aformulabis_{3-j}}$, then we add to
 $\aformulabis_{3-j}$ the conjunct $\Gdiamond{\geq 0}{\aformulater}  \vee \neg (\Gdiamond{\geq 0}{\aformulater})$, and we repeat the process
until $\submax{\aformulabis_1} = \submax{\aformulabis_2}$.
Moreover, we assume that the propositional variables not in the scope of a modality are among $\avarprop_1$, \ldots, $\avarprop_{\alpha}$.

In order to compute $\aformulabis$, we perform the following steps.

\begin{enumerate}
\item For each $i \in \set{1,2}$, let $\hat{\aformulabis_i}$ be a formula logically equivalent to $\aformulabis_i$
such that $\hat{\aformulabis_i}$ is in disjunctive normal form (DNF) with respect to the atoms in  $\maxgmod{\aformulabis_i} \cup \set{\avarprop_1, \ldots, \avarprop_{\alpha}}$.
Assume that $\maxgmod{\aformulabis_i} = \set{\Gdiamond{\geq k_1}{\aformulater_1}, \ldots,
     \Gdiamond{\geq k_{n'}}{\aformulater_{n'}}}$ with
    $\set{\aformulater_1^{\star}, \ldots,
     \aformulater_n^{\star}} = \set{\aformulater_1, \ldots,
     \aformulater_{n'}}$, i.e. some $\aformulater_i^{\star}$ may occur more than once but with
     different graded rank.

Let $\Bool = \set{0,1}$. Given a formula $\aformula$, we write $\aformula^1$ for 
$\aformula$ and $\aformula^0$ for $\neg \aformula$.
Hence, the  formula $\hat{\aformulabis_i}$ satisfies
$$
\hat{\aformulabis_i} \subseteq \bigvee_{\amap \colon \interval{1}{n'+\alpha} \rightarrow \Bool} ((\Gdiamond{\geq k_1}{\aformulater_1})^{\amap(1)} \wedge \cdots \wedge
(\Gdiamond{\geq k_{n'}}{\aformulater_{n'}})^{\amap(n')}) \wedge (\avarprop_1^{\amap(n'+1)} \wedge \cdots \wedge \avarprop_{\alpha}^{\amap(n'+\alpha)}),
$$
where the relation $\subseteq$ in that context means that $\hat{\aformulabis_i}$ is subdisjunction of the generalised disjunction on the right-hand side.
Note that $\newbd{0}{\aformulabis_i} = \newbd{0}{\hat{\aformulabis_i}}$.
\item The second step consists in partitioning the modalities so that $\tilde{\aformulabis_i}$ is obtained from
$\hat{\aformulabis_i}$  by replacing any occurrence of $(\Gdiamond{\geq k_j}{\aformulater_j})^{\amap(j)}$ by
$$
(\bigvee_{\amapbis \colon \interval{1}{k_j} \rightarrow \set{\asetbis \mid \aformulater_j \in \asetbis \ {\rm and} \ \asetbis
\subseteq \set{\aformulater_1^{\star}, \ldots, \aformulater_n^{\star}}}} \ \ \ \bigwedge_{\asetbis \in \range{\amapbis}} \Gdiamond{\geq \card{\amapbis^{-1}(\asetbis)}}{
(\asetbis \wedge \bar{\asetbis})})^{\amap(j)},
$$
where $\asetbis$ stands for $\bigwedge_{\aformulabis \in \asetbis} \aformulabis$ and
 $\bar{\asetbis}$ stands for $\bigwedge_{\aformulabis \in (\set{\aformulater_1^{\star}, \ldots, \aformulater_n^{\star}} \setminus \asetbis)} \neg \aformulabis$.
\end{enumerate}
It is easy to check that $\aformulabis_i \equiv \hat{\aformulabis_i}$ and $\hat{\aformulabis_i} \equiv \tilde{\aformulabis_i}$.
We write $\hat{\tilde{\aformulabis_i}}$ to denote $\tilde{\aformulabis_i}$ in DNF of the form below
$$
\hat{\tilde{\aformulabis_i}} \subseteq \bigvee_{\amap \colon \interval{1}{n''+\alpha} \rightarrow \Bool} ((\Gdiamond{\geq l_1}{\aformulater_1^{\star \star}})^{\amap(1)} \wedge \cdots \wedge
(\Gdiamond{\geq l_{n''}}{\aformulater_{n''}^{\star \star}})^{\amap(n'')}) \wedge (\avarprop_1^{\amap(n''+1)} \wedge \cdots \wedge \avarprop_{\alpha}^{\amap(n''+\alpha)}),
$$
with $l_i \leq \widehat{k}$ (because $\card{\amapbis^{-1}(\asetbis)}$ above is always bounded
by $\widehat{k}$), and there are at most $2^n$ distinct $\aformulater_j^{\star \star}$.
Consequently, $\hat{\tilde{\aformulabis_1}} \chopop \hat{\tilde{\aformulabis_2}}$ is logically equivalent to a disjunction of the form
below as the disjunction distributes over the composition operator:
$$
\bigvee
\Big(
(\avarprop_1^{\amap(n''+1)} \wedge \cdots \wedge \avarprop_{\alpha}^{\amap(n''+\alpha)}) \wedge
(\avarprop_1^{\amap'(n''+1)} \wedge \cdots \wedge \avarprop_{\alpha}^{\amap'(n''+\alpha)}) \wedge
$$
$$
\big((\Gdiamond{\geq l_1}{\aformulater_1^{\star \star}})^{\amap(1)} \wedge \cdots \wedge
(\Gdiamond{\geq l_{n''}}{\aformulater_{n''}^{\star \star}})^{\amap(n'')}
\chopop
(\Gdiamond{\geq l_1}{\aformulater_1^{\star \star}})^{\amap'(1)} \wedge \cdots \wedge
(\Gdiamond{\geq l_{n''}}{\aformulater_{n''}^{\star \star}})^{\amap'(n'')}
\big)
\Big).
$$
Observe that $(\avarprop \wedge \aformulabis) \chopop (\avarprop' \wedge \aformulabis')$ is logically equivalent
to $\avarprop \wedge \avarprop' \wedge (\aformulabis \chopop \aformulabis')$.
By Lemma~\ref{lemma:elimination-three-two}, the subformula with outermost connective $\chopop$ can be rewritten
as a \GML formula $\aformula$ with  graded rank at most twice the maximal graded rank (i.e. $2 \times \widehat{k}$)
and with $\card{\submax{\aformula}} \leq 2^n$. Note that the condition of being in good shape is guaranteed by
construction of $\tilde{\aformulabis_i}$. The formula $\aformulabis$ is obtained by applying
Lemma~\ref{lemma:elimination-three-two} on the large disjunction above as much as needed.
It is now easy to check that $\newbd{0}{\aformulabis} \leq (2 \times \widehat{k}) \times 2^n$.
\end{proof}

\subsection{Proof of \Cref{lemma:small-model-newbd}}

\begin{proof} 
The proof is by induction on the modal degree of $\aformula$ and we show that
the branching degree of the models is at most $\maxbd{\aformula}$ (which allows us
to get the number of worlds at most $\maxbd{\aformula}^{\md{\aformula}+1}$
as only nodes reachable in at most $\md{\aformula}$ steps are relevant for satisfaction). 
The base case with $\md{\aformula} =0$
is by an easy verification as then $\maxbd{\aformula} = \newbd{0}{\aformula} = 0$ 
and therefore satisfaction of $\aformula$ can be witnessed
on a single node model. For the induction step, let us suppose that for all formulae $\aformulabis$ of modal
depth less than $d$, if $\aformulabis$ has a model then it has model in which  each node has at 
most $\maxbd{\aformulabis}$ children.

Let $\aformula$ be a satisfiable formula in \GML of modal depth $d+1$.
Let $\maxgmod{\aformula} = \set{\aformulabis_1, \ldots,\aformulabis_n}$ and
$\avarprop_1, \ldots, \avarprop_m$ be
the propositional variables in $\aformula$ that are not in the scope of a graded modality.
We write $\DNF{\aformula}$ to denote the set of formulae in disjunctive normal form
logically equivalent to $\aformula$ with atomic formulae among
$\set{\aformulabis_1, \ldots,\aformulabis_n,\avarprop_1, \ldots, \avarprop_m}$.
We exclude from $\DNF{\aformula}$ the conjunctions and disjunctions with repetitions as well as
conjunctions that do not respect the  conditions below to avoid obvious inconsistencies.
Typically, the conjunctions are of the form  (modulo AC and without repetitions)
$$
\Gdiamond{\geq k_1}{\aformula_1} \wedge \ldots \wedge \Gdiamond{\geq k_m}{\aformula_m}
\wedge
\neg \Gdiamond{\geq k_1'}{\aformula_1'} \wedge \ldots \wedge \neg \Gdiamond{\geq k_m'}{\aformula_{m'}'}
\wedge
\aliteral_1 \wedge \cdots \wedge \aliteral_{m''},
$$
where the $\aliteral_i$'s are literals built over $\avarprop_1, \ldots, \avarprop_m$.  Without loss of generality,
we assume that if $\aformula_i = \aformula'_{j}$, then $k'_j > k_i$ and there are no contradictory
literals in $\aliteral_1 \wedge \cdots \wedge \aliteral_{m''}$.

Let $\aformula' \in \DNF{\aformula}$.
As $\aformula'$ is satisfiable too, there is a conjunction $\aformula''$ in $\aformula'$  that is satisfiable, say of the form below:
$$
\aformula'' =
\Gdiamond{\geq k_1}{\aformula_1} \wedge \ldots \wedge \Gdiamond{\geq k_m}{\aformula_m}
\wedge
\neg \Gdiamond{\geq k_1'}{\aformula_1'} \wedge \ldots \wedge \neg \Gdiamond{\geq k_m'}{\aformula_{m'}'}
\wedge
\aliteral_1 \wedge \cdots \wedge \aliteral_{m''}.
$$
By definition of  $\bd{0,\aformula'}$, we have $(k_1 + \cdots + k_m) =  \bd{0,\aformula'} \leq \maxbd{\aformula}$.
Let $\amodel = \triple{\worlds}{\arelation}{\avaluation}$ be a model and $\aworld \in \worlds$ such that
$\amodel, \aworld \models \aformula''$.
By definition of $\models$, for each $i \in \interval{1}{m}$, there is a set $\aset_i$ made
of $k_i$ children of $\aworld$ such that each child in  $\aset_i$  satisfies $\aformula_i$.
Let  $\amodel' = \triple{\worlds'}{\arelation'}{\avaluation'}$ be the model such that
$\worlds' \egdef \set{\aworld} \cup \set{\aworld' \mid \aworld' \in \arelation^*(\aworld''),
\aworld'' \in \aset_1 \cup \cdots \cup \aset_m}$,
$\arelation' = \arelation \cap (\worlds' \times \worlds')$ and $\avaluation'$ is the restriction of
$\avaluation$ to $\worlds'$. It is easy to verify that $\amodel', \aworld \models
\aformula''$ and $\aworld$ has at most $(k_1 + \cdots + k_m)$ children in $\amodel'$.
By the induction hypothesis, for each $i \in \interval{1}{m}$,
there is a model $\amodel_i = \triple{\worlds_i}{\arelation_i}{\avaluation_i}$ and
$\aworld_i \in \worlds_i$ (say $\pair{\worlds_i}{\arelation_i}$ is rooted at $\aworld_i$)
such that $\amodel_i, \aworld_i \models \aformula_i$ and
each node in $\amodel_i$ has at most $\maxbd{\aformula_i}$ children.
As $\aformula_i$ is a subformula of $\aformula$, by definition of $\bd{\aformula}$, we have also
$\maxbd{\aformula_i} \leq \maxbd{\aformula}$.
Let us build the model $\amodel''$ obtained from $\amodel'$ such that for all $i \in \interval{1}{m}$
and for all children $\aworld' \in \aset_i$, we replace the subtree rooted at $\aworld'$ in $\amodel'$
by a copy of $\amodel_i$. It is then easy to verify that $\amodel'', \aworld \models \aformula''$
and each node in $\amodel''$ has at most $\maxbd{\aformula}$ children.
This completes the proof for the induction step.
\end{proof}







\subsection{Proof of \Cref{lemma:f-gml}}

\begin{proof} 
Let $\aformula$ be a formula in
 $\sf{F}$ built over the \GML formulae in
$\set{\Gdiamond{\geq k_1} \aformulater_1, \ldots,
     \Gdiamond{\geq k_{n'}} \aformulater_{n'}}$ and the propositional variables $\avarprop_1$, \ldots, $\avarprop_{\alpha}$.
We write $\set{\aformulater_1^{\star}, \ldots,
     \aformulater_n^{\star}}$ to denote the set $\set{\aformulater_1, \ldots,
     \aformulater_{n'}}$ (therefore $n \leq n'$).
Without loss of generality, we can assume that for all subformulae $(\aformulabis_1 \chopop \aformulabis_2)$
of $\aformula$, we have $\submax{\aformulabis_1} = \submax{\aformulabis_2} = \set{\aformulater_1^{\star}, \ldots,
     \aformulater_n^{\star}}$ (see \Cref{proof:lemma-f-one} for the definition fo $\submax{\aformula}$). In the previous equality, we need to define $\submax{\aformulabis}$ for
the formulae $\aformulabis$ in $\sf{F}$ (as it was done only for \GML formulae so far).
Assuming that  $\aformulabis$ in $\sf{F}$ is built over
$\set{\Gdiamond{\geq l_1} \aformula_1, \ldots,
     \Gdiamond{\geq l_{m}} \aformula_{m}}$, we set
$\submax{\aformulabis_1} \egdef \set{\aformula_1, \ldots, \aformula_{m}}$.

In the case the assumption above is not satisfied, we proceed as follows to lead to a logically equivalent
formula satisfying the condition, at a polynomial computational cost only.
\begin{enumerate}
\itemsep 0 cm
\item Compute the outermost \GML formulae of the form $\Gdiamond{\geq k} \aformulabis$ in $\aformula$.
\item Let $\Phi$ be the tautology $\bigwedge (\Gdiamond{\geq k} \aformulabis \vee \neg \Gdiamond{\geq k} \aformulabis)$
      where the generalised conjunction goes through all the above $\Gdiamond{\geq k} \aformulabis$ in $\aformula$.
\item In a bottom-up fashion, replace $\aformulabis_1 \chopop \aformulabis_2$ by
      $(\aformulabis_1 \wedge \Phi) \chopop (\aformulabis_2 \wedge \Phi)$.
\end{enumerate}
The resulting formula is of polynomial size in the size of $\aformula$. So, in the sequel, we can assume that
for all subformulae $(\aformulabis_1 \chopop \aformulabis_2)$
of $\aformula$, we have $\submax{\aformulabis_1} = \submax{\aformulabis_2} = \set{\aformulater_1^{\star}, \ldots,
     \aformulater_n^{\star}}$.

Let $CD = \cd{\aformula}$ (composition degree of $\aformula$).
In order to define $\aformulabis$ from $\aformula$, we construct
a sequence of formulae
    $\aformula = \aformula_0, \ldots, \aformula_{M} = \aformulabis$ such that:
    \begin{enumerate}
    \itemsep 0 cm
    \item The number of occurrences of $\chopop$ decreases strictly from $\aformula_{i}$ to $\aformula_{i+1}$.
    \item Suppose that $\aformulater_1' \chopop \aformulater_2'$ is a subformula of $\aformula_i$
          at the composition depth $CD' \leq CD$
            such that
           $\aformulater_1', \aformulater_2'$ are  \GML formulae
          and any subformula of $\aformulater_1' \wedge \aformulater_2'$  of the form
          $\Gdiamond{\geq k} \aformulater$ has $k \leq \widehat{k} \times 2^{(CD-CD')}$ and $\newbd{0}{\aformulater}
          \leq n \times \newbd{1}{\aformula}$.
          By using Lemma~\ref{lemma:f-one} and its proof, we replace $\aformulater_1' \chopop \aformulater_2'$
         by the formula $A$ in \GML with $\newbd{0}{A} \leq \widehat{k} \times 2^{(CD+1-CD')} \times 2^n$,
         $\newbd{1}{A} \leq n \times \newbd{1}{\aformula}$ and for all
         $m \geq 2$, $\newbd{m}{A} = \newbd{m}{\aformula}$.
    \end{enumerate}

Let us explain below how to perform the transformation in (2.).
It is worth noting that all the subformulae $\Gdiamond{\geq k} \aformulater$ belonging to $\maxgmod{\aformulater_j'}$ for some $j \in
\set{1,2}$
and obtained by a transformation using Lemma~\ref{lemma:f-one}, has $\aformulater$ already equal to
$\asetbis \wedge \bar{\asetbis}$ for some $\asetbis \subseteq \set{\aformulater_1^{\star}, \ldots,
     \aformulater_n^{\star}}$ and $k \leq  \widehat{k} \times 2^{(CD-CD')}$.
In order to compute $A$  from  $\aformulater_1' \chopop \aformulater_2'$, we perform the following steps.
\begin{enumerate}

\item Let $\tilde{\aformulater_j'}$ be the formula obtained from $\aformulater_j'$
      by replacing any occurrence of $\Gdiamond{\geq k}{\aformulater}$ with $\aformulater \in
      \set{\aformulater_1^{\star}, \ldots,
     \aformulater_n^{\star}}$, by
$$
\bigvee_{\amapbis \colon \interval{1}{k_j} \rightarrow \set{\asetbis \mid \aformulater \in \asetbis \ {\rm and} \ \asetbis
\subseteq \set{\aformulater_1^{\star}, \ldots, \aformulater_n^{\star}}}} \ \ \ \bigwedge_{\asetbis \in \range{\amapbis}} \Gdiamond{\geq \card{\amapbis^{-1}(\asetbis)}}{
(\asetbis \wedge \bar{\asetbis})}.
$$
Hence, if $\aformulater$ were already of the form $\asetbis \wedge \bar{\asetbis}$
in $\aformulater_j'$, nothing is done at this stage.

\item It is easy to check that $\aformulabis_j' \equiv \tilde{\aformulabis_j'}$.
We write $\hat{\aformulabis_j'}$ to denote $\tilde{\aformulabis_j'}$ in DNF of the form below
$$
\hat{\aformulabis_j'} \subseteq \bigvee_{\amap \colon \interval{1}{n''+\alpha} \rightarrow \Bool} ((\Gdiamond{\geq l_1}{\aformulater_1^{\star \star}})^{\amap(1)} \wedge \cdots \wedge
(\Gdiamond{\geq l_{n''}}{\aformulater_{n''}^{\star \star}})^{\amap(n'')}) \wedge (\avarprop_1^{\amap(n''+1)} \wedge \cdots \wedge \avarprop_{\alpha}^{\amap(n''+\alpha)}),
$$
with $l_k \leq  \widehat{k} \times 2^{(CD-CD')}$, and there are at most $2^n$ distinct $\aformulater_k^{\star \star}$.
Consequently, $\hat{\aformulater_1'} \chopop \hat{\aformulater_2'}$ is logically equivalent to a disjunction of the form:
$$
\bigvee_{\amap, \amap'}
\Big(
(\avarprop_1^{\amap(n''+1)} \wedge \cdots \wedge \avarprop_{\alpha}^{\amap(n''+\alpha)}) \wedge
(\avarprop_1^{\amap'(n''+1)} \wedge \cdots \wedge \avarprop_{\alpha}^{\amap'(n''+\alpha)}) \wedge
$$
$$
\big((\Gdiamond{\geq l_1}{\aformulater_1^{\star \star}})^{\amap(1)} \wedge \cdots \wedge
(\Gdiamond{\geq l_{n''}}{\aformulater_{n''}^{\star \star}})^{\amap(n'')}
\chopop
(\Gdiamond{\geq l_1}{\aformulater_1^{\star \star}})^{\amap'(1)} \wedge \cdots \wedge
(\Gdiamond{\geq l_{n''}}{\aformulater_{n''}^{\star \star}})^{\amap'(n'')}
\big)
\Big).
$$
By Lemma~\ref{lemma:elimination-three-two}, the subformula with outermost connective $\chopop$ can be rewritten
as a \GML formula $A$ with  graded rank at most twice the maximal graded rank (i.e. $2 \times \widehat{k} \times  2^{(CD-CD')}$)
and with $\card{\submax{A}} \leq 2^n$. Note that the condition of being in good shape is guaranteed by
construction of $\tilde{\aformulater_j'}$.
\end{enumerate}

The formula $\aformulabis$ is obtained from $\aformula$ by applying the above transformations. As the number
of occurrences of $\chopop$ decreases strictly, we get some formula $\aformula_M$ in \GML logically equivalent to $\aformula$.
It remains to  check that the outcome formula $\aformulabis = \aformula_M$ satisfies the announced quantitative properties.
\end{proof}

\subsection{Proof of \Cref{lemma:small-gml-formula}}
 Given a formula $\aformula$ in \modallogicCC or in $\sf{F}$, recall that we write $\cd{\aformula}$ to denote
its \defstyle{composition degree}, i.e. the maximal number of imbrications of $\chopop$ in $\aformula$.
Similarly, we write $\dweight{\aformula}$ to denote
its \defstyle{diamond weight}, i.e. the  number of distinct subformulae of $\aformula$
whose outermost connective is a modality $\Diamond$ or $\Gdiamond{\geq k}$.

The following lemma subsumes \Cref{lemma:small-gml-formula}.

\begin{lemma}[Exponential-size model property]
\label{lemma:exponential-size-property-CC}
Let $\aformula$ be a formula in \modallogicCC.
Then, there is a \GML formula $\aformula'$ such that $\aformula' \equiv \aformula$
and $\maxbd{\aformula'} \leq \gr{\aformula} \times (\dweight{\aformula})^{\md{\aformula}} \times 2^{\cd{\aformula}} \times 2^{\dweight{\aformula}}$ and
$\md{\aformula'} \leq \md{\aformula}$.
\end{lemma}

We recall that $\newbd{m}{\aformula}$ can be understood as the maximal
$\newbd{0}{\aformulabis}$ for some subformula $\aformulabis$ occurring at the modal depth $m$ within 
$\aformula$ and 
$\maxbd{\aformula}$ is equal to $\max \set{\newbd{m}{\aformula}
\mid  m \in \interval{0}{\md{\aformula}}}$.

\begin{proof} 
Based on Lemma~\ref{lemma:f-gml} and on its proof, one can show the following property.
Then, we shall explain how to compute  $\aformula'$  from $\aformula$.

Let $\aformula$ be a formula in the fragment $\sf{F}$ built over \GML formulae in
$\set{\Gdiamond{\geq k_1} \aformulater_1, \ldots,
     \Gdiamond{\geq k_n} \aformulater_n}$, $k_{max} = \max \set{k_1, \ldots, k_n}$
and $\maxbd{\aformulater_i} \leq B$ for all $i \in \interval{1}{n}$ for some $B \geq 0$.
By \Cref{lemma:f-gml}, there is a \GML
formula $\aformulabis$ such that
\begin{enumerate}
\itemsep 0 cm
\item $\aformula \equiv \aformulabis$,
\item $\newbd{0}{\aformulabis} \leq k_{max}  \times 2^{\cd{\aformula}} \times 2^n$,
\item $\newbd{1}{\aformulabis} \leq n \times B$,
\item $\newbd{m}{\aformulabis} \leq B$ for all $m \geq 2$,
\item $\md{\aformulabis} \leq \md{\aformula}$.
\end{enumerate}
Consequently, $\maxbd{\aformulabis} \leq \max \set{ k_{max}  \times 2^{\cd{\aformula}} \times 2^n, n \times B}$.
Let $\aformula$ be an \modallogicCC formula with $D = \md{\aformula}$.
In order to define $\aformula'$ from $\aformula$, we define a sequence of formulae
    $\aformula = \aformula_0, \ldots, \aformula_{M} = \aformula'$ such that:
    \begin{enumerate}
    \itemsep 0 cm
    \item The number of occurrences of $\chopop$ decreases strictly from $\aformula_{i}$ to $\aformula_{i+1}$.
    \item Suppose that $\Gdiamond{\geq k} \aformulabis$ is a subformula of $\aformula_i$ at modal depth $D' \leq D$ such that
          $\aformulabis$ belongs to the fragment $\sf{F}$ and it contains at least one occurrence of $\chopop$.
          If $\aformulabis$ is built upon
          $\Gdiamond{\geq k_1} \aformulabis_1, \ldots,\Gdiamond{\geq k_n} \aformulabis_n$, then
          $n \leq \dweight{\aformula}$ and
          for all $i \in \interval{1}{n}$, we have  $\maxbd{\aformulabis_i} \leq
          \gr{\aformula} \times (\dweight{\aformula})^{D-D'} \times 2^{\cd{\aformula}} \times 2^{\dweight{\aformula}}$.
    \end{enumerate}

    Let us explain how the substitutions are operated.
    If $\aformula$ belongs to the fragment $\sf{F}$, then we apply Lemma~\ref{lemma:f-gml} getting
    $\aformulabis \equiv \aformula$ with $\aformulabis$ in \GML and
    $
    \maxbd{\aformulabis} \leq
    \max(
     \gr{\aformula}  \times 2^{\cd{\aformula}} \times 2^{\dweight{\aformula}},
     \dweight{\aformula} \times (\dweight{\aformula} \times  \gr{\aformula})
    )
    $.

    Now assume that $\aformula_i$ contains some occurrences of $\chopop$ in the scope of a graded modality.
    There is necessarily a subformula $\Gdiamond{\geq k} \aformulabis$ of $\aformula_i$, say at modal depth $D' \leq D$ such that
    $\aformulabis$ belongs to the fragment $\sf{F}$ and it contains at least one occurrence of $\chopop$.
    We can assume that $\aformulabis$ is built from
    $\Gdiamond{\geq k_1} \aformulabis_1, \ldots,\Gdiamond{\geq k_n} \aformulabis_n$ with $n \leq \dweight{\aformula}$
    and by the induction hypothesis,
    $\maxbd{\aformulabis_i} \leq
          \gr{\aformula} \times (\dweight{\aformula})^{D-D'} \times 2^{\cd{\aformula}} \times 2^{\dweight{\aformula}}$.
    By the variant of Lemma~\ref{lemma:f-gml} stated above,
    there is $\aformulabis'$ in \GML such that
    \begin{enumerate}
    \itemsep 0 cm
    \item $\aformulabis \equiv \aformulabis'$,
    \item $\newbd{0}{\aformulabis'} \leq k_{max}  \times 2^{\cd{\aformula}} \times 2^n \leq \gr{\aformula} \times  2^{\cd{\aformula}} \times 2^{\dweight{\aformula}}$,
    \item $\newbd{1}{\aformulabis'} \leq n \times  \gr{\aformula} \times (\dweight{\aformula})^{D-D'} \times 2^{\cd{\aformula}} \times 2^{\dweight{\aformula}} \leq
    \gr{\aformula} \times (\dweight{\aformula})^{D+1-D'} \times 2^{\cd{\aformula}} \times 2^{\dweight{\aformula}}$.
    \item $\newbd{m}{\aformulabis'} \leq \gr{\aformula} \times (\dweight{\aformula})^{D-D'} \times 2^{\cd{\aformula}} \times 2^{\dweight{\aformula}}$ for all $m \geq 2$,
    \item $\md{\aformulabis'} \leq \md{\aformula}$.
    \end{enumerate}
    Let $\aformula_{i+1}$ be obtained from $\aformula_{i}$ by replacing $\Gdiamond{\geq k} \aformulabis$ by $\Gdiamond{\geq k} \aformulabis'$.
    Since the substitution is performed in a bottom-up manner, still, if
    $\Gdiamond{\geq k} \aformulater$ is a subformula  of $\aformula_{i+1}$ such that
    $\aformulater$ belongs to the fragment $\sf{F}$, it contains at least one occurrence of $\chopop$
    and it is built over
    $\Gdiamond{\geq k_1} \aformulater_1, \ldots,\Gdiamond{\geq k_{\alpha}} \aformulater_{\alpha}$ then $\alpha \leq \dweight{\aformula}$.
\end{proof}

\subsection{Proof of Lemma~\ref{lemma:correctness-PL}}


This section contains the proof of Lemma~\ref{lemma:correctness-PL} 
and its first part  is dedicated to preliminary definitions and results. 

Given $\apropset = \set{\avarprop_1, \dots, \avarprop_m}$ and
a finite forest $\amodel = \triple{\worlds}{\arelation}{\avaluation}$, 
for all $\aworld', \aworld'' \in \worlds$, we write $\aworld' \approx_{\apropset} \aworld''$ iff
for all $i \in \interval{1}{m}$, we have $\amodel, \aworld' \models \avarprop_i$ iff 
 $\amodel, \aworld'' \models \avarprop_i$, i.e. $\aworld'$ and $\aworld''$ agree on the
truth values of all the propositional variables in $\apropset$. 
As done in Section~\ref{section-aexppol}, we recall that $\apropsetbis = \set{ \avarpropbis_1, \dots, \avarpropbis_{n+1}}$. 

\begin{lemma} 
\label{lemma-plcopies}
Let 
$\emptyset \neq \aset \subseteq \interval{1}{n+1}$ and 
$\pair{\amodel}{\aworld}$ be a pointed forest such that
$\amodel, \aworld \models \tluniq{\apropsetbis}$. We have
$\amodel, \aworld \models \tlcopies{\aset}$ iff for all $\aworld' 
\in \arelation(\aworld) \cap (\bigcup_{k \in \aset} \avaluation(\avarpropbis_k))$, 
$
\aset \subseteq \set{k \in \interval{1}{n+1} 
\mid \ {\rm there \ is} \ \aworld'' \in \arelation(\aworld) \ {\rm such \ that} \ \aworld' \approx_{\apropset} \aworld'' \ {\rm and} \
\amodel, \aworld'' \models \avarpropbis_k}
$.
\end{lemma}

The second condition can be restated as follows: whenever a child of $\aworld$ satisfies a valuation with respect to $\apropset$
and belongs to $(\bigcup_{k \in \aset} \avaluation(\avarpropbis_k))$, 
then the valuation is satisfied in a child of $\aworld$ satisfying $\avarpropbis_k$ for all $k \in \aset$.
We recall that  $\tlcopies{\aset}$ is defined as follows. 
\begin{center}
      \scalebox{0.9}{$
      \displaystyle\bigwedge_{\mathclap{\ \ k \neq k' \in \aset}}\!\!\neg \big(
      \Box \avarpropbis_{k} \chopop
      (
      \Gdiamond{= 1} \avarpropbis_{k} \wedge
      \neg
     (\true \chopop \Gdiamond{= 1} \avarpropbis_{k} \wedge \Gdiamond{= 1} \avarpropbis_{k'}  \wedge
      \!\bigwedge_{\mathclap{j \in \interval{1}{m}}}\!\!\Diamond \avarprop_j \Rightarrow \Box \avarprop_j
     )
      )
       \big).
       $}
\end{center}

\begin{proof} In order to show the main equivalence of the statement, we proceed by showing intermediate properties for
subformulae of $\tlcopies{\aset}$. Actually, we shall state the properties, assuming that their proof are by an easy verification. 
In what follows, we always assume that $\pair{\amodel}{\aworld}$ be a pointed forest such that
$\amodel, \aworld \models \tluniq{\apropsetbis}$.

\begin{description}

\item[(unicity)] The first intermediate property is related to the formula $\tluniq{\apropsetbis}$, which allows us to state
a unicity property. We have $\amodel, \aworld \models \tluniq{\apropsetbis}$ with
$\tluniq{\apropsetbis}$ equal to 
$\Box (\bigwedge_{i \neq i' \in \interval{1}{n+1}} \neg (\avarpropbis_i \wedge \avarpropbis_{i'}) \wedge \bigvee_{i \in \interval{1}{n+1}} \avarpropbis_i)$
iff for all $\aworld' \in \arelation(\aworld)$, there is a unique $i \in \interval{1}{n+1}$ such that
$\amodel, \aworld' \models \avarpropbis_{i}$.

\item[(uniformity)] The second property  is related to the subformula 
$\bigwedge_{j \in \interval{1}{m}}\!\!\Diamond \avarprop_j \Rightarrow \Box \avarprop_j$ that states a uniformity condition. 
We have $\amodel, \aworld \models \bigwedge_{j \in \interval{1}{m}}\!\!\Diamond \avarprop_j \Rightarrow \Box \avarprop_j$
iff for all $\aworld', \aworld'' \in \arelation(\aworld)$, we have $\aworld' \approx_{\apropset} \aworld''$. 

\item[(two-witnesses)] Let $k \neq k' \in \aset$ and $\aformulabis_{k,k'}$ be the
formula $(\true \chopop \Gdiamond{= 1} \avarpropbis_{k} \wedge \Gdiamond{= 1} \avarpropbis_{k'}  \wedge
      \!\bigwedge_{j \in \interval{1}{m}}\!\!\Diamond \avarprop_j \Rightarrow \Box \avarprop_j
     )$. We have $\amodel, \aworld \models \aformulabis_{k,k'}$  iff there are $\aworld' \neq \aworld'' \in \arelation(\aworld)$ such that
     $\amodel, \aworld' \models \avarpropbis_{k}$, $\amodel, \aworld'' \models \avarpropbis_{k'}$
     and $\aworld' \approx_{\apropset} \aworld''$. 

\item[(no-witness-1)] Again, let $k \neq k' \in \aset$. We have $\amodel, \aworld \models \Gdiamond{= 1} \avarpropbis_{k} \wedge \neg \aformulabis_{k,k'}$ iff
     there is a unique $\aworld' \in \arelation(\aworld)$ such that $\amodel, \aworld' \models \avarpropbis_k$
     and there is no $\aworld''  \in \arelation(\aworld)$  such that $\amodel, \aworld'' \models \avarpropbis_{k'}$
     and $\aworld' \approx_{\apropset} \aworld''$. 

\item[(no-witness-2)] Finally, we have $\amodel, \aworld \models \Box \avarpropbis_{k} \chopop (\Gdiamond{= 1} \avarpropbis_{k} \wedge 
\neg \aformulabis_{k,k'})$
there is  $\aworld' \in \arelation(\aworld)$ such that $\amodel, \aworld' \models \avarpropbis_k$
     and there is no $\aworld''  \in \arelation(\aworld)$  such that $\amodel, \aworld'' \models \avarpropbis_{k'}$
     and $\aworld' \approx_{\apropset} \aworld''$. 
\end{description}
Consequently, $\amodel, \aworld \models \tlcopies{\aset}$ iff for all $k \neq k' \in \aset$, there is no
$\aworld' \in \arelation(\aworld)$ such that $\amodel, \aworld' \models \avarpropbis_k$ and for which 
there is no $\aworld''  \in \arelation(\aworld)$  such that $\amodel, \aworld'' \models \avarpropbis_{k'}$
     and $\aworld' \approx_{\apropset} \aworld''$. Otherwise said, for all $\aworld' \in \arelation(\aworld)$  such that 
 $\amodel, \aworld' \models \avarpropbis_k$, there is $\aworld''  \in \arelation(\aworld)$ such that 
 $\amodel, \aworld'' \models \avarpropbis_{k'}$
     and $\aworld' \approx_{\apropset} \aworld''$ ($\apropset$ and $\apropsetbis$ are disjoint). 
\end{proof}

Let $\pair{\amodel}{\aworld}$ be a pointed forest satisfying $\tluniq{\apropsetbis}$, $\ateam$ be a team  built upon $\apropset$ and
$\emptyset \neq \aset \subseteq \interval{1}{n+1}$. We write $\pair{\amodel}{\aworld} \equiv_{\apropset}^{\aset} \ateam$ iff
the conditions below are satisfied.
\begin{enumerate}
\item For all valuations $\aplvaluation \in \ateam$, for all $k \in \aset$, there is $\aworld' \in \arelation(\aworld)$ such that
for all $i \in \interval{1}{m}$, we have $\amodel, \aworld' \models \avarprop_i$ iff $\aplvaluation(\avarprop_i) = \top$ (written $\amodel, \aworld'
\models \aplvaluation$) and $\amodel, \aworld' \models \avarpropbis_{k}$. 
\item For all valuations $\aplvaluation$ such that 
      (for all $k \in \aset$, there is $\aworld'_{k} \in \arelation(\aworld)$ such that
      $\amodel, \aworld'_k
\models \aplvaluation$ and $\amodel, \aworld'_k
\models \avarpropbis_{k}$), we have $\aplvaluation \in \ateam$. 
\end{enumerate}
Hence, when $\pair{\amodel}{\aworld} \equiv_{\apropset}^{\aset} \ateam$, the children of $\aworld$ encodes the team $\ateam$ with the property
that each encoding of $\aplvaluation \in \ateam$ is witnessed by $\card{\aset}$ witness worlds.

Given an PL[\plcnot] formula $\aformula$, its \defstyle{$\plvee$-weight}, written $\veeweight{\aformula}$, is defined
as the number of occurrences of $\plvee$ in $\aformula$.

\begin{lemma}
\label{lemma-correctness-team-logic}
Let $\emptyset \neq \aset \subseteq \interval{1}{n+1}$,
$\pair{\amodel}{\aworld}$ be a pointed forest such that $\amodel, \aworld \models \tluniq{\apropsetbis} \wedge \tlcopies{\aset}$
and $\ateam$ be a team built over $\apropset$ such that $\pair{\amodel}{\aworld} \equiv_{\apropset}^{\aset} \ateam$. For all 
PL[\plcnot] formula $\aformulabis$ built over $\apropset$ such that  $\veeweight{\aformulabis} \leq \card{\aset}-1$, we have
$\ateam \models \aformulabis$ iff $\amodel, \aworld \models \atranslation(\aformulabis, \aset)$. 
\end{lemma}

\begin{proof} The proof is by structural induction.  

\begin{description}
\item[Base case with $\aformulabis = \avarprop_i$, $i \in \interval{1}{m}$.] First, assume that $\ateam \models \avarprop_i$, which means that for
all valuations $\aplvaluation \in \ateam$, we have $\aplvaluation(\avarprop_i) = \top$. {\em Ad absurdum}, suppose that
there is $\aworld' \in \arelation(\aworld) \cap (\bigcup_{k \in \aset} \avaluation(\avarpropbis_k))$,
such that $\amodel, \aworld' \not \models \avarprop_i$. Let $\aplvaluation$ be the valuation over $\apropset$
satisfied by $\aworld'$. 
As $\amodel, \aworld \models  \tlcopies{\aset}$,  by Lemma~\ref{lemma-plcopies}, 
the valuation $\aplvaluation$ is satisfied in a  child of $\aworld$  satisfying $\avarpropbis_k$ for all $k \in \aset$.
By (2.) in the definition of $\equiv_{\apropset}^{\aset}$, this implies that $\aplvaluation \in \ateam$, which leads to
a contradiction.  Consequently, for all $\aworld' \in \arelation(\aworld) \cap (\bigcup_{k \in \aset} \avaluation(\avarpropbis_k))$,
we have $\amodel, \aworld'  \models \avarprop_i$, which can be expressed precisely with $\amodel, \aworld \models 
\Box ((\bigvee_{j \in \aset} \avarpropbis_j) \Rightarrow \avarprop_i)$. Hence, $\amodel, \aworld \models \atranslation(\avarprop_i, \aset)$
by definition of $\atranslation$. 
For the proof of the other direction, we assume that  $\amodel, \aworld \models 
\Box ((\bigvee_{j \in \aset} \avarpropbis_j) \Rightarrow \avarprop_i)$ and one can show $\ateam \models \avarprop_i$
by using this time (1.). 
Indeed, {\em ad absurdum}, suppose that $\ateam \not \models \avarprop_i$. So, there is a
valuation $\aplvaluation$ such that $\aplvaluation(\avarprop_i) = \perp$. 
By (1.),  for all $k \in \aset$, there is $\aworld'_{k} \in \arelation(\aworld)$ such that
$\amodel, \aworld'_{k} \not \models \avarprop_i$  and $\amodel, \aworld'_{k} \models \avarpropbis_{k}$. 
Since  $\aworld'_{k} \in \arelation(\aworld)$, $\amodel, \aworld'_{k} \models \avarpropbis_{k}$
and $\amodel, \aworld \models 
\Box ((\bigvee_{j \in \aset} \avarpropbis_j) \Rightarrow \avarprop_i)$, we get 
$\amodel, \aworld'_{k} \models \avarprop_i$, which leads to a contradiction. 

\item[Base case with $\aformulabis = \plneg \avarprop_i$, $i \in \interval{1}{m}$.] Similar to the case $\aformulabis = \avarprop_i$.
\item[Induction step.] The cases in the induction step for which the outermost connective of $\aformulabis$ is
either $\wedge$ or $\plcnot$ are by an easy verification. Let us consider the case $\aformulabis = \aformulabis_1 \plvee \aformulabis_2$. 
Observe that $\veeweight{\aformulabis} = \veeweight{\aformulabis_1} + \veeweight{\aformulabis_2} + 1$
      and recall that $\veeweight{\aformulabis} \leq \card{\aset}-1$. 
      Consequently, $\veeweight{\aformulabis_1} + \veeweight{\aformulabis_2} +2 \leq \card{\aset}$ 
      and let $\aset_i = \mathfrak{c}_i(\aset,\veeweight{\aformulabis_1}+1, \veeweight{\aformulabis_2} + 1)$ for $i \in \set{1,2}$. 
     \begin{description}
     \item[Assume $\ateam \models \aformulabis_1 \plvee \aformulabis_2$.] 
       By definition of $\models$ for PL[\plcnot], 
      there are $\ateam_1$ and $\ateam_2$ such that $\ateam = \ateam_1 \cup \ateam_{2}$,   
      $\ateam_1 \models \aformulabis_1$ and $\ateam_2 \models \aformulabis_2$.
      Let us define $\amodel_1 = \triple{\worlds}{\arelation_1}{\avaluation_1}$ and 
      $\amodel_2 = \triple{\worlds}{\arelation_2}{\avaluation_2}$ such that $\amodel = \amodel_1 +_{\aworld} \amodel_2$ and satisfying
      the conditions below (only the relevant part is explicitly specified). 
      \begin{itemize}
      \item Assume $\aplvaluation \in \ateam_1 \cap \ateam_2$. As $\pair{\amodel}{\aworld} \equiv_{\apropset}^{\aset} \ateam$, 
            for all $k \in \aset$, there is $\aworld'_k \in \arelation(\aworld)$ such that
            $\amodel, \aworld'_k \models \aplvaluation$ and $\amodel, \aworld'_k \models \avarpropbis_{k}$.
            For all $i \in \set{1,2}$ and  $k \in \aset$, for all  $\aworld' \in \arelation(\aworld) \cap \avaluation(\avarpropbis_k)$
            such that $\amodel, \aworld' \models \aplvaluation$,
             if $k \in \aset_i$, then $\pair{\aworld}{\aworld'} \in \arelation_i$
            by definition, otherwise $\pair{\aworld}{\aworld'} \in \arelation_{3-i}$.
            For all $\aworld' \in \arelation(\aworld)$ such that $\aworld' \not \in  (\bigcup_{k \in \aset} \avaluation(\avarpropbis_k))$ and 
             $\amodel, \aworld' \models \aplvaluation$, 
            it is irrelevant whether $\pair{\aworld}{\aworld'}$ belongs to $\arelation_1$ or to  $\arelation_2$.
      \item Assume that  $\aplvaluation \in \ateam_j \setminus \ateam_{3-j}$ for some $j \in \set{1,2}$.    
             For all  $\aworld' \in \arelation(\aworld)$
            such that $\amodel, \aworld' \models \aplvaluation$, $\pair{\aworld}{\aworld'} \in \arelation_j$
            by definition.
      \end{itemize}
      One can check that $\amodel_1, \aworld \equiv_{\apropset}^{\aset_1} \ateam_1$,  $\amodel_2, \aworld \equiv_{\apropset}^{\aset_2} \ateam_2$,
      $\veeweight{\aformulabis_1} \leq \card{\aset_1} -1$ and  $\veeweight{\aformulabis_2} \leq \card{\aset_2} -1$.
      By the induction hypothesis, we have $\amodel_1, \aworld \models \atranslation(\aformulabis_1, \aset_1)$ and 
      $\amodel_2, \aworld \models \atranslation(\aformulabis_2, \aset_2)$. Moreover, as $\amodel, \aworld \models \tlcopies{\aset}$, 
      it is also easy to check that  $\amodel_1, \aworld \models \tlcopies{\aset_1}$ and $\amodel_2, \aworld \models \tlcopies{\aset_2}$. 
      Hence, $\amodel, \aworld \models (\atranslation(\aformulabis_1,\aset_1) \land 
      \tlcopies{\aset_1}) \chopop\, (\atranslation(\aformulabis_2,\aset_2)
      \land 
       \tlcopies{\aset_2})$, i.e.  $\amodel, \aworld \models \atranslation(\aformulabis, \aset)$ by definition of $\atranslation$.

     \item[Assume $\amodel, \aworld \models \atranslation(\aformulabis_1 \plvee \aformulabis_2, \aset)$.] There are $\amodel_1$, $\amodel_2$
     such that $\amodel = \amodel_1 +_{\aworld} \amodel_2$, $\amodel_1, \aworld \models \tlcopies{\aset_1} \wedge \atranslation(\aformulabis_1, \aset_1)$
     and $\amodel_2, \aworld \models \tlcopies{\aset_2} \wedge \atranslation(\aformulabis_2, \aset_2)$.
     Let us define $\ateam_1$ and $\ateam_2$ such that $\ateam = \ateam_1 \cup \ateam_2$,
     $\amodel_1, \aworld \equiv_{\apropset}^{\aset_1} \ateam_1$ and  $\amodel_2, \aworld \equiv_{\apropset}^{\aset_2} \ateam_2$.
     Let $\aplvaluation \in \ateam$ and $j \in \set{1,2}$. We have $\aplvaluation \in \ateam_j$ $\equivdef$
     for all $k \in \aset_j$, there is $\aworld'_{k} \in \arelation_j(\aworld)$ such that
      $\amodel_j, \aworld'_k
\models \aplvaluation$ and $\amodel_j, \aworld'_k \models \avarpropbis_{k}$.
     As $\amodel, \aworld \models \tlcopies{\aset}$ and $\aset = \aset_1 \uplus \aset_2$, 
     one can verify that the definition of $\ateam_1$ and $\ateam_2$ is well-designed
     and the teams $\ateam_1$ and $\ateam_2$ satisfy the expected properties. 
     Using that  $\veeweight{\aformulabis_1} +1 \leq \card{\aset_1}$ and  $\veeweight{\aformulabis_2} +1 \leq \card{\aset_2}$, 
     by the induction hypothesis, we have $\ateam_1 \models \aformulabis_1$ and  $\ateam_2 \models \aformulabis_2$. 
     Consequently, $\ateam \models \aformulabis$. 
     \end{description}
\end{description}

\end{proof}

The proof of Lemma~\ref{lemma:correctness-PL} is now by an easy verification. 

\begin{proof} (Lemma~\ref{lemma:correctness-PL}) Let $\aformula$ be an PL[\plcnot] formula built upon $\apropset = \set{\avarprop_1, \dots, \avarprop_m}$
with $\veeweight{\aformula} = n$ and $\apropsetbis = \set{ \avarpropbis_1, \dots, \avarpropbis_{n+1}}$.

Suppose that $\aformula$ is satisfiable, meaning that there is a team $\ateam = \set{\aplvaluation_1, \ldots, \aplvaluation_K}$ satisfying $\aformula$. 
Let $\amodel =  \triple{\worlds}{\arelation}{\avaluation}$ be the finite forest such that
$\aworld = \set{0} \cup \interval{1}{K} \times \interval{1}{n+1}$, 
$\arelation = \set{\pair{0}{\pair{i}{j}} \mid \pair{i}{j} \in \interval{1}{K} \times \interval{1}{n+1}}$,
and $\avaluation$ is a valuation such that,
\begin{itemize}
\item $\avaluation(\avarpropbis_j) = \interval{1}{K} \times \set{j}$ for all $j \in \interval{1}{n+1}$,
\item $\avaluation(\avarprop_s) = \set{\pair{i}{j} \mid \aplvaluation_{i}(\avarprop_s) = \top}$ for all $s \in \interval{1}{m}$. 
\end{itemize}
One can show that $\amodel, \aworld \models \tluniq{\apropsetbis} \wedge \tlcopies{\interval{1}{n+1}}$ and 
$\amodel, \aworld \equiv_{\apropset}^{\interval{1}{n+1}} \ateam$. As $\veeweight{\aformula} = \card{\interval{1}{n+1}}-1$ ($=n$),
by Lemma~\ref{lemma-correctness-team-logic}, we have $\amodel, \aworld  \models \atranslation(\aformula, \interval{1}{n+1})$. 

Conversely, suppose that $\tluniq{\apropsetbis} \wedge \tlcopies{\interval{1}{n+1}} \wedge \atranslation(\aformula, \interval{1}{n+1})$
is satisfiable, meaning that there is a pointed forest $\pair{\amodel}{\aworld}$ satisfying it with $\amodel =  \triple{\worlds}{\arelation}{\avaluation}$.
We define the team $\ateam$ such that for all valuations $\aplvaluation$ built over $\apropset$, 
$\aplvaluation$ belongs to $\ateam$ iff there is $\aworld' \in \arelation(\aworld)$ such that $\amodel, \aworld' \models \avarpropbis_k$ for some
$k \in \interval{1}{n+1}$ and $\amodel, \aworld' \models \aplvaluation$. 
Again, one can check that $\amodel, \aworld \equiv_{\apropset}^{\interval{1}{n+1}} \ateam$ (here we use the fact the
$\amodel, \aworld \models \tluniq{\apropsetbis} \wedge \tlcopies{\interval{1}{n+1}}$) and 
by Lemma~\ref{lemma-correctness-team-logic}, we have $\ateam \models \aformula$. 
\end{proof}


\section{Proofs of \Cref{section-tower-SC}}

\subsection{Correctness of $\init{j}$, $\nominal{\aaux}{i}$, $\atnom{\aaux}{i} \aformula$ and
$\twonoms{\aaux}{\aauxbis}{i}$}

In the following statements and proofs, let $\amodel = \triple{\worlds}{\arelation}{\avaluation}$ be a finite forest and $\aworld \in \worlds$.

\begin{lemma}
  \label{lemma:init}
Let $j \geq 1$.
$\amodel,\aworld \models \init{j}$
if and only if for every $0 \leq i \leq j$, every $\aworld' \in \arelation^i(\aworld)$ and every $\aaux \in \Aux$,
\begin{enumerate}
  \item if \ $\amodel, \aworld' \models \avartree$ then
  $\forall \aworld_1',\aworld_2' \in \arelation(\aworld')$,
  if $\amodel, \aworld_1' \models \aaux$ and $\amodel, \aworld_2' \models \aaux$
  then $\aworld_1' = \aworld_2'$ (i.e. at most one child of $\aworld'$ satisfies $\aaux$);
  \item for every $\aworld'' \in \arelation(\aworld')$, if \ $\amodel, \aworld'' \models \aaux$, then
    $\arelation(\aworld'') = \emptyset$ (i.e. $\aworld''$ does not have children) and it cannot be that \
    $\amodel, \aworld'' \models \aauxbis$ for some $\aauxbis \in \Aux$ syntactically different from $\aaux$ (i.e. among the propositions in $\Aux$, $\aworld''$ only satisfies $\aaux$).
\end{enumerate}
Moreover, given $\amodel' \sqsubseteq \amodel$, $\amodel',\aworld \models \init{j}$.
\end{lemma}

\begin{proof}{(sketch)}.
  Recall that $\init{j}$ is defined as follows:
  \begin{nscenter}
  $
  \Boxbox^{j} {\displaystyle\bigwedge_{{\aaux \in \Aux}}}
  \Big(
    \big(
      \avartree \implies \lnot (\Diamond \aaux \separate \Diamond \aaux)
    \big)
  \land
    \Box\big(
      \aaux \implies \Box \bottom \land {\displaystyle\bigwedge_{\mathclap{\aauxbis \in \Aux \setminus\{\aaux\}}}} \lnot \aauxbis
    \big)
   \Big)
  $
  \end{nscenter}
  The proof is straightforward (and hence here only sketched).
  Indeed, the statement ``for every $0 \leq i \leq j$, every $\aworld' \in \arelation^i(\aworld)$ and every $\aaux \in \Aux$''
  is captured by the prefix $\Boxbox^{j} {\bigwedge_{{\aaux \in \Aux}}}$ of $\init{j}$.
  Then, (1) corresponds to the conjunct $\avartree \implies \lnot (\Diamond \aaux \separate \Diamond \aaux)$ whereas
  (2) corresponds to the conjunct $\Box\big(
    \aaux \implies \Box \bottom \land {\bigwedge_{{\aauxbis \in \Aux \setminus\{\aaux\}}}} \lnot \aauxbis
  \big)$.
\end{proof}

\begin{lemma}
  \label{lemma:nom}
  Let $\aaux \in \Aux$ and $0 < i \leq j \in \Nat$.
  Suppose $\amodel,\aworld \models \init{j}$.

  $\amodel,\aworld \models \nominal{\aaux}{i}$ if and only if $\aaux$ is a nominal
  for the depth $i$. Recall that $\aaux$ is a nominal for the depth $i$ if there is exactly one $\avartree$-world in $\arelation^i(\aworld)$ having a child satisfying $\aaux$.
\end{lemma}

\begin{proof}
Recall that $\nominal{\aaux}{i}$ is defined as follows:
\begin{nscenter}
  $
  \begin{aligned}
      \HMDiamond{\avartree}^i\Diamond\aaux \land
    {\bigwedge_{\mathclap{k \in \interval{0}{i-1}}}} \HMBox{\avartree}^k \lnot \big(\HMDiamond{\avartree}^{i-k}\Diamond\aaux \separate \HMDiamond{\avartree}^{i-k}\Diamond \aaux\big).
  \end{aligned}
  $
  \end{nscenter}

  ($\Rightarrow$): Suppose $\amodel,\aworld\models\nominal{\aaux}{i}$, then by definition of $\models$ and the relativised modality $\HMDiamond{\avartree}$, there exists a path of $\avartree$-worlds $\aworld_1,\aworld_2,\ldots,\aworld_{i}$, such that $\aworld\arelation\aworld_1\arelation\aworld_2\ldots\arelation\aworld_{i}$, and there exists $\aworld'$ such that
  $(\aworld_{i},\aworld')\in R$ and $\amodel, \aworld'\models\aaux$.
  The second conjunct of $\nominal{\aaux}{i}$ guarantees that there is only one  such paths, leading to $\aworld_{i}$ being a nominal for the depth $i$.
  Indeed, suppose \emph{ad absurdum} that there is a second world $\aworld_i' \in \arelation^i(\aworld)$, distinct from $\aworld_i$, such that $\amodel,\aworld_i' \models \Diamond \aaux$.
  Since $\amodel,\aworld \models \init{j}$, $\aworld_i'$ must be a $\avartree$-node
  and there must be a path of $\avartree$-worlds $\aworld_1',\aworld_2',\ldots,\aworld_{i}'$ such
  that $\aworld\arelation\aworld_1'\arelation\aworld_2'\ldots\arelation\aworld_{i}'$.
  Then, there must be $k \in \interval{0}{i-1}$ such that
  for every $j \leq k$, $\aworld_j = \aworld_j'$, and for every $l \in \interval{j+1}{i}$, $\aworld_l \neq \aworld_l'$.
  By considering the pointed forest $\pair{\amodel}{\aworld_k}$,
  we can easily show that $\amodel,\aworld_k \models \HMDiamond{\avartree}^{i-k}\Diamond\aaux \separate \HMDiamond{\avartree}^{i-k}\Diamond \aaux$.
  This implies that $\amodel,\aworld \models \HMDiamond{\avartree}^k \big(\HMDiamond{\avartree}^{i-k}\Diamond\aaux \separate \HMDiamond{\avartree}^{i-k}\Diamond \aaux\big)$, in contradiction with the second conjunct of $\nominal{\aaux}{i}$. Hence, $\aworld_i'$ cannot be distinct from $\aworld_i$.

  ($\Leftarrow$): This direction is analogous. Suppose that $\amodel,\aworld \models \init{j}$
  and  $\aaux$ is a nominal for the depth $i$.
  By definition,  there is a unique $\avartree$-world $\aworld'$ in $\arelation^i(\aworld)$ having a child satisfying $\aaux$.
  Since $\amodel,\aworld \models \init{j}$, the path from $\aworld$ to $\aworld'$ must only witness $\avartree$-nodes. Hence $\amodel,\aworld \models \HMDiamond{\avartree}^i\Diamond\aaux$.
  Moreover, by the uniqueness of this path we conclude that
  $\amodel,\aworld \models {\bigwedge_{{k \in \interval{0}{i-1}}}} \HMBox{\avartree}^k \lnot \big(\HMDiamond{\avartree}^{i-k}\Diamond\aaux \separate \HMDiamond{\avartree}^{i-k}\Diamond \aaux\big)$ also holds.
  Thus, $\amodel,\aworld\models\nominal{\aaux}{i}$.
\end{proof}

\begin{lemma}
  \label{lemma:at}
  Let $\aaux \in \Aux$ and $0 < i \leq j \in \Nat$.
  Suppose $\amodel,\aworld \models \init{j} \land \nominal{\aaux}{i}$.

  $\amodel,\aworld \models \atnom{\aaux}{i} \aformula$
  if and only if the world (say $\aworld'$) corresponding to the nominal $\aaux$ for the depth $i$ is such that $\amodel,\aworld' \models \aformula$.
\end{lemma}

\begin{proof}
Both directions are straightforward.
Recall that $\atnom{\aaux}{i} \aformula$ is defined as
  $\HMDiamond{\avartree}^{i}(\Diamond\aaux\wedge\varphi)$.
Moreover, as we are working under the hypothesis that $\amodel,\aworld \models \init{j} \land \nominal{\aaux}{i}$, by \Cref{lemma:nom}, $\aaux$ is a nominal for the depth $i$. In the following, let $\aworld'$ be the world in $\arelation^i(\aworld)$ corresponding to the nominal $\aaux$ (i.e. $\aworld'$ has an $\aaux$-child).

($\Rightarrow$): Suppose $\amodel,\aworld \models \atnom{\aaux}{i} \aformula$.
By definition, there is $\aworld'' \in \arelation^i(\aworld)$ s.t. $\amodel,\aworld'' \models \Diamond \aaux \land \aformula$. Since $\aaux$ is a nominal for the depth $i$, we conclude that $\aworld' = \aworld''$ and hence $\amodel,\aworld'' \models \aformula$.

($\Leftarrow$): Suppose that $\aworld'$ is such that $\amodel,\aworld' \models \aformula$. By definition,
$\aworld'$ is the world corresponding to the nominal $\aaux$ (for the depth $i$).
Hence $\amodel,\aworld' \models \Diamond \aaux$. Since $\aworld' \in \arelation^i(\aworld)$ by $\amodel,\aworld \models \init{j}$ we conclude that
there is a path of $\avartree$-nodes from $\aworld$ to $\aworld'$, of length $i$.
Thus, $\amodel,\aworld \models \HMDiamond{\avartree}^{i}(\Diamond\aaux\wedge\varphi)$.
\end{proof}

\begin{lemma}\label{lemma:twonoms}
  Let $\aaux \neq \aauxbis \in \Aux$ and $0 < i \leq j \in \Nat$.
  Suppose $\amodel,\aworld \models \init{j}$.

  $\amodel,\aworld \models \twonoms{\aaux}{\aauxbis}{i}$ if and only if $\aaux$ and $\aauxbis$ are nominals
  for the depth $i$, corresponding to two different worlds.
\end{lemma}

\begin{proof}
Given \Cref{lemma:nom,lemma:at}, this proof is straightforward.
Recall that $\twonoms{\aaux}{\aauxbis}{i} \egdef \nominal{\aaux}{i} \wedge \nominal{\aauxbis}{i} \wedge \neg\atnom{\aaux}{i}\Diamond\aauxbis$.

($\Rightarrow$): Suppose $\amodel,\aworld \models \twonoms{\aaux}{\aauxbis}{i}$.
By~\Cref{lemma:nom} $\aaux$ and
$\aauxbis$ are nominals for depth $i$. Let $\aworld_{\aaux}$ (resp. $\aworld_{\aauxbis}$) be the world in $\arelation^i(\aworld)$ corresponding to the nominal $\aaux$ (resp. $\aauxbis$). Notice that, in particular, $\amodel,\aworld_{\aauxbis} \models \Diamond \aauxbis$.
By $\amodel,\aworld \models \neg\atnom{\aaux}{i}\Diamond\aauxbis$  and~\Cref{lemma:at},
we conclude that $\amodel,\aworld_{\aaux} \not\models \Diamond \aauxbis$.
Thus, $\aworld_{\aaux} \neq \aworld_{\aauxbis}$.

($\Leftarrow$): This direction is analogous and simply relies on~\Cref{lemma:nom,lemma:at}.
\end{proof}

\subsection{Formal semantics of the inductively defined formulae used for $\complete{j}$}\label{appendix:inductive-formulae}
Let us formalise the expected semantics of the formulae introduced in order to define $\complete{j}$, and whose definition is inductive.
Let $\amodel = \triple{\worlds}{\arelation}{\avaluation}$ be a finite forest and $\aworld \in \worlds$. Let $1 \leq i \leq j$ and let $\aaux \neq \aauxbis \in \Aux$.

\begin{description}
\item[(\,$\fork{\aaux}{\aauxbis}{i}{j}$\,):]  Suppose $\amodel,\aworld \models \init{j}$.

  $\amodel,\aworld \models \fork{\aaux}{\aauxbis}{i}{j}$ if and only if
  {\bfseries \itshape{(i)}} $\aworld$ has exactly two $\avartree$-children and exactly two paths of $\avartree$-nodes, both of length $i$;
  {\bfseries \itshape{(ii)}} one of these two paths ends on a world (say $\aworld_{\aaux}$) corresponding to the nominal $\aaux$ whereas the other ends on a world (say $\aworld_{\aauxbis}$) corresponding to the nominal $\aauxbis$;
  {\bfseries \itshape{(iii)}} if $i < j$ then $\pair{\amodel}{\aworld_{\aaux}}$ and $\pair{\amodel}{\aworld_{\aauxbis}}$ satisfy
   $\completeplus{j - i} \egdef \complete{j - i} \land \HMBox{\avartree}(\Diamond \avarleft \land \Diamond \avarselect \land \Diamond \avarright)$.
\item[(\,$\less{\aaux}{\aauxbis}{i}{j}$\,):] Suppose $\amodel,\aworld \models \init{j} \land \fork{\aaux}{\aauxbis}{i}{j}$.

$\amodel,\aworld \models \less{\aaux}{\aauxbis}{i}{j}$ if and only if
there are two distinct $\avartree$-nodes $\aworld_{\aaux},\aworld_{\aauxbis} \in \arelation^i(\aworld)$ such that $\aworld_{\aaux}$ corresponds to the nominal $\aaux$, $\aworld_{\aauxbis}$ corresponds to the nominal $\aauxbis$ and
$\nb{\aworld_{\aaux}} < \nb{\aworld_{\aauxbis}}$.
\item[(\,$\successor{\aaux}{\aauxbis}{j}$\,):] Suppose $\amodel,\aworld \models \init{j} \land \fork{\aaux}{\aauxbis}{1}{j}$.

$\amodel,\aworld \models \successor{\aaux}{\aauxbis}{j}$ if and only if
there are two distinct $\avartree$-nodes $\aworld_{\aaux},\aworld_{\aauxbis} \in \arelation(\aworld)$ such that $\aworld_{\aaux}$ corresponds to the nominal $\aaux$, $\aworld_{\aauxbis}$ corresponds to the nominal $\aauxbis$ and
$\nb{\aworld_{\aauxbis}} = \nb{\aworld_{\aaux}}+1$.

\item[(\,$\pU{j}$\,):]
Suppose $\amodel,\aworld \models \init{j} \land \pS{j} \land \pA$.

$\amodel,\aworld \models \pU{j}$ if and only if
$\pair{\amodel}{\aworld}$ satisfies
\ref{prop:apU}, i.e.\ distinct $\avartree$-nodes in $\arelation(\aworld)$
encode different numbers.

\item[(\,$\pC{j}$\,):]
Suppose $\amodel,\aworld \models \init{j} \land \pS{j} \land \pA$.

$\amodel,\aworld \models \pC{j}$ if and only if
$\pair{\amodel}{\aworld}$ satisfies
\ref{prop:apC}, i.e.\ for every $\avartree$-node $\aworld_1 \in \arelation(\aworld)$, if $\nb{\aworld_1} < \amapter(j,n)-1$ then $\nb{\aworld_2} = \nb{\aworld_1}+1$ for some $\avartree$-node $\aworld_2 \in \arelation(\aworld)$.
\item[(\,$\complete{j}$\,):]
Suppose $\amodel,\aworld \models \init{j}$.

$\amodel,\aworld \models \complete{j}$ if and only if
$\pair{\amodel}{\aworld}$ satisfies \ref{prop:apS}, \ref{prop:apZ},
\ref{prop:apU}, \ref{prop:apC} and \ref{prop:apA}.
\end{description}

The formulae $\pS{j}$, $\pA$ and $\pZ{j}$ ($j \geq 1$) are also required in order to define correctly $\complete{j}$. However their definition and proof of correctness are straightforward. Hence we omit the proofs, and simply state the expected semantics of these formulae. It should be noted that a formal proof of $\pZ{j}$ relies on $\complete{j-1}$, which (as we will see multiple times in the next sections), we can assume to be correctly defined by inductive hypothesis (on $j$). 

\begin{lemma}\label{lemma:basicform}
Let $j \geq 1$.
Let $\amodel = \triple{\worlds}{\arelation}{\avaluation}$ be a finite forest and $\aworld \in \worlds$.
\begin{itemize}
  \item $\amodel,\aword \models \pS{j}$ if and only if $\pair{\amodel}{\aworld}$ satisfies \ref{prop:apS}, i.e.\ every $\avartree$-node in $\arelation(\aworld)$ satisfies $\complete{j-1}$.
  \item $\amodel,\aword \models \pA$ if and only if
   $\pair{\amodel}{\aworld}$ satisfies \ref{prop:apA}, i.e.\
  $\aworld$ is a $\avartree$-node, every $\avartree$-node in $\arelation(\aworld)$ has one $\avariable$-child and one $\avariablebis$-child, and every $\avartree$-node in $\arelation^2(\aworld)$ has three children satisfying $\avarleft$, $\avarright$ and $\avarselect$, respectively.
  \item Suppose $\amodel,\aword \models \pS{j}$. $\amodel,\aworld \models \pZ{j}$ if and only if
  $\pair{\amodel}{\aworld}$ satisfies \ref{prop:apZ}, i.e.\
    there is a $\avartree$-node $\tilde\aworld \in \arelation(\aworld)$ s.t.\ $\nb{\tilde\aworld} = 0$.
\end{itemize}
\end{lemma}

We now prove the correctness of all the formulae listed above, starting from the base case where $j = 1$ or $i = j$,
to then show the proof for $1 \leq i < j$.

\subsection{Base case $i = j$ / $j = 1$: Correctness of $\fork{\aaux}{\aauxbis}{j}{j}$, $\less{\aaux}{\aauxbis}{j}{j}$ and $\successor{\aaux}{\aauxbis}{1}$}

In the following statements and proofs, let  $\amodel = \triple{\worlds}{\arelation}{\avaluation}$ be a finite forest and $\aworld \in \worlds$.

\begin{lemma}\label{lemma:forkjj}
Let $\aaux \neq \aauxbis \in \Aux$ and $j \geq 1$. Suppose $\amodel,\aworld \models \init{j}$.

$\amodel,\aworld \models \fork{\aaux}{\aauxbis}{j}{j}$ if and only if
\begin{enumerate}
\item $\aworld$ has exactly two $\avartree$-children and exactly two paths of $\avartree$-nodes, both of length $j$, ending in two $\avartree$-nodes (say $\aworld_1$ and $\aworld_2$);
\item $\aworld_1$ corresponds to the nominal $\aaux$ (for the depth $j$),
whereas $\aworld_2$ corresponds to the nominal $\aauxbis$ (for the depth $j$).
\end{enumerate}
\end{lemma}

\begin{proof}
Reall that $\fork{\aaux}{\aauxbis}{j}{j}$ is defined as $\Diamond_{=2} \avartree\,
{\land}\,\HMBox{\avartree}\Boxbox^{j-2}\!(\avartree{\implies}\Diamond_{=1} \avartree)\,
{\land}\,\twonoms{\aaux}{\aauxbis}{j}$.

($\Rightarrow$):
Suppose $\amodel,\aworld\models\fork{\aaux}{\aauxbis}{j}{j}$.
By $\amodel,\aworld\models \Diamond_{=2} \avartree$, $\aworld$ has exactly two $\avartree$-children (let us say $\aworld_1'$ and $\aworld_2'$).
Then, by $\amodel,\aworld\models\HMBox{\avartree}\Boxbox^{j-2}\!(\avartree{\implies}\Diamond_{=1} \avartree)$,
it is easy to show that
\begin{itemize}
\item there is exactly one path of $\avartree$-nodes of length $j-1$,
starting in $\aworld_1'$ and ending in a $\avartree$-node $\aworld_1 \in \arelation^j(\aworld)$;
\item there is exactly one path of $\avartree$-nodes of length $j-1$,
starting in $\aworld_2'$ and ending in a $\avartree$-node $\aworld_2 \in \arelation^j(\aworld)$.
\end{itemize}
Then, the property (1) of the statement is verified and $\{\aworld_1,\aworld_2\} = \arelation^j(\aworld)$.
The property (2) of the statement is then verified by simply applying \Cref{lemma:twonoms}.

($\Leftarrow$): This direction is straightforward. In short, from (1) we conclude that $\amodel,\aworld \models \Diamond_{=2} \avartree\,
{\land}\,\HMBox{\avartree}\Boxbox^{j-2}\!(\avartree{\implies}\Diamond_{=1} \avartree)$, whereas from (2) together with \Cref{lemma:twonoms} we have $\amodel,\aworld \models \twonoms{\aaux}{\aauxbis}{j}$.
\end{proof}

\begin{lemma}\label{lemma:lessjj}
Let $\aaux \neq \aauxbis \in \Aux$ and $j \geq 1$. Suppose $\amodel,\aworld \models \init{j} \land \fork{\aaux}{\aauxbis}{j}{j}$.

$\amodel,\aworld \models \less{\aaux}{\aauxbis}{j}{j}$ if and only if
there are two distinct $\avartree$-nodes $\aworld_{\aaux},\aworld_{\aauxbis} \in \arelation^j(\aworld)$ such that $\aworld_{\aaux}$ corresponds to the nominal $\aaux$, $\aworld_{\aauxbis}$ corresponds to the nominal $\aauxbis$ and
$\nb{\aworld_{\aaux}} < \nb{\aworld_{\aauxbis}}$.
\end{lemma}

\begin{proof}
  Recall that $\less{\aaux}{\aauxbis}{j}{j}$ is defined as
  ${\bigvee_{{u \in \interval{1}{n}}}}
    \big(
      \atnom{\aaux}{j}\lnot \avarprop_u
      \land \atnom{\aauxbis}{j} \ \avarprop_u \land {\bigwedge_{{v \in \interval{u+1}{n}}}}
      (\atnom{\aaux}{j} \ \avarprop_v \iff  \atnom{\aauxbis}{j} \ \avarprop_v)
    \big)$.
  The proof uses standard properties of numbers encoded in binary.
  Let $x,y$ be two natural numbers that can be represented in binary by using $n$ bits. Let us denote with $x_i$ (resp. $y_i$) the $i$-th bit of the binary representation of $x$ (resp. $y$).
  We have that $x < y$ if and only if
    \begin{enumerate}[label=(\Alph*)]
      \item there is a position $i \in \interval{1}{n}$ such that $x_i = 0$ and $y_i = 1$;
      \item for every position $j > i$, $x_j = 0$ $\iff$ $y_j = 0$.
    \end{enumerate}
  The formula $\less{\aaux}{\aauxbis}{j}{j}$ uses exactly this characterisation in order to state that $\nb{\aworld_{\aaux}} < \nb{\aworld_{\aauxbis}}$.

  In the following, since we are working under the hypothesis that $\amodel,\aworld \models \init{j} \land \fork{\aaux}{\aauxbis}{j}{j}$, let $\aworld_{\aaux}$ (resp. $\aworld_{\aauxbis}$) be the world corresponding to the nominal $\aaux$ (resp. $\aauxbis$), w.r.t.\ the depth $j$.

($\Rightarrow$):
Suppose $\amodel,\aworld \models \less{\aaux}{\aauxbis}{j}{j}$.
Then there is $u \in \interval{1}{n}$ s.t.\
$\amodel,\aworld\models
{
    \atnom{\aaux}{j}\lnot \avarprop_u
    \land \atnom{\aauxbis}{j} \ \avarprop_u \land {\bigwedge_{v \in \interval{u+1}{n}}}
    (\atnom{\aaux}{j} \ \avarprop_v \iff  \atnom{\aauxbis}{j} \ \avarprop_v)
}$.
By \Cref{lemma:at} and $\amodel,\aworld \models \atnom{\aaux}{j}\lnot \avarprop_u
\land \atnom{\aauxbis}{j} \ \avarprop_u$ we conclude that
$\amodel,\aworld_{\aaux} \models \lnot \avarprop_u$ and $\amodel,\aworld_{\aauxbis} \models \avarprop_u$.
Hence, the $u$-th bit is $0$ in the number encoded by $\aworld_{\aaux}$, whereas it is $1$ in the number encoded by $\aworld_{\aauxbis}$, as required by (A).
Similarly, by \Cref{lemma:at} and
$\amodel,\aworld \models {\bigwedge_{v \in \interval{u+1}{n}}}
(\atnom{\aaux}{j} \ \avarprop_v \iff  \atnom{\aauxbis}{j} \ \avarprop_v)$,
we conclude that for every $v \in \interval{u+1}{n}$,
$\amodel,\aworld_{\aaux} \models \avarprop_v$ if and only if $\amodel,\aworld_{\aauxbis} \models \avarprop_v$.
This corresponds to the property (B) above, leading to $\nb{\aworld_{\aaux}} < \nb{\aworld_{\aauxbis}}$.

($\Leftarrow$): This direction follows similar arguments (backwards).
\end{proof}

\begin{lemma}\label{lemma:succone}
Let $\aaux \neq \aauxbis \in \Aux$.
Suppose $\amodel,\aworld \models \init{1} \land \fork{\aaux}{\aauxbis}{1}{1}$.

$\amodel,\aworld \models \successor{\aaux}{\aauxbis}{1}$ if and only if
there are two distinct $\avartree$-nodes $\aworld_{\aaux},\aworld_{\aauxbis} \in \arelation(\aworld)$ such that $\aworld_{\aaux}$ corresponds to the nominal $\aaux$, $\aworld_{\aauxbis}$ corresponds to the nominal $\aauxbis$ and
$\nb{\aworld_{\aauxbis}} = \nb{\aworld_{\aaux}}+1$.
\end{lemma}

\begin{proof}
  Recall the definition of $\successor{\aaux}{\aauxbis}{1}$:
  \[
      {\bigvee_{\mathclap{\quad u \in \interval{1}{n}}}}
        \!\!\big(
          \atnom{\aaux}{1}(\lnot \avarprop_u {\land}\!\!\bigwedge_{\mathclap{v \in \interval{1}{u-1}}}\!\avarprop_v)
          \land
          \atnom{\aauxbis}{1} (\avarprop_u {\land}\!\!\bigwedge_{\mathclap{v \in \interval{1}{u-1}}}\!\!\lnot \avarprop_v)
          {\land}\!
          {\bigwedge_{\mathclap{\quad \quad v \in \interval{u+1}{n}}}}\!
          (\atnom{\aaux}{1} \ \avarprop_v\,{\iff}\,\atnom{\aauxbis}{1} \ \avarprop_v)
        \big)
  \]
  The proof uses standard properties of numbers encoded in binary.
  Let $x,y$ be two natural numbers that can be represented in binary by using $n$ bits. Let us denote with $x_i$ (resp. $y_i$) the $i$-th bit of the binary representation of $x$ (resp. $y$).
  We have that $y = x + 1$ if and only if
    \begin{enumerate}[label=(\Alph*)]
      \item there is a position $i \in \interval{1}{n}$ such that $x_i = 0$ and $y_i = 1$;
      \item for every position $j > i$, $x_j = 0$ $\iff$ $y_j = 0$;
      \item for every position $j < i$, $x_j = 1$ and $y_j = 0$.
    \end{enumerate}
  Notice that (A) and (B) are as in the characterisation of $ x < y$ given in \Cref{lemma:lessjj}.
  The formula $\successor{\aaux}{\aauxbis}{1}$ uses exactly this characterisation in order to state that $\nb{\aworld_{\aauxbis}} = \nb{\aworld_{\aaux}}+1$.

  Since we are working under the hypothesis that $\amodel,\aworld \models \init{1} \land \fork{\aaux}{\aauxbis}{1}{1}$, there are two distinct worlds $\aworld_{\aaux}$ and $\aworld_{\aauxbis}$ corresponding to the two nominals $\aaux$ and $\aauxbis$ for the depth $1$, respectively.
  Then, the proof of this lemma follows closely the proof of \Cref{lemma:lessjj},
  and enforcing (C) by means of the subformula
  $\atnom{\aaux}{1}(\lnot \avarprop_u {\land}\!\!\bigwedge_{{v \in \interval{1}{u-1}}}\!\avarprop_v)
  \land
  \atnom{\aauxbis}{1} (\avarprop_u {\land}\!\!\bigwedge_{{v \in \interval{1}{u-1}}}\!\!\lnot \avarprop_v)$.
\end{proof}

\subsection{Base case $i = j$ / $j = 1$: Correctness of $\pU{1}$ and $\pC{1}$}

In the following statements and proofs, let  $\amodel = \triple{\worlds}{\arelation}{\avaluation}$ be a finite forest and $\aworld \in \worlds$.

\begin{lemma}\label{lemma:uniqone}
Suppose $\amodel,\aworld \models \init{1} \land \pA$.

$\amodel,\aworld \models \pU{1}$ if and only if
$\pair{\amodel}{\aworld}$ satisfies
\ref{prop:apU_1}, i.e.\ distinct $\avartree$-nodes in $\arelation(\aworld)$
encode different numbers.
\end{lemma}

\begin{proof}
Let us recall that
  $
  \pU{1} \egdef \lnot \big(\true \separate (\fork{\avariable}{\avariablebis}{1}{1} \land \equivalent{\avariable}{\avariablebis}{1}{1})\big)
  $
  where $\equivalent{\avariable}{\avariablebis}{1}{1}$ stands for
  $\lnot (\less{\avariable}{\avariablebis}{1}{1} \lor \less{\avariablebis}{\avariable}{1}{1})$.

($\Rightarrow$): Conversely, suppose that there are two distinct $\avartree$-nodes $\aworld_{\avariable}$ and $\aworld_{\avariablebis}$ encoding the same number.
Since $\amodel,\aworld \models \init{1} \land \pA$, every world in $\arelation(\aworld)$ has exactly one child satisfying $\avariable$ and exactly one (different) child satisfying $\avariablebis$.
Let us then consider the submodel $\amodel' = \triple{\worlds}{\arelation_1}{\avaluation}$
where $\arelation_1(\aworld) = \{\aworld_{\avariable},\aworld_{\avariablebis}\}$,
$\arelation_1(\aworld_{\avariable}) = \{\aworld_1\}$ and  $\arelation_1(\aworld_{\avariablebis}) = \{\aworld_2\}$,
so that $\aworld_1$ satisfies $\avariable$ whereas $\aworld_2$ satisfies $\avariablebis$.
By \Cref{lemma:forkjj}, $\amodel',\aworld \models \fork{\avariable}{\avariablebis}{1}{1}$. By hypothesis, $\nb{\aworld_{\avariable}} = \nb{\aworld_{\avariablebis}}$ and therefore we also have $\amodel',\aworld \models \equivalent{\avariable}{\avariablebis}{1}{1}$.
Thus, by definition, $\amodel,\aworld \not\models \pU{1}$.

($\Leftarrow$): Again conversely, suppose that $\amodel,\aworld \not\models \pU{1}$ and therefore $\amodel,\aworld \models \true \separate (\fork{\avariable}{\avariablebis}{1}{1} \land \equivalent{\avariable}{\avariablebis}{1}{1})$.
Then, by definition there is a submodel $\amodel' = \triple{\worlds}{\arelation_1}{\avaluation}$ of $\amodel$
such that $\amodel',\aworld \models \fork{\avariable}{\avariablebis}{1}{1} \land \equivalent{\avariable}{\avariablebis}{1}{1}$.
Moreover, since the satisfaction of $\init{1}$ is monotonic w.r.t.\ submodels, we have $\amodel',\aworld \models \init{1}$.
We can then apply \Cref{lemma:forkjj,lemma:lessjj} in order to conclude that there are two distinct worlds $\aworld_{\avariable}$ and $\aworld_{\avariablebis}$ in $\arelation'(\aworld)$ such that $\nb{\aworld_{\avariable}} = \nb{\aworld_{\avariablebis}}$.
Since the encoding of a number (for $j = 1$) only depends on the satisfaction of
the propositional symbols $\avarprop_1,\dots,\avarprop_n$ on a certain world,
we conclude that the same property holds for $\amodel$:
the two worlds $\aworld_{\avariable}$ and $\aworld_{\avariablebis}$ in $\arelation(\aworld)$ are such that $\nb{\aworld_{\avariable}} = \nb{\aworld_{\avariablebis}}$.
Therefore, $\pair{\amodel}{\aworld}$ does not satisfy \ref{prop:apU_1}.
\end{proof}

\begin{lemma}\label{lemma:complone}
Suppose $\amodel,\aworld \models \init{1} \land \pA$.

$\amodel,\aworld \models \pC{1}$ if and only if
$\pair{\amodel}{\aworld}$ satisfies
\ref{prop:apC_1}, i.e.\ for every $\avartree$-node $\aworld_1 \in \arelation(\aworld)$, if $\nb{\aworld_1} < 2^n-1$ then $\nb{\aworld_2} = \nb{\aworld_1}+1$ for some $\avartree$-node $\aworld_2 \in \arelation(\aworld)$.
\end{lemma}

\begin{proof}
Recall that $\pC{1}$ is defined as: 
  \begin{nscenter}
  $
  \lnot
    \big(\Box \bottom \separate
      \big(
        \HMBox{\avartree}\Diamond \avariablebis \land \atnom{\avariable}{1} \lnot \one_{1} \land \lnot ( \true \separate (\fork{\avariable}{\avariablebis}{1}{1} \land \successor{\avariable}{\avariablebis}{1}))
      \big)
    \big).
  $
\end{nscenter}

($\Rightarrow$):
Suppose $\amodel,\aworld \models \pC{1}$.
By definition of $\models$, this implies that
for any $\amodel'=\triple{\worlds}{\arelation'}{\avaluation}$ submodel of $\amodel$
such that $\arelation'(\aworld)=\arelation(\aworld)$, if
$\amodel',\aworld\models\HMBox{\avartree}\Diamond \avariablebis \land \atnom{\avariable}{1} \lnot \one_{1}$, then
$\amodel',\aworld\models \true \separate (\fork{\avariable}{\avariablebis}{1}{1} \land \successor{\avariable}{\avariablebis}{1})$.
Then, let us pick a $\avartree$-node $\aworld_\avariable\in\arelation'(\aworld)=\arelation(\aworld)$ such that $\nb{\aworld_\avariable}<2^n-1$. We show that there must be a world $\aworld_\avariablebis \in \arelation'(\aworld)$ such that $\nb{\aworld_\avariablebis} = \nb{\aworld_\avariable} + 1$.
Let us consider the submodel $\amodel'' = \triple{\worlds}{\arelation'}{\avaluation}$ of $\amodel$ such that for every $\overline{\aworld} \in \worlds$,
if $\overline{\aworld} \neq \aworld_\avariable$ then $\arelation'(\overline{\aworld}) = \arelation(\overline{\aworld})$ and otherwise $\arelation'(\aworld_\avariable) = \{\aworld_{1}\}$ where $\aworld_{1}$ is the only $\Aux$-child of $\aworld_\avariable$ (w.r.t.\ $\arelation$) satisfying $\avariable$.
Notice that $\aworld_{1}$ exists and it is unique by $\amodel,\aworld \models \init{1} \land \pA$.
Moreover, $\aworld_\avariable$ corresponds in $\amodel'$ to the nominal $\avariable$ for the depth $1$.
Again by $\amodel,\aworld \models \init{1} \land \pA$, we conclude that $\amodel',\aworld \models \HMBox{\avartree}\Diamond \avariablebis$.
Moreover, since $\nb{\aworld_\avariable}<2^n-1$, by \Cref{lemma:at} we have $\amodel',\aworld \models \atnom{\avariable}{1} \lnot \one_{1}$.
Hence by hypothesis, $\amodel',\aworld \models \true \separate (\fork{\avariable}{\avariablebis}{1}{1} \land \successor{\avariable}{\avariablebis}{1})$.
Then, let $\amodel'' = \triple{\worlds}{\arelation''}{\avaluation} \sqsubseteq \amodel'$
be such that $\amodel'',\aworld \models \fork{\avariable}{\avariablebis}{1}{1} \land \successor{\avariable}{\avariablebis}{1}$.
By \Cref{lemma:forkjj,lemma:succone}, there is $\aworld_\avariablebis\in\arelation''(\aworld)$ such that $\nb{\aworld_\avariablebis}=\nb{\aworld_\avariable}+1$.
Since the encoding of a number (for $j = 1$) only depends on the satisfaction of
the propositional symbols $\avarprop_1,\dots,\avarprop_n$ on a certain world,
we conclude that the same property holds for $\amodel$.
Thus, $\pair{\amodel}{\aworld}$ satisfies \ref{prop:apC_1}.

($\Leftarrow$): Suppose that $\pair{\amodel}{\aworld}$ satisfies \ref{prop:apC_1}, and \emph{ad absurdum} assume that $\amodel,\aworld \not\models \pC{1}$,
hence $\amodel,\aworld \models \Box \bottom \separate
    \big(
      \HMBox{\avartree}\Diamond \avariablebis \land \atnom{\avariable}{1} \lnot \one_{1} \land \lnot ( \true \separate (\fork{\avariable}{\avariablebis}{1}{1} \land \successor{\avariable}{\avariablebis}{1}))
    \big).
$
Then, there is a submodel $\amodel' = \triple{\worlds}{\arelation'}{\avaluation}$ of $\amodel$ such that $\arelation'(\aworld) = \arelation(\aworld)$ and
$\amodel',\aworld \models
    \HMBox{\avartree}\Diamond \avariablebis \land \atnom{\avariable}{1} \lnot \one_{1} \land \lnot ( \true \separate (\fork{\avariable}{\avariablebis}{1}{1} \land \successor{\avariable}{\avariablebis}{1}))$.
Notice that this formula does not enforce $\avariable$ to be a nominal for the depth $1$, however from $\amodel',\aworld \models \atnom{\avariable}{1} \lnot \one_{1}$
we deduce that there is at least one $\avartree$-node $\aworld_{\avariable}$ such that $\amodel',\aworld_{\avariable} \models \Diamond \avariable \land \lnot \one_{1}$.
Then, $\nb{\aworld_{\avariable}} < 2^n-1$ and by hypothesis there is a $\avartree$-node $\aworld_{\avariablebis}$ such that $\nb{\aworld_{\avariablebis}} = \nb{\aworld_{\avariable}} + 1$.
Let us consider now the submodel $\amodel'' = \triple{\worlds}{\arelation''}{\avaluation}$ of $\amodel'$
where $\arelation''(\aworld) = \{\aworld_{\avariable},\aworld_{\avariablebis}\}$,
$\arelation''(\aworld_{\avariable}) = \{\aworld_1\}$ and $\arelation''(\aworld_{\avariablebis}) = \{\aworld_2\}$,
where $\aworld_1$ (resp. $\aworld_2$) is the only $\Aux$-children of $\aworld_{\avariable}$ (resp. $\aworld_{\avariablebis}$)
that satisfies $\avariable$ (resp. $\avariablebis$).
The existence of $\aworld_1$ and $\aworld_2$ is guaranteed
by $\amodel',\aworld_{\avariable} \models \Diamond \avariable \land \lnot \one_{1}$
and $\amodel',\aworld \models \HMBox{\avartree}\Diamond \avariablebis$.
By \Cref{lemma:forkjj},
$\amodel'',\aworld \models \fork{\avariable}{\avariablebis}{1}{1}$.
Moreover, as the encoding of a number (for $j = 1$) only depends on the satisfaction of
the propositional symbols $\avarprop_1,\dots,\avarprop_n$ on a certain world,
$\amodel'',\aworld \models \successor{\avariable}{\avariablebis}{1}$.
Then, we conclude that $\amodel',\aworld \models \true \separate (\fork{\avariable}{\avariablebis}{1}{1} \land \successor{\avariable}{\avariablebis}{1})$, in contradiction with
$\amodel',\aworld \models
    \HMBox{\avartree}\Diamond \avariablebis \land \atnom{\avariable}{1} \lnot \one_{1} \land \lnot ( \true \separate (\fork{\avariable}{\avariablebis}{1}{1} \land \successor{\avariable}{\avariablebis}{1}))$.
Thus, $\amodel,\aworld \models \pC{1}$.
\end{proof}

\subsection{Proof of \Cref{lemma:tower-hardness-base-case} and satisfiability of $\complete{1}$}

\begin{proof}(\Cref{lemma:tower-hardness-base-case})
Follows directly from \Cref{lemma:basicform,lemma:uniqone,lemma:complone}.
\end{proof}

A quick check of $\init{1}$ and the conditions \ref{prop:apS_1}, \ref{prop:apZ_1},
\ref{prop:apU_1}, \ref{prop:apC_1} and \ref{prop:apA} should convince the reader that they are simultaneously satisfiable, leading to $\init{1} \land \complete{1}$ being satisfiable.
However, in the following we provide an explicit model satisfiying this formula

\begin{lemma}\label{A:lemma:complete-one-satisfiable}
$\init{1} \land \complete{1}$ is satisfiable.
\end{lemma}

\begin{proof}
  Consider the finite forest $\amodel = \triple{\worlds}{\arelation}{\avaluation}$ and a world $\aworld$ such that
  \begin{enumerate}
    \item $\arelation$ is the minimal set of pairs such that $\arelation(\aworld) = \{\aworld_0 ,\dots, \aworld_{2^n-1}\}$ (where $\aworld_0 ,\dots, \aworld_{2^n-1}$ are all distinct worlds),
    and for every $i \in \interval{0}{2^n-1}$, $\arelation(\aworld_i) = \{\aworld_i^\avariable,\aworld_i^\avariablebis\}$ (again, $\aworld_i^\avariable,\aworld_i^\avariablebis$ distincts);
    \item $\worlds = \{\aworld\} \cup \arelation(\aworld) \cup \bigcup_{\aworld' \in \arelation(\aworld)} \arelation(\aworld')$;
    \item $\avaluation(\avariable) = \{\aworld_0^\avariable,\dots,\aworld_{2^n-1}^\avariable\}$, $\avaluation(\avariablebis) = \{\aworld_0^\avariablebis,\dots,\aworld_{2^n-1}^\avariablebis\}$
    and for every $i \in \interval{0}{2^n-1}$ and $j \in \interval{1}{n}$, $\aworld_i \in \avaluation(\avarprop_j)$ if and only if the $j$-th bit in the binary encoding of $i$ is $1$.
  \end{enumerate}
  It is easy to check that $\pair{\amodel}{\aworld}$ satisfies $\init{1}$ as well as \ref{prop:apS_1}, \ref{prop:apZ_1},
  \ref{prop:apU_1}, \ref{prop:apC_1} and \ref{prop:apA}. Thus, by \Cref{lemma:tower-hardness-base-case} $\amodel,\aworld \models \init{1} \land \complete{1}$.
\end{proof}

\subsection{Inductive case $1 \leq i < j$ : Correctness of $\fork{\aaux}{\aauxbis}{i}{j}$, $\lsrpartition{j}$, $\less{\aaux}{\aauxbis}{i}{j}$ and $\successor{\aaux}{\aauxbis}{j}$}

In the following statements and proofs, let $\amodel = \triple{\worlds}{\arelation}{\avaluation}$ be a finite forest and $\aworld \in \worlds$. Let $1 \leq i < j$.
We show the correctness of the definitions of $\fork{\aaux}{\aauxbis}{i}{j}$, $\lsrpartition{j}$, $\less{\aaux}{\aauxbis}{i}{j}$ and $\successor{\aaux}{\aauxbis}{j}$, under the inductive hypothesis
that all the statements in \Cref{appendix:inductive-formulae}
holds for all $i',j' \in \Nat$ such that $1 \geq i' \geq j' \geq j$ and ($j' < j$ or $j' - i' < j - i$).

First of all, assume for a moment that $\complete{j}$ is correctly defined, with semantics as in \ref{appendix:inductive-formulae}. Then the following result holds.

\begin{lemma}\label{A:the-property-of-types}
Let $0 \leq i \leq j$ with $j \geq 2$.
Let $\amodel = \triple{\worlds}{\arelation}{\avaluation}$ and $\aworld \in \worlds$
such that $\amodel,\aworld \models \init{j} \land \complete{j}$.
Consider a world $\aworld' \in \arelation^i(\aworld)$ and a number $m \in \interval{0}{\amapter(j-i,n)-1}$.
Lastly, suppose $\amodel' \sqsubseteq \amodel$ such that $\amodel',\aworld' \models \complete{j-i}$.
Then,
\begin{nscenter}
$\nbexp{\aworld'}{j-i} = m$ w.r.t.\ $\pair{\amodel}{\aworld'}$
if and only if $\nbexp{\aworld'}{j-i} = m$ w.r.t.\ $\pair{\amodel'}{\aworld'}$.
\end{nscenter}
\end{lemma}

\begin{proof}
  The proof is rather straightforward. From the semantics of $\complete{j}$, w.r.t.\ any of the two models $\pair{\amodel}{\aworld'}$ or $\pair{\amodel'}{\aworld'}$,
  $\nbexp{\aworld'}{j-i}$ is encoded by using
  \begin{enumerate}
  \item the $\avartree$-nodes reachable from $\aworld'$ in at most $j-i$ steps;
  \item the $\{\avariable,\avariablebis\}$-nodes reachable from $\aworld'$ in exactly $2$ steps;
  \item the $\Aux$-nodes reachable from $\aworld'$ in at least $3$ steps and at most $j-i+1$ steps.
  \end{enumerate}
  Let $\amodel' = \triple{\worlds}{\arelation_1}{\avaluation}$.
  From $\amodel',\aworld' \models \complete{j-i}$ we can show that the accessibility to all these nodes is preserved between $\pair{\amodel}{\aworld'}$ and $\pair{\amodel'}{\aworld'}$, leading to the result (or rather, that losing the accessibility to any of these nodes leads to a model not satisfying $\complete{j-i}$).
  Indeed,
  \begin{enumerate}
    \item suppose that there is a $\avartree$-node $\overline{\aworld} \in \arelation^{k}(\aworld')$,
    with $k \in \interval{1}{j-i}$,  not in $\arelation_1^k(\aworld')$.
    Let $\overline\aworld_1$ be the parent of $\overline{\aworld}$ in $\arelation$. Then in particular, $\overline\aworld_1 \in \arelation^{k-1}(\aworld')$ and $\pair{\overline\aworld_1}{\overline\aworld} \in \arelation$.
    Since $\overline\aworld \not\in \arelation_1^k(\aworld')$, we conclude that
    $\pair{\amodel'}{\overline\aworld_1}$ does not satisfy \ref{prop:apC} and therefore ${\amodel'},{\overline\aworld_1} \not \models \complete{j-i-k}$.
    Then, $\pair{\amodel'}{\aworld'}$ cannot satisfy \ref{prop:apS}, in contradiction with ${\amodel'},\aworld' \models \complete{j-i}$;

    \item suppose that one $\{\avariable,\avariablebis\}$-node in $\arelation^2(\aworld')$ is not in $\arelation_1^2(\aworld')$. Then trivially $\pair{\amodel'}{\aworld'}$ cannot satisfy \ref{prop:apA}, in contradiction with $\amodel',\aworld' \models \complete{j}$;
    \item similarly, suppose that one $\Aux$-node in $\arelation^k(\aworld')$, where $k \in \interval{3}{j-i+1}$, is not in $\arelation_1^2(\aworld')$. Then again $\pair{\amodel'}{\aworld'}$ cannot satisfy \ref{prop:apA}, in contradiction with $\amodel',\aworld' \models \complete{j}$. \qedhere
  \end{enumerate}
\end{proof}

\begin{lemma}\label{lemma:forkij}
Let $\aaux \neq \aauxbis \in \Aux$ and
$1 \leq i < j$. Suppose $\amodel,\aworld \models \init{j}$.

$\amodel,\aworld \models \fork{\aaux}{\aauxbis}{i}{j}$ if and only if
{\bfseries \itshape{(i)}} $\aworld$ has exactly two $\avartree$-children and exactly two paths of $\avartree$-nodes, both of length $i$;
{\bfseries \itshape{(ii)}} one of these two paths ends on a world (say $\aworld_{\aaux}$) corresponding to the nominal $\aaux$ whereas the other ends on a world (say $\aworld_{\aauxbis}$) corresponding to the nominal $\aauxbis$;
{\bfseries \itshape{(iii)}} $\pair{\amodel}{\aworld_{\aaux}}$ and $\pair{\amodel}{\aworld_{\aauxbis}}$ satisfy
 $\completeplus{j - i} \egdef \complete{j - i} \land \HMBox{\avartree}(\Diamond \avarleft \land \Diamond \avarselect \land \Diamond \avarright)$.
\end{lemma}

\begin{proof}
  Recall that $\fork{\aaux}{\aauxbis}{i}{j}$ is defined as $\fork{\aaux}{\aauxbis}{i}{i} \land \HMBox{\avartree}^i\completeplus{j-i}$.
We have:
\begin{itemize}
  \item $\amodel,\aworld \models \fork{\aaux}{\aauxbis}{i}{i}$ if and only if
  (by \Cref{lemma:forkjj})
  {\bfseries \itshape{(i)}} $\aworld$ has exactly two $\avartree$-children and
  exactly two paths of $\avartree$-nodes, both of length $j$;
  {\bfseries \itshape{(ii)}} one of these two paths ends on a world corresponding
  to the nominal $\aaux$ whereas the other ends on a world corresponding to the
  nominal $\aauxbis$.
  \item  Let $\aworld_{\aaux},\aworld_{\aauxbis}\in\arelation^i(\aworld)$, since
  $\amodel,\aworld \models \HMBox{\avartree}^i\completeplus{j - i}$ we get
  $\amodel,\aworld' \models \completeplus{j - i}$, for $\aworld'\in\set{\aworld_{\aaux},\aworld_{\aauxbis}}$.
\end{itemize}
This concludes the proof.
\end{proof}

\begin{lemma}\label{lemma:lsr}
Let 
$1 \leq i < j$.
Suppose $\amodel,\aworld \models \init{j}$.

$\amodel,\aworld \models \lsrpartition{j-i}$
if and only if
\begin{enumerate}
  \item $\amodel,\aworld \models \complete{j-i}$;
  \item every $\avartree$-node in $\arelation(\aworld)$ has exactly one $\Aux$-child satisfying an atomic proposition from $\{\avarleft,\avarselect,\avarright\}$;
  \item exactly one $\avartree$-node in $\arelation(\aworld)$ (say $\aworld_{\avarselect}$) has an $\Aux$-child satisfying $\avarselect$;
  \item given $\aworld' \in \arelation(\aworld)$, $\aworld'$ has an $\Aux$-child satisfying $\avarleft$ if and only if $\nb{\aworld'} > \nb{\aworld_{\avarselect}}$;
  \item given $\aworld' \in \arelation(\aworld)$, $\aworld'$ has an $\Aux$-child satisfying $\avarright$
  if and only if
  $\nb{\aworld'} < \nb{\aworld_{\avarselect}}$.
\end{enumerate}
\end{lemma}

\begin{proof}
This proof is rather straightforward. Recall that $\lsrpartition{j-i}$ is defined as
\[
\complete{j-i} \land \HMBox{\avartree}\Diamond_{=1}(\avarleft \lor \avarselect \lor \avarright)
\land \nominal{\avarselect}{1} \land
 \lnot (\true \separate (\fork{\avarselect}{\avarleft}{1}{j-i} \land \lnot \less{\avarselect}{\avarleft}{1}{j-i})) \land
\lnot (\true \separate (\fork{\avarselect}{\avarright}{1}{j-i} \land \lnot \less{\avarright}{\avarselect}{1}{j-i})).
\]
Then,
\begin{itemize}
  \item the first conjunct of $\lsrpartition{j-i}$, i.e.\ $\complete{j-i}$, directly realises the requirement (1);
  \item the second conjunct of $\lsrpartition{j-i}$, i.e.\ $\HMBox{\avartree}\Diamond_{=1}(\avarleft \lor \avarselect \lor \avarright)$,
  directly realises the requirement (2);
  \item the third conjunct of $\lsrpartition{j-i}$, i.e.\ $\nominal{\avarselect}{1}$,
  directly realised the requirement (3);
  \item the fourth conjunct of $\lsrpartition{j-i}$ realises the requirement (4).
  Suppose $\amodel,\aworld\models \lnot (\true \separate (\fork{\avarselect}{\avarleft}{1}{j-i} \land \lnot \less{\avarselect}{\avarleft}{1}{j-i}))$. Then,
  for all submodels $\amodel'\sqsubseteq\amodel$,
  if $\amodel',\aworld\models\fork{\avarselect}{\avarleft}{1}{j-i}$ then
  $\amodel',\aworld\models\less{\avarselect}{\avarleft}{1}{j-i}$. Let $\aworld'\in\arelation(\aworld)$ be
  such that $\aworld'$ has an
  $\Aux$-child satisfying $\avarleft$. Then by \Cref{lemma:forkij} $\amodel,\aworld\models\fork{\avarselect}{\avarleft}{1}{j-1}$ and as a consequence
  $\amodel,\aworld\models\less{\avarselect}{\avarleft}{1}{j-i}$.
  Let us consider $\amodel' = \triple{\worlds}{\arelation'}{\worlds}$
  obtained from $\amodel$ by removing from $\arelation$ every pair $\pair{\aworld_1}{\aworld_2} \in \arelation$ such that
  \begin{itemize}
    \item $\aworld_1$ and $\aworld_2$ are $\avartree$-nodes;
    \item $\pair{\aworld_1}{\aworld_2}$ does not belong to the path from $\aworld$ to $\aworld_\avarselect$, nor to the path from $\aworld$ to $\aworld'$;
    \item $\pair{\aworld_1}{\aworld_2}$ does not belong to any path starting from $\aworld_\avarselect$ or $\aworld'$.
  \end{itemize}
  Then, we can show that $\amodel',\aworld \models \fork{\avarselect}{\avarleft}{1}{j-i}$ and therefore, by hypothesis,
  $\amodel',\aworld \models \less{\avarselect}{\avarleft}{1}{j-i}$.
  By inductive hypothesis, from $\less{\avarselect}{\avarleft}{1}{j-i}$ we conclude that $\nb{\aworld'}>\nb{\aworld_\avarselect}$ with respect to $\pair{\amodel'}{\aworld}$.
  Now, from $\amodel',\aworld \models \fork{\avarselect}{\avarleft}{1}{j-i}$
  we also conclude that $\amodel',\aworld_{\avarselect} \models \complete{j-i}$ and $\amodel',\aworld' \models \complete{j-i}$.
  Then, by \ref{A:the-property-of-types}, $\nb{\aworld'}>\nb{\aworld_\avarselect}$ also holds with respect to  $\pair{\amodel}{\aworld}$.
  The other direction is analogous;
  \item the fifth conjunct of $\lsrpartition{j-i}$ realises the requirement (5). The proof is similar to the one for the requirement (4), just above. \qedhere
\end{itemize}
\end{proof}

We prove a technical lemma that will help us with the proof of correctness of $\less{\aaux}{\aauxbis}{i}{j}$ and $\successor{\aaux}{\aauxbis}{j}$.
\begin{lemma}\label{A:lemma:left-and-right}
  Let $\aaux \neq \aauxbis \in \Aux$ and
  $1 \leq i < j$. Suppose that $\pair{\amodel}{\aworld}$ is such that $\arelation^i(\aworld) = \{\aworld_{\aaux},\aworld_{\aauxbis}\}$ for some $\avartree$-nodes $\aworld_{\aaux}$ and $\aworld_{\aauxbis}$ in $\worlds$, and these two worlds satisfy the conditions of $\lsrpartition{j-i}$, i.e. for every $b \in \{\aaux,\aauxbis\}$
  \begin{enumerate}[label=(\Alph*)]
    \item $\amodel,\aworld_b \models \complete{j-i}$;
    \item every $\avartree$-node in $\arelation(\aworld_b)$ has exactly one $\Aux$-child satisfying an atomic proposition from $\{\avarleft,\avarselect,\avarright\}$;
    \item exactly one $\avartree$-node in $\arelation(\aworld_b)$ (say $\aworld_{b,\avarselect}$) has an $\Aux$-child satisfying $\avarselect$;
    \item given $\aworld' \in \arelation(\aworld_b)$, $\aworld'$ has an $\Aux$-child satisfying $\avarleft$ if and only if $\nb{\aworld'} > \nb{\aworld_{b,\avarselect}}$;
    \item given $\aworld' \in \arelation(\aworld_b)$, $\aworld'$ has an $\Aux$-child satisfying $\avarright$
    if and only if
    $\nb{\aworld'} < \nb{\aworld_{b,\avarselect}}$.
  \end{enumerate}
  Then,
  \begin{enumerate}[label=\Roman*.]
    \item $\amodel,\aworld \models \selectpred{\aaux}{\aauxbis}{i}{j}$ if and only if $\nb{\aworld_{\aaux,\avarselect}} = \nb{\aworld_{\aauxbis,\avarselect}}$, $\amodel,\aworld_{\aaux,\avarselect} \models \lnot \treeval$ and
    $\amodel,\aworld_{\aauxbis,\avarselect} \models \treeval$;
    \item $\amodel,\aworld \models \leftpred{\aaux}{\aauxbis}{i}{j}$ if and only if for every two worlds $\aworld_{\aaux,\avarleft} \in \arelation(\aworld_{\aaux})$
    and
    $\aworld_{\aauxbis,\avarleft} \in \arelation(\aworld_{\aauxbis})$ such that
    $\nb{\aworld_{\aaux,\avarleft}} > \nb{\aworld_{\aaux,\avarselect}}$ and
    $\nb{\aworld_{\aauxbis,\avarleft}} > \nb{\aworld_{\aauxbis,\avarselect}}$,
    if $\nb{\aworld_{\aaux,\avarleft}} = \nb{\aworld_{\aauxbis,\avarleft}}$ then,
      $\amodel,\aworld_{\aaux,\avarleft} \models \treeval$
      if and only if
      $\amodel,\aworld_{\aauxbis,\avarleft} \models \treeval$.
    \item If $i = 1$ then, $\amodel,\aworld \models \rightpred{\aaux}{\aauxbis}$
    if and only if
    \begin{itemize}
      \item for every world $\aworld_{\aaux,\avarright} \in \arelation(\aworld_{\aaux})$, if $\nb{\aworld_{\aaux,\avarright}} < \nb{\aworld_{\aaux,\avarselect}}$ then $\amodel,\aworld_{\aaux,\avarright} \models \treeval$;
      \item for every world $\aworld_{\aauxbis,\avarright} \in \arelation(\aworld_{\aauxbis})$, if $\nb{\aworld_{\aauxbis,\avarright}} < \nb{\aworld_{\aauxbis,\avarselect}}$ then $\amodel,\aworld_{\aauxbis,\avarright} \models \lnot\treeval$.
    \end{itemize}
  \end{enumerate}
\end{lemma}

\begin{proof}
  We will prove each item.
   \begin{enumerate}[label=(Proof of \Roman*),align = left]
    \item We recall that $\selectpred{\aaux}{\aauxbis}{i}{j}$ is defined as
    \begin{nscenter}
      $\true \separate \big( \fork{\avariable}{\avariablebis}{i+1}{j} \land
      \atnom{\aaux}{i}\HMDiamond{\avartree}(\Diamond \avarselect \land \Diamond \avariable)\, \land \atnom{\aauxbis}{i}\HMDiamond{\avartree}(\Diamond \avarselect \land \Diamond \avariablebis) \land \equivalent{\avariable}{\avariablebis}{i+1}{j} \land \atnom{\avariable}{i+1}\lnot \treeval \land \atnom{\avariablebis}{i+1}\treeval
      \big)$.
    \end{nscenter}

      ($\Rightarrow$): Suppose $\amodel,\aworld \models \selectpred{\aaux}{\aauxbis}{i}{j}$. By unfolding the definition just given,
      there exists $\amodel'=\tup{\worlds,\arelation_1,\avaluation}$, such that $\amodel'\sqsubseteq\amodel$
      and:
      \begin{enumerate}
        \item $\aworld$ has exactly two $\avartree$-children and exactly two paths of $\avartree$-nodes, both of length $i+1$;
        \item one of these two paths ends on a world (say $\aworld_{\avariable}$) corresponding to the nominal $\avariable$ whereas the other ends on a world (say $\aworld_{\avariablebis}$) corresponding to the nominal $\avariablebis$;
        \item there exists a $\avartree$-world $\aworld_\aaux\in\arelation_1^i(\aworld)$ corresponding to the nominal $\aaux$
        such that $\amodel',\aworld_\aaux\models\HMDiamond{\avartree}(\Diamond\avarselect\land\Diamond\avariable)$;
        \item there exists a $\avartree$-world $\aworld_\aauxbis\in\arelation_1^i(\aworld)$ corresponding to the nominal $\aauxbis$
        such that $\amodel',\aworld_\aauxbis\models\HMDiamond{\avartree}(\Diamond\avarselect\land\Diamond\avariablebis)$;
        \item $\amodel',\aworld\models\equivalent{\avariable}{\avariablebis}{i+1}{j}$;
        \item $\amodel',\aworld_\avariable\models\neg\treeval$ and $\amodel',\aworld_\avariablebis\models\treeval$.
      \end{enumerate}

      Let $\aworld_{\aaux,\avarselect}\in\arelation_1(\aworld_\aaux)$ and $\aworld_{\aauxbis,\avarselect}\in\arelation_1(\aworld_\aauxbis)$ be such that
      they are the only $\avartree$-children of $\aworld_\aaux$ and $\aworld_\aauxbis$
      respectively, having a child satisfying $\avarselect$ (notice they exist due
      to hypothesis $(C)$). Notice by item b. above, there exists $\aworld'\in\arelation_1(\aworld_\aaux)$ such that
      $\amodel',\aworld'\models\avartree$ and $\amodel',\aworld\models\Diamond\avarselect\land\Diamond\avariable$.
      Since $\aworld_{\aaux,\avarselect}$ is the only child of $\aworld_\aaux$
      having an $\avarselect$-child, then $\aworld_{\aaux,\avarselect}=\aworld'$,
      and as a consequence $\amodel',\aworld_{\aaux,\avarselect}\models\Diamond\avariable$.
      The same argument can be applied by using item c. above in order to get
      $\amodel'\aworld_{\aauxbis,\avarselect}\models\Diamond\avariablebis$.
      By item a. and b. above, we have that the corresponding $\aworld_\avariable$
      and $\aworld_\avariablebis$ must be the unique $\avartree$-worlds at distance
      $i+1$ of $\aworld$ having $\avariable$ and $\avariablebis$ children, respectively.
      Therefore, we have necessarily $\aworld_{\aaux,\avarselect}=\aworld_\avariable$ and $\aworld_{\aauxbis,\avarselect}=\aworld_\avariablebis$, so $\amodel,\aworld_{\aaux,\avarselect}\models\neg\treeval$ and $\amodel,\aworld_{\aauxbis,\avarselect}\models\treeval$ as wanted (by using item f. above).

      Finally, by applying inductive hypothesis on item e.,
      together with \Cref{A:the-property-of-types},
      we get
      $\nb{\aworld_{\aaux,\avarselect}} = \nb{\aworld_{\aauxbis,\avarselect}}$,
      which concludes the proof of this direction.

      ($\Leftarrow$): For this direction, we can use a similar argument backwards.

    \item We recall that $\leftpred{\aaux}{\aauxbis}{i}{j}$ is defined as
    \begin{nscenter}
      $\lnot \big( \true \separate \big(
      \fork{\avariable}{\avariablebis}{i+1}{j} \land
      \atnom{\aaux}{i}\HMDiamond{\avartree}(\Diamond \avarleft \land \Diamond \avariable)\, \land
      \atnom{\aauxbis}{i}\!\HMDiamond{\avartree}(\Diamond \avarleft \land \Diamond \avariablebis)
      \land  \equivalent{\avariable}{\avariablebis}{i+1}{j}{\land}
      \lnot (\atnom{\avariable}{i+1}\treeval \iff \atnom{\avariablebis}{i+1}\treeval)
      \big)\big)$.
    \end{nscenter}
    Notice also that by definition of the satisfaction relation $\models$, we have that $\amodel,\aworld\models\leftpred{\aaux}{\aauxbis}{i}{j}$ if and only if for all $\amodel'=\tup{\worlds,\arelation_1,\avaluation}$ such that $\amodel'\sqsubseteq\amodel$, we have
    \begin{nscenter}
      $
      \amodel',\aworld\models (\fork{\avariable}{\avariablebis}{i+1}{j} \land
      \atnom{\aaux}{i}\HMDiamond{\avartree}(\Diamond \avarleft \land \Diamond \avariable)\, \land
      \atnom{\aauxbis}{i}\!\HMDiamond{\avartree}(\Diamond \avarleft \land \Diamond \avariablebis)
      \land  \equivalent{\avariable}{\avariablebis}{i+1}{j}) \implies
      (\atnom{\avariable}{i+1}\treeval \iff \atnom{\avariablebis}{i+1}\treeval)
      $
    \end{nscenter}

    ($\Rightarrow$): Suppose $\amodel,\aworld\models\leftpred{\aaux}{\aauxbis}{i}{j}$. Then, for all $\amodel'=\tup{\worlds,\arelation_1,\avaluation}$ such that $\amodel'\sqsubseteq\amodel$, if the following conditions hold
    \begin{enumerate}
      \item $\aworld$ has exactly two $\avartree$-children and exactly two paths of $\avartree$-nodes, both of length $i+1$;
      \item one of these two paths ends on a world (say $\aworld_{\avariable}$) corresponding to the nominal $\avariable$ whereas the other ends on a world (say $\aworld_{\avariablebis}$) corresponding to the nominal $\avariablebis$;
      \item there exists a $\avartree$-world $\aworld_\aaux\in\arelation_1^i(\aworld)$ corresponding to the nominal $\aaux$
      such that $\amodel',\aworld_\aaux\models\HMDiamond{\avartree}(\Diamond\avarleft\land\Diamond\avariable)$;
      \item there exists a $\avartree$-world $\aworld_\aauxbis\in\arelation_1^i(\aworld)$ corresponding to the nominal $\aauxbis$
      such that $\amodel',\aworld_\aauxbis\models\HMDiamond{\avartree}(\Diamond\avarleft\land\Diamond\avariablebis)$;
      \item $\amodel',\aworld\models\equivalent{\avariable}{\avariablebis}{i+1}{j}$;
    \end{enumerate}
    then it follows that

    \begin{enumerate}
      \setcounter{enumii}{5}
      \item $\amodel',\aworld_\avariable \models \treeval$ iff $\amodel',\aworld_\avariablebis \models \treeval$.
    \end{enumerate}
    By hypothesis,  there exist $\aworld_\aaux,\aworld_\aauxbis$ at
    distance $i$ from $\aworld$ corresponding to nominals $\aaux$ and $\aauxbis$, respectively.
    Let $\aworld_{\aaux,\avarleft}\in\arelation(\aworld_\aaux))$ and $\aworld_{\aauxbis,\avarleft}\in\arelation(\aworld_\aauxbis))$
    such that $\nb{\aworld_{\aaux,\avarleft}}>\nb{\aworld_{\aaux,\avarselect}}$
    and $\nb{\aworld_{\aauxbis,\avarleft}}>\nb{\aworld_{\aauxbis,\avarselect}}$.
    If we are able to satisfy all the conditions a.--e. above, we can conclude
    what we want.  Suppose $\nb{\aworld_{\aaux,\avarleft}}=\nb{\aworld_{\aauxbis,\avarleft}}$. By the induction hypothesis, together with \Cref{A:the-property-of-types}, we get $\amodel,\aworld\models\equivalent{\avariable}{\avariablebis}{i+1}{j}$. Also, since by hypothesis
    $\amodel,\aworld_b\models\complete{j-i}$, for $\aworld_b\in\set{\aworld_\aaux,\aworld_\aauxbis}$, then it is easy to check that
    we satisfy the remaining conditions above. Therefore we can conclude
    $\amodel',\aworld_\avariable \models \treeval$ iff $\amodel',\aworld_\avariablebis \models \treeval$.

    ($\Leftarrow$): The other direction uses similar steps backwards.

    \item We recall that $\rightpred{\aaux}{\aauxbis} \egdef \atnom{\aaux}{1}\HMBox{\avartree}(\Diamond \avarright \implies \treeval)
    \land \atnom{\aauxbis}{1}\HMBox{\avartree}(\Diamond \avarright \implies \lnot \treeval)$.

    ($\Rightarrow$): Suppose $\amodel,\aworld\models\rightpred{\aaux}{\aauxbis}$. By unfolding the definition above, there exist two distinct
    $\avartree$-nodes $\aworld_\aaux,\aworld_\aauxbis\in\arelation(\aworld)$,
    corresponding to nominals $\aaux$ and $\aauxbis$ respectively, such that:
    \begin{enumerate}
      \item $\amodel,\aworld_\aaux\models\HMBox{\avartree}(\Diamond \avarright \implies \treeval)$, and
      \item $\amodel,\aworld_\aauxbis\models\HMBox{\avartree}(\Diamond \avarright \implies \neg\treeval)$.
    \end{enumerate}
    By item $(C)$ in the hypothesis, we know that there is exactly one
    $\avartree$-node in $\arelation(\aworld_\aaux)$ (say $\aworld_{\aaux,\avarselect}$) having an $\Aux$-child satisfying $\avarselect$.
    Let $\aworld_{\aaux,\avarright}\in\arelation(\aworld_\aaux)$ be such that
    $\nb{\aworld_{\aaux,\avarright}} <\nb{\aworld_{\aaux,\avarselect}}$.
    By item $(E)$ in the hypothesis, there exists $\aworld'\in\arelation(\aworld_{\aaux,\avarright})$ such that $\amodel,\aworld'\models\avarright$, so  $\amodel,\aworld_{\aaux,\avarright}\models\Diamond\avarright$. As a consequence, by the item a. above, we have $\amodel,\aworld_{\aaux,\avarright}\models\treeval$.

    By applying the same reasoning with $\aworld_{\aauxbis,\avarright}\in\arelation(\aworld_\aauxbis)$ such that
    $\nb{\aworld_{\aauxbis,\avarright}} <\nb{\aworld_{\aauxbis,\avarselect}}$, and
    the item b. above, we get $\amodel,\aworld_{\aauxbis,\avarright}\models\neg\treeval$.

    ($\Leftarrow$): This direction uses similar arguments (backwards).
   \end{enumerate}
\end{proof}

\begin{lemma}\label{lemma:lessij}
  Let $\aaux \neq \aauxbis \in \Aux$ and
  $1 \leq i < j$.
  Suppose $\amodel,\aworld \models \init{j} \land \fork{\aaux}{\aauxbis}{i}{j}$.

  $\amodel,\aworld \models \less{\aaux}{\aauxbis}{i}{j}$ if and only if
  there are two distinct $\avartree$-nodes $\aworld_{\aaux},\aworld_{\aauxbis} \in \arelation^i(\aworld)$ such that $\aworld_{\aaux}$ corresponds to the nominal $\aaux$, $\aworld_{\aauxbis}$ corresponds to the nominal $\aauxbis$ and
  $\nb{\aworld_{\aaux}} < \nb{\aworld_{\aauxbis}}$.
  \end{lemma}

\begin{proof}
  Recall that $\less{\aaux}{\aauxbis}{i}{j}$ is defined as $\true \separate (
    \twonoms{\aaux}{\aauxbis}{i} \land \HMBox{\avartree}^i\lsrpartition{j-i} \land
  \selectpred{\aaux}{\aauxbis}{i}{j} \land \leftpred{\aaux}{\aauxbis}{i}{j}
  )$.
  As in \Cref{lemma:lessjj}, the proof uses standard properties of numbers encoded in binary.
  Again, let $x,y$ be two natural numbers that can be represented in binary by using $n$ bits. Let us denote with $x_i$ (resp. $y_i$) the $i$-th bit of the binary representation of $x$ (resp. $y$).
  We have that $x < y$ if and only if
    \begin{enumerate}[label=(\Alph*)]
      \item there is a position $i \in \interval{1}{n}$ such that $x_i = 0$ and $y_i = 1$;
      \item for every position $j > i$, $x_j = 0$ $\iff$ $y_j = 0$.
    \end{enumerate}
  The formula $\less{\aaux}{\aauxbis}{i}{j}$ uses exactly this characterisation in order to state that $\nb{\aworld_{\aaux}} < \nb{\aworld_{\aauxbis}}$.

  Suppose $\amodel,\aworld \models \init{j} \land \fork{\aaux}{\aauxbis}{i}{j}$.
  From \Cref{lemma:forkij}, in $\pair{\amodel}{\aworld}$ it holds that
  \begin{enumerate}[label=(\roman*)]
  \item $\aworld$ has exactly two $\avartree$-children and exactly two paths of $\avartree$-nodes, both of length $i$;
  \item one of these two paths ends on a world (say $\aworld_{\aaux}$) corresponding to the nominal $\aaux$ whereas the other ends on a world (say $\aworld_{\aauxbis}$) corresponding to the nominal $\aauxbis$;
  \item $\pair{\amodel}{\aworld_{\aaux}}$ and $\pair{\amodel}{\aworld_{\aauxbis}}$ satisfy
   $\completeplus{j - i} \egdef \complete{j - i} \land \HMBox{\avartree}(\Diamond \avarleft \land \Diamond \avarselect \land \Diamond \avarright)$.
  \end{enumerate}
  To complete the proof, we prove each direction separately. 

  ($\Rightarrow$):
  Suppose $\amodel,\aworld \models \less{\aaux}{\aauxbis}{i}{j}$.
  Then (by $\models$) there exists $\amodel'=\tup{\worlds,\arelation',\avaluation}$, such that $\amodel'\sqsubseteq\model$ and
  \begin{nscenter}
    $\amodel',\aworld\models \twonoms{\aaux}{\aauxbis}{i} \land \HMBox{\avartree}^i\lsrpartition{j-i} \land
    \selectpred{\aaux}{\aauxbis}{i}{j} \land \leftpred{\aaux}{\aauxbis}{i}{j}$.
  \end{nscenter}
  Then, from (i)--(iii), we can conclude that in $\pair{\amodel'}{\aworld}$, the two worlds $\aworld_{\aaux}$ and $\aworld_{\aauxbis}$ (corresponding to the nominals $\aaux$ and $\aauxbis$ in $\pair{\amodel}{\aworld}$) are exactly the ones responsible for the satisfaction of $\twonoms{\aaux}{\aauxbis}{i}$. Moreover, from
  $\amodel',\aworld \models \HMBox{\avartree}^i\lsrpartition{j-i}$ and \Cref{lemma:lsr}, we have $\amodel',\aworld_{\aaux} \models \complete{j-i}$.
  Then, by \Cref{A:the-property-of-types} we conclude that $\aworld_{\aaux}$ encodes the same number w.r.t. $\pair{\amodel}{\aworld}$ and $\pair{\amodel'}{\aworld}$.
  The same property holds for $\aworld_{\aauxbis}$, since again by $\amodel',\aworld \models \HMBox{\avartree}^i\lsrpartition{j-i}$ and \Cref{lemma:lsr}, we have $\amodel',\aworld_{\aauxbis} \models \complete{j-i}$.
  Lastly, again from \Cref{lemma:lsr},
  \begin{enumerate}
    \item every $\avartree$-node in $\arelation'(\aworld_{\aaux})$ and $\arelation'(\aworld_{\aauxbis})$
    has exactly one $\Aux$-child satisfying an atomic proposition from $\{\avarleft,\avarselect,\avarright\}$;
    \item exactly one $\avartree$-node in $\arelation'(\aworld_{\aaux})$ (say $\aworld_{\aaux,\avarselect}$) has an $\Aux$-child satisfying $\avarselect$.
    Similarly, exactly one $\avartree$-node in $\arelation'(\aworld_{\aauxbis})$ (say $\aworld_{\aauxbis,\avarselect}$) has an $\Aux$-child satisfying $\avarselect$.
    \item given $\aworld_{\aaux,\avarleft} \in \arelation'(\aworld_{\aaux})$ (resp. $\aworld_{\aauxbis,\avarleft} \in \arelation'(\aworld_{\aauxbis})$),
    it has an $\Aux$-child satisfying $\avarleft$ if and only if $\nb{\aworld_{\aaux,\avarleft}} > \nb{\aworld_{\aaux,\avarselect}}$ (resp.  $\nb{\aworld_{\aauxbis,\avarleft}} > \nb{\aworld_{\aauxbis,\avarselect}}$).
  \end{enumerate}
  Recall that the  number $\nb{\aworld_{\aaux}}$ (resp. $\nb{\aworld_{\aauxbis}}$) is represented by
  the binary encoding of the truth values of $\treeval$ on the $\avartree$-children of $\aworld_{\aaux}$ (resp. $\aworld_{\aauxbis}$) which, since $\pair{\amodel'}{\aworld_{\aaux}} \models \complete{j-i}$ (resp. $\pair{\amodel'}{\aworld_{\aauxbis}} \models \complete{j-i}$), are $\amapter(j-i,n)$ children implicitly ordered by the number they, in turn, encode.
  As $\pair{\amodel'}{\aworld}$ satisfies the hypothesis of \Cref{A:lemma:left-and-right},
  from $\amodel',\aworld \models \selectpred{\aaux}{\aauxbis}{i}{j} \land \leftpred{\aaux}{\aauxbis}{i}{j}$ we conclude that
  \begin{itemize}
    \item $\nb{\aworld_{\aaux,\avarselect}} = \nb{\aworld_{\aauxbis,\avarselect}}$, $\amodel, \aworld_{\aaux,\avarselect} \models \lnot \treeval$ and $\amodel, \aworld_{\aauxbis,\avarselect} \models \treeval$.
    Thus, in the binary representation of $\nb{\aworld_{\aaux}}$, the $\nb{\aworld_{\aaux,\avarselect}}$-bit is $0$, whereas in the binary representation of $\nb{\aworld_{\aauxbis}}$, it is $1$. Hence, the property (A) of numbers encoded in binary holds for $\nb{\aworld_{\aaux}}$ and $\nb{\aworld_{\aauxbis}}$;
    \item for all  worlds $\aworld_{\aaux,\avarleft} \in \arelation(\aworld_{\aaux})$
    and
    $\aworld_{\aauxbis,\avarleft} \in \arelation(\aworld_{\aauxbis})$ such that
    $\nb{\aworld_{\aaux,\avarleft}} > \nb{\aworld_{\aaux,\avarselect}}$ and
    $\nb{\aworld_{\aauxbis,\avarleft}} > \nb{\aworld_{\aauxbis,\avarselect}}$,
    if $\nb{\aworld_{\aaux,\avarleft}} = \nb{\aworld_{\aauxbis,\avarleft}}$ then
    \begin{nscenter}
      $\amodel,\aworld_{\aaux,\avarleft} \models \treeval$
      if and only if
      $\amodel,\aworld_{\aauxbis,\avarleft} \models \treeval$.
    \end{nscenter}
    Thus, the binary representation of $\nb{\aworld_{\aaux}}$ and $\nb{\aworld_{\aauxbis}}$, is the same when restricted to the bits that are more significant than $\nb{\aworld_{\aaux,\avarselect}}$ (which is equal to $\nb{\aworld_{\aauxbis,\avarselect}}$ by the previous case).
    Hence, the property (B) is also verified by $\nb{\aworld_{\aaux}}$ and $\nb{\aworld_{\aauxbis}}$.
  \end{itemize}
  Directly, we then conclude that $\nb{\aworld_{\aaux}} < \nb{\aworld_{\aauxbis}}$.

  ($\Leftarrow$): The right-to-left direction is proven analogously by essentially relying on \Cref{A:lemma:left-and-right} (I and II).
\end{proof}

\begin{lemma}\label{lemma:succj}
Let $\aaux \neq \aauxbis \in \Aux$ and $1 \leq i < j$.
Suppose $\amodel,\aworld \models \init{j} \land \fork{\aaux}{\aauxbis}{1}{j}$.

$\amodel,\aworld \models \successor{\aaux}{\aauxbis}{j}$ if and only if
there are two distinct $\avartree$-nodes $\aworld_{\aaux},\aworld_{\aauxbis} \in \arelation(\aworld)$ such that $\aworld_{\aaux}$ corresponds to the nominal $\aaux$, $\aworld_{\aauxbis}$ corresponds to the nominal $\aauxbis$ and
$\nb{\aworld_{\aauxbis}} = \nb{\aworld_{\aaux}}+1$.
\end{lemma}

\begin{proof}
  We recall the definition of $\successor{\aaux}{\aauxbis}{j}$ (where we expand the definition of $\selectleftpred{\aaux}{\aauxbis}{1}{j}$ given in the body of the paper):
   \begin{nscenter}
    $
    \successor{\aaux}{\aauxbis}{j} \egdef \true  {\separate} \big(
      \twonoms{\aaux}{\aauxbis}{1} \land \HMBox{\avartree}\lsrpartition{j-1} \land \selectpred{\aaux}{\aauxbis}{1}{j} \land
      \leftpred{\aaux}{\aauxbis}{1}{j}
      \land
      \rightpred{\aaux}{\aauxbis}
    \big).
    $
   \end{nscenter}
   As in \Cref{lemma:succone}, the proof uses standard properties of numbers encoded in binary.
   Again, let $x,y$ be two natural numbers that can be represented in binary by using $n$ bits. Let us denote with $x_i$ (resp. $y_i$) the $i$-th bit of the binary representation of $x$ (resp. $y$).
   We have that $y = x + 1$ if and only if
     \begin{enumerate}[label=(\Alph*)]
       \item there is a position $i \in \interval{1}{n}$ such that $x_i = 0$ and $y_i = 1$;
       \item for every position $j > i$, $x_j = 0$ $\iff$ $y_j = 0$;
       \item for every position $j < i$, $x_j = 1$ and $y_j = 0$.
     \end{enumerate}
   The formula $\successor{\aaux}{\aauxbis}{j}$ uses exactly this characterisation in order to state that $\nb{\aworld_{\aauxbis}} = \nb{\aworld_{\aaux}}+1$.

   One can see that the formula $\successor{\aaux}{\aauxbis}{j}$ can be obtained (syntactically) from the formula $\less{\aaux}{\aauxbis}{1}{j} \egdef \true \separate (
     \twonoms{\aaux}{\aauxbis}{1} \land \HMBox{\avartree}^i\lsrpartition{j-1} \land
   \selectpred{\aaux}{\aauxbis}{1}{j} \land \leftpred{\aaux}{\aauxbis}{1}{j})$
   by simply adding the conjunct
   $\rightpred{\aaux}{\aauxbis}$ to the right of $\leftpred{\aaux}{\aauxbis}{1}{j}$.
   Because of this, it is easy to see that
   the proof of this lemma follows very closely the structure of the proof of \Cref{lemma:lessij}.
   Indeed, to prove (A) and (B) we essentially rely on \Cref{A:lemma:left-and-right} (I and II), whereas to
   prove (C) we rely on the third point of \Cref{A:lemma:left-and-right}.
\end{proof}

\subsection{Inductive case $1 \leq i < j$ : Correctness of $\pU{j}$ and $\pC{j}$}

Let $\amodel = \triple{\worlds}{\arelation}{\avaluation}$ be a finite forest and $\aworld \in \worlds$.

\begin{lemma}\label{lemma:uniqj}
Let $j \geq 2$.
Suppose $\amodel,\aworld \models \init{j} \land \pA$.

$\amodel,\aworld \models \pU{j}$ if and only if
$\pair{\amodel}{\aworld}$ satisfies
\ref{prop:apU}, i.e.\ distinct $\avartree$-nodes in $\arelation(\aworld)$
encode different numbers.
\end{lemma}

\begin{proof}
As in \Cref{lemma:uniqone}, but using
\Cref{lemma:lessij} on the inductive formula $\equivalent{\avariable}{\avariablebis}{1}{j}$.
\end{proof}

\begin{lemma}\label{lemma:complj}
Let $j \geq 2$.
Suppose $\amodel,\aworld \models \init{j} \land \pA$.

$\amodel,\aworld \models \pC{j}$ if and only if
$\pair{\amodel}{\aworld}$ satisfies
\ref{prop:apC}, i.e.\ for every $\avartree$-node $\aworld_1 \in \arelation(\aworld)$, if $\nb{\aworld_1} < \amapter(j,n)-1$ then $\nb{\aworld_2} = \nb{\aworld_1}+1$ for some $\avartree$-node $\aworld_2 \in \arelation(\aworld)$.
\end{lemma}

\begin{proof}
As in \Cref{lemma:complone}, but using \Cref{lemma:succj} and the formula
$\completeplus{j-1}$ in order to properly evaluate $\fork{\avariable}{\avariablebis}{1}{j}$.
\end{proof}

\subsection{Proof of \Cref{lemma:tower-hardness-inductive}}

\begin{proof}
It follows directly from \Cref{lemma:basicform,lemma:uniqj,lemma:complj}.
\end{proof}

Again, a quick check of $\init{j}$ and the conditions \ref{prop:apS}, \ref{prop:apZ},
\ref{prop:apU}, \ref{prop:apC} and \ref{prop:apA} should be enough to convince the reader that they are simultaneously satisfiable, making $\init{j} \land \complete{j}$ also satisfiable.
However, in the following we show a model satisfying $\init{j} \land \complete{j}$.

\begin{lemma}\label{A:lemma:complete-j-satisfiable}
  Let $j \geq 2$. $\init{j} \land \complete{j}$ is satisfiable.
\end{lemma}

\begin{proof}
  Let $j \geq 2$.
  By induction on $j$, we suppose that $\init{j-1} \land \complete{j-1}$ is satisfiable (we already treated the base case for $j=1$ in \Cref{A:lemma:complete-one-satisfiable}).
  Let us consider $\aworld_0,\dots,\aworld_{\amapter(j,n)-1}$ distinct worlds.
  By the induction hypothesis, we can construct $\amapter(j,n)$ models $\amodel_i = \triple{\worlds_i}{\arelation_i}{\avaluation_i}$ ($i \in \interval{0}{\amapter(j,n)-1}$),
  so that $\aworld_i \in \worlds_i$ and $\amodel_i,\aworld_i \models \init{j-1} \land \complete{j-1}$.
  W.l.o.g. we can assume, for each two disjoint $i,j \in \interval{0}{\amapter(j,n)-1}$, $\worlds_i \cap \worlds_j = \emptyset$.
  Similarly, we can assume that each $\amodel_i$ is minimal, i.e. for every $\amodel' \sqsubseteq \amodel_i$ different from $\amodel'$, $\amodel',\aworld_i \not\models \init{j-1} \land \complete{j-1}$.
  This implies that $\aworld_i$ does not have any $\Aux$-children, and every $\avartree$-node in $\arelation_i(\aworld_i)$ does not have $\{\avarleft,\avarselect,\avarright\}$-children (as these two properties are not guaranteed by \ref{prop:apA}).

  Let $\aworld$ be a fresh world not appearing in the aforementioned models.
  Similarly, for every $i \in \interval{0}{\amapter(j,n)-1}$, let $\aworld_i^\avariable$ and $\aworld_i^\avariablebis$ be fresh worlds.
  Lastly, we also introduce, for every world $\overline{\aworld} \in \arelation_i(\aworld_i)$, three (distinct) new worlds $\aworld_{\overline{\aworld}}^\avarleft$, $\aworld_{\overline{\aworld}}^\avarselect$ and $\aworld_{\overline{\aworld}}^\avarright$.

  Then, let us consider the model $\amodel = \triple{\worlds}{\arelation}{\avaluation}$ defined as follows:
  \begin{enumerate}
    \item $\worlds \egdef \{\aworld\} \cup \worlds_i \cup
    \{ {\aworld_i^\avariable}, {\aworld_i^\avariablebis} \mid i \in \interval{0}{\amapter(j,n)-1} \} \cup
    \{ {\aworld_{\overline{\aworld}^\avarleft}},
    {\aworld_{\overline{\aworld}^\avarselect}},
    {\aworld_{\overline{\aworld}^\avarright}}, \mid i \in
    \interval{0}{\amapter(j,n)-1}, \overline{\aworld} \in \arelation_i(\aworld_i) \}$
    \item $\arelation \egdef
    \begin{aligned}[t]
      & \{\pair{\aworld}{\aworld_0},\dots,\pair{\aworld}{\aworld_{\amapter(j,n)-1}}\} \cup \textstyle\bigcup_{i \in \interval{0}{\amapter(j,n)-1}} \arelation_i \cup
      \{ \pair{\aworld_i}{\aworld_i^\avariable},\pair{\aworld_i}{\aworld_i^\avariablebis} \mid i \in \interval{0}{\amapter(j,n)-1} \}\\
      & \cup \{ \pair{\overline{\aworld}}{\aworld_{\overline{\aworld}^\avarleft}}, \pair{\overline{\aworld}}{\aworld_{\overline{\aworld}^\avarselect}},
      \pair{\overline{\aworld}}{\aworld_{\overline{\aworld}^\avarright}}, \mid i \in
      \interval{0}{\amapter(j,n)-1}, \overline{\aworld} \in \arelation_i(\aworld_i) \}
    \end{aligned}$
    \item $\avaluation$ is such that
    \begin{itemize}
      \item for every $i \in
      \interval{0}{\amapter(j,n)-1}$, $\avarprop \in \AP$ and every $\aworld' \in \arelation_i^2(\aworld_i)$, $\aworld' \in \avaluation(\avarprop)$ if and only if $\aworld' \in \avaluation_i(\avarprop)$. Hence, w.r.t. $\pair{\amodel}{\aworld}$, the evaluations w.r.t. worlds in $\arelation_i^3(\aworld) \cap \worlds_i$
      is unchanged compared to the one in $\pair{\amodel_i}{\aworld_i}$.
      \item For every $i \in \interval{0}{\amapter(j,n)-1}$ and every $\aworld' \in \arelation_i(\aworld_i)$, $\aworld' \in \avaluation(\treeval)$ if and only if w.r.t. $\pair{\amodel_i}{\aworld_i}$, the $\nb{\aworld'}$-bit in the binary representation of $i$ is $1$. Notice that this will lead to $\nb{\aworld_i} = i$.
      \item For every $i \in \interval{0}{\amapter(j,n)-1}$ and $\aaux \in \Aux$, $\aworld_i^\avariable \in \avaluation(\aaux)$ if and only if $\aaux = \avariable$. Similarly,
      $\aworld_i^\avariablebis \in \avaluation(\aaux)$ if and only if $\aaux = \avariablebis$.
      Thus, every $\aworld_i^\avariable$ is a $\avariable$-node, whereas every $\aworld_i^\avariablebis$ is a $\avariablebis$-node.
      \item For every $\aaux \in \Aux$, $\aworld \not \in \avaluation(\aaux)$ and for every $i \in \interval{0}{\amapter(j,n)-1}$, $\aworld_i\not \in \avaluation(\aaux)$.
      Moreover, for every $\overline{\aworld} \in \arelation_i(\aworld_i)$, $\overline{\aworld} \not\in \avaluation(\aaux)$ (notice that, by minimality, $\overline{\aworld}$ is a $\avartree$-node also in $\amodel_i$).
      Thus, $\aworld$, $\aworld_i$ and $\overline{\aworld}$ (as above) are all $\avartree$-nodes.
      \item For every $\aaux \in \Aux$, $\aworld \not \in \avaluation(\aaux)$ and for every $i \in \interval{0}{\amapter(j,n)-1}$ and $\overline{\aworld} \in \arelation_i(\aworld_i)$,
      (1) $\aworld_{\overline{\aworld}}^\avarleft \in \avaluation(\aaux)$ iff $\aaux = \avarleft$,
      (2) $\aworld_{\overline{\aworld}}^\avarselect \in \avaluation(\aaux)$ iff $\aaux = \avarselect$,
      (3) $\aworld_{\overline{\aworld}}^\avarright \in \avaluation(\aaux)$ iff $\aaux = \avarright$.
      Hence, every $\aworld_{\overline{\aworld}}^\avarleft$, $\aworld_{\overline{\aworld}}^\avarselect $ and $\aworld_{\overline{\aworld}}^\avarright$ (as above) is a $\avarleft$-node, $\avarselect$-node and $\avarright$-node, respectively.
    \end{itemize}
  \end{enumerate}
  We can check that $\pair{\amodel}{\aworld}$ satisfies $\init{j}$ as well as \ref{prop:apS}, \ref{prop:apZ},
  \ref{prop:apU}, \ref{prop:apC} and \ref{prop:apA}. Thus, by \Cref{lemma:tower-hardness-inductive}, $\amodel,\aworld \models \init{j} \land \complete{j}$.
\end{proof}

\subsection{Definitions and Proofs of \Cref{subsection:tiling-grid}}
We develop the material from~\Cref{subsection:tiling-grid}, providing all the necessary details.
As usual, in the following we let $\amodel = \triple{\worlds}{\arelation}{\avaluation}$ be a finite forest and consider one of its worlds $\aworld \in \worlds$.

Let $k \geq 2$ and let $\pair{\cTT}{\atile}$ be an instance of
$\tiling_k$, where $\cTT = \triple{\cT}{\cH}{\cV}$ and $\atile \in \cT$.
In the following, we define a formula $\atiling{k}{\cTT,\atile}$ such that the following lemma holds.

\begin{lemma}\label{lemma:tiling_correct_encoding}
$\pair{\cTT}{\atile}$ as a solution for $\tiling_k$ if and only if
the formula $\atiling{k}{\cTT,\atile}$ is satisfiable.
\end{lemma}

Recall that a solution for $\pair{\cTT}{\atile}$ w.r.t. $\tiling_k$ is a map
$\tau: \interval{0}{\amapter(k,n)-1}\!\times\!\interval{0}{\amapter(k,n)-1}\!\to\!\cT$
satisfying~\ref{tiling_c:1} and~\ref{tiling_c:2}.
W.l.o.g.\ we assume $\cT$ to be a set of atomic propositions, disjoint from
$\{\avarprop_1,\dots,\avarprop_n,\treeval\} \cup \Aux$ used in the definition of $\complete{j}$.
Let us first describe how to represent a grid $\interval{0}{\amapter(k,n)-1}^2$ in the pointed forest
$\pair{\amodel}{\aworld}$.
We use the same ideas needed in order to define $\complete{k}$, but with some minor modifications.
As previously stated, if $\amodel,\aworld \models \complete{k}$ then
given a $\avartree$-node $\aworld' \in \arelation(\aworld)$,
the number $\nb{\aworld'} \in \interval{0}{\amapter(k,n)-1}$ is encoded using the $\avartree$-children of $\aworld'$, where the numbers encoded by these children represent positions
in the binary encoding of $\nb{\aworld'}$.
Instead of being a single number, a position in the grid is a pair of numbers  $(h,v) \in \interval{0}{\amapter(k,n)-1}^2$.
Hence, in a model $\pair{\amodel}{\aworld}$ satisfying $\atiling{k}{\cTT,\atile}$
we require that $\aworld' \in \arelation(\aworld)$ encodes two numbers $\nbexp{\aworld'}{\cH}$ and $\nbexp{\aworld'}{\cV}$, and say that $\aworld'$ encodes the position $(h,v)$ if and only if $\nbexp{\aworld'}{\cH} = h$ and $\nbexp{\aworld'}{\cV} = v$.
Since both numbers are from $\interval{0}{\amapter(k,n)-1}$, the same amount of $\avartree$-children as in $\complete{k}$ can be used in order to encode both $\nbexp{\aworld'}{\cH}$ and $\nbexp{\aworld'}{\cV}$.
Thus, we rely on the formula $\complete{k-1}$  to force $\aworld'$ to have the correct amount of $\avartree$-children, by requiring it to hold in $\pair{\amodel}{\aworld'}$.
Similarly to what is done previously for $\complete{j}$ ($j \geq 2$), we encode the numbers $\nbexp{\aworld'}{\cH}$ and $\nbexp{\aworld'}{\cV}$
by using the truth value, on the $\avartree$-children of $\aworld'$, of two new atomic propositions $\treeval_\cH$ and $\treeval_\cV$, respectively.
Then, we  use similar formulae to $\pZ{k}$, $\pU{k}$ and $\pC{k}$ in order to state that $\aworld$ witnesses exactly one child for each position in the grid. Once the grid is encoded, the tiling conditions can be enforced rather easily.

We introduce the formula $\agrid{k}{\cTT}$ that characterises the set of models encoding the $\interval{0}{\amapter(k,n)-1}^2$ grid.
A model $\pair{\amodel = \triple{\worlds}{\arelation}{\avaluation}}{\aworld}$ satisfying
$\agrid{k}{\cTT}$ is such that:
\begin{description}
  \item[\desclabel{(\apZ$_{\cTT,k}$)}{A:tiling_enc:apZ}] One $\avartree$-node in $\arelation(\aworld)$ encodes the position $\pair{0}{0}$, i.e. there is an $\avartree$-node $\tilde\aworld{\in}\,\arelation(\aworld)$ s.t.\
    $\nbexp{\tilde\aworld}{\cH} = \nbexp{\tilde\aworld}{\cV} = 0$;
  \item[\desclabel{(\apU$_{\cTT,k}$)}{A:tiling_enc:apU}] for all two distinct $\avartree$-nodes $\aworld_1,\aworld_2 \in \arelation(\aworld)$, $\nbexp{\aworld_1}{\cH} \neq \nbexp{\aworld_2}{\cH}$ or  $\nbexp{\aworld_1}{\cH} \neq \nbexp{\aworld_2}{\cH}$;
  \item[\desclabel{(\apC$_{\cTT,k}$)}{A:tiling_enc:apC}] for every $\avartree$-node $\aworld_1 \in \arelation(\aworld)$,
  \begin{itemize}
    \item if $\nbexp{\aworld_1}{\cH} < \amapter(j,n)-1$ then $\nbexp{\aworld_2}{\cH} = \nbexp{\aworld_1}{\cH}+1$ and $\nbexp{\aworld_2}{\cV} = \nbexp{\aworld_1}{\cV}$, for some $\avartree$-node $\aworld_2 \in \arelation(\aworld)$;
    \item if $\nbexp{\aworld_1}{\cV} < \amapter(j,n)-1$ then $\nbexp{\aworld_2}{\cV} = \nbexp{\aworld_1}{\cV}+1$ and $\nbexp{\aworld_2}{\cH} = \nbexp{\aworld_1}{\cH}$, for some $\avartree$-node $\aworld_2 \in \arelation(\aworld)$;
  \end{itemize}
  \item[\desclabel{(\ainit/\apS/\apA)}{A:tiling_enc:apS}] $\pair{\amodel}{\aworld}$ satisfies $\init{k}$,  $\pS{k}$ and $\pA$;
\end{description}
It is easy to see that, with these conditions, $\pair{\amodel}{\aworld}$ correctly encodes the grid.
The definition of $\agrid{k}{\cTT}$ follows rather closely the definition of $\complete{j}$.
It is defined as
\[
    \pZT{k}{\cTT} \land \pUT{k}{\cTT} \land \pCTBB{k}{\cTT} \land \init{k} \land \pS{k} \land \pA
\]
where each conjunct expresses the homonymous property above.
In order to define the first three conjuncts of $\agrid{k}{\cTT}$  (hence completing its definition)
we start by defining the formulae $\bequivalent{\aaux}{\aauxbis}{k}{\heapdim}$ and $\bsuccessor{\aaux}{\aauxbis}{k}{\heapdim}$, where $\heapdim \in \{\cH,\cV\}$. Similarly to $\equivalent{\aaux}{\aauxbis}{1}{k}$ and $\successor{\aaux}{\aauxbis}{k}$,
Given a model $\pair{\amodel = \triple{\worlds}{\arelation}{\avaluation}}{\aworld}$ satisfying $\fork{\aaux}{\aauxbis}{1}{k}$, and the two $\avartree$-nodes $\aworld_{\aaux},\aworld_{\aauxbis} \in \arelation(\aworld)$
corresponding to the nominals $\aaux$ and $\aauxbis$, respectively,
\begin{itemize}
\item $\bequivalent{\aaux}{\aauxbis}{k}{\heapdim}$ states that $\nbexp{\aworld_{\aaux}}{\heapdim} = \nbexp{\aworld_{\aauxbis}}{\heapdim}$;
\item $\bsuccessor{\aaux}{\aauxbis}{k}{\heapdim}$ states that $\nbexp{\aworld_{\aauxbis}}{\heapdim} = \nbexp{\aworld_{\aaux}}{\heapdim}+1$.
\end{itemize}
To encode $\bequivalent{\aaux}{\aauxbis}{k}{\heapdim}$ we simply require that for all two $\avartree$-children $\aworld_{\avariable} \in \arelation(\aworld_{\aaux})$ and $\aworld_{\avariablebis} \in \arelation(\aworld_{\aauxbis})$, if
$\nb{\aworld_{\avariable}} = \nb{\aworld_{\avariablebis}}$ then $\aworld_{\avariable}$ and $\aworld_{\avariablebis}$
agree on the satisfaction of $\treeval_\heapdim$. In formula:
\[
\begin{aligned}
  \bequivalent{\aaux}{\aauxbis}{k}{\heapdim}
  \egdef
  \lnot \big( \true \separate (\fork{\avariable}{\avariablebis}{2}{k}
  \land \atnom{\aaux}{1}{\HMDiamond{\avartree} \Diamond \avariable}
  \land \atnom{\aauxbis}{1}{\HMDiamond{\avartree} \Diamond \avariablebis}
   \land \equivalent{\avariable}{\avariablebis}{2}{k}
  \land \lnot (\atnom{\avariable}{2}\treeval_\heapdim \iff \atnom{\avariablebis}{2}\treeval_\heapdim)
  )\big).
\end{aligned}
\]

\begin{lemma}\label{A:lemma:bequivalent}
  Let $\aaux \neq \aauxbis \in \Aux$ and
  $k \geq 2$.
  Suppose $\amodel,\aworld \models \init{k} \land \fork{\aaux}{\aauxbis}{1}{k}$.

  $\amodel,\aworld \models \bequivalent{\aaux}{\aauxbis}{k}{\heapdim}$ if and only if
  there are two distinct $\avartree$-nodes $\aworld_{\aaux},\aworld_{\aauxbis} \in \arelation^i(\aworld)$ such that $\aworld_{\aaux}$ corresponds to the nominal $\aaux$, $\aworld_{\aauxbis}$ corresponds to the nominal $\aauxbis$ and
  $\nbexp{\aworld_{\aaux}}{\heapdim} = \nbexp{\aworld_{\aauxbis}}{\heapdim}$.
\end{lemma}

\begin{proof}
This proof is similar to the one of \Cref{A:lemma:left-and-right} (II).
Since $\amodel,\aworld \models \init{k} \land \fork{\aaux}{\aauxbis}{1}{k}$, by \Cref{lemma:forkij}
there are two worlds $\aworld_{\aaux}$ and $\aworld_{\aauxbis}$ in $\arelation(\aworld)$ corresponding to the nominals (for the depth $1$) $\aaux$ and $\aauxbis$, respectively.

($\Rightarrow$):
Suppose $\amodel,\aworld \models \bequivalent{\aaux}{\aauxbis}{k}{\heapdim}$.
Then, for every $\amodel' = \triple{\worlds}{\arelation_1}{\avaluation}$,
if $\amodel' \sqsubseteq \amodel$ and
$\amodel',\aworld \models \fork{\avariable}{\avariablebis}{2}{k}
\land \atnom{\aaux}{1}{\HMDiamond{\avartree} \Diamond \avariable}
\land \atnom{\aauxbis}{1}{\HMDiamond{\avartree} \Diamond \avariablebis}
 \land \equivalent{\avariable}{\avariablebis}{2}{k}$
 then $\amodel',\aworld \models \atnom{\avariable}{2}\treeval_\heapdim \iff \atnom{\avariablebis}{2}\treeval_\heapdim$.
 Now, from $\amodel,\aworld \models \fork{\aaux}{\aauxbis}{1}{k}$ we have
 $\amodel,\aworld_{\aaux} \models \complete{k-1}$ and $\amodel,\aworld_{\aauxbis} \models \complete{k-1}$ (notice that then, all the worlds in $\arelation(\aworld_{\aaux}) \cup \arelation(\aworld_{\aauxbis})$ satisfy $\complete{k-2}$).
 Thus, let us consider any two worlds $\aworld_\avariable$ and $\aworld_\avariablebis$ such that
 \begin{itemize}
  \item $\aworld_\avariable \in \arelation(\aworld_{\aaux})$ and $\aworld_\avariablebis \in \arelation(\aworld_{\aauxbis})$;
  \item $\nbexp{\aworld_{\avariable}}{k-1} = \nbexp{\aworld_{\avariablebis}}{k-1}$.
 \end{itemize}
 We show that $\amodel,\aworld_{\avariable} \models \treeval_{\heapdim}$ if and only if $\amodel,\aworld_{\avariablebis} \models \treeval_{\heapdim}$,
 thus concluding that $\nbexp{\aworld_{\aaux}}{\heapdim} = \nbexp{\aworld_\aauxbis}{\heapdim}$.
 Let us consider the finite forest $\amodel' = \triple{\worlds}{\arelation_1}{\avaluation}$
 where $\arelation_1$ is obtained from $\arelation$ by removing every edge $\pair{\aworld_b}{\aworld'} \in \arelation$ where $b \in \{\aaux,\aauxbis\}$, and $\aworld'$ is a $\avartree$-node different from $\aworld_\avariable$ and $\aworld_\avariablebis$.
 We also remove the edge $\pair{\aworld_\avariable}{\aworld'} \in \arelation$ where $\aword'$ is the only $\avariablebis$-child of $\aworld_\avariable$, as well as $\pair{\aworld_{\avariablebis}}{\aworld''}$ where $\aworld''$ is the only $\avariable$-child of $\aworld_\avariablebis$.
 The existence of these nodes is guaranteed by $\amodel,\aworld_{\aaux} \models \complete{k-1}$ and $\amodel,\aworld_{\aauxbis} \models \complete{k-1}$.
 By \Cref{lemma:forkij} we have $\amodel',\aworld \models \fork{\avariable}{\avariablebis}{2}{k}$, where $\aworld_{\avariable}$ corresponds to the nominal (at depth $2$) $\avariable$, whereas $\aworld_{\avariablebis}$
 corresponds to the nominal (at depth $2$) $\avariablebis$.
 Moreover, \Cref{lemma:forkij} ensures that $\amodel,\aworld_\avariable \models \complete{k-2}$ and $\amodel,\aworld_\avariablebis \models \complete{k-2}$, hence
 by \Cref{A:the-property-of-types} we conclude that $\aworld_{\avariable}$ (resp. $\aworld_{\avariablebis}$) encodes the same number w.r.t. $\pair{\amodel}{\aworld}$
 and $\pair{\amodel'}{\aworld}$.
 Again from the definition of $\arelation_1$ it is easy to see that $\amodel',\aworld \models \atnom{\aaux}{1}{\HMDiamond{\avartree} \Diamond \avariable}
 \land \atnom{\aauxbis}{1}{\HMDiamond{\avartree} \Diamond \avariablebis}$.
 Lastly, by hypothesis on $\aworld_\avariable$ and $\aworld_\avariablebis$, together with \Cref{lemma:lessij} and $\equivalent{\avariable}{\avariablebis}{2}{k} \egdef \lnot (\less{\avariable}{\avariablebis}{2}{k} \lor \less{\avariablebis}{\avariable}{2}{k})$,
 we conclude that $\amodel',\aworld \models \equivalent{\avariable}{\avariablebis}{2}{k}$.
 Thus, by hypothesis, $\amodel',\aworld \models \atnom{\avariable}{2}\treeval_\heapdim \iff \atnom{\avariablebis}{2}\treeval_\heapdim$, concluding the proof.

 ($\Leftarrow$): This direction is proved analogously by essentially relying on \Cref{lemma:lessij} and \Cref{A:the-property-of-types}.
\end{proof}

The formula $\bsuccessor{\aaux}{\aauxbis}{k}{\heapdim}$ can be defined by slightly modifying the formula $\successor{\aaux}{\aauxbis}{k}$. We start by defining the formulae $\bleftpred{\aaux}{\aauxbis}{k}{\heapdim}$, $\bselectpred{\aaux}{\aauxbis}{k}{\heapdim}$ and $\brightpred{\aaux}{\aauxbis}{\heapdim}$
with semantics similar to $\leftpred{\aaux}{\aauxbis}{1}{k}$, $\selectpred{\aaux}{\aauxbis}{1}{k}$
and $\rightpred{\aaux}{\aauxbis}$, respectively, but where, for a given $\avartree$-node in $\arelation^2(\aworld)$,
we are interested in the satisfaction of $\treeval_\heapdim$ instead of $\treeval$. For example, the formula $\bselectpred{\aaux}{\aauxbis}{k}{\heapdim}$ is defined as
\[
\begin{aligned}
  \bselectpred{\aaux}{\aauxbis}{k}{\heapdim}\egdef
  \true \separate \big( \fork{\avariable}{\avariablebis}{2}{k} \land
  \atnom{\aaux}{1}\HMDiamond{\avartree}(\Diamond \avarselect \land \Diamond \avariable)\, \land \atnom{\aauxbis}{1}\HMDiamond{\avartree}(\Diamond \avarselect \land \Diamond \avariablebis) \land \equivalent{\avariable}{\avariablebis}{2}{k} \land \atnom{\avariable}{2}\lnot \treeval_\heapdim \land \atnom{\avariablebis}{2}\treeval_\heapdim
  \big)
\end{aligned}
\]
i.e. by replacing the two last conjuncts of $\selectpred{\aaux}{\aauxbis}{1}{k}$, $\atnom{\avariable}{2}\lnot \treeval $ and $\atnom{\avariablebis}{2}\treeval$ with $\atnom{\avariable}{2}\lnot \treeval_\heapdim$ and $\atnom{\avariablebis}{2}\treeval_\heapdim$, respectively.
Similarly, $\bleftpred{\aaux}{\aauxbis}{k}{\heapdim}$ is defined from $\leftpred{\aaux}{\aauxbis}{1}{k}$ by replacing the last conjunct of this formula, i.e.\ $\lnot (\atnom{\avariable}{2}\treeval \iff \atnom{\avariablebis}{2}\treeval)$, by
$\lnot (\atnom{\avariable}{2}\treeval_\heapdim \iff \atnom{\avariablebis}{2}\treeval_\heapdim)$.
Lastly, $\brightpred{\aaux}{\aauxbis}{\heapdim}$ is defined from $\rightpred{\aaux}{\aauxbis}$ by replacing every occurrence of $\treeval$ by $\treeval_\heapdim$.
The formula $\bsuccessor{\aaux}{\aauxbis}{k}{\heapdim}$ is then defined as follows:
\[
\begin{aligned}
\bsuccessor{\aaux}{\aauxbis}{k}{\heapdim} \egdef
\true \separate
  \big(
     \twonoms{\aaux}{\aauxbis}{1} \land \HMBox{\avartree}\lsrpartition{k-1} \land  \bleftpred{\aaux}{\aauxbis}{k}{\heapdim}\,{\land}\,\bselectpred{\aaux}{\aauxbis}{k}{\heapdim}\,{\land}\,
      \brightpred{\aaux}{\aauxbis}{\heapdim}
  \big).
\end{aligned}
\]

\begin{lemma}\label{A:lemma:bsuccessor}
  Let $\aaux \neq \aauxbis \in \Aux$ and
  $k \geq 2$.
  Suppose $\amodel,\aworld \models \init{k} \land \fork{\aaux}{\aauxbis}{1}{k}$.

  $\amodel,\aworld \models \bsuccessor{\aaux}{\aauxbis}{k}{\heapdim}$ if and only if
  there are two distinct $\avartree$-nodes $\aworld_{\aaux},\aworld_{\aauxbis} \in \arelation^i(\aworld)$ such that $\aworld_{\aaux}$ corresponds to the nominal $\aaux$, $\aworld_{\aauxbis}$ corresponds to the nominal $\aauxbis$ and
  $\nbexp{\aworld_{\aauxbis}}{\heapdim} = \nbexp{\aworld_{\aaux}}{\heapdim} + 1$.
\end{lemma}

\begin{proof}
The proof unfolds as the proofs of \Cref{lemma:succone,lemma:succj}.
\end{proof}

We are now ready to define the formulae $\pZT{k}{\cTT}$, $\pUT{k}{\cTT}$ and $\pCTBB{k}{\cTT}$, achieving the conditions \ref{A:tiling_enc:apZ}, \ref{A:tiling_enc:apU} and \ref{A:tiling_enc:apC}, respectively.
All these formulae follow closely the definitions of $\pZ{k}$, $\pU{k}$ and $\pC{k}$ of the previous sections, hence we refer to these latter formulae for an informal description on how they work. The formula $\pZT{k}{\cTT}$ is simply defined as $\HMDiamond{\avartree}(\HMBox{\avartree}(\lnot \treeval_\cH \land \lnot \treeval_\cV))$.
\begin{lemma}\label{A:lemma:grid-zero}
  $\amodel,\aworld \models \pZT{k}{\cTT}$ if and only if $\pair{\amodel}{\aworld}$ satisfies \ref{A:tiling_enc:apZ}.
\end{lemma}

\begin{proof}
The proof is straightforward, by definition of $\pZT{k}{\cTT}$ and how $\pair{0}{0}$ is encoded in the grid.
\end{proof}

The formula $\pUT{k}{\cTT}$ is defined from $\pU{k}$ by simply replacing $\equivalent{\avariable}{\avariablebis}{1}{k}$
with $\bequivalent{\avariable}{\avariablebis}{k}{\cH} \land \bequivalent{\avariable}{\avariablebis}{k}{\cV}$:
\[
\pUT{k}{\cTT} = \lnot \big(\true \separate (\fork{\avariable}{\avariablebis}{1}{k} \land \bequivalent{\avariable}{\avariablebis}{k}{\cH} \land\bequivalent{\avariable}{\avariablebis}{k}{\cV})\big)
\]

\begin{lemma}\label{A:lemma:grid-unique}
  Let $k \geq 2$.
  Suppose $\amodel,\aworld \models \init{k} \land \pA$.

  $\amodel,\aworld \models \pU{k}$ if and only if
  $\pair{\amodel}{\aworld}$ satisfies
  \ref{A:tiling_enc:apU}, i.e.\ distinct $\avartree$-nodes in $\arelation(\aworld)$
  encode different pairs of numbers.
\end{lemma}

\begin{proof}
This lemma is proven as \Cref{lemma:uniqone} and \Cref{lemma:uniqj}, by relying on \Cref{A:lemma:bequivalent} in order to show that, given two distinct worlds $\aworld_{\avariable}$ and $\aworld_\avariablebis$ corresponding to nominals (for the depth $1$) $\avariable$ and $\avariablebis$, respectively,
$\bequivalent{\avariable}{\avariablebis}{k}{\cH} \land\bequivalent{\avariable}{\avariablebis}{k}{\cV}$ holds if and only if $\nbexp{\aworld_\avariable}{\cH} = \nbexp{\aworld_{\avariablebis}}{\cH}$ and
$\nbexp{\aworld_\avariable}{\cV} = \nbexp{\aworld_{\avariablebis}}{\cV}$.
\end{proof}

Lastly, $\pCTBB{k}{\cTT} \egdef \pCTB{k}{\cTT}{\cH} \land \pCTB{k}{\cTT}{\cV}$ where
\[
\begin{aligned}
\pCTB{k}{\cTT}{\cH} \egdef \lnot
  \Big( \Box \false\!\separate\!
    \Big(\HMBox{\avartree} (\completeplus{k{-}1}\,{\land}\,\Diamond \avariablebis)\,{\land}\,
    \nominal{\avariable}{1} \land
 \atnom{\avariable}{1} \lnot \one^{\cH}_{k} \land \lnot \big( \true \separate (\fork{\avariable}{\avariablebis}{1}{j} \land
 \bsuccessor{\avariable}{\avariablebis}{k}{\cH}
 \land \bequivalent{\avariable}{\avariablebis}{k}{\cV}
 )\big)\Big)\Big)
\end{aligned}
\]
and $\pCTB{k}{\cTT}{\cV}$ is defined form $\pCTB{k}{\cTT}{\cH}$ by replacing $\one^{\cH}_{k}$, $\bsuccessor{\avariable}{\avariablebis}{k}{\cH}$ and
$\bequivalent{\avariable}{\avariablebis}{k}{\cV}$ with
$\one^{\cV}_{k}$, $\bsuccessor{\avariable}{\avariablebis}{k}{\cV}$ and
$\bequivalent{\avariable}{\avariablebis}{k}{\cH}$, respectively.
Here, $\one^{\heapdim}_{k}$ ($\heapdim \in \{\cH,\cV\}$) is defined as $\HMBox{\avartree} \treeval_\heapdim$,
and hence it is satisfied by the $\avartree$-nodes $\aworld' \in \arelation(\aworld)$ such that $\nbexp{\aworld'}{\heapdim} = \amapter(k,n)-1$.

\begin{lemma}\label{A:lemma:grid-complete}
  Let $k \geq 2$.
  Suppose $\amodel,\aworld \models \init{k} \land \pA$.
  $\amodel,\aworld \models \pCTBB{k}{\cTT}$ if and only if $\pair{\amodel}{\aworld}$ satisfies \ref{A:tiling_enc:apC}.

  More precisely,
  \begin{enumerate}
    \item $\amodel,\aworld \models \pCTB{k}{\cTT}{\cH}$ if and only if
    for every $\avartree$-node $\aworld_1 \in \arelation(\aworld)$, if $\nbexp{\aworld_1}{\cH} < \amapter(j,n)-1$ then
    there is a $\avartree$-node $\aworld_2 \in \arelation(\aworld)$ such that
    $\nbexp{\aworld_2}{\cH} = \nbexp{\aworld_1}{\cH}+1$
    and $\nbexp{\aworld_2}{\cV} = \nbexp{\aworld_1}{\cV}$;
    \item $\amodel,\aworld \models \pCTB{k}{\cTT}{\cV}$ if and only if
    for every $\avartree$-node $\aworld_1 \in \arelation(\aworld)$, if $\nbexp{\aworld_1}{\cV} < \amapter(j,n)-1$ then
    there is a $\avartree$-node $\aworld_2 \in \arelation(\aworld)$ such that
    $\nbexp{\aworld_2}{\cH} = \nbexp{\aworld_1}{\cH}$
    and $\nbexp{\aworld_2}{\cV} = \nbexp{\aworld_1}{\cV}+1$.
  \end{enumerate}
\end{lemma}

\begin{proof}
  Both (1) and (2) are proved as \Cref{lemma:complone} and \Cref{lemma:complj},
  with the sole difference that we rely on \Cref{A:lemma:bequivalent} and \Cref{A:lemma:bsuccessor}
  in order to show that, given two distinct worlds $\aworld_{\avariable}$ and $\aworld_\avariablebis$ corresponding to nominals (for the depth $1$) $\avariable$ and $\avariablebis$, respectively,
  $\bsuccessor{\avariable}{\avariablebis}{k}{\cH}
  \land \bequivalent{\avariable}{\avariablebis}{k}{\cV}$ holds if and only if $\nbexp{\aworld_\avariable}{\cH} = \nbexp{\aworld_{\avariablebis}}{\cH}+1$ and
  $\nbexp{\aworld_\avariable}{\cV} = \nbexp{\aworld_{\avariablebis}}{\cV}$ (in the proof of 1).
  Similarly, (in the proof of 2)
  $\bsuccessor{\avariable}{\avariablebis}{k}{\cV}
  \land \bequivalent{\avariable}{\avariablebis}{k}{\cH}$ holds if and only if $\nbexp{\aworld_\avariable}{\cH} = \nbexp{\aworld_{\avariablebis}}{\cH}$ and
  $\nbexp{\aworld_\avariable}{\cV} = \nbexp{\aworld_{\avariablebis}}{\cV}+1$.
\end{proof}

This concludes the definition of $\agrid{k}{\cTT}$.
It is proved correct in the following lemma.

\begin{lemma}\label{A:lemma:grid:correctness}
$\amodel,\aworld \models \agrid{k}{\cTT}$ if and only if $\pair{\amodel}{\aworld}$ satisfies \ref{A:tiling_enc:apZ}, \ref{A:tiling_enc:apU},
\ref{A:tiling_enc:apC} and \ref{A:tiling_enc:apS}.
\end{lemma}

\begin{proof}
  Directly from \Cref{A:lemma:grid-zero,A:lemma:grid-unique,A:lemma:grid-complete,lemma:init,lemma:basicform}.
\end{proof}

\begin{corollary}\label{A:corollary:grid-satisfiable}
 $\agrid{k}{\cTT}$ is satisfiable.
\end{corollary}

\begin{proof}(sketch)
  The satisfiability of $\agrid{k}{\cTT}$ can be established by \Cref{A:lemma:grid:correctness}  as \ref{A:tiling_enc:apZ}, \ref{A:tiling_enc:apU},
  \ref{A:tiling_enc:apC} and \ref{A:tiling_enc:apS} can be simultaneously satisfied.
  A model satisfying these constraints can be defined similarly to what is done in
  \Cref{A:lemma:complete-j-satisfiable}, the main difference being that $\amapter(k,n)^2$ $\avartree$-nodes need to be considered, instead of just $\amapter(k,n)$.
\end{proof}

We can now proceed to the encoding of the tiling conditions~\ref{tiling_c:1} and~\ref{tiling_c:2}.
Given a model $\pair{\amodel = \triple{\worlds}{\arelation}{\avaluation}}{\aworld}$ satisfying
$\agrid{k}{\cTT}$, the existence of a solution for $\pair{\cTT}{\atile}$, w.r.t.\ $\tiling_k$, can be expressed with the following conditions:
\begin{description}
  \item[\desclabel{(\apone$_{\cTT}$)}{A:tiling_enc:one_tile}] every $\avarprop$-node in $\arelation(\aworld)$ satisfies exactly one tile in $\cT$;
  \item[\desclabel{(\apfirst$_{\cTT,\atile}$)}{A:tiling_enc:zero_c}] for $\tilde\aworld{\in}\,\arelation(\aworld)$, if $\nbexp{\tilde\aworld}{\cH} {=} \nbexp{\tilde\aworld}{\cV} {=} 0$ then $\tilde{\aworld} \in \avaluation(\atile)$;
  \item [\desclabel{(\aphor$_{\cTT}$)}{A:tiling_enc:hori}] for all $\aworld_1,\aworld_2 \in \arelation(\aworld)$,
    if $\nbexp{\aworld_2}{\cH} = \nbexp{\aworld_1}{\cH}+1$ and $\nbexp{\aworld_2}{\cV} = \nbexp{\aworld_1}{\cV}$ then there is $\pair{\atile_1}{\atile_2} \in \cH$ such that
    $\aworld_1 \in \avaluation(\atile_1)$ and $\aworld_2 \in \avaluation(\atile_2)$;
  \item [\desclabel{(\apvert$_{\cTT}$)}{A:tiling_enc:vert}] for all $\aworld_1,\aworld_2 \in \arelation(\aworld)$,
    if $\nbexp{\aworld_2}{\cV} = \nbexp{\aworld_1}{\cV}+1$ and $\nbexp{\aworld_2}{\cH} = \nbexp{\aworld_1}{\cH}$ then there is $\pair{\atile_1}{\atile_2} \in \cV$ such that
    $\aworld_1 \in \avaluation(\atile_1)$ and $\aworld_2 \in \avaluation(\atile_2)$.
\end{description}
Then, the formula $\atiling{k}{\cTT,\atile}$ can be defined as
\[
\agrid{k}{\cTT} \land \pone{\cTT} \land \pfirst{k}{\cTT,\atile} \land \phor{k}{\cTT} \land \pvert{k}{\cTT}
\]
where the last four conjuncts express the homonymous property above.
Given the toolkit of formulae introduced up to now, these four formulae are easy to define.
$\pone{\cTT}$ is simply defined as
$\HMBox{\avartree}\bigvee_{\atile_1 \in \cT} (\atile_1 \land \bigwedge_{\atile_2 \in \cT} \lnot \atile_2)$.
Similarly, $\pfirst{k}{\cTT,\atile}$ is also straightforward to define:
\[
\pfirst{k}{\cTT,\atile} \egdef \HMBox{\avartree}\big(\HMBox{\avartree}(\lnot \treeval_\cH \land \lnot \treeval_\cV) \implies \atile\big).
\]
Notice that, in this formula, we use the fact that
the $\avartree$-node $\aworld' \in \arelation(\aworld)$ encoding $\pair{0}{0}$ is the only one, among the $\avartree$-children of $\aworld$, satisfying $\HMBox{\avartree}(\lnot \treeval_\cH \land \lnot \treeval_\cV)$.

\begin{lemma}\label{A:lemma:tiling:zero-and-unique}
Let $k \geq 2$ and suppose $\amodel,\aworld \models \agrid{k}{\cTT}$. Then,
\begin{enumerate}[label=\Roman*.]
  \item $\amodel,\aworld \models \pone{\cTT}$ if and only if $\pair{\amodel}{\aworld}$ satisfies \ref{A:tiling_enc:one_tile};
  \item $\amodel,\aworld \models \pfirst{k}{\cTT,\atile}$ if and only if $\pair{\amodel}{\aworld}$ satisfies \ref{A:tiling_enc:zero_c}.
\end{enumerate}
\end{lemma}
\begin{proof}
  Both I and II are easily proven directly from the definition of $\pone{\cTT}$ and $\pfirst{k}{\cTT,\atile}$.
\end{proof}

For the formula $\phor{k}{\cTT}$, we essentially state that there cannot be two $\avartree$-nodes $\aworld_1,\aworld_2 \in \arelation(\aworld)$ such that $\aworld_2$ encodes the position $\pair{\nbexp{\aworld_1}{\cH}+1}{\nbexp{\aworld_1}{\cV}}$
and $\aworld_1 \in \avaluation(\atile_1)$, $\aworld_2 \in \avaluation(\atile_2)$ does not hold for any $\pair{\atile_1}{\atile_2} \in \cH$. In formula:
\[
\begin{aligned}
\phor{k}{\cTT} \egdef \lnot
\big(
\true
  \separate
  \big(
  \fork{\avariable}{\avariablebis}{1}{k} \land
  \bsuccessor{\avariable}{\avariablebis}{k}{\cH}
  \land \bequivalent{\avariable}{\avariablebis}{k}{\cV}
   \land
  \lnot \textstyle\bigvee_{{\pair{\atile_1}{\atile_2} \in \cH}} (\atnom{\avariable}{1} \atile_1 \land \atnom{\avariablebis}{1} \atile_2)
  \big)
\big).
\end{aligned}
\]
Lastly, $\pvert{k}{\cTT}$ is defined as $\phor{k}{\cTT}$, but replacing $\cH$ by $\cV$ and vice-versa:
\[
\begin{aligned}
\pvert{k}{\cTT} \egdef \lnot
\big(
\true
  \separate
  \big(
  \fork{\avariable}{\avariablebis}{1}{k} \land
  \bsuccessor{\avariable}{\avariablebis}{k}{\cV}
  \land \bequivalent{\avariable}{\avariablebis}{k}{\cH}
   \land
  \lnot \textstyle\bigvee_{{\pair{\atile_1}{\atile_2} \in \cV}} (\atnom{\avariable}{1} \atile_1 \land \atnom{\avariablebis}{1} \atile_2)
  \big)
\big).
\end{aligned}
\]

\begin{lemma}\label{A:lemma:tiling:hor-and-vert}
Let $k \geq 2$ and suppose $\amodel,\aworld \models \agrid{k}{\cTT}$. Then,
  \begin{enumerate}[label=\Roman*.]
    \item $\amodel,\aworld \models \phor{k}{\cTT}$ if and only if $\pair{\amodel}{\aworld}$ satisfies \ref{A:tiling_enc:hori};
    \item $\amodel,\aworld \models \pvert{k}{\cTT}$ if and only if $\pair{\amodel}{\aworld}$ satisfies \ref{A:tiling_enc:vert}.
  \end{enumerate}
\end{lemma}

\begin{proof}
  We show the proof for I, the one for II being analogous.
  Recall that \ref{A:tiling_enc:hori} stands for:
  \begin{nscenter}
    $\forall$\,$\aworld_1,\aworld_2 \in \arelation(\aworld)$,
      if $\nbexp{\aworld_2}{\cH} = \nbexp{\aworld_1}{\cH}+1$ and $\nbexp{\aworld_2}{\cV} = \nbexp{\aworld_1}{\cV}$ then there is $\pair{\atile_1}{\atile_2} \in \cH$ s.t.\
      $\aworld_1 \in \avaluation(\atile_1)$ and $\aworld_2 \in \avaluation(\atile_2)$.
  \end{nscenter}
  Suppose $\amodel,\aworld \models \agrid{k}{\cTT}$. Then in particular $\amodel,\aworld \models \complete{k}$ and every world $\aworld' \in \arelation(\aworld)$ encodes a pair
  of numbers $\pair{\nbexp{\aworld}{\cH}}{\nbexp{\aworld}{\cV}} \in \interval{0}{\amapter(k,n)-1}^2$.

  ($\Rightarrow$):
  Suppose $\amodel,\aworld \models \phor{k}{\cTT}$. Then, by definition, for every
  $\amodel'\sqsubseteq\amodel$,
  if $\amodel',\aworld \models \fork{\avariable}{\avariablebis}{1}{k} \land
  \bsuccessor{\avariable}{\avariablebis}{k}{\cH}
  \land \bequivalent{\avariable}{\avariablebis}{k}{\cV}$ then
  $\amodel',\aworld \models \textstyle\bigvee_{{\pair{\atile_1}{\atile_2} \in \cH}} (\atnom{\avariable}{1} \atile_1 \land \atnom{\avariablebis}{1} \atile_2)$.
  Consider now two worlds $\forall$\,$\aworld_\avariable,\aworld_\avariablebis \in \arelation(\aworld)$
  such that $\nbexp{\aworld_\avariablebis}{\cH} = \nbexp{\aworld_\avariable}{\cH}+1$ and $\nbexp{\aworld_\avariablebis}{\cV} = \nbexp{\aworld_\avariable}{\cV}$.
  Let $\amodel' = \triple{\worlds}{\arelation_1}{\avaluation}$
  be the submodel of $\amodel$ where $\arelation_1$ is defined from $\arelation$ by removing the following pairs of worlds:
  \begin{itemize}
    \item $\pair{\aworld}{\aworld'} \in \arelation$ where $\aworld'$ is different from $\aworld_1$ and $\aworld_2$;
    \item $\pair{\aworld_\avariable}{\aworld''} \in \arelation$ where $\aworld''$ is the only $\Aux$-child of $\aworld_\avariable$ satisfying $\avariablebis$ (this world exists as $\amodel,\aworld \models \complete{k}$);
    \item $\pair{\aworld_\avariablebis}{\aworld'''} \in \arelation$ where $\aworld'''$ is the only $\Aux$-child of $\aworld_\avariablebis$ satisfying $\avariable$ (again, this world exists as $\amodel,\aworld \models \complete{k}$).
  \end{itemize}
  We can easily check that the pointed forest $\pair{\amodel'}{\aworld}$ satisfies $\fork{\avariable}{\avariablebis}{1}{k}$, where $\aworld_\avariable$ and $\aworld_\avariablebis$ correspond to two nominals (for the depth $1$) $\avariable$ and $\avariablebis$, respectively.
  Thus, $\amodel',\aworld_\avariable \models \complete{k-1}$ and $\amodel',\aworld_\avariablebis \models \complete{k-1}$.
  Therefore, by
  \Cref{A:the-property-of-types} (which can be easily extended in order to consider pairs of numbers described with $\treeval_\cH$ and $\treeval_\cV$, instead of a single number described with $\treeval$), we conclude that $\aworld_\avariable$ and $\aworld_\avariablebis$ keep encoding the same two pairs of numbers when $\amodel$ is modified to $\amodel'$.
  Then, since by hypothesis $\nbexp{\aworld_\avariablebis}{\cH} = \nbexp{\aworld_\avariable}{\cH}+1$ and
  $\nbexp{\aworld_\avariablebis}{\cV} = \nbexp{\aworld_\avariable}{\cV}$,
  by \Cref{A:lemma:bsuccessor,A:lemma:bequivalent} we conclude that $\amodel',\aworld \models \bsuccessor{\avariable}{\avariablebis}{k}{\cH}
  \land \bequivalent{\avariable}{\avariablebis}{k}{\cV}$.
  Then, by hypothesis $\amodel,\aworld \models \phor{k}{\cTT}$, we conclude that
  $\amodel',\aworld \models \textstyle\bigvee_{{\pair{\atile_1}{\atile_2} \in \cH}} (\atnom{\avariable}{1} \atile_1 \land \atnom{\avariablebis}{1} \atile_2)$.
  Thus, there must be a pair $\pair{\atile_1}{\atile_2} \in \cH$
  such that $\amodel',\aworld \models \atnom{\avariable}{1} \atile_1 \land \atnom{\avariablebis}{1} \atile_2$.
  Since $\aworld_\avariable$ (resp. $\aworld_\avariablebis$)
  corresponds to the nominal (for the depth $1$) $\avariable$ (resp. $\avariablebis$),
  we conclude that $\amodel,\aworld_{\avariable} \models \atile_1$ and $\amodel,\aworld_{\avariablebis} \models \atile_2$.
  By definition, this implies that $\pair{\amodel}{\aworld}$ satisfies \ref{A:tiling_enc:hori}.

  ($\Leftarrow$): This direction is rather straightforward and, analogously to the left-to-right direction, relies on \Cref{A:the-property-of-types,A:lemma:bsuccessor,A:lemma:bequivalent}.
  Briefly, suppose that $\pair{\amodel}{\aworld}$ satisfies \ref{A:tiling_enc:hori} and, \emph{ad absurdum}, assume that $\amodel,\aworld \not\models \phor{k}{\cTT}$.
  Therefore,
  \begin{nscenter}
    $\amodel,\aworld \models \true
      \separate
      \big(
      \fork{\avariable}{\avariablebis}{1}{k} \land
      \bsuccessor{\avariable}{\avariablebis}{k}{\cH}
      \land \bequivalent{\avariable}{\avariablebis}{k}{\cV}
       \land
      \lnot \textstyle\bigvee_{{\pair{\atile_1}{\atile_2} \in \cH}} (\atnom{\avariable}{1} \atile_1 \land \atnom{\avariablebis}{1} \atile_2)
      \big)$.
  \end{nscenter}
  Then, there is a submodel $\amodel' = \triple{\worlds}{\arelation}{\avaluation}$ of $\amodel$ such that
  $\amodel',\aworld \models \fork{\avariable}{\avariablebis}{1}{k} \land
  \bsuccessor{\avariable}{\avariablebis}{k}{\cH}
  \land \bequivalent{\avariable}{\avariablebis}{k}{\cV}
   \land
  \lnot \textstyle\bigvee_{{\pair{\atile_1}{\atile_2} \in \cH}} (\atnom{\avariable}{1} \atile_1 \land \atnom{\avariablebis}{1} \atile_2)$.
  By $\amodel',\aworld \models \fork{\avariable}{\avariablebis}{1}{k}$
  we conclude that there are two worlds $\aworld_\avariable$ and $\aworld_\avariablebis$ corresponding to two nominals (depth 1) $\avariable$ and $\avariablebis$, respectively.
  Moreover, by \ref{A:the-property-of-types}, these worlds encode the same two numbers w.r.t. $\pair{\amodel}{\aworld}$ and $\pair{\amodel'}{\aworld}$.
  From $\amodel',\aworld \models
  \bsuccessor{\avariable}{\avariablebis}{k}{\cH}
  \land \bequivalent{\avariable}{\avariablebis}{k}{\cV}$
  and the fact that $\pair{\amodel}{\aworld}$ satisfies \ref{A:tiling_enc:hori},
  together with \Cref{A:lemma:bsuccessor,A:lemma:bequivalent}
  we conclude that there is a pair $\pair{\atile_1}{\atile_2} \in \cH$ such that
  $\aworld_\avariable \in \avaluation(\atile_1)$ and $\aworld_\avariablebis \in \avaluation(\atile_2)$.
  However, this contradicts $\amodel',\aworld \models
  \lnot \textstyle\bigvee_{{\pair{\atile_1}{\atile_2} \in \cH}} (\atnom{\avariable}{1} \atile_1 \land \atnom{\avariablebis}{1} \atile_2)$.
  Thus, $\amodel,\aworld \models \phor{k}{\cTT}$.
\end{proof}

This concludes the definition of $\atiling{k}{\cTT,\atile}$.

\begin{lemma}\label{A:lemma:tiling:correcntess}
  $\amodel,\aworld \models \atiling{k}{\cTT,\atile}$ if and only if
  $\pair{\amodel}{\aworld}$ satisfies
  \ref{A:tiling_enc:apZ},
  \ref{A:tiling_enc:apU},
  \ref{A:tiling_enc:apC},
  \ref{A:tiling_enc:apS},
  \ref{A:tiling_enc:one_tile},
  \ref{A:tiling_enc:zero_c},
  \ref{A:tiling_enc:hori} and \ref{A:tiling_enc:vert}.
\end{lemma}

\begin{proof}
  Directly from \Cref{A:lemma:grid:correctness,A:lemma:tiling:zero-and-unique,A:lemma:tiling:hor-and-vert}.
\end{proof}

We can now prove \Cref{lemma:tiling_correct_encoding} (shown below), leading directly to \Cref{theorem:tower-completeness-SC}.\\[-5pt]

\begin{lemma*}[\ref{lemma:tiling_correct_encoding}]
Let $k \geq 2$ and let $\pair{\cTT}{\atile}$ be an instance of
$\tiling_k$, where $\cTT = \triple{\cT}{\cH}{\cV}$ and $\atile \in \cT$. Then,
\begin{nscenter}
$\pair{\cTT}{\atile}$ as a solution for $\tiling_k$ if and only if
the formula $\atiling{k}{\cTT,\atile}$ is satisfiable.
\end{nscenter}
\end{lemma*}

\begin{proof}
  ($\Rightarrow$): Suppose that $\pair{\cTT}{\atile}$ has a solution $\tau : \interval{0}{\amapter(k,n)-1}^2 \to \cT$.
  Let $\amodel = \triple{\worlds}{\arelation}{\avaluation}$ and $\aworld \in \worlds$
  be such that $\amodel,\aworld \models \agrid{k}{\cTT}$ (such a pointed forest exists by \Cref{A:corollary:grid-satisfiable}).
  We slightly modify $\avaluation$ so that the resulting model still satisfies  $\agrid{k}{\cTT}$, but also satisfies \ref{A:tiling_enc:one_tile},
  \ref{A:tiling_enc:zero_c},
  \ref{A:tiling_enc:hori} and \ref{A:tiling_enc:vert}.
  This can be done rather straightforwardly.
  Indeed, since $\amodel,\aworld \models \agrid{k}{\cTT}$, by \Cref{A:lemma:grid:correctness} every $\avartree$-node $\aworld' \in \arelation(\aworld)$
  encodes a pair of numbers $\pair{\nbexp{\aworld'}{\cH}}{\nbexp{\aworld'}{\cV}} \in \interval{0}{\amapter(k,n)-1}$.
  Then, let us consider the model $\amodel' = \triple{\worlds}{\arelation}{\avaluation'}$ such that
  \begin{enumerate}
    \item for every $\avarprop \in \AP \setminus \cT$, $\avaluation'(\avarprop) = \avaluation(\avarprop)$. This property leads to $\amodel',\aworld \models \agrid{k}{\cTT}$, since $\agrid{k}{\cTT}$ is written with propositional symbols not appearing in $\cT$.
    \item for every $\atile \in \cT$ and $\aworld' \in \arelation(\aworld)$,
    $\aworld' \in \avaluation(\atile)$ if and only if $\tau(\nbexp{\aworld'}{\cH},\nbexp{\aworld'}{\cV}) = \atile$.
  \end{enumerate}
  The second condition allows us to conclude that $\pair{\amodel'}{\aworld}$ satisfies  \ref{A:tiling_enc:one_tile},
  \ref{A:tiling_enc:zero_c},
  \ref{A:tiling_enc:hori} and \ref{A:tiling_enc:vert}.
  Indeed, \ref{A:tiling_enc:one_tile} holds as $\tau$ is functional;
  \ref{A:tiling_enc:zero_c} holds as $\tau$ satisfies \ref{tiling_c:1};
  whereas   \ref{A:tiling_enc:hori} and \ref{A:tiling_enc:vert} hold as $\tau$ satisfies \ref{tiling_c:2}.
  Thus, $\pair{\amodel'}{\aworld} \models \atiling{k}{\cTT,\atile}$ and therefore $\atiling{k}{\cTT,\atile}$ is satisfiable.

  ($\Leftarrow$): Suppose $\atiling{k}{\cTT,\atile}$ satisfiable and let $\amodel = \triple{\worlds}{\arelation}{\avaluation}$ and $\aworld \in \worlds$ s.t. $\amodel,\aworld \models \atiling{k}{\cTT,\atile}$.
  Let us consider the relation $\tau \subseteq \interval{0}{\amapter(k,n)-1} \times \interval{0}{\amapter(k,n)-1} \times \cT$ defined as
  \begin{nscenter}
    $\triple{i}{j}{\atile'} \in \tau$ if and only if there is $\aworld' \in \arelation(\aworld)$ s.t.\ $\nbexp{\aworld'}{\cH} = i$, $\nbexp{\aworld'}{\cV} = j$ and $\aworld' \in \avaluation(\atile')$.
  \end{nscenter}
  Directly by \Cref{A:lemma:tiling:correcntess} we have that:
  \begin{enumerate}[label=\Roman*.]
    \item from \ref{A:tiling_enc:apU} and \ref{A:tiling_enc:one_tile}, $\tau$ is (possibly weakly) functional in its first two components,
    i.e. for every $\pair{i}{j} \in \interval{0}{\amapter(k,n)-1}^2$ there is at most one $\atile'$ such that $\triple{i}{j}{\atile'} \in \tau$;
    \item from \ref{A:tiling_enc:apZ} and \ref{A:tiling_enc:apC}, $\tau$ is total (hence not weakly functional),
    i.e. cannot be that there is $\pair{i}{j} \in \interval{0}{\amapter(k,n)-1}^2$
    such that for every $\atile' \in \cT$, $\triple{i}{j}{\atile'} \not\in \tau$.
    Together with I, this means that $\tau$ is a map;
    \item from \ref{A:tiling_enc:zero_c}, $\triple{0}{0}{\atile} \in \tau$;
    \item from \ref{A:tiling_enc:hori} and \ref{A:tiling_enc:vert}, for all
    $i\in \interval{0}{\amapter(k,n)-1}$ and $j \in \interval{0}{\amapter(k,n)-2}$, $\pair{\tau(j,i)}{\tau(j+1,i)} \in \cH\ $ and $\ \pair{\tau(i,j)}{\tau(i,j+1)} \in \cV$.
  \end{enumerate}
  Therefore, we conclude that $\tau$ is a solution for $\tiling_k$.
\end{proof}

\section{Proofs of \Cref{section-expressivity-SC}}

To show the existence of a formula in \GML that is equivalent to $\SabDiamond \aformula$, we rely on the
indistinguishability relation \GML, called g-bisimulation and studied in~\cite{DeRijke00}.
So, let us first recall what is a g-bisimulation.
Let $\amodel = \triple{\worlds}{\arelation}{\avaluation}$ and $\amodel' = \triple{\worlds'}{\arelation'}{\avaluation'}$ be two
finite forests. Let $m \in \Nat ,k \in \Nat^{>0}$ and $\apropset \subseteq_{\fin} \varprop$.
A \emph{g-bisimulation} up to $\triple{m}{k}{\apropset}$ between $\amodel$ and $\amodel'$ is a sequence of $m+1$ $k$-uple $\abisim^0 = (\abisim^0_1,\abisim^0_2,\dots,\abisim^0_k)$, $\dots$, $\abisim^m = (\abisim^m_1,\abisim^m_2,\dots,\abisim^m_k)$ satisfying:
\begin{description}
\item[\rulelab{init}{gbisim:init}] $\abisim^0_1$ is not empty and for every $i \in [1,k]$ and $j \in \interval{0}{m}$, $\abisim^j_i \subseteq \powerset{\worlds} \times \powerset{\worlds'}$;
\item[\rulelab{refine}{gbisim:refine}] for every $i \in [1,k]$ and $j \in [1,m]$, $\abisim^j_i \subseteq \abisim^{j-1}_{i}$;
\item[\rulelab{size}{gbisim:size}] if $X \abisim^j_i Y$ then $\card{X} = \card{Y} = i$;
\item[\rulelab{atoms}{gbisim:atoms}] if $\{\aworld\}\abisim^0_1\{\aworld'\}$ then for every $\avarprop \in \apropset$, $\aworld \in \avaluation(\avarprop)$ if and only if $\aworld' \in \avaluation'(\avarprop)$;
\item[\rulelab{m-forth}{gbisim:m-forth}] if $\{\aworld\} \abisim^{j+1}_1 \{\aworld'\}$ and $X {\subseteq} \arelation(\aworld)$ with $\card{X} {\in} \interval{1}{k}$, then there is $Y {\subseteq} \arelation'(\aworld')$ such that $X \abisim^j_{\card{X}}Y$;
\item[\rulelab{m-back}{gbisim:m-back}] if $\{\aworld\} \abisim^{j+1}_1 \{\aworld'\}$ and $Y {\subseteq} \arelation'(\aworld')$ with $\card{Y} {\in} \interval{1}{k}$, then there is $X {\subseteq} \arelation(\aworld)$ such that 
$X \abisim^j_{\card{Y}}Y$;
\item[\rulelab{g-forth}{gbisim:g-forth}] if $X \abisim^j_i Y$ and $\aworld \in X$, then there is $\aworld' \in Y$ such that $\{\aworld\}\abisim^j_1 \{\aworld'\}$;
\item[\rulelab{g-back}{gbisim:g-back}] if $X \abisim^j_i Y$ and $\aworld' \in Y$, then there is $\aworld \in X$ such that $\{\aworld\}\abisim^j_1 \{\aworld'\}$.
\end{description}
We write $\amodel,\aworld \bisrel{m,k}{\apropset} \amodel',\aworld'$ 
and we say that the two models are g-bisimilar iff there is a g-bisimulation up to $\triple{m}{k}{\apropset}$ between $\amodel$ and $\amodel'$, say $\abisim^0,\dots,\abisim^m$, such that $\{\aworld\} \abisim^m_1 \{\aworld'\}$.
We write  $\acharformula{\amodel,\aworld}{m,k}{\apropset}$ to denote 
 the set of formulae in \GML of rank $\pair{m}{k}$ and with
propositional symbols from $\apropset$ that are satisfied in $\amodel,\aworld$,
i.e.  $\acharformula{\amodel,\aworld}{m,k}{\apropset} \egdef \set{\aformulabis \in \GML[m,k,\apropset] \mid
\amodel, \aworld \models \aformulabis}$.
We write $\atypeset{m,k}{\apropset}$ the quotient set induced by  the equivalence relation $ \bisrel{m,k}{\apropset}$.
Let us summarise the main results from~\cite{DeRijke00}.
\begin{proposition}[\cite{DeRijke00}]\label{prop:rijke00}%
\begin{minipage}[t]{0.7\linewidth}
\vspace{-0.3cm}
\begin{enumerate}
\item $\acharformula{\amodel,\aworld}{m,k}{\apropset}$ contains finitely many non-equivalent formulae.
\item
 $\amodel,\aworld \bisrel{m,k}{\apropset} \amodel',\aworld'$ if and only if \ $\acharformula{\amodel,\aworld}{m,k}{\apropset} = \acharformula{\amodel',\aworld'}{m,k}{\apropset}$.
 \item $\bisrel{m,k}{\apropset}$ is a finite index equivalence relation. $\atypeset{m,k}{\apropset}$ is finite.
\end{enumerate}
\end{minipage}
\end{proposition}
So, $\equiv_{m,k}^{\apropset}$ and $\bisrel{m,k}{\apropset}$ are identical relations
(see the definitions for $\equiv_{m,k}^{\apropset}$ and $\GML[m,k,\apropset]$ in~\Cref{section-sc-into-gml})
 and there is a finite
set $\set{\aformulater_1, \ldots, \aformulater_{Q}} \subseteq \GML[m,k,\apropset]$ such that
\begin{itemize}
\item $\aformulater_1 \vee \cdots \vee \aformulater_{Q}$ is valid, and each $\aformulater_i$ is satisfiable,
\item for all $i \neq j \in \interval{1}{Q}$, $\aformulater_i \wedge \aformulater_j$ is unsatisfiable,
\item $\pair{\amodel}{\aworld} \equiv_{m,k}^{\apropset} \pair{\amodel'}{\aworld'}$ iff there is $i$ such that
      $\pair{\amodel}{\aworld} \models \aformulater_i$ and $\pair{\amodel'}{\aworld'} \models \aformulater_i$.
\end{itemize}
 Hence, $\aformulater_i$ characterises one equivalence class of  $\equiv_{m,k}^{\apropset}$ (or equivalently of $\bisrel{m,k}{\apropset}$).

 In what follows, recall that $\arelation|_{\aworld} \egdef \{\pair{\aworld'}{\aworld''} \in \arelation \mid \aworld' \subseteq \arelation^*(\aworld)\}$.

\begin{lemma}\label{lemma:gbisimcutrel}
Let $m \in \Nat$, $k \in \Nat^{>0}$ and $\apropset \subseteq_{\fin} \varprop$. Let $\amodel = \triple{\worlds}{\arelation}{\avaluation}$ be a finite forest and let $\aworld \in \worlds$.
Then, $\amodel,\aworld \bisrel{m,k}{\apropset} \triple{\worlds}{\arelation|_{\aworld}}{\avaluation},\aworld$.
\end{lemma}

\begin{proof}
As $\bisrel{m,k}{\apropset}$ is an equivalence relation (Proposition~\ref{prop:rijke00}.3), it is reflexive and hence $\amodel,\aworld \bisrel{m,k}{\apropset} \amodel,\aworld$.
There is therefore a g-bisimulation up to $\triple{m}{k}{\apropset}$ between $\amodel$ and itself, say $\abisim^0,\dots,\abisim^m$ where $\abisim^i = (\abisim^i_1,\dots,\abisim^i_k)$ for every $i \in \interval{0}{m}$, such that $\{\aworld\} \abisim^m_1 \{\aworld\}$.
Consider now the restriction of $\abisim^i_j$, where $i \in \interval{0}{m}$ and $j \in \interval{1}{k}$, to those sets where every element is reachable from $\aworld$. Formally, we define
$\widehat{\abisim^i_j} = \{ \pair{X}{Y} \in \abisim^i_j \mid X \cup Y \subseteq \arelation^*(\aworld) \}$.
It is easy to show that $\widehat{\abisim^0},\dots,\widehat{\abisim^m}$, where $\widehat{\abisim^i} = (\widehat{\abisim^i_1},\dots,\widehat{\abisim^i_k})$ for every $i \in \interval{0}{m}$,
is a g-bisimulation up to $\triple{m}{k}{\apropset}$ between $\amodel$ and $\triple{\worlds}{\arelation|_{\aworld}}{\avaluation}$.
Moreover, as  $\{\aworld\} \widehat{\abisim^m_1} \{\aworld\}$ by definition, we conclude that $\amodel,\aworld \bisrel{m,k}{\apropset} \triple{\worlds}{\arelation|_{\aworld}}{\avaluation},\aworld$.
\end{proof}

\subsection{Proof of \Cref{lemma:sabotage-elimination}}
In the following, we denote with
$\atowertypeset{m,k}{\apropset}$
the set $\atypeset{m,\amap(m,k)}{\apropset}$.
Then, notice that 
$\atowertypeset{m,k}{\apropset} = \atypeset{0,k}{\apropset}$ for $m=0$, and otherwise ($m \geq 1$)
$\atowertypeset{m,k}{\apropset} = \atypeset{m,k \times (\card{\atowertypeset{m-1,k}{\apropset}}+1)}{\apropset}$.
Since $\atypeset{m',k'}{\apropset'}$ is finite for all $m',k'$ and $\apropset'$,
$\atowertypeset{m,k}{\apropset}$ is well-defined and finite.
\Cref{lemma:sabotage-elimination} can be reformulated using $\atowertypeset{m,k}{\apropset}$ as follows.

\begin{lemma*}
Let $m,k\in \Nat$ and $\apropset \subseteq_\fin \varprop$.
Let $\pair{\amodel}{\aworld}$, $\pair{\amodel'}{\aworld'}$ be pointed forests such that $\amodel= \triple{\worlds}{\arelation}{\avaluation}$ and $\amodel' = \triple{\worlds'}{\arelation'}{\avaluation'}$.
If $\{\pair{\amodel}{\aworld}, \pair{\amodel'}{\aworld'}\} \subseteq \atermset$ for some $\atermset \in \atowertypeset{m,k}{\apropset}$,  then
for every $\arelation_1 \subseteq \arelation$ there is $\arelation_1' \subseteq \arelation'$ s.t.\
$\pair{\triple{\worlds}{\arelation_1}{\avaluation}}{\aworld}
\typeeqclass{m,k}{\apropset}
\pair{\triple{\worlds'}{\arelation_1'}{\avaluation'}}{\aworld'}$, and
if $\arelation_1(\aworld) = \arelation(\aworld)$ then $\arelation_1'(\aworld') = \arelation'(\aworld')$.
\end{lemma*}

\begin{proof}
In the case $k=0$, any formula in $\GML[m,0,\apropset]$ is equivalent to
a formula in the propositional calculus built over propositional variables
in $\apropset$ as $\Gdiamond{\geq 0} \aformulabis$ is logically equivalent to
$\top$. Hence, the lemma trivially holds.

Otherwise  ($k \geq 1$), we prove semantically the lemma
as $\typeeqclass{m,k}{\apropset}$ and $\bisrel{m,k}{\apropset}$ are identical relations.
The proof is by induction on the modal depth $m$.
The induction step is articulated in three main steps:
\begin{description}
\item[(I)] definition and proof of various properties of the two models,
\item[(II)] definition of a strategy to reduce $\arelation'$ to $\arelation_1'$ that closely follows the relationship between $\arelation$ and $\arelation_1$ with respect to the children of $\aworld$ and,
\item[(III)] a proof that the relation $\arelation_1'$ is such that $\triple{\worlds}{\arelation_1}{\avaluation},\aworld \bisrel{m,k}{\apropset} \triple{\worlds'}{\arelation_1'}{\avaluation'},\aworld'$.
By construction, we also obtain that if $\arelation_1(\aworld) = \arelation(\aworld)$ then $\arelation_1'(\aworld') = \arelation'(\aworld')$.
\end{description}
Let us begin with the base case.
\begin{description}
\item[Base case: $m = 0$.]
The base case is straightforward from the following property of g-bisimulations.
When $m = 0$, given $\widehat{\amodel} = \triple{\widehat{\worlds}}{\widehat{\arelation}}{\widehat{\avaluation}}$,
$\widehat{\arelation}_1 \subseteq \widehat{\arelation}$, $\widehat{\aworld} \in \widehat{\worlds}$
and $\widehat{k} \in \Nat$, we have  $\widehat{\amodel},\widehat{\aworld} \bisrel{0,\widehat{k}}{\apropset}
\triple{\widehat{\worlds}}{\widehat{\arelation_1}}{\widehat{\avaluation}},\widehat{\aworld}$.
This statement holds as it can be easily shown that the set of relations $\abisim^0 = (\abisim^0_1,\dots,\abisim^0_{\widehat{k}})$ where $\abisim^0_1 = \{(\aworld,\aworld)\}$ and $\abisim^0_j = \emptyset$ for $j \in \interval{2}{\widehat{k}}$
satisfies all the requirements for being a g-bisimulation.

Then, with respect to the statement of the lemma, by definition, we have
$\triple{\worlds}{\arelation_1}{\avaluation},\aworld \bisrel{0,k}{\apropset} \amodel,\aworld$.
Now, by definition $\atowertypeset{0,k}{\apropset} = \atypeset{0,k}{\apropset}$ and by hypothesis there is $\atermset \in \atypeset{0,k}{\apropset}$ such that $\{\pair{\amodel}{\aworld}, \pair{\amodel'}{\aworld'}\} \subseteq \atermset$.
By definition of $\atypeset{0,k}{\apropset}$, we have 
$$
\amodel,\aworld \bisrel{0,k}{\apropset}\amodel',\aworld'.
$$
As $\bisrel{0,k}{\apropset}$ is an equivalence relation, we conclude $\triple{\worlds}{\arelation_1}{\avaluation},\aworld \bisrel{0,k}{\apropset} \amodel',\aworld'$ and therefore it is sufficient to take $\arelation_1' \egdef \arelation'$
to end the proof. Note that in this case, $\arelation_1'(\aworld') = \arelation'(\aworld')$ holds too.

\item[Induction case.]
In particular, we have $m > 1$ and $\atowertypeset{m,k}{\apropset} = \atypeset{m,k\times (\card{\atowertypeset{m-1,k}{\apropset}}+1)}{\apropset}$.
Moreover, by hypothesis there exists
$\atermset \in \atypeset{m,k\times (\card{\atowertypeset{m-1,k}{\apropset}}+1)}{\apropset}$
such that $\{\pair{\amodel}{\aworld}, \pair{\amodel'}{\aworld'}\} \subseteq \atermset$.
By definition, we have
$$
\amodel,\aworld \bisrel{m,k \times (\card{\atowertypeset{m-1,k}{\apropset}}+1)}{\apropset} \amodel',\aworld'.
$$

Let us explain the main idea of the proof. Let us pick  one child $\aworld_1$ of $\aworld$ in $\amodel$.
Obviously, the pointed forest $\pair{\amodel}{\aworld_1}$ belongs to
a specific equivalence class $\atermset \in \atowertypeset{m-1,k}{\apropset}$.
The effect of reducing $\arelation$
to $\arelation_1$ is that $\aworld_1$, together with the updated model,
``jumps''\footnote{We always put the word ``jump'' in quotes as it is used in an informal way.} to an equivalence class
$\atermset_1 \in \atypeset{m-1,k}{\apropset}$.
Obviously, $\pair{\amodel}{\aworld_1}$ already belongs to a class in $\atypeset{m-1,k}{\apropset}$.
However (from the statement of the lemma), we are only interested in $\atypeset{m-1,k}{\apropset}$ when considering $\arelation_1$, whereas we focus on $\atowertypeset{m-1,k}{\apropset}$ when studying $\arelation$.
To prove the result, we have to show that there is a child $\aworld_1'$ of $\aworld'$ in $\amodel'$ so that $\pair{\amodel'}{\aworld_1'}$ is in the same equivalence class $\atermset$ of $\pair{\amodel}{\aworld_1}$ and
to show that it is possible to update $\arelation'$ to make $\aworld_1'$ (together with the updated model) ``jump'' to the equivalence class $\atermset_1$.
However, we need to do this for all the children of $\aworld$ and $\aworld'$, respecting the constraints of being
a g-bisimulation.
 The key step is to show that the graded rank $k \times (\card{\atowertypeset{m-1,k}{\apropset}}+1)$
is all we need to find enough children in $\arelation'(\aworld')$ and to be able to construct a relation
$\arelation_1'$ so that the resulting models are g-bisimilar up to $\triple{m}{k}{\apropset}$.
Let us now formalise the proof, which requires some intermediate steps that are below $\boxed{\text{highlighted}}$.

We start by considering a single equivalence class $\atermset \in \atowertypeset{m-1,k}{\apropset}$ (in fact, our proof is done modularly on these classes).
We introduce the two following sets:
\begin{itemize}
\item $\arelation(\aworld)|_{\atermset} \egdef \{ \aworld_1 \in \arelation(\aworld) \mid
\pair{\amodel}{\aworld_1} \in \atermset\}$.
\item $\arelation'(\aworld')|_{\atermset} \egdef \{ \aworld_1' \in \arelation'(\aworld') \mid \pair{\amodel'}{\aworld_1'} \in \atermset\}$.
\end{itemize}
It is fairly simple to see that the following property holds:
\begin{nscenter}
$\boxed{\text{\rulelab{($\star$)}{SDelim:cardequiv}\qquad\qquad
$\min(\card{\arelation(\aworld)|_\atermset}, k \times (\card{\atowertypeset{m-1,k}{\apropset}}+1)) = \min(\card{\arelation'(\aworld')|_\atermset}, k \times (\card{\atowertypeset{m-1,k}{\apropset}}+1))$}}$
\end{nscenter}
Indeed, ad absurdum, suppose  that
\begin{nscenter}
\rulelab{($\dagger$)}{SDelim:abshyp1}\qquad\qquad$\card{\arelation(\aworld)|_\atermset} < k \times 
(\card{\atowertypeset{m-1,k}{\apropset}}+1)$ and $\card{\arelation(\aworld)|_\atermset} < \card{\arelation'(\aworld')|_\atermset}$
\end{nscenter}

The other case
$\card{\arelation'(\aworld')|_\atermset} < k \times (\card{\atowertypeset{m-1,k}{\apropset}}+1)$ 
and $\card{\arelation'(\aworld')|_\atermset} < \card{\arelation(\aworld)|_\atermset}$ is analogous and therefore
its treatment is omitted below.
Since it holds by hypothesis that $\amodel,\aworld \bisrel{m,k \times (\card{\atowertypeset{m-1,k}{\apropset}}+1)}{\apropset} \amodel',\aworld'$,
there is a g-bisimulation up to $\triple{m}{k \times (\card{\atowertypeset{m-1,k}{\apropset}}+1)}{\apropset}$ between $\amodel$
and $\amodel'$, say $\abisim^0,\dots,\abisim^m$, such that $\{\aworld\} \abisim^m_1 \{\aworld'\}$.
\begin{itemize}
\item From (\ref{gbisim:m-back}), by taking $Y$ as a subset of
$\arelation'(\aworld')|_\atermset$ such that
\begin{nscenter}
$\card{Y} = \min(\card{\arelation'(\aworld')|_\atermset},k \times( \card{\atowertypeset{m-1,k}{\apropset}}+1))$,
\end{nscenter}
it must hold that there is a subset $X \subseteq \arelation(\aworld)$ such that $X \abisim^{m-1}_{\card{Y}} Y$.
\item From (\ref{gbisim:size}), $\card{X} = \card{Y}$. Hence, by \ref{SDelim:abshyp1} there must be a world $\aworld_2 \in X$ s.t.\ $\pair{\amodel}{\aworld_2} \not \in \atermset$.
\item From (\ref{gbisim:g-forth}), there is $\aworld_2' \in Y$ such that $\{\aworld_2\} \abisim^{m-1}_1 \{\aworld_2'\}$.
\item As $\{\aworld_2\} \abisim^{m-1}_1 \{\aworld_2'\}$, from the definition of g-bisimulation it holds that
\begin{nscenter}
$\amodel,\aworld_2 \bisrel{m-1,k \times (\card{\atowertypeset{m-1,k}{\apropset}}+1)}{\apropset} \amodel',\aworld_2'$.
\end{nscenter}
\item Again by definition of g-bisimulation, it is easy to see that if two models are in the same equivalence class w.r.t. 
$\bisrel{m',k'}{\apropset}$ then they are in the same equivalence class w.r.t. $\bisrel{m',k''}{\apropset}$ for every $k'' \leq k'$.
Therefore $\amodel,\aworld_2 \bisrel{m-1,k \times (\card{\atowertypeset{m-2,k}{\apropset}}+1)}{\apropset} \amodel',\aworld_2'$.
Notice that the set of equivalence classes induced by $\bisrel{m-1,k \times (\card{\atowertypeset{m-2,k}{\apropset}}+1)}{\apropset}$ is $\atowertypeset{m-1,k}{\apropset}$.
We conclude that $\pair{\amodel}{\aworld_2}$ and $\pair{\amodel'}{\aworld_2'}$ belong to the same class in $\atowertypeset{m-1,k}{\apropset}$.
However, this leads to a contradiction
as we have $\aworld_2 \not\in \atermset$ and $\aworld_2' \in \atermset$ (where $\atermset \in \atowertypeset{m-1,k}{\apropset}$).
\end{itemize}
This concludes the proof of \ref{SDelim:cardequiv}.

Given an equivalence class $\atermset'$ in $\atypeset{m-1,k}{\apropset}$, we define the set below
\begin{nscenter}
$\arelation_1(\aworld)|_{\atermset \blacktriangleright \atermset'} \egdef \arelation(\aworld)|_{\atermset} \cap \arelation_1(\aworld)|_{\atermset'}$.
\end{nscenter}
Following the proof idea presented above, a world $ \aworld_1 \in \arelation_1(\aworld)|_{\atermset \blacktriangleright \atermset'}$ is a child of $\aworld$ such that $\pair{\amodel}{\aworld_1}$ is in the class $\atermset$ and ``jumps'' to the class $\atermset'$ when updating the accessibility relation from $\arelation$ to $\arelation_1$.
In what follows, we denote with $\arelation|_{\aworld_1}$ the restriction of $\arelation$ to those worlds reachable from $\aworld_1$, i.e. the set $\{\pair{\aworld_2}{\aworld_3} \in \arelation \mid \{\aworld_2,\aworld_3\} \subseteq \arelation^*(\aworld_1)\}$, as defined in the statement of Lemma~\ref{lemma:gbisimcutrel}.
We also consider similar restrictions for $\arelation'$ and $\arelation_1'$.
 We are interested in the following key property:
\begin{nscenter}
$\boxed{\text{\rulelab{($\star\star$)}{SDelim:updatetype}\qquad\qquad
$\begin{aligned}[t]&\text{for every }\aworld_1 \in \arelation_1(\aworld)|_{\atermset \blacktriangleright \atermset'} \text{ and } \aworld_1' \in \arelation'(\aworld')|_\atermset \text{ there is }
\arelation_{1,\aworld_1'}' \subseteq \arelation'|_{\aworld_1'}\\
&\text{such that } \triple{\worlds}{\arelation_1|_{\aworld_1}}{\avaluation},\aworld_1 \bisrel{m-1,k}{\apropset} \triple{\worlds'}{\arelation_{1,\aworld_1'}'}{\avaluation'},\aworld_1'
\end{aligned}
$}}$
\end{nscenter}
Let us prove~\ref{SDelim:updatetype}.
By definition, we have $\aworld_1 \in \arelation(\aworld)|_{\atermset}$ and $\aworld_1' \in \arelation'(\aworld')|_{\atermset}$.
Therefore, $\{\pair{\amodel}{\aworld_1}, \pair{\amodel'}{\aworld_1'}\} \subseteq \atermset \in \atowertypeset{m-1,k}{\apropset}$.
By Lemma~\ref{lemma:gbisimcutrel}, it follows that $\triple{\worlds}{\arelation|_{\aworld_1}}{\avaluation}, \aworld_1$
and  $\triple{\worlds'}{\arelation'|_{\aworld_1'}}{\avaluation'}, \aworld'_1$ are also in $\atermset$.
Moreover, by definition $\arelation_1|_{\aworld_1} \subseteq \arelation|_{\aworld_1}$.
Then, we can use the induction hypothesis (notice that the modal degree is now $m-1$)
to conclude that there is $\arelation_{1,\aworld_1'}' \subseteq \arelation'|_{\aworld_1'}$
such that
$\triple{\worlds}{\arelation_1|_{\aworld_1}}{\avaluation},\aworld_1 \bisrel{m-1,k}{\apropset} \triple{\worlds'}{\arelation_{1,\aworld_1'}'}{\avaluation'},\aworld_1'$,
concluding the proof of \ref{SDelim:updatetype}.
This intermediate result gives us an important information: every single ``jump'' (as informally expressed above) done while updating the accessibility relation of $\amodel$ can be mimicked by updating $\amodel'$. An important missing piece is proving that all jumps can be simultaneously mimicked.
In order to prove this, we start by considering
the following partition of $\arelation(\aworld)|_\atermset$:
\begin{nscenter}
$\arelation(\aworld)_{\blacktriangleright\arelation_1}^\atermset \egdef \{ \arelation_1(\aworld)|_{\atermset \blacktriangleright \atermset'} \mid \atermset' \in \atypeset{m-1,k}{\apropset} \} \cup \{ \arelation(\aworld)|_{\atermset} \setminus \arelation_1(\aworld)\}.$
\end{nscenter}
Informally, $\arelation(\aworld)_{\blacktriangleright\arelation_1}^\atermset$ partitions the children of $\aworld$ in $\arelation(\aworld)|_\atermset$ in different sets depending on what is the set $\atermset' \in \atypeset{m-1,k}{\apropset}$ they ``jump'' to. One additional set, i.e. $\arelation(\aworld)|_{\atermset} \setminus \arelation_1(\aworld)$, contains all the children of $\aworld$ in  $\arelation(\aworld)|_\atermset$ that are lost when updating $\arelation$ to $\arelation_1$.
To be completely formal, let us first prove that $\arelation(\aworld)_{\blacktriangleright\arelation_1}^\atermset$ is a partition of $\arelation(\aworld)|_\atermset$.
Indeed, $\arelation(\aworld)|_\atermset$ can be written as $(\arelation(\aworld)|_\atermset \cap \arelation_1(\aworld)) \cup (\arelation(\aworld)|_\atermset \setminus \arelation_1(\aworld))$. Moreover, by definition of $\atypeset{m-1,k}{\apropset}$ as the quotient set of $\bisrel{m-1,k}{\apropset}$, we have
$\arelation_1(\aworld) = \bigcup_{\atermset' \in \atypeset{m-1,k}{\apropset}} \arelation_1(\aworld)|_{\atermset'}$.
Lastly, $\arelation(\aworld)|_\atermset \cap \bigcup_{\atermset' \in \atypeset{m-1,k}{\apropset}} \arelation_1(\aworld)|_{\atermset'}$ is equivalent to $ \bigcup_{\atermset' \in \atypeset{m-1,k}{\apropset}} (\arelation(\aworld)|_\atermset \cap \arelation_1(\aworld)|_{\atermset'})$, which
leads to the definition of the partition $\arelation(\aworld)_{\blacktriangleright\arelation_1}^\atermset$  from the definition of
$\arelation_1(\aworld)|_{\atermset \blacktriangleright \atermset'}$ together with the remaining component $\arelation(\aworld)|_{\atermset} \setminus \arelation_1(\aworld)$. The figure below presents schematically  the results we have shown so far, only considering the children of $\aworld$ in
$\arelation(\aworld)|_\atermset$ (on the left) and the children of $\aworld'$ in $\arelation'(\aworld')|_\atermset$ (on the right).

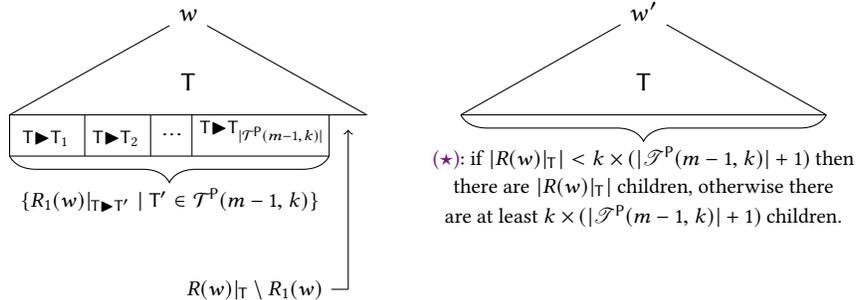
\begin{figure}[h]
\centering
\begin{tikzpicture}[auto]
\node (w) {$\aworld$};

\node (a1) [below left = 1cm and 2cm of w] {};
\node (a2) [below right = 1cm and 2cm of w] {};

\node (aa) [right = 4cm of a1] {};
\node (aaa) [below = 0.3cm of aa] {};

\node (h1) [right = 0.875cm of a1.center] {};
\node (h1a) [below = 0.3cm of h1] {};

\node (h2) [right = 1.75cm of a1.center] {};
\node (h2a) [below = 0.3cm of h2] {};

\node (h3) [left = 1.7 of aa.center] {};
\node (h3a) [below = 0.3cm of h3] {};

\node (a1a) [below = 0.3cm of a1] {};

\node (b1) [right = 1cm of a2] {};
\node (w2) [above right= 1cm and 2cm of b1] {$\aworld'$};
\node (b2) [below right = 1cm and 2cm of w2] {};

\draw (w) -- (a1.center);
\draw (a1.center) -- node[above=0.2cm] {$\atermset$} (a2.center);
\draw (a2.center) -- (w);
\draw (aa.center) -- (aaa.center);
\draw (a1.center) -- (a1a.center);
\draw (a1a.center) -- (aaa.center);

\draw (h1.center) -- (h1a.center);
\draw (h2.center) -- (h2a.center);
\draw (h3.center) -- (h3a.center);

\draw (w2) -- (b1.center);
\draw (b1.center) -- node[above=0.2cm] {$\atermset$} (b2.center);
\draw (b2.center) -- (w2);

\node (r1) [below right = -1pt and -2pt of a1] {$\scriptstyle{\atermset \blacktriangleright \atermset_1}$};
\node (r2) [below right = -1pt and -4pt of h1] {$\scriptstyle{\atermset \blacktriangleright \atermset_2}$};
\node (r2) [below right = 0pt and -3pt of h2] {$\scriptstyle{\dots}$};
\node (r2) [below right = -5pt and -4pt of h3] {$\scriptstyle{\atermset \blacktriangleright \atermset_{\card{\atypeset{m-1,k}{\apropset}}}}$};

\draw [decorate,decoration={brace,amplitude=10pt}]
(aaa.center) -- node[below = 0.3cm]
{\footnotesize{$\{ \arelation_1(\aworld)|_{\atermset \blacktriangleright \atermset'} \mid \atermset' \in \atypeset{m-1,k}{\apropset} \}$}} (a1a.center);

\node (h4) [below left = -5pt and 0cm of a2] {};
\node (h4b) [below = 2cm of h4] {};
\node (h4c) [left = 0.1cm of h4b] {\footnotesize{$\arelation(\aworld)|_{\atermset} \setminus \arelation_1(\aworld)$}};

\draw[->] (h4b.center) -- (h4);
\draw (h4b.center) -- (h4c);

\draw [decorate,decoration={brace,amplitude=10pt}]
(b2.center) -- node[below = 0.3cm, align=left]
{\footnotesize{\shortstack{\ref{SDelim:cardequiv}: if $\card{\arelation(\aworld)|_\atermset} < k \times (\card{\atowertypeset{m-1,k}{\apropset}}+1)$ then\\ there are $\card{\arelation(\aworld)|_\atermset}$ children,
otherwise there\\ are at least $k \times (\card{\atowertypeset{m-1,k}{\apropset}}+1)$ children.}}} (b1.center);
\end{tikzpicture}
\end{figure}

To work towards the definition of $\arelation_1'$ (as in the statement of the lemma), we now deal with the children in $\arelation'(\aworld')|_{\atermset}$ and find suitable subsets of $\arelation_1'$ in order to define a partition of $\arelation'(\aworld')|_{\atermset}$ that is similar to
$\arelation(\aworld)_{\blacktriangleright\arelation_1}^\atermset$ (where ``similar'' here means that, later,
we will be able to construct a g-bisimulation using this partition). More precisely, we show that:
\vspace{5pt}
\fbox{
\begin{minipage}{0.96\linewidth}
\noindent\rulelab{($\star\star\star$)}{SDelim:construction} it is possible to \emph{construct} a family of sets
\begin{nscenter}
$\begin{aligned}[t] &\arelation'(\aworld')|_{\atermset \leadsto \atermset'} \qquad\qquad \text{for every } \atermset' \in \atypeset{m-1,k}{\apropset}\\
&\agarbage_\atermset
\end{aligned}$
\end{nscenter}
satisfying the following properties.
\begin{enumerate}
\item For every $\atermset' \in \atypeset{m-1,k}{\apropset}$,
$\arelation'(\aworld')|_{\atermset \leadsto \atermset'}$ is a set of pairs $\pair{\arelation'_{1,\aworld_1'}}{\aworld_1'}$ s.t.\
$\aworld_1' \in \arelation'(\aworld')|_{\atermset}$,
$\arelation'_{1,\aworld_1'} {\subseteq} \arelation' $,
$\pair{\triple{\worlds'}{\arelation'_{1,\aworld_1'}}{\avaluation'}}{\aworld_1'} \in \atermset'$,
and for all
$\pair{\aworld_2'}{\aworld_3'} \in \arelation'_{1,\aworld_1'}$, $\{\aworld_2',\aworld_3'\} \subseteq {\arelation'}^*(\aworld_1')$.
\item $\agarbage_\atermset \subseteq \arelation'(\aworld')|_\atermset$.
\item Every $\aworld_1' \in \arelation'(\aworld')|_\atermset$ appears in exactly one set among $\arelation'(\aworld')|_{\atermset \leadsto \atermset'}$ (for every $\atermset' \in \atypeset{m-1,k}{\apropset}$) and
$\agarbage_\atermset$. Then, these sets underlie a partition of $\arelation'(\aworld')|_\atermset$.
\item For every $\atermset' \in \atypeset{m-1,k}{\apropset}$, $\min(\card{\arelation_1(\aworld)|_{\atermset \blacktriangleright \atermset'}}, k) = \min(\card{\arelation'(\aworld')|_{\atermset \leadsto \atermset'}},k)$.
\item $\min(\card{\arelation(\aworld)|_{\atermset} \setminus \arelation_1(\aworld)},k) = \min(\card{\agarbage_\atermset},k)$.
\end{enumerate}
\end{minipage}}
\vspace{5pt}

Let us informally explain these properties (apart from the second and third properties, which are self-explanatory).
The first property basically requires us to modify $\arelation'$ so that the children of $\arelation'(\aworld')|_{\atermset}$ ``jumps'' to specific sets in $\atypeset{m-1,k}{\apropset}$, in line with the developments that lead to the proof of \ref{SDelim:updatetype}.
Instead, the set $\agarbage_\atermset$  is dedicated to those worlds that should be made unaccessible from $\aworld'$.
The updates to $\arelation'$ cannot be arbitrary, and this is where the fourth and fifth properties come into play.
These properties impose cardinality constraints on the sets we construct, in line with the graded rank $k$ that is used in the
equivalence relation $\bisrel{m,k}{\apropset}$.
For example, suppose that for a given set $\atermset'$ we have $\card{\arelation_1(\aworld)|_{\atermset \blacktriangleright \atermset'}} < k$.
Then, we need to select exactly $\card{\arelation_1(\aworld)|_{\atermset \blacktriangleright \atermset'}}$ children in $\arelation'(\aworld')|_\atermset$ and modify $\arelation'$ so that all of them can be used to define the set $\arelation'(\aworld')|_{\atermset \leadsto \atermset'}$.
If instead $\card{\arelation_1(\aworld)|_{\atermset \blacktriangleright \atermset'}} \geq k$, it is possible to select an arbitrary amount of children from $\arelation'(\aworld')|_\atermset$, as long as they are at least $k$.
Again, after selecting these children we need to modify $\arelation'$ so that they define the set $\arelation'(\aworld')|_{\atermset \leadsto \atermset'}$.
To comply with these two last properties we rely on \ref{SDelim:cardequiv}.
The proof of \ref{SDelim:construction} distinguishes two cases (which are very similar in substance):
\begin{description}
\item[\labelitemi\ $\card{\arelation(\aworld)|_{\atermset}} < k \times (\card{\atowertypeset{m-1,k}{\apropset}}+1)$.] By~\ref{SDelim:cardequiv}
it follows that $\card{\arelation'(\aworld')|_{\atermset}} = \card{\arelation(\aworld)|_{\atermset}}$. This case is the easiest one.
Consider a bijection $\amap: \arelation(\aworld)|_{\atermset} \to \arelation'(\aworld')|_{\atermset}$.
Then define $\agarbage_\atermset$ as the set $\{ \amap(\aworld_1) \mid \aworld_1 \in \arelation(\aworld)|_{\atermset} \setminus \arelation_1(\aworld)\}$.
By doing this, trivially the second and fifth properties required by \ref{SDelim:construction} are satisfied.
In order to define the sets of the form  $\arelation'(\aworld')|_{\atermset \leadsto \atermset'}$, we start by an initialisation
to the empty set $\emptyset$ and then we populate them.
Iteratively, for every $\atermset' \in \atypeset{m-1,k}{\apropset}$ and every $\aworld_1 \in \arelation_1(\aworld)|_{\atermset \blacktriangleright \atermset'}$, consider $\amap(\aworld_1)$. By~\ref{SDelim:updatetype}, there is
$\arelation_{1,\amap(\aworld_1)}' \subseteq \arelation'|_{\amap(\aworld_1)}$ such that $\triple{\worlds}{\arelation_1|_{\aworld_1}}{\avaluation},\aworld_1 \bisrel{m-1,k}{\apropset} \triple{\worlds'}{\arelation_{1,\amap(\aworld_1)}'}{\avaluation'},\amap(\aworld_1)$.
By Lemma~\ref{lemma:gbisimcutrel}, it follows that
$\triple{\worlds}{\arelation_1}{\avaluation},\aworld_1 \bisrel{m-1,k}{\apropset} \triple{\worlds'}{\arelation_{1,\amap(\aworld_1)}'}{\avaluation'},\amap(\aworld_1)$
and therefore $\pair{\triple{\worlds'}{\arelation_{1,\amap(\aworld_1)}'}{\avaluation'}}{\amap(\aworld_1)} \in \atermset'$.
Then, add to $\arelation'(\aworld')|_{\atermset \leadsto \atermset'}$ the pair $\pair{\arelation_{1,\amap(\aworld_1)}'}{\amap(\aworld_1)}$. Notice that this pair satisfies the constraints required in the first property of~\ref{SDelim:construction}.
After the iterations over all $\atermset' \in \atypeset{m-1,k}{\apropset}$ and over all
$\aworld_1 \in \arelation_1(\aworld)|_{\atermset \blacktriangleright \atermset'}$, the construction is completed.
As we are guided by the bijection $\amap$, we obtain that every $\aworld_1' \in \arelation'(\aworld')|_\atermset$ appears
in exactly one set among $\arelation'(\aworld')|_{\atermset \leadsto \atermset'}$  for some $\atermset' \in \atypeset{m-1,k}{\apropset}$ or in
$\agarbage_\atermset$ (condition 3 of~\ref{SDelim:construction}).
Moreover (again thanks to the bijection $\amap$) it holds that for every  $\atermset' \in \atypeset{m-1,k}{\apropset}$,
$\card{\arelation'(\aworld')|_{\atermset \leadsto \atermset'}} = \card{\arelation_1(\aworld)|_{\atermset \blacktriangleright \atermset'}}$,
which implies condition 4 of~\ref{SDelim:construction}. Hence,~\ref{SDelim:construction} is proved.
\item[\labelitemi\ $\card{\arelation(\aworld)|_{\atermset}} \geq k \times (\card{\atowertypeset{m-1,k}{\apropset}}+1)$.]
By \ref{SDelim:cardequiv}, it follows that $\card{\arelation'(\aworld')|_{\atermset}} \geq  k \times (\card{\atowertypeset{m-1,k}{\apropset}}+1)$ too.
For this case, it is easy to show that there is a set in the partition
$\arelation(\aworld)_{\blacktriangleright\arelation_1}^\atermset$ of $\arelation(\aworld)|_\atermset$ that has cardinality at least $k$.
Indeed, ad absurdum, suppose  all the sets in $\arelation(\aworld)_{\blacktriangleright\arelation_1}^\atermset$ are of cardinality less than $k$.
As $\arelation(\aworld)_{\blacktriangleright\arelation_1}^\atermset$ partitions $\arelation(\aworld)|_\atermset$ and it contains $\card{\atypeset{m-1,}{\apropset}}+1$ sets
(where the $+1$ refers to the set $\arelation(\aworld)|_{\atermset} \setminus \arelation_1(\aworld)$) this would imply that
$\card{\arelation(\aworld)|_\atermset} \leq (k-1) \times (\card{\atypeset{m-1,k}{\apropset}}+1)$.
This leads to a contradiction as by definition $\card{\atypeset{m-1,k}{\apropset}} \leq \card{\atowertypeset{m-1,k}{\apropset}}$ and we are in the case where $\card{\arelation(\aworld)|_{\atermset}} \geq k \times (\card{\atowertypeset{m-1,k}{\apropset}}+1)$.
 Hence, let $\Omega$ be a set in
$\arelation(\aworld)_{\blacktriangleright\arelation_1}^\atermset$ that has at least $k$ elements.

 For the construction, we initialise all the sets $\arelation'(\aworld')|_{\atermset \leadsto \atermset'}$ and $\agarbage_\atermset$
to the empty set $\emptyset$ and we show how to populate them. Moreover, we introduce an auxiliary set $\Delta$ which is initially
equal to $\arelation'(\aworld')|_\atermset$ and keeps track of which elements of this latter set have not been already used in the
construction (and are hence available).
The set $\Delta$  can be understood as a copy of $\arelation'(\aworld')|_\atermset$ with unmarked elements and marked elements.
Unmarked elements are the worlds yet to be handled by the algorithm.
 Iteratively,
 \begin{enumerate}
 \item consider some $\atermset' \in \atypeset{m-1,k}{\apropset}$ such that $\arelation_1(\aworld)|_{\atermset \blacktriangleright \atermset'} \neq \Omega$ and that was not already treated;
 \item select $\beta = \min(\card{\arelation_1(\aworld)|_{\atermset \blacktriangleright \atermset'}},k)$ worlds, say $\aworld_1',\dots,\aworld_\beta'$ from the pool of available worlds $\Delta$.
 \item As in the previous case of the proof, by~\ref{SDelim:updatetype} we have that for each $i \in \interval{1}{\beta}$ there is
 $\arelation_{1,\aworld_i'}' \subseteq \arelation'|_{\aworld_i'}$ such that for every $\aworld_1 \in \arelation_1(\aworld)|_{\atermset \blacktriangleright \atermset'}$ it holds that
 \begin{nscenter}
 $\triple{\worlds}{\arelation_1|_{\aworld_1}}{\avaluation},\aworld_1 \bisrel{m-1,k}{\apropset} \triple{\worlds'}{\arelation_{1,\aworld_i'}'}{\avaluation'},\aworld_i'$.
 \end{nscenter}
 By Lemma~\ref{lemma:gbisimcutrel}, it follows also that
 $\triple{\worlds}{\arelation_1}{\avaluation},\aworld_1 \bisrel{m-1,k}{\apropset} \triple{\worlds'}{\arelation_{1,\aworld_i'}'}{\avaluation'},\aworld_i'$
 and therefore $\pair{\triple{\worlds'}{\arelation_{1,\aworld_i'}'}{\avaluation'}}{\aworld_i'} \in \atermset'$.
 Then, define the set $\arelation'(\aworld')|_{\atermset \leadsto \atermset'}$ as
 \begin{nscenter}
  $\{ \pair{\arelation_{1,\aworld_i'}'}{\aworld_i'} \mid i \in \interval{1}{\beta} \}$.
 \end{nscenter}
  Notice that by construction this set satisfies the first and fourth properties of~\ref{SDelim:construction}.
 \item Remove $\aworld_1',\dots,\aworld_\beta'$ from $\Delta$ (they will not be used in the successive iterations).
 \end{enumerate}
 After this iterative construction, only two sets still need to be handled:
$\Omega$ and $\arelation(\aworld)|_{\atermset} \setminus \arelation_1(\aworld)$.
In the case these two sets are different, we proceed as follows.
 \begin{enumerate}
    \item We start by considering $\arelation(\aworld)|_{\atermset} \setminus \arelation_1(\aworld)$, and we select
    $\beta = \min(\card{\arelation(\aworld)|_{\atermset} \setminus \arelation_1(\aworld)},k)$ worlds,
    say $\aworld_1',\dots,\aworld_\beta'$ from the pool of available worlds $\Delta$.
    \item We define $\agarbage_\atermset$ as $\{\aworld_1',\dots,\aworld_\beta'\}$ and remove these worlds from $\Delta$. 
By construction, $\agarbage_\atermset$ satisfies the second and fifth properties of~\ref{SDelim:construction}.
    \item We  consider $\Omega$. A few things should be noted now.
      \begin{itemize}
        \item There is $\atermset' \in \atypeset{m-1,k}{\apropset}$ such that
              $\Omega = \arelation_1(\aworld)|_{\atermset \blacktriangleright \atermset'}$, and by definition of $\Omega$,
              we have $\card{\arelation_1(\aworld)|_{\atermset \blacktriangleright \atermset'}} \geq k$.
        \item At this point of the construction, we dealt with $\card{\atypeset{m-1,k}{\apropset}}$ of the $\card{\atypeset{m-1,k}{\apropset}}+1$ sets
              needed for the construction. For each of these sets we used at most $k$ new worlds of $\arelation'(\aworld')|_\atermset$. Hence, as
              $\card{\arelation'(\aworld')|_\atermset} \geq k \times (\card{\atowertypeset{m-1,k}{\apropset}}+1)$
        and $\card{\atowertypeset{m-1,k}{\apropset}} \geq \card{\atypeset{m-1,k}{\apropset}}$, we conclude that $\Delta$ has at least $k$ elements.
      \end{itemize}
    \item Consider the set $\Delta$. By~\ref{SDelim:updatetype} we have that for each $\aworld_1' \in \Delta$ there is
    $\arelation_{1,\aworld_1'}' \subseteq \arelation'|_{\aworld_1'}$ such that for every $\aworld_1 \in \arelation_1(\aworld)|_{\atermset \blacktriangleright \atermset'}$ it holds that
    \begin{nscenter}
    $\triple{\worlds}{\arelation_1|_{\aworld_1}}{\avaluation},\aworld_1 \bisrel{m-1,k}{\apropset} \triple{\worlds'}{\arelation_{1,\aworld_1'}'}{\avaluation'},\aworld_1'$.
    \end{nscenter}
    By Lemma~\ref{lemma:gbisimcutrel}, it follows that
    $\triple{\worlds}{\arelation_1}{\avaluation},\aworld_1 \bisrel{m-1,k}{\apropset} \triple{\worlds'}{\arelation_{1,\aworld_1'}'}{\avaluation'},\aworld_1'$
    and therefore $\pair{\triple{\worlds'}{\arelation_{1,\aworld_1'}'}{\avaluation'}}{\aworld_1'} \in \atermset'$.
    Then, define the set $\arelation'(\aworld')|_{\atermset \leadsto \atermset'}$ as
    \begin{nscenter}
     $\{ \pair{\arelation_{1,\aworld_1'}'}{\aworld_1'} \mid \aworld_1' \in \Delta \}$.
    \end{nscenter}
     By construction, this set satisfies the first and fourth properties of~\ref{SDelim:construction} (recall that both $\arelation'(\aworld')|_{\atermset \leadsto \atermset'}$ and $\arelation_1(\aworld)|_{\atermset \blacktriangleright \atermset'}$  have at least $k$ elements, see the previous point).
     \item Empty $\Delta$ as every remaining world in it is now used.
     We completed the construction in the case of $\Omega \neq \arelation(\aworld)|_{\atermset} \setminus \arelation_1(\aworld)$.
 \end{enumerate}
 In the case $\Omega = \arelation(\aworld)|_{\atermset} \setminus \arelation_1(\aworld)$, the construction is trivially completed by adding to $\agarbage_\atermset$ every world in $\Delta$. Notice that for the same considerations done before (point 3 of the construction for $\Omega \neq \arelation(\aworld)|_{\atermset} \setminus \arelation_1(\aworld)$) it holds that $\Delta$ has at least $k$ elements.
 Hence, $\agarbage_\atermset$ satisfies both the second and the fifth properties of \ref{SDelim:construction}. Again, as a last step, we empty $\Delta$ as every remaining world is now used.

 During the definition of the construction, we already detailed why the first, second, fourth and fifth properties of~\ref{SDelim:construction} are satisfied.
 The same holds true for the third one, as we relied on the set $\Delta$ to never use twice the same world, and at the end of the construction $\Delta$ was always empty.
\end{description}
Therefore~\ref{SDelim:construction} holds.
A last note about this construction: from the first and third properties of \ref{SDelim:construction}, in particular that ``for all
$\pair{\aworld_2'}{\aworld_3'} \in \arelation'_{1,\aworld_1'},\ \{\aworld_2',\aworld_3'\} \subseteq {\arelation'}^*(\aworld_1')  \}$'', it is easy to see that
for all
$\pair{\arelation'_{1,\aworld_1'}}{\aworld_1'} \in \arelation'(\aworld')|_{\atermset \leadsto \atermset_1}$ and
$\pair{\arelation'_{1,\aworld_2'}}{\aworld_2'} \in \arelation'(\aworld')|_{\atermset \leadsto \atermset_2}$ with $\aworld_1' \neq \aworld_2'$,
we have $\arelation'_{1,\aworld_1'} \cap \arelation'_{1,\aworld_2'} = \emptyset$.
Keeping this in mind, we are now ready to construct $\arelation_1'$.

We consider every $\atermset \in \atowertypeset{m-1,k}{\apropset}$ and apply \ref{SDelim:construction} to construct the sets  $\arelation'(\aworld')|_{\atermset \leadsto \atermset'}$ (for every $\atermset' \in \atypeset{m-1,k}{\apropset}$) and $\agarbage_\atermset$.
We then define $\arelation_1'$ as
\begin{nscenter}
$\arelation_1' \egdef \displaystyle\bigcup_{\substack{\atermset \in \atowertypeset{m-1,k}{\apropset}\\ \atermset' \in \atypeset{m-1,k}{\apropset}\\ \pair{\arelation_{1,\aworld_1'}'}{\aworld_1'} \in \arelation'(\aworld')|_{\atermset \leadsto \atermset'}}}
\{\pair{\aworld'}{\aworld_1'}\} \cup \arelation_{1,\aworld_1'}'$.
\end{nscenter}
\end{description}
Clearly. we have that $\arelation_1' \subseteq \arelation_1$. Moreover, from the 
properties of~\ref{SDelim:construction}, it holds that
 for every $\aworld_1' \in \arelation_1'(\aworld)$, $\arelation_1'|_{\aworld_1'} = \arelation_{1,\aworld_1'}'$.
In order to conclude the proof, we need to show that
\begin{enumerate}
\item $\triple{\worlds}{\arelation_1}{\avaluation},\aworld \bisrel{m,k}{\apropset} \triple{\worlds'}{\arelation_1'}{\avaluation'},\aworld'$;
\item if $\arelation_1(\aworld) = \arelation(\aworld)$ then $\arelation_1'(\aworld') = \arelation'(\aworld')$.
\end{enumerate}
Let us first prove (2) by using the fifth property of \ref{SDelim:construction}.
Suppose $\arelation_1(\aworld) = \arelation(\aworld)$ and hence $\arelation(\aworld) \setminus \arelation_1(\aworld) = \emptyset$.
It is easy to see that
$\arelation(\aworld) \setminus \arelation_1(\aworld)$ can also be written as $\bigcup_{\atermset \in \atowertypeset{m-1,k}{\apropset}} (\arelation(\aworld)|_{\atermset} \setminus \arelation_1(\aworld))$.
We conclude that $\card{\arelation(\aworld)|_{\atermset} \setminus \arelation_1(\aworld)} = 0$ for every $\atermset \in \atowertypeset{m-1,k}{\apropset}$.
Similarly, $\arelation'(\aworld') \setminus \arelation_1'(\aworld')$ can be shown to be equivalent to
$\bigcup_{\atermset \in \atowertypeset{m-1,k}{\apropset}} (\arelation'(\aworld')|_{\atermset} \setminus \arelation_1'(\aworld'))$.
Notice that for every $\atermset \in \atowertypeset{m-1,k}{\apropset}$, a world $\aworld_1' \in \arelation'(\aworld')|_{\atermset} \setminus \arelation_1'(\aworld')$ cannot be inside a pair of $\arelation'(\aworld')|_{\atermset \leadsto \atermset'}$ (for any $\atermset' \in \atypeset{m-1,k}{\apropset}$).
Indeed, if this was the case, then $\pair{\aworld'}{\aworld_1'} \in \arelation_1'$ (see definition of $\arelation_1'$) in contradiction with $\aworld_1' \in \arelation'(\aworld')|_{\atermset} \setminus \arelation_1'(\aworld')$.
Then $\aworld_1' \in \agarbage_{\atermset}$ and we conclude that $\arelation'(\aworld')|_{\atermset} \setminus \arelation_1'(\aworld') = \agarbage_{\atermset}$ and $\arelation'(\aworld') \setminus \arelation_1'(\aworld') = \bigcup_{\atermset \in \atowertypeset{m-1,k}{\apropset}} \agarbage_{\atermset}$.
By construction, every world $\aworld_1' \in \arelation'(\aworld)$ can appear in at most one set in $\{\agarbage_{\atermset} \mid \atermset' \in \atypeset{m-1,k}{\apropset} \}$ and hence
$\card{\arelation'(\aworld') \setminus \arelation_1'(\aworld')} = \sum_{\atermset \in \atowertypeset{m-1,k}{\apropset}} \card{\agarbage_\atermset}$.
We can now apply the fifth  property of \ref{SDelim:construction}, i.e.
\begin{nscenter}
$\min(\card{\arelation(\aworld)|_{\atermset} \setminus \arelation_1(\aworld)},k) = \min(\card{\agarbage_\atermset},k)$
\end{nscenter}
that, together with $k \geq 1$ (see the beginning of the proof) and $\card{\arelation(\aworld)|_{\atermset} \setminus \arelation_1(\aworld)} = 0$ leads to $\card{\arelation'(\aworld') \setminus \arelation_1'(\aworld')} = 0$. As by definition $\arelation_1'(\aworld') \subseteq \arelation'(\aworld')$, this ends the proof of (2).

In order to conclude the proof, let us prove (1) and this is done by constructing
a g-bisimulation $\abisim^0,\dots,\abisim^m$
up to $\triple{m}{k}{\apropset}$ between $\triple{\worlds}{\arelation_1}{\avaluation}$ and $\triple{\worlds'}{\arelation_1'}{\avaluation'}$ such that $\{\aworld\} \abisim^m_1 \{\aworld'\}$.
Here, we iteratively construct the g-bisimulation starting from the sets $\abisim^j_i = \{\pair{\aworld}{\aworld'}\}$ (for every $i \in \interval{1}{k}$ and $j \in \interval{0}{m}$).
During the construction we make sure to always preserve the satisfaction of the conditions (\ref{gbisim:init}), (\ref{gbisim:refine}), (\ref{gbisim:size}) and (\ref{gbisim:atoms}).
Notice that these conditions hold for our initial sequence of relations. In particular, (\ref{gbisim:atoms}) holds as by hypothesis there is
$\atermset \in \atypeset{m,k\times (\card{\atowertypeset{m-1,k}{\apropset}}+1)}{\apropset}$
such that $\{\pair{\amodel}{\aworld}, \pair{\amodel'}{\aworld'}\} \subseteq \atermset$ and hence
$\amodel,\aworld \bisrel{m,k \times (\card{\atowertypeset{m-1,k}{\apropset}}+1)}{\apropset} \amodel',\aworld'$.
The construction can be split into four steps:
\begin{description}
\item[\rulelab{m-forth-step}{construction:m-forth}]
Let $X \subseteq \arelation_1(\aworld)$ be a set such that $\card{X} \in \interval{1}{k}$.
As required by the condition (\ref{gbisim:m-forth}), we want to pair this set with a suitable subset $Y \subseteq \arelation_1'(\aworld)$ of cardinality $\card{X}$ so that it is possible to then satisfy the conditions (\ref{gbisim:g-forth}) and (\ref{gbisim:g-back}).
Let us consider the partition of $X$ defined as $\{ X_{\atermset \blacktriangleright \atermset'} \mid \atermset \in \atowertypeset{m-1,k}{\apropset} \text{ and } \atermset' \in \atypeset{m-1,k}{\apropset} \}$ where $X_{\atermset \blacktriangleright \atermset'} = X \cap \arelation_1(\aworld)|_{\atermset \blacktriangleright \atermset'}$.
We consider the set $\arelation'(\aworld')|_{\atermset \leadsto \atermset'}$ and select $\card{X_{\atermset \blacktriangleright \atermset'}}$ worlds appearing in one of its pairs (which  are of the form $\pair{\arelation_{1,\aworld_1'}'}{\aworld_1'}$).
Let $Y_{\atermset \leadsto \atermset'}$ be the set of these selected worlds.
By \ref{SDelim:construction} this set is guaranteed to exist and is such that every world $\aworld_1'$ in it is also in
$\arelation_1'(\aworld')$.
Let $Y = \bigcup_{\atermset \in \atowertypeset{m-1,k}{\apropset}, \atermset' \in \atypeset{m-1,k}{\apropset} } Y_{\atermset \leadsto \atermset'}$. It is easy to see that $\card{X} = \card{Y}$.
For every $j \in \interval{0}{m-1}$ we add $\pair{X}{Y}$ to $\abisim^j_{\card{X}}$.
\item[\rulelab{m-back-step}{construction:m-back}]
Let $Y \subseteq \arelation_1'(\aworld)$ be a set such that $\card{Y} \in \interval{1}{k}$.
Let us follow the condition (\ref{gbisim:m-back}) symmetrically to what was done for the condition (\ref{gbisim:m-forth}) in the previous step of the construction.
Let us first consider the partition of $Y$ defined as $\{ Y_{\atermset \leadsto \atermset'} \mid \atermset \in \atowertypeset{m-1,k}{\apropset} \text{ and } \atermset' \in \atypeset{m-1,k}{\apropset} \}$ where
$Y_{\atermset \leadsto \atermset'} = Y \cap \{ \aworld_1' \mid \pair{\arelation_{1,\aworld_1'}'}{\aworld_1'} \in \arelation'(\aworld')|_{\atermset \leadsto \atermset'} \text{ for some } \arelation_{1,\aworld_1'}' \}$.
We select a subset $X_{\atermset \blacktriangleright \atermset'}$ of $\arelation_1(\aworld)|_{\atermset \blacktriangleright \atermset'}$ having cardinality $\card{Y_{\atermset \leadsto \atermset'}}$, which is guaranteed to exist by~\ref{SDelim:construction}.
Let $X =  \bigcup_{\atermset \in \atowertypeset{m-1,k}{\apropset}, \atermset' \in \atypeset{m-1,k}{\apropset} } X_{\atermset \blacktriangleright \atermset'}$. It is easy to see that $\card{Y} = \card{X}$.
For every $j \in \interval{0}{m-1}$ we add $\pair{X}{Y}$ to $\abisim^j_{\card{Y}}$.
\item[\rulelab{g-forth-step}{construction:g-forth}]
From the first two steps of the construction, the set $\abisim^j_i$ was updated with new pairs $\pair{X}{Y}$ where every element in $X$ is from $\arelation_1(\aworld)$ and every element of $Y$ is from $\arelation_1'(\aworld)$.
Consider then one of these pairs $\pair{X}{Y}$ and
let $\aworld_1 \in X$.
There is $\atermset \in \atowertypeset{m-1,k}{\apropset}$ and $\atermset' \in \atypeset{m-1,k}{\apropset}$ such that
$\aworld_1 \in \arelation_1(\aworld)|_{\atermset \blacktriangleright \atermset'}$.
By construction (first and second steps above),  there is $\aworld_1' \in Y$ such that for some $\arelation_{1,\aworld_1'}' \subseteq \arelation_1'$ it holds that $\pair{\arelation_{1,\aworld_1'}'}{\aworld_1'} \in \arelation'(\aworld')|_{\atermset \leadsto \atermset'}$.
Again, by applying \ref{SDelim:construction} we obtain that $\triple{\worlds}{\arelation_1}{\avaluation},\aworld_1 \bisrel{m-1,k}{\apropset} \triple{\worlds'}{\arelation_{1,\aworld_1'}}{\avaluation'}, \aworld_1'$.
Since by definition $\arelation_{1,\aworld_1'}' = \arelation_1'|_{\aworld_1'}$ and from Lemma~\ref{lemma:gbisimcutrel} we obtain $\triple{\worlds}{\arelation_1}{\avaluation},\aworld_1 \bisrel{m-1,k}{\apropset} \triple{\worlds'}{\arelation_1'}{\avaluation'}, \aworld_1'$.
Then, let $\abisimtwo^0,\dots,\abisimtwo^{m-1}$ be the g-bisimulation up to $\triple{m-1}{k}{\apropset}$ between $\triple{\worlds}{\arelation_1}{\avaluation}$ and $\triple{\worlds'}{\arelation_1'}{\avaluation'}$ such that $\{\aworld_1\} \abisimtwo_1^{m-1} \{\aworld_1'\}$.
For every $i \in \interval{1}{k}$ and every $j \in \interval{0}{m-1}$, update $\abisim^j_i$ to $\abisim^j_i \cup \abisimtwo^j_i$.
\item[\rulelab{g-back-step}{construction:g-back}]
Symmetrically to the previous point of the construction, let us consider again a pair $\pair{X}{Y}$ introduced by one of the two steps (\ref{construction:m-forth}) and (\ref{construction:m-back}). Let $\aworld_1' \in Y$.
Then there is  $\atermset \in \atowertypeset{m-1,k}{\apropset}$ and $\atermset' \in \atypeset{m-1,k}{\apropset}$ and $\arelation_{1,\aworld_1'}' \subseteq \arelation_1'$ such that $\pair{\arelation_{1,\aworld_1'}'}{\aworld_1'} \in \arelation'(\aworld')|_{\atermset \leadsto \atermset'}$.
By construction (steps (\ref{construction:m-forth}) and (\ref{construction:m-back})), there is $\aworld_1 \in X$ such that $\aworld_1 \in \arelation'(\aworld)|_{\atermset \blacktriangleright \atermset'}$.
Then by \ref{SDelim:construction}, we obtain that $\triple{\worlds}{\arelation_1}{\avaluation},\aworld_1 \bisrel{m-1,k}{\apropset} \triple{\worlds'}{\arelation_{1,\aworld_1'}}{\avaluation'}, \aworld_1'$.
Again, by definition $\arelation_{1,\aworld_1'}' = \arelation_1'|_{\aworld_1'}$ and from Lemma~\ref{lemma:gbisimcutrel} we obtain $\triple{\worlds}{\arelation_1}{\avaluation},\aworld_1 \bisrel{m-1,k}{\apropset} \triple{\worlds'}{\arelation_1'}{\avaluation'}, \aworld_1'$.
Then, let $\abisimtwo^0,\dots,\abisimtwo^{m-1}$ be the g-bisimulation up to $\triple{m-1}{k}{\apropset}$ between $\triple{\worlds}{\arelation_1}{\avaluation}$ and $\triple{\worlds'}{\arelation_1'}{\avaluation'}$ such that $\{\aworld_1\} \abisimtwo_1^{m-1} \{\aworld_1'\}$.
For every $i \in \interval{1}{k}$ and every $j \in \interval{0}{m-1}$, update $\abisim^j_i$ to $\abisim^j_i \cup \abisimtwo^j_i$.
\end{description}
It is simple to see that this construction leads to a sequence of relations $\abisim^0,\dots,\abisim^m$ that is a g-bisimulation up to $\triple{m}{k}{\apropset}$ between $\triple{\worlds}{\arelation_1}{\avaluation}$ and $\triple{\worlds'}{\arelation_1'}{\avaluation'}$ such that $\{\aworld\}\abisim_1^m \{\aworld'\}$.
Indeed, the conditions (\ref{gbisim:init}), (\ref{gbisim:refine}), (\ref{gbisim:size}) and (\ref{gbisim:atoms}) hold at any point during the construction. For the other condition, let $\pair{X}{Y}$ be a pair in some $\abisim^j_i$.
If it was not introduced by the first two steps of the construction, then $\pair{X}{Y}$ is a member of some set $\abisimtwo^j_i \subseteq \abisim^j_i$ that is used in a g-bisimulation whose elements are all used to construct $\abisim^0,\dots,\abisim^m$ (third and fourth
point of the proof). Hence, w.r.t. $\pair{X}{Y}$ no condition can be violated. If instead $\pair{X}{Y}$ is added to the g-bisimulation during the first and second point of the construction, then by construction it is easy to check that it satisfies all the conditions.
Therefore $\triple{\worlds}{\arelation_1}{\avaluation},\aworld \bisrel{m,k}{\apropset} \triple{\worlds'}{\arelation_1'}{\avaluation'},\aworld'$, which ends the proof of the whole lemma.
\end{proof}

\subsection{Proof of \Cref{lemma:sabotage-equivalent-formula}}
 \begin{proof}
If $k = 0$, then the proof is by an easy verification as the formula $\aformula$ from the statement is logically equivalent
to a formula from the propositional calculus (each subformula $\Gdiamond{\geq 0} \aformulabis$ is logically equivalent to
$\top$).
Otherwise ($k \geq 1$), let $k^+ = k \times (\card{\atowertypeset{m-1,k}{\apropset}}+1)$. As,
$\typeeqclass{m,k^+}{\apropset}$ and $\bisrel{m,k^+}{\apropset}$ are identical relations, there is a finite
set $\set{\aformulater_1, \ldots, \aformulater_{Q}} \subseteq \GML[m,k^+,\apropset]$ such that
\begin{itemize}
\item $\aformulater_1 \vee \cdots \vee \aformulater_{Q}$ is valid, and each $\aformulater_i$ is satisfiable,
\item for all $i \neq j \in \interval{1}{Q}$, $\aformulater_i \wedge \aformulater_j$ is unsatisfiable,
\item $\pair{\amodel}{\aworld} \typeeqclass{m,k^+}{\apropset} \pair{\amodel'}{\aworld'}$ iff there is $i$ such that
      $\pair{\amodel}{\aworld} \models \aformulater_i$ and $\pair{\amodel'}{\aworld'} \models \aformulater_i$.
\end{itemize}
This is a direct consequence of~\Cref{prop:rijke00} containing results established in~\cite{DeRijke00}.
Let $\aformulabis$ be the formula
$\bigvee \set{\aformulater_i \mid \exists \ \amodel, \aworld \ {\rm s.t.} \ \amodel,\aworld \models \aformulater_i \wedge  \SabDiamond \aformula}$.
An empty disjunction is understood as $\perp$.

Now, we show that $\aformulabis$ is logically equivalent to $\SabDiamond \aformula$.
Suppose that $\amodel, \aworld \models \SabDiamond \aformula$. As $\aformulater_1 \vee \cdots \vee \aformulater_{Q}$ is valid,
there is $i \in \interval{1}{Q}$ such that $\amodel, \aworld \models \aformulater_i$. Therefore
$\aformulater_i$ occurs in $\aformulabis$ and consequently, $\amodel, \aworld \models \aformulabis$.

Conversely, suppose that $\amodel,\aworld \models \aformulabis$ with $\amodel = \triple{\worlds}{\arelation}{\avaluation}$.
So, there is $\aformulater_i$ occuring in  $\aformulabis$ such that  $\amodel,\aworld \models \aformulater_i$ and
there exist a model $\amodel' = \triple{\worlds'}{\arelation'}{\avaluation'}$ and $\aworld' \in \worlds'$
such that $\amodel',\aworld' \models \aformulater_i \wedge \SabDiamond \aformula$. So,
$\pair{\amodel}{\aworld} \typeeqclass{m,k^+}{\apropset} \pair{\amodel'}{\aworld'}$.
By the definition of the satisfaction relation $\models$,
there is $\arelation_1' \subseteq \arelation'$ such that $\arelation_1'(\aworld') = \arelation'(\aworld')$ and
$\triple{\worlds'}{\arelation_1'}{\avaluation'}, \aworld' \models \aformula$.
All the assumptions of Lemma~\ref{lemma:sabotage-elimination} apply and therefore,
there is $\arelation_1 \subseteq \arelation$ such that $\arelation_1(\aworld) = \arelation(\aworld)$,
$\triple{\worlds}{\arelation_1}{\avaluation}, \aworld \bisrel{m,k}{\apropset} \triple{\worlds'}{\arelation_1'}{\avaluation'}, \aworld'$
and $\triple{\worlds}{\arelation_1}{\avaluation}, \aworld \typeeqclass{m,k}{\apropset} \triple{\worlds'}{\arelation_1'}{\avaluation'}, \aworld'$. 
As $\aformula$ belongs to $\GML[m,k,\apropset]$, we also get that $\triple{\worlds}{\arelation_1}{\avaluation}, \aworld \models \aformula$.
But then by definition of $\models$, we conclude that $\amodel, \aworld \models \SabDiamond \aformula$.
\end{proof}

\subsection{Proof of \modallogicSC $\preceq$ \GML}
\begin{lemma}\label{lemma:SC-leq-GML}
\modallogicSC $\preceq$ \GML.
\end{lemma}

\begin{proof}
Let $\aformula$ be a formula in \modallogicSC.
As $\Diamond \aformulabis \equiv \Gdiamond{\geq 1} \aformulabis$, we can replace
every occurrence of the modality $\Diamond$  appearing in $\aformula$ with the modality $\Gdiamond{\geq 1}{}$.
Moreover,
by Lemma~\ref{lemma:SC-CC-tautology}, we can replace every subformula of the form $\aformulabis \separate \aformulater$ with the formula $\SabDiamond(\aformulabis \chopop \aformulater)$.
In this way, we obtain a formula $\aformula'$ that is equivalent to $\aformula$ and where all the modalities are of the form $\Gdiamond{\geq 1}{}$, $\chopop$ and $\SabDiamond$.
If $\aformula'$ has no occurrence of $\chopop$ or $\SabDiamond$, we are done.
Otherwise, let $\aformulabis$ be a subformula of $\aformula'$ of the form $\SabDiamond(\aformula_1 \chopop \aformula_2)$ where $\aformula_1$ and $\aformula_2$ are in \GML.
\begin{itemize}
\item By \Cref{theorem:clean-cut}, there is a formula $\aformulabis'$ in \GML such that $\aformulabis' \equiv \aformula_1 \chopop \aformula_2$.
\item By Lemma~\ref{lemma:sabotage-equivalent-formula}
there is a formula $\aformulabis''$ in \GML such that $\aformulabis'' \equiv \SabDiamond \aformulabis'$.
\end{itemize}
One can show that $\aformula' \equiv \aformula'[\aformulabis \leftarrow \aformulabis'']$, where  $\aformula'[\aformulabis \leftarrow \aformulabis'']$
is obtained from $\aformula'$ by replacing every occurrence of $\aformulabis$ by $\aformulabis''$.
Note that the number of occurrences of $\SabDiamond$ and $\chopop$ in
$\aformula'[\aformulabis \leftarrow \aformulabis'']$ is strictly less than the number of occurrences of $\SabDiamond$ and $\chopop$ in
$\aformula'$.
By repeating such a type of replacement,
eventually we obtain a formula $\aformula''$ in \GML such that $\aformula' \equiv \aformula''$.
Indeed, all the occurrences of $\SabDiamond$ and $\chopop$ only appear as instances  of the
pattern $\SabDiamond(\aformulabis \chopop \aformulater)$.
Hence, we get a formula in \GML logically equivalent  to $\aformula$.
\end{proof}

\subsection{Proof that $\aformulaerankset{m,s}{\apropset}$ is finite up to logical equivalence}

\begin{lemma}\label{lemma:SC-finite-up-to-logEq}
 $\aformulaerankset{m,s}{\apropset}$ is finite up to logical equivalence.
\end{lemma}

\begin{proof}
This proof is standard and relies on the analogous result from classical logic \cite{Libkin04}:
\begin{nscenter}
\begin{minipage}{0.8\linewidth}
\begin{itemize}
\item[\rulelab{($\star$)}{property:finite-formulae-pc}]
given a finite set of formulae $\asetformulae$ there are only finitely many Boolean combinations
of formulae from $\asetformulae$, up to logical equivalence.
\end{itemize}
\end{minipage}
\end{nscenter}
The proof of the lemma is by induction on $\pair{m}{s}$.
For the base case, i.e.\ $\pair{0}{0}$, every formula of $\aformulaerankset{0,0}{\apropset}$ is by definition a Boolean combination of formulae from $\apropset$. Then by \ref{property:finite-formulae-pc} this set of formulae is clearly finite up to logical equivalence.
For the induction step, we divide the set of formulae of $\aformulaerankset{m,s}{\apropset}$ in three disjoint sets and
we show that each of them is finite up to logical equivalence.
\begin{enumerate}
\item We consider the set of formulae dominated by the  operator $\Diamond$, i.e.\ the set of every formula $\aformula$ that is syntactically equivalent to $\Diamond \aformulabis$ for some $\aformulabis \in \aformulaerankset{m-1,s}{\apropset}$.
By the induction hypothesis, there are only finitely many such $\aformulabis$ up to logical equivalence.
Hence, the set of formulae dominated by $\Diamond$ is finite up to logical equivalence.
\item We consider the set of formulae dominated by the  operator $\separate$, i.e.\ the set of every formula $\aformula$ that is syntactically equivalent to $\aformulabis \separate \aformulater$ for some $\aformulabis \in \aformulaerankset{m,s'}{\apropset}$ and $\aformulabis \in 
\aformulaerankset{m,s''}{\apropset}$ such that $\max(s',s'') = s{-}1$.
By the induction hypothesis, there are only finitely many such $\aformulabis$ and $\aformulater$ up to logical equivalence.
Hence, the set of formulae dominated by the  operator $\separate$ is finite up to logical equivalence.
\item Lastly, we consider the set of formulae of 
$\aformulaerankset{m,s}{\apropset}$ that are not dominated by $\Diamond$ or $\separate$ operators.
Each formula $\aformula$ of this set is therefore a Boolean combination of formulae $\aformula_1,\dots,\aformula_n$ of 
$\aformulaerankset{m,s}{\apropset}$ that are  dominated by $\Diamond$ or 
$\separate$ operators (hence every of these formulae are different form $\aformula$).
From the previous two cases, the set of such $\aformula_1,\dots,\aformula_n$ formulae is finite up to logical equivalence.
Then, by \ref{property:finite-formulae-pc} we conclude that the set of formulae of $\aformulaerankset{m,s}{\apropset}$ 
that are not dominated by $\Diamond$ or $\separate$ operators is also finite up to logical equivalence, concluding the proof. \qedhere
\end{enumerate}
\end{proof}

\subsection{Characteristic formulae}
As usual, thanks to Lemma~\ref{lemma:SC-finite-up-to-logEq}, given a pointed forest $\pair{\amodel}{\aworld}$,
we can define a finite \defstyle{characteristic formula}
$\acharformulaSC{\amodel,\aworld}{m,s}{\apropset}$ in $\aformulaerankset{m,s}{\apropset}$
that is logically equivalent to the infinite conjunction $\bigwedge \{ \aformula \in \aformulaerankset{m,s}{\apropset} \mid \amodel,\aworld \models \aformula \}$.
Notice that $\acharformulaSC{\amodel,\aworld}{m,s}{\apropset}$ is in $\aformulaerankset{m,s}{\apropset}$.
Moreover, we can prove the following result.

\begin{lemma}\label{lemma:SC-char-formula}
Let $\pair{\amodel}{\aworld}$ and $\pair{\amodel'}{\aworld'}$ be two pointed forests. For every rank $\triple{m}{s}{\apropset}$ it holds that
\begin{itemize}
  \item $\amodel,\aworld \models \acharformulaSC{\amodel,\aworld}{m,s}{\apropset}$;
  \item $\amodel,\aworld \models \acharformulaSC{\amodel',\aworld'}{m,s}{\apropset}$ iff \ $\amodel',\aworld' \models \acharformulaSC{\amodel,\aworld}{m,s}{\apropset}$.
\end{itemize}
\end{lemma}

\begin{proof}
This proof is standard. The first part of the lemma follows directly by definition of the characteristic formulae.
For the second part, by symmetry we just need to show one direction. Assume that $\amodel,\aworld \models
\acharformulaSC{\amodel',\aworld'}{m,s}{\apropset}$.
Let $\aformulabis \in \aformulaerankset{m,s}{\apropset}$ such that $\amodel,\aworld \models \aformulabis$.
To prove the result it is sufficient to show that then $\amodel',\aworld' \models \aformulabis$.
Ad absurdum, suppose that  $\amodel',\aworld' \not\models \aformulabis$.
By definition $\amodel',\aworld' \models \lnot \aformulabis$ and notice that $\lnot \aformulabis \in \aformulaerankset{m,s}{\apropset}$. Therefore from the equivalence
\begin{nscenter}
$\acharformulaSC{\amodel',\aworld'}{m,s}{\apropset} \equiv \bigwedge \{ \aformula \in \aformulaerankset{m,s}{\apropset} \mid \amodel',\aworld' \models \aformula \},$
\end{nscenter}
it is easy to see that $\acharformulaSC{\amodel',\aworld'}{m,s}{\apropset} \implies \lnot \aformulabis$ is a tautology.
From $\amodel,\aworld \models \acharformulaSC{\amodel',\aworld'}{m,s}{\apropset}$ we then derive that $\amodel,\aworld \models \lnot \aformulabis$, in contradiction with the hypothesis $\amodel,\aworld \models \aformulabis$. Hence,  $\amodel',\aworld' \models \aformulabis$.
\end{proof}

\subsection{Proof of \Cref{lemma:EF-game-sound-complete}}
\begin{proof}
We first prove that the games are sound (right to left direction).
\begin{nscenter}
$\boxed{
\text{ If there is } \aformula \in \aformulaerankset{m,s}{\apropset}
\text{ s.t.\ } \amodel,\aworld \models \aformula \text{ and } \amodel',\aworld' \not \models \aformula
\text{ then } \pair{\amodel}{\aworld} \neggamerel{m,s}{\apropset} \pair{\amodel'}{\aworld'}
}$
\end{nscenter}
The proof is rather standard and is done by structural induction on $\aformula$.
\begin{description}
\item[Base case: $\aformula = \avarprop$, where $\avarprop \in \apropset$.] Then by hypothesis $\amodel,\aworld \models \avarprop$ and $\amodel',\aworld' \not \models \avarprop$ and the spoiler wins from the condition of the game imposed before each round.
\item[Induction case: $\aformula = \aformulabis \land \aformulater$.] By hypothesis $\amodel,\aworld \models \aformulabis \land \aformulater$ whereas $\amodel',\aworld' \not \models \aformulabis$ or $\model',\aworld' \not \models \aformulater$.
In both cases ($\amodel',\aworld' \not \models \aformulabis$ or $\model',\aworld' \not \models \aformulater$), by the induction hypothesis the spoiler has a winning strategy for $\triple{\pair{\amodel}{\aworld}}{\pair{\amodel'}{\aworld'}}{\triple{m}{s}{\apropset}}$, i.e.\ $\pair{\amodel}{\aworld} \neggamerel{m,s}{\apropset} \pair{\amodel'}{\aworld'}$.
\item[Induction case: $\aformula = \lnot \aformulabis$.]
By hypothesis $\amodel,\aworld \not\models \aformulabis$ whereas $\amodel',\aworld' \models \aformulabis$.
Then by symmetry and by the induction hypothesis $\pair{\amodel}{\aworld} \neggamerel{m,s}{\apropset} \pair{\amodel'}{\aworld'}$.
\item[Induction case: $\aformula = \Diamond \aformulabis$.]
By hypothesis $\amodel,\aworld \models \Diamond \aformulabis$ and $\amodel',\aworld' \not\models \Diamond \aformulabis$.
Then there is a world $\aworld_1$ accessible from $\aworld$ and such that $\amodel,\aworld_1 \models \aformulabis$.
Moreover by definition the modal depth of $\Diamond \aformulabis$ is at least $1$ and the spoiler can play a modal move.
Then, the spoiler chooses the structure $\pair{\amodel}{\aworld}$ and chooses exactly $\aworld_1$.
The duplicator has then to reply by choosing a world $\aworld_1'$ accessible from $\aworld'$ (otherwise the spoiler wins and the result clearly follows). Since $\amodel',\aworld' \not \models \Diamond \aformulabis$, it holds that $\amodel',\aworld_1' \not \models \aformulabis$.
By the induction hypothesis, it holds that $\pair{\amodel}{\aworld_1} \neggamerel{m-1,s}{\apropset} \pair{\amodel'}{\aworld_1'}$.
Hence, by choosing $\aworld_1$, the spoiler builds a winning strategy for the game
$\triple{\pair{\amodel}{\aworld}}{\pair{\amodel'}{\aworld'}}{\triple{m}{s}{\apropset}}$.
\item[Induction case: $\aformula = \aformulabis \separate \aformulater$.]
By hypothesis, $\amodel,\aworld \models \aformulabis \separate \aformulater$ and $\amodel',\aworld' \not\models \aformulabis \separate \aformulater$.
Then, there are $\amodel_1$ and $\amodel_2$ such that
$\amodel_1 + \amodel_2 = \amodel$, $\amodel_1,\aworld \models \aformulabis$ and $\amodel_2, \aworld \models \aformulater$. Moreover, by definition,
the number of nested stars in $\aformulabis \separate \aformulater$ is at least $1$ and therefore the spoiler can play a spatial move.
The spoiler chooses the structure $\pair{\amodel}{\aworld}$ and chooses exactly $\amodel_1$ and $\amodel_2$.
The duplicator has then to reply by choosing two structures $\amodel_1'$ and $\amodel_2'$ such that
$\amodel_1' + \amodel_2' = \amodel'$.
Since  $\amodel',\aworld' \not\models \aformulabis \separate \aformulater$, either $\amodel_1',\aworld' \not\models \aformulabis$ or $\amodel_2',\aworld' \not\models \aformulater$.
If the former holds, then by the induction hypothesis, $\pair{\amodel_1}{\aworld} \neggamerel{m,s-1}{\apropset} \pair{\amodel_1'}{\aworld'}$.
Hence, by choosing to continue the game on $\triple{\pair{\amodel_1}{\aworld}}{\pair{\amodel_1'}{\aworld'}}{\triple{m}{s-1}{\apropset}}$ the spoiler built a winning strategy for the game $\triple{\pair{\amodel}{\aworld}}{\pair{\amodel'}{\aworld'}}{\triple{m}{s}{\apropset}}$.
Symmetrically, if instead $\amodel_2',\aworld' \not\models \aformulater$ then by the induction hypothesis
$\pair{\amodel_2}{\aworld} \neggamerel{m,s-1}{\apropset} \pair{\amodel_2'}{\aworld'}$.
Hence, by choosing to continue the game on $\triple{\pair{\amodel_2}{\aworld}}{\pair{\amodel_2'}{\aworld'}}{\triple{m}{s-1}{\apropset}}$, the spoiler
builds a winning strategy for the game $\triple{\pair{\amodel}{\aworld}}{\pair{\amodel'}{\aworld'}}{\triple{m}{s}{\apropset}}$.
In either case, we conclude that $\pair{\amodel}{\aworld} \neggamerel{m,s}{\apropset} \pair{\amodel'}{\aworld'}$.
\end{description}
We now prove that the games are complete (left to right direction).
\begin{nscenter}
$\boxed{
\text{ If }  \pair{\amodel}{\aworld} \neggamerel{m,s}{\apropset} \pair{\amodel'}{\aworld'}
\text{ then there is }
\aformula \in \aformulaerankset{m,s}{\apropset}
\text{ s.t.\ } \amodel,\aworld \models \aformula \text{ and } \amodel',\aworld' \not \models \aformula
}$
\end{nscenter}
Again, the proof is rather standard and it is by induction on $\pair{m}{s}$ and by cases on the first move that the spoiler makes in his winning stategy for the game $\triple{\pair{\amodel}{\aworld}}{\pair{\amodel'}{\aworld'}}{\triple{m}{s}{\apropset}}$.
\begin{description}
\item[Base case: $m = 0$ and $s = 0$.]
Since the spoiler has a winning strategy, in particular it wins the game of rank $\triple{0}{0}{\apropset}$ and therefore by definition of the game
it must hold that
there is a propositional symbol $\avarprop \in \apropset$ such that $\amodel,\aworld \models \avarprop$ iff $\amodel',\aworld' \not \models \avarprop$.
If $\amodel,\aworld \models \avarprop$, then $\aformula$ (as in the statement) is $\avarprop$.
Otherwise (i.e.\ $\amodel',\aworld' \models \avarprop$) we take $\aformula = \lnot \avarprop$.

Notice that this case also holds for games on arbitrary rank $\triple{m}{s}{\apropset}$:  
the spoiler wins simply from the conditions of the game that are imposed before each round.
\item[Induction case: the spoiler plays a modal move.] Notice that then $m \geq 1$.
Suppose that, by following its strategy, the spoiler chooses $\pair{\amodel}{\aworld}$ and a world $\aworld_1$ accessible from $\aworld$.
By Lemma~\ref{lemma:SC-char-formula}, we have that $\amodel,\aworld_1 \models \acharformulaSC{\amodel,\aworld_1}{m-1,s}{\apropset}$.
Let $\aformula$ be defined as the formula $\Diamond\acharformulaSC{\amodel,\aworld_1}{m-1,s}{\apropset}$.
By definition, $\amodel,\aworld \models \aformula$ and $\aformula \in \aformulaerankset{m,s}{\apropset}$.
Ad absurdum, suppose that $\amodel',\aworld' \models \aformula$.
Then there is a world $\aworld_1'$ accessible from $\aworld'$ such that $\amodel',\aworld_1' \models \acharformulaSC{\amodel,\aworld_1}{m-1,s}{\apropset}$.
By Lemma~\ref{lemma:SC-char-formula} there is no formula in $\aformulaerankset{m-1,s}{\apropset}$ that can discriminate between $\pair{\amodel}{\aworld_1}$ and $\pair{\amodel'}{\aworld_1'}$.
As our games are determined,
by the induction hypothesis this implies that the duplicator has a winning strategy for the game $\triple{\pair{\amodel}{\aworld_1}}{\pair{\amodel'}{\aworld_1'}}{\triple{m-1}{s}{\apropset}}$.
This is contradictory, as by hypothesis the spoiler has a winning strategy and the move it played is part of this strategy. Hence,
$\amodel,\aworld \models \aformula$ and $\amodel',\aworld' \not\models \aformula$.

The proof is analogous for the case where the spoiler chooses $\pair{\amodel'}{\aworld'}$ and a world $\aworld_1'$ accessible from $\aworld$.
In this case we obtain $\amodel,\aworld \not\models \aformulabis$ and $\amodel',\aworld' \models \aformulabis$, where
$\aformulabis$ is defined as $\Diamond\acharformulaSC{\amodel',\aworld_1'}{m-1,s}{\apropset}$.
Hence, we take $\aformula$ (as in the statement) defined as $\lnot \aformulabis$.

\item[Induction case: the spoiler plays a spatial move.] Notice that then $s \geq 1$.
Suppose that, by following its strategy, the spoiler chooses $\pair{\amodel}{\aworld}$ and two finite forests $\amodel_1$ and $\amodel_2$
such that
$\amodel_1 + \amodel_2 = \amodel$.
Recall that, by Lemma~\ref{lemma:SC-char-formula},
$\amodel_1,\aworld \models \acharformulaSC{\amodel_1,\aworld}{m,s-1}{\apropset}$ and
$\amodel_2,\aworld \models \acharformulaSC{\amodel_2,\aworld}{m,s-1}{\apropset}$.
Let $\aformula$ be defined as $\acharformulaSC{\amodel_1,\aworld}{m,s-1}{\apropset} \separate \acharformulaSC{\amodel_2,\aworld}{m,s-1}{\apropset}$.
By definition $\amodel,\aworld \models \aformula$ and $\aformula \in \aformulaerankset{m,s}{\apropset}$.
Ad absurdum, suppose that $\amodel',\aworld' \models \aformula$.
Then there are $\amodel_1'$ and $\amodel_2'$ such that
$\amodel_1' + \amodel_2' = \amodel'$,
$\amodel_1',\aworld' \models \acharformulaSC{\amodel_1,\aworld}{m,s-1}{\apropset}$
and
$\amodel_2',\aworld' \models \acharformulaSC{\amodel_2,\aworld}{m,s-1}{\apropset}$.
Then, by Lemma~\ref{lemma:SC-char-formula},
there is no formula in $\aformulaerankset{m,s-1}{\apropset}$ that can discriminate between $\pair{\amodel_1}{\aworld}$ and $\pair{\amodel_1'}{\aworld'}$, or that can discriminate between $\pair{\amodel_2}{\aworld}$ and $\pair{\amodel_2'}{\aworld'}$.
As our games are determined, by the induction hypothesis this implies that the duplicator has a winning strategy for both the games
$\triple{\pair{\amodel_1}{\aworld}}{\pair{\amodel_1'}{\aworld'}}{\triple{m}{s-1}{\apropset}}$
and
$\triple{\pair{\amodel_2}{\aworld}}{\pair{\amodel_2'}{\aworld'}}{\triple{m}{s-1}{\apropset}}$.
This leads to a contradiction, as by hypothesis the spoiler has a winning strategy and the move it played is part of this strategy.
Hence, $\amodel,\aworld \models \aformula$ and $\amodel',\aworld' \not\models \aformula$.

The proof is analogous for the case where the spoiler chooses  $\pair{\amodel'}{\aworld'}$ and two finite forests $\amodel_1'$ and $\amodel_2'$ such that 
$\amodel_1' + \amodel_2' = \amodel'$.
In this case we obtain $\amodel,\aworld \not\models \aformulabis$ and $\amodel',\aworld' \models \aformulabis$ where
$\aformulabis$ is defined as $\acharformulaSC{\amodel_1',\aworld'}{m,s-1}{\apropset} \separate \acharformulaSC{\amodel_2',\aworld'}{m,s-1}{\apropset}$.
Hence, we take $\aformula$ (as in the statement) defined as $\lnot \aformulabis$. \qedhere
\end{description}
\end{proof}

\subsection{Proof of \Cref{lemma:GML-more-SC}}

\begin{proof}
As usual, the  non-expressivity of $\Gdiamond{=2}\Gdiamond{=1}\true$ is shown by proving that for every rank $\triple{m}{s}{\apropset}$ there are two structures $\pair{\amodel}{\aworld}$ and $\pair{\amodel'}{\aworld'}$ such that
\begin{itemize}
\item $\pair{\amodel}{\aworld} \gamerel{m,s}{\apropset} \pair{\amodel'}{\aworld'}$, and
\item $\amodel,\aworld \models \Gdiamond{=2}\Gdiamond{=1}\true$ whereas $\amodel',\aworld' \not\models \Gdiamond{=2}\Gdiamond{=1}\true$.
\end{itemize}
Here, we divide the proof into two parts, named below \ref{lemmaprop:GML-more-SL-1} and \ref{lemmaprop:GML-more-SL-2}.
We start with some preliminary definitions.
Let $\amodel = \triple{\worlds}{\arelation}{\avaluation}$ be a finite forest and $\aworld \in \worlds$.
We denote with $\arelation(\aworld)_{=n}$ the set of worlds in $\arelation(\aworld)$ having exactly $n$ children, i.e. $\{ \aworld_1 \in \arelation(\aworld) \mid \card{\arelation(\aworld_1)} = n \}$.
During the proof, we only use pointed forests $\pair{\amodel}{\aworld}$ satisfying the following properties:
\begin{enumerate}[label=\Roman*]
\item\label{constr:GML-more-SL-1} $\avaluation(\avarprop) = \emptyset$ for every $\avarprop \in \varprop$;
\item\label{constr:GML-more-SL-2} $\arelation(\aworld)_{=0}$, $\arelation(\aworld)_{=1}$ and $\arelation(\aworld)_{=2}$ form a partition of $\arelation(\aworld)$;
\item\label{constr:GML-more-SL-3} $\arelation^3(\aworld) = \emptyset$, i.e. the set of worlds reachable from $\aworld$ in at least three steps is empty.
\end{enumerate}
Below, we represent schematically the models satisfying the
properties~\ref{constr:GML-more-SL-1},~\ref{constr:GML-more-SL-2}~and~\ref{constr:GML-more-SL-3}
(notice that each world does not satisfy any propositional symbol).
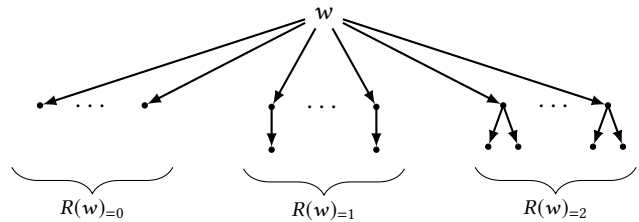
\begin{nscenter}
\begin{tikzpicture}
\node (w) {$\aworld$};
\node[dot] (w1) [below left = 1cm and 3.5cm of w] {};
\node (h1) [right = 0.3cm of w1] {$\dots$};
\node[dot] (w2) [right = 0.3cm of h1] {};

\node (h2) [below = 0.9cm of w] {$\dots$};
\node[dot] (w3) [left = 0.3cm of h2] {};
\node[dot] (w4) [right = 0.3cm of h2] {};

\node[dot] (w5) [below right = 1cm and 3.5cm of w] {};
\node (h3) [left = 0.3cm of w5] {$\dots$};
\node[dot] (w6) [left = 0.3cm of h3] {};

\node[dot] (ww1) [below = 0.5cm of w3] {};
\node[dot] (ww2) [below = 0.5cm of w4] {};

\node[dot] (ww3) [below left = 0.5cm and 0.15cm of w5] {};
\node[dot] (ww4) [below right = 0.5cm and 0.15cm of w5] {};

\node[dot] (ww5) [below left = 0.5cm and 0.15cm of w6] {};
\node[dot] (ww6) [below right = 0.5cm and 0.15cm of w6] {};

\draw[pto] (w) -- (w1);
\draw[pto] (w) -- (w2);

\draw[pto] (w) -- (w3);
\draw[pto] (w) -- (w4);

\draw[pto] (w) -- (w5);
\draw[pto] (w) -- (w6);

\draw[pto] (w3) -- (ww1);
\draw[pto] (w4) -- (ww2);

\draw[pto] (w5) -- (ww3);
\draw[pto] (w5) -- (ww4);

\draw[pto] (w6) -- (ww5);
\draw[pto] (w6) -- (ww6);

\node (h4) [below left = 0.55cm and 0.1cm of w1] {};
\node (h5) [below right = 0.55cm and 0.1cm of w2] {};

\node (h6) [below left = 0.55cm and 0.1cm of w3] {};
\node (h7) [below right = 0.55cm and 0.1cm of w4] {};

\node (h8) [below left = 0.55cm and 0.1cm of w6] {};
\node (h9) [below right = 0.55cm and 0.1cm of w5] {};

\draw [decorate,decoration={brace,amplitude=10pt}]
(h5.south east) -- node[below = 0.3cm]
{\footnotesize{$\arelation(\aworld)_{=0}$}} (h4.south west);

\draw [decorate,decoration={brace,amplitude=10pt}]
(h7.south east) -- node[below = 0.3cm]
{\footnotesize{$\arelation(\aworld)_{=1}$}} (h6.south west);

\draw [decorate,decoration={brace,amplitude=10pt}]
(h9.south east) -- node[below = 0.3cm]
{\footnotesize{$\arelation(\aworld)_{=2}$}} (h8.south west);

\end{tikzpicture}
\end{nscenter}
Let us consider two models $\amodel_1 = \triple{\worlds}{\arelation_1}{\avaluation}$ and 
$\amodel_2 = \triple{\worlds}{\arelation_2}{\avaluation_2}$ such that $\amodel_1 + \amodel_2 = \amodel$.
We pinpoint three important properties of the models we are considering.
\begin{description}
\item[\rulelab{S1}{GML-m-SL:S1}] Every world in $\arelation(\aworld)_{=0}$ is either in $\arelation_1(\aworld)_{=0}$ or $\arelation_2(\aworld)_{=0}$;
\item[\rulelab{S2}{GML-m-SL:S2}] Every world $\aworld_1 \in \arelation(\aworld)_{=1}$ is in $\arelation_1(\aworld)_{=0}$, $\arelation_2(\aworld)_{=0}$, $\arelation_1(\aworld)_{=1}$ or in $\arelation_2(\aworld)_{=1}$.
Indeed, suppose $\pair{\aworld}{\aworld_1} \in \arelation_i$ (for some $i \in \{1,2\}$).
If $\aworld_1$ is in the domain of the same relation $\arelation_i$ then $\aworld_1 \in \arelation_i(\aworld)_{=1}$. Otherwise ($\aworld_1$ is in the domain of $\arelation_{3-i}$) then  $\aworld_1 \in \arelation_i(\aworld)_{=0}$.
\item[\rulelab{S3}{GML-m-SL:S3}] Every world in $\arelation(\aworld)_{=2}$ is in $\arelation_1(\aworld)_{=0}$, $\arelation_2(\aworld)_{=0}$, $\arelation_1(\aworld)_{=1}$, $\arelation_2(\aworld)_{=1}$, $\arelation_1(\aworld)_{=2}$ or $\arelation_2(\aworld)_{=2}$.
The justification is similar to the one given above for
$\arelation(\aworld)_{=1}$.
\end{description}
We first prove the following property:
\begin{nscenter}
$\boxed{\text{\rulelab{(A)}{lemmaprop:GML-more-SL-1}\qquad
$\begin{aligned}[t]&\text{Given a rank } \triple{m}{s}{\apropset}\text{ and two pointed forests }\pair{\amodel = \triple{\worlds}{\arelation}{\avaluation}}{\aworld}\text{ and}\\
&\pair{\amodel' = \triple{\worlds'}{\arelation'}{\avaluation'}}{\aworld'} \text{ satisfying \ref{constr:GML-more-SL-1}, \ref{constr:GML-more-SL-2} and \ref{constr:GML-more-SL-3}, if}\\
&\ \text{\labelitemi\ } \min(\card{\arelation(\aworld)_{=0}},2^s) = \min(\card{\arelation'(\aworld')_{=0}},2^s);\\
&\ \text{\labelitemi\ } \min(\card{\arelation(\aworld)_{=1}},2^s(s+1)) = \min(\card{\arelation'(\aworld')_{=1}},2^s(s+1));\\
&\ \text{\labelitemi\ } \min(\card{\arelation(\aworld)_{=2}},2^{s-1}(s+1)(s+2)) = \min(\card{\arelation'(\aworld')_{=2}},2^{s-1}(s+1)(s+2))\\
&\text{then } \pair{\amodel}{\aworld} \gamerel{m,s}{\apropset} \pair{\amodel'}{\aworld'}
\end{aligned}
$}}$
\end{nscenter}

First, as worlds in our models do not satisfy any propositional symbol,
the spoiler cannot win because of distinct propositional valuations.
The proof is by cases on $m$ and on the moves done by the spoiler, and by induction on $s$.
First, suppose $m = 0$. Then it is easy to see that the duplicator has a winning strategy.
Indeed, as $m = 0$, the spoiler cannot play the modal move and therefore cannot change the current worlds $\aworld$ and $\aworld'$.
Then, after $s$ spatial moves the game will be in the state $\pair{\amodel_1}{\aworld}$ and $\pair{\amodel_1'}{\aworld'}$ w.r.t.\ the rank $\triple{0}{0}{\apropset}$. From \ref{constr:GML-more-SL-1} we conclude that the duplicator wins.

Suppose now $m \geq 1$ and  the spoiler decides to perform a modal move. Notice that, in particular, this case also takes care of the case where $s = 0$ and the spoiler is forced to play a modal move.
Moreover, suppose that the spoiler chooses $\pair{\amodel}{\aworld}$ (the case where it picks $\pair{\amodel'}{\aworld'}$ is analogous).
We have to distinguish the following situations. 
\begin{itemize}
\item Suppose that the spoiler chooses a world $\aworld_1 \in \arelation(\aworld)_{=0}$. Then $\card{\arelation(\aworld)_{=0}} \geq 1$ and by hypothesis $\min(\card{\arelation(\aworld)_{=0}},2^s) = \min(\card{\arelation'(\aworld')_{=0}},2^s)$,
it follows that $\card{\arelation'(\aworld')_{=0}} \geq 1$.
It is then sufficient for the duplicator to choose $\aworld_1 \in \arelation'(\aworld')_{=0}$ to guarantee him a victory, as the subtrees rooted in $\aworld_1$ and $\aworld_1'$ are isomorphic.
\item Suppose that the spoiler chooses a world $\aworld_1 \in \arelation(\aworld)_{=1}$. Then $\card{\arelation(\aworld)_{=1}} \geq 1$ and by hypothesis $\min(\card{\arelation(\aworld)_{=1}},2^s(s+1)) = \min(\card{\arelation'(\aworld')_{=1}},2^s(s+1))$,
it follows that $\card{\arelation'(\aworld')_{=1}} \geq 1$. Then again,
it is sufficient for the duplicator to choose $\aworld_1 \in \arelation'(\aworld')_{=1}$ to guarantee him a victory, as the subtrees rooted in $\aworld_1$ and $\aworld_1'$ are isomorphic.
\item Suppose that the spoiler chooses a world $\aworld_1 \in \arelation(\aworld)_{=2}$. Then $\card{\arelation(\aworld)_{=2}} \geq 1$ and by hypothesis $\min(\card{\arelation(\aworld)_{=2}},2^{s-1}(s+1)(s+2)) = \min(\card{\arelation'(\aworld')_{=2}},2^{s-1}(s+1)(s+2))$,
it follows that $\card{\arelation'(\aworld')_{=2}} \geq 1$ (notice here that $2^{s-1}(s+1)(s+2) = 1$ for $s = 0$).
Then again,
it is sufficient for the duplicator to choose $\aworld_1 \in \arelation'(\aworld')_{=2}$ to guarantee him a victory, as the subtrees rooted in $\aworld_1$ and $\aworld_1'$ are isomorphic.
\end{itemize}

As stated before, the case where the spoiler decides to perform a modal move also captures the base case of the induction on $s$. Then, it remains to show the case where $s \geq 1$ and the spoiler decides to do a spatial move.
Again suppose that the spoiler chooses $\pair{\amodel}{\aworld}$ (the case where it picks $\pair{\amodel'}{\aworld'}$ is analogous).
It then picks two structures $\amodel_1 = \triple{\worlds}{\arelation_1}{\avaluation}$ and $\amodel_2 = \triple{\worlds}{\arelation_2}{\avaluation}$ such that
 $\amodel_1 + \amodel_2 = \amodel$.
Notice that these two structures are such what both $\pair{\amodel_1}{\aworld}$ and $\pair{\amodel_2}{\aworld}$ satisfy
\ref{constr:GML-more-SL-1}, \ref{constr:GML-more-SL-2} and \ref{constr:GML-more-SL-3}, as it is easy to see that these three properties are all preserved when taking submodels.
The duplicator has now to pick two structures $\amodel_1' = \triple{\worlds'}{\arelation_1'}{\avaluation'}$ and $\amodel_2' = \triple{\worlds'}{\arelation_2'}{\avaluation'}$ such that
$\amodel_1' + \amodel_2' = \amodel'$ and that guarantees him a victory. It does so by constructing $\arelation_1'$ and $\arelation_2'$ as follows (from the empty set):
\begin{description}
\item[Split of $\arelation'(\aworld)_{=0}$.]
We introduce the sets
\begin{nscenter}
$\begin{aligned}[t]
\arelation_1(\aworld)|_{0 \blacktriangleright 0} &\egdef \arelation_1(\aworld)_{=0} \cap \arelation(\aworld)_{=0}\\
\arelation_2(\aworld)|_{0 \blacktriangleright 0} &\egdef \arelation_2(\aworld)_{=0} \cap \arelation(\aworld)_{=0}.
\end{aligned}$
\end{nscenter}
It is easy to see that these sets are pairwise disjoint.
From (\ref{GML-m-SL:S1}) it follows that
\begin{nscenter}
$\arelation(\aworld)_{=0} = (\arelation_1(\aworld)_{=0} \cap \arelation(\aworld)_{=0}) \cup (\arelation_2(\aworld)_{=0} \cap \arelation(\aworld)_{=0})$.
\end{nscenter}
The duplicator start by partitioning $\arelation'(\aworld)_{=0}$ into two sets $Z_{1}$ and $Z_{2}$ according to the cardinalities of the two components of $\arelation(\aworld)_{=0}$ highlighted above, namely the two sets $\arelation_1(\aworld)_{=0} \cap \arelation(\aworld)_{=0}$ and $\arelation_2(\aworld)_{=0} \cap \arelation(\aworld)_{=0}$.
\begin{itemize}
\item Suppose that $\card{\arelation_1(\aworld)|_{0 \blacktriangleright 0}} < 2^{s-1}$ and $\card{\arelation_2(\aworld)|_{0 \blacktriangleright 0}} < 2^{s-1}$.
 Hence, $\card{\arelation(\aworld)_{=0}} < 2^s$ and by hypothesis $\card{\arelation'(\aworld')_{=0}} = \card{\arelation(\aworld)_{=0}}$.
Then the split of $\arelation'(\aworld)_{=0}$ into $Z_{1}$ and $Z_{2}$ is made so that $\card{Z_{1}} = \card{\arelation_1(\aworld)|_{0 \blacktriangleright 0}}$ and
$\card{Z_{2}} = \card{\arelation_2(\aworld)|_{0 \blacktriangleright 0}}$.
\item Suppose that there is $i \in \{1,2\}$ such that $\card{\arelation_i(\aworld)|_{0 \blacktriangleright 0}} < 2^{s-1}$ and
$\card{\arelation_j(\aworld)|_{0 \blacktriangleright 0}} \geq 2^{s-1}$,
where $j = 3-i$ is the index of the other set.
Then the split of $\arelation'(\aworld)_{=0}$ into $Z_{i}$ and $Z_{j}$ is made so that $\card{Z_{i}} = \card{\arelation_i(\aworld)|_{0 \blacktriangleright 0}}$.
Notice that by hypothesis on the cardinality of $\arelation'(\aworld)_{=0}$ it holds that $\card{Z_j} \geq 2^{s-1}$ (otherwise $\min(\card{\arelation(\aworld)_{=0}},2^s) \neq \min(\card{\arelation'(\aworld')_{=0}},2^s)$).
\item Suppose that $\card{\arelation_1(\aworld)|_{0 \blacktriangleright 0}} \geq 2^{s-1}$ and $\card{\arelation_2(\aworld)|_{0 \blacktriangleright 0}} \geq 2^{s-1}$.
Then the split of $\arelation'(\aworld)_{=0}$ into $Z_{1}$ and $Z_{2}$ is made so that $\card{Z_{1}} = 2^{s-1}$. Notice that by hypothesis on the cardinality of $\arelation'(\aworld)_{=0}$ it holds that $\card{Z_j} \geq 2^{s-1}$.
\end{itemize}
For each $\aworld_1' \in Z_1$, the duplicator adds $\pair{\aworld'}{\aworld_1'}$ to $\arelation_1'$.
For each $\aworld_2' \in Z_2$, it adds $\pair{\aworld'}{\aworld_2'}$ to $\arelation_2'$.
Notice that by construction the two sets introduced are always such that
\begin{description}
\item[\rulelab{Z1}{GML-m-SL:Z1}] $\min(\card{\arelation_1(\aworld)|_{0 \blacktriangleright 0}},2^{s-1}) = \min(\card{Z_1},2^{s-1})$
\item[\rulelab{Z2}{GML-m-SL:Z2}] $\min(\card{\arelation_2(\aworld)|_{0 \blacktriangleright 0}},2^{s-1}) = \min(\card{Z_2},2^{s-1})$.
\end{description}
\item[Split of $\arelation'(\aworld)_{=1}$.]
We introduce the following sets:
\begin{nscenter}
$
\begin{aligned}[t]
\arelation_1(\aworld)|_{1 \blacktriangleright 0} \egdef \arelation_1(\aworld)_{=0} \cap \arelation(\aworld)_{=1} &\qquad\qquad
\arelation_2(\aworld)|_{1 \blacktriangleright 0} \egdef \arelation_2(\aworld)_{=0} \cap \arelation(\aworld)_{=1}\\
\arelation_1(\aworld)|_{1 \blacktriangleright 1} \egdef \arelation_1(\aworld)_{=1} \cap \arelation(\aworld)_{=1} &\qquad\qquad
\arelation_2(\aworld)|_{1 \blacktriangleright 1} \egdef \arelation_2(\aworld)_{=1} \cap \arelation(\aworld)_{=1}.
\end{aligned}$
\end{nscenter}
It is easy to see that these sets are pairwise disjoint.
From (\ref{GML-m-SL:S2}) it follows that
\begin{nscenter}
$
\arelation(\aworld)_{=1} = \arelation_1(\aworld)|_{1 \blacktriangleright 0} \cup  \arelation_2(\aworld)|_{1 \blacktriangleright 0}  \cup  \arelation_1(\aworld)|_{1 \blacktriangleright 1}  \cup  \arelation_2(\aworld)|_{1 \blacktriangleright 1}.
$
\end{nscenter}
The duplicator starts by partitioning $\arelation'(\aworld)_{=1}$ into four sets $Z_1'$, $Z_2'$, $O_1$ and $O_2$ according to the cardinalities of the four sets above (`Z' for `zero', `O' for 'one').
In order to shorten the presentation, instead of concretely make explicit all the cases as we did in the previous point of the construction, we
treat them ``schematically''.
Let $\mathcal{\aset} = \{ \arelation_1(\aworld)|_{1 \blacktriangleright 0}, \arelation_2(\aworld)|_{1 \blacktriangleright 0}, \arelation_1(\aworld)|_{1 \blacktriangleright 1}, \arelation_2(\aworld)|_{1 \blacktriangleright 1}\}$ and
let $\amap$ be the bijection
\begin{nscenter}
$\amap(\arelation_1(\aworld)|_{1 \blacktriangleright 0}) \egdef Z_1'$, \quad $\amap(\arelation_2(\aworld)|_{1 \blacktriangleright 0}) \egdef Z_2'$
\quad
$\amap(\arelation_1(\aworld)|_{1 \blacktriangleright 1}) \egdef O_1$, \quad $\amap(\arelation_2(\aworld)|_{1 \blacktriangleright 1}) \egdef O_2$.
\end{nscenter}
Moreover, we define ($\mathcal{B}$ stands for ``bound'')
\begin{nscenter}
$\begin{aligned}[t]
&\boundd{\arelation_1(\aworld)|_{1 \blacktriangleright 0}} \ \egdef \ \boundd{\arelation_2(\aworld)|_{1 \blacktriangleright 0}} \ \egdef \ 2^{s-1}\\
&\boundd{\arelation_1(\aworld)|_{1 \blacktriangleright 1}} \ \egdef \ \boundd{\arelation_2(\aworld)|_{1 \blacktriangleright 1}} \ \egdef \ 2^{s-1}s.
\end{aligned}$
\end{nscenter}
So, these definitions (actually notations) are helpful at the metalevel. Besides,
notice that, from $s \geq 1$, it holds that $2^{s-1}$ and $2^{s-1}s$ are both at least $1$.

  \begin{itemize}
    \item Suppose that for every set $S \in \mathcal{\aset}$ it holds that $\card{S} < \boundd{S}$. Then, since it holds that
    \begin{nscenter}
    $\card{\arelation(\aworld)_{=1}} = \card{\arelation_1(\aworld)|_{1 \blacktriangleright 0}} + \card{\arelation_2(\aworld)|_{1 \blacktriangleright 0}} + \card{\arelation_1(\aworld)|_{1 \blacktriangleright 1}} + \card{\arelation_2(\aworld)|_{1 \blacktriangleright 1}}$
    \end{nscenter}
    it holds that $\card{\arelation(\aworld)_{=1}} < 2^{s-1} + 2^{s-1} + 2^{s-1}s + 2^{s-1}s = 2^{s}(s+1)$ and therefore by hypothesis we conclude that $\card{\arelation(\aworld)_{=1}} = \card{\arelation'(\aworld')_{=1}}$.
    Then, the split of $\arelation'(\aworld')_{=1}$ into $Z_1'$, $Z_2'$, $O_1$ and $O_2$ is made so that for every
    $S \in \mathcal{\aset}$, $\card{\amap(S)} = \card{S}$.
    \item Suppose instead that there is $\widehat{S} \in \mathcal{\aset}$ such that $\card{\widehat{S}} \geq \boundd{\widehat{S}}$.
    Then, the split of $\arelation'(\aworld')_{=1}$ into $Z_1'$, $Z_2'$, $O_1$ and $O_2$ is made so that for every
    $S \in \mathcal{\aset} \setminus \set{\widehat{S}}$, $\card{\amap(S)} = \min(\card{S},\boundd{S})$.
    From the hypothesis
    \begin{nscenter}
      $\min(\card{\arelation(\aworld)_{=1}},2^s(s+1)) = \min(\card{\arelation'(\aworld')_{=1}},2^s(s+1))$
    \end{nscenter}
    we conclude that this construction can be effectively made and it is such that $\card{\amap(\widehat{S})} \geq \boundd{\widehat{S}}$.
  \end{itemize}
  For each $\aworld_1' \in Z_1'$, the duplicator adds $\pair{\aworld'}{\aworld_1'}$ to $\arelation_1'$
  and the only element of $\arelation'|_{\aworld_1'}$ to $\arelation_2'$.
  For each $\aworld_2' \in Z_2'$, it adds $\pair{\aworld'}{\aworld_2'}$ to $\arelation_2'$
   and the only element of $\arelation'|_{\aworld_2'}$ to $\arelation_1'$.
  For each $\aworld_1' \in O_1$, it adds $\pair{\aworld'}{\aworld_1'}$ and the only element of $\arelation'|_{\aworld_1'}$ to $\arelation_1'$.
  Lastly, for each $\aworld_2' \in O_2$, it adds $\pair{\aworld'}{\aworld_2'}$ and the only element of $\arelation'|_{\aworld_2'}$ to $\arelation_2'$.
  Notice that by construction the four sets introduced are always such that
  \begin{description}
  \item[\rulelab{Z11}{GML-m-SL:Z1p}] $\min(\card{\arelation_1(\aworld)|_{1 \blacktriangleright 0}},2^{s-1}) = \min(\card{Z_1'},2^{s-1})$
  \item[\rulelab{Z21}{GML-m-SL:Z2p}] $\min(\card{\arelation_2(\aworld)|_{1 \blacktriangleright 0}},2^{s-1}) = \min(\card{Z_2'},2^{s-1})$
  \item[\rulelab{O1}{GML-m-SL:O1}] $\min(\card{\arelation_1(\aworld)|_{1 \blacktriangleright 1}},2^{s-1}s) = \min(\card{O_1},2^{s-1}s)$
  \item[\rulelab{O2}{GML-m-SL:O2}] $\min(\card{\arelation_2(\aworld)|_{1 \blacktriangleright 1}},2^{s-1}s) = \min(\card{O_2},2^{s-1}s)$
  \end{description}
  or, more schematically, for every $S \in \mathcal{\aset}$, $\min(\card{S},\boundd{S}) = \min(\card{\amap(S)},\boundd{S})$.
\item[Split of $\arelation'(\aworld)_{=2}$.]
Similarly to the previous steps, we introduce the following sets:
\begin{nscenter}
$
\begin{aligned}[t]
\arelation_1(\aworld)|_{2 \blacktriangleright 0} &\egdef \arelation_1(\aworld)_{=0} \cap \arelation(\aworld)_{=2} &\qquad\qquad
\arelation_2(\aworld)|_{2 \blacktriangleright 0} &\egdef \arelation_2(\aworld)_{=0} \cap \arelation(\aworld)_{=2}\\
\arelation_1(\aworld)|_{2 \blacktriangleright 1} &\egdef \arelation_1(\aworld)_{=1} \cap \arelation(\aworld)_{=2} &\qquad\qquad
\arelation_2(\aworld)|_{2 \blacktriangleright 1} &\egdef \arelation_2(\aworld)_{=1} \cap \arelation(\aworld)_{=2}\\
\arelation_1(\aworld)|_{2 \blacktriangleright 2} &\egdef \arelation_1(\aworld)_{=2} \cap \arelation(\aworld)_{=2} &\qquad\qquad
\arelation_2(\aworld)|_{2 \blacktriangleright 2} &\egdef \arelation_2(\aworld)_{=2} \cap \arelation(\aworld)_{=2}.
\end{aligned}$
\end{nscenter}
It is easy to see that these sets are pairwise disjoint.
From (\ref{GML-m-SL:S3}) it follows that
\begin{nscenter}
$
\arelation(\aworld)_{=2} = \arelation_1(\aworld)|_{2 \blacktriangleright 0} \cup  \arelation_2(\aworld)|_{2 \blacktriangleright 0}  \cup  \arelation_1(\aworld)|_{2 \blacktriangleright 1}  \cup  \arelation_2(\aworld)|_{2 \blacktriangleright 1} \cup
 \arelation_1(\aworld)|_{2 \blacktriangleright 2}  \cup  \arelation_2(\aworld)|_{2 \blacktriangleright 2}
$
\end{nscenter}
The duplicator starts by partitioning $\arelation'(\aworld)_{=2}$ into six sets $Z_1''$, $Z_2''$, $O_1'$, $O_2'$, $T_1$ and $T_2$ according to the cardinalities of the six sets above (`T' for `two').
Again, to shorten the presentation we introduce the set
\begin{nscenter}
$\mathcal{\aset} = \{ \arelation_1(\aworld)|_{2 \blacktriangleright 0}, \arelation_2(\aworld)|_{2 \blacktriangleright 0}, \arelation_1(\aworld)|_{2 \blacktriangleright 1}, \arelation_2(\aworld)|_{2 \blacktriangleright 1},
\arelation_1(\aworld)|_{2 \blacktriangleright 2}, \arelation_2(\aworld)|_{2 \blacktriangleright 2}\}$,
\end{nscenter}
and the bijection $\amap$ such that 
\begin{nscenter}
$\amap(\arelation_1(\aworld)|_{2 \blacktriangleright 0}) \egdef Z_1''$, \quad $\amap(\arelation_2(\aworld)|_{2 \blacktriangleright 0}) \egdef Z_2''$
\quad
$\amap(\arelation_1(\aworld)|_{2 \blacktriangleright 1}) \egdef O_1'$, \\
$\amap(\arelation_2(\aworld)|_{2 \blacktriangleright 1}) \egdef O_2'$,
\quad $\amap(\arelation_1(\aworld)|_{2 \blacktriangleright 2}) \egdef T_1$,
\quad $\amap(\arelation_2(\aworld)|_{2 \blacktriangleright 2}) \egdef T_2$.
\end{nscenter}
Moreover, we define
\begin{nscenter}
$\begin{aligned}[t]
&\boundd{\arelation_1(\aworld)|_{2 \blacktriangleright 0}} \ \egdef \ \boundd{\arelation_2(\aworld)|_{2 \blacktriangleright 0}} \ \egdef \ 2^{s-1}\\
&\boundd{\arelation_1(\aworld)|_{2 \blacktriangleright 1}} \ \egdef \ \boundd{\arelation_2(\aworld)|_{2 \blacktriangleright 1}} \ \egdef \ 2^{s-1}s\\
&\boundd{\arelation_1(\aworld)|_{2 \blacktriangleright 2}} \ \egdef \ \boundd{\arelation_2(\aworld)|_{2 \blacktriangleright 2}} \ \egdef \ 2^{s-2}s(s+1)\\
\end{aligned}$
\end{nscenter}
Notice that, from $s \geq 1$, it holds that $2^{s-1}$, $2^{s-1}s$ and $2^{s-2}s(s+1)$ are both at least $1$.
  \begin{itemize}
    \item Suppose that for every set $S \in \mathcal{\aset}$ it holds that $\card{S} < \boundd{S}$. Then, since $\card{\arelation(\aworld)_{=2}}$ is
    \begin{nscenter}
    $\card{\arelation_1(\aworld)|_{2 \blacktriangleright 0}} + \card{\arelation_2(\aworld)|_{2 \blacktriangleright 0}} + \card{\arelation_1(\aworld)|_{2 \blacktriangleright 1}} + \card{\arelation_2(\aworld)|_{2 \blacktriangleright 1}}  + \card{\arelation_1(\aworld)|_{2 \blacktriangleright 2}} + \card{\arelation_2(\aworld)|_{2 \blacktriangleright 2}}$
    \end{nscenter}
    it holds that
    \begin{nscenter}
    $\card{\arelation(\aworld)_{=2}} < 2 \times 2^{s-1} + 2 \times 2^{s-1}s + 2 \times 2^{s-2}s(s+1) = 2^{s-1}(s+1)(s+2)$
    \end{nscenter}
     and therefore by hypothesis we conclude that $\card{\arelation(\aworld)_{=2}} = \card{\arelation'(\aworld')_{=2}}$.
    Then, the split of $\arelation'(\aworld')_{=2}$ into $Z_1''$, $Z_2''$, $O_1'$, $O_2'$, $T_1$ and $T_2$ is made so that for every
    $S \in \mathcal{\aset}$, $\card{\amap(S)} = \card{S}$.
    \item Suppose instead that there is $\widehat{S} \in \mathcal{\aset}$ such that $\card{\widehat{S}} \geq \boundd{\widehat{S}}$.
    Then, the split of $\arelation'(\aworld')_{=2}$ into $Z_1''$, $Z_2''$, $O_1'$, $O_2'$, $T_1$ and $T_2$ is made so that for every $S \in \mathcal{\aset} \setminus \widehat{S}$, $\card{\amap(S)} = \min(\card{S},\boundd{S})$.
    From the hypothesis
    \begin{nscenter}
      $\min(\card{\arelation(\aworld)_{=2}},2^{s-1}(s+1)(s+2)) = \min(\card{\arelation'(\aworld')_{=2}},2^{s-1}(s+1)(s+2))$
    \end{nscenter}
    we conclude that this construction can be effectively made and it is such that $\card{\amap(\widehat{S})} \geq \boundd{\widehat{S}}$.
  \end{itemize}
  Then, the duplicator updates  $\arelation_1'$ and $\arelation_2'$ as follows:
  \begin{itemize}
  \item For each $\aworld_1' \in Z_1''$, the duplicator adds $\pair{\aworld'}{\aworld_1'}$ to $\arelation_1'$
  and the two elements of $\arelation'|_{\aworld_1'}$ to $\arelation_2'$.
  \item For each $\aworld_2' \in Z_2''$, it adds $\pair{\aworld'}{\aworld_2'}$ to $\arelation_2'$
  and the two elements of $\arelation'|_{\aworld_2'}$ to $\arelation_1'$.
  \item For each $\aworld_1' \in O_1'$, it adds $\pair{\aworld'}{\aworld_1'}$ and one of the two elements of
  $\arelation'|_{\aworld_1'}$ to $\arelation_1'$.
  The other element of $\arelation'|_{\aworld_1'}$ is assigned to $\arelation_2'$.
  \item For each $\aworld_2' \in O_2'$, it adds $\pair{\aworld'}{\aworld_2'}$ and one of the two elements of
$\arelation'|_{\aworld_2'}$ to $\arelation_2'$.
  The other element of $\arelation'|_{\aworld_2'}$ is assigned to $\arelation_1'$.
  \item For each $\aworld_2' \in T_1$, it adds $\pair{\aworld'}{\aworld_2'}$ to $\arelation_1'$ and
         the two elements of $\arelation'|_{\aworld_2'}$ to $\arelation_1'$.
  \item For each $\aworld_2' \in T_2$, it adds $\pair{\aworld'}{\aworld_2'}$ to $\arelation_2'$ and
        the two elements of $\arelation'|_{\aworld_2'}$ to $\arelation_2'$.
  \end{itemize}
  Notice that by construction the six sets introduced are always such that
  \begin{description}
  \item[\rulelab{Z12}{GML-m-SL:Z1pp}] $\min(\card{\arelation_1(\aworld)|_{2 \blacktriangleright 0}},2^{s-1}) = \min(\card{Z_1''},2^{s-1})$
  \item[\rulelab{Z22}{GML-m-SL:Z2pp}] $\min(\card{\arelation_2(\aworld)|_{2 \blacktriangleright 0}},2^{s-1}) = \min(\card{Z_2''},2^{s-1})$
  \item[\rulelab{O11}{GML-m-SL:O1p}] $\min(\card{\arelation_1(\aworld)|_{2 \blacktriangleright 1}},2^{s-1}s) = \min(\card{O_1'},2^{s-1}s)$
  \item[\rulelab{O21}{GML-m-SL:O2p}] $\min(\card{\arelation_2(\aworld)|_{2 \blacktriangleright 1}},2^{s-1}s) = \min(\card{O_2'},2^{s-1}s)$
  \item[\rulelab{T1}{GML-m-SL:T1}] $\min(\card{\arelation_1(\aworld)|_{2 \blacktriangleright 2}},2^{s-2}s(s+1)) = \min(\card{T_1},2^{s-2}s(s+1))$
  \item[\rulelab{T2}{GML-m-SL:T2}] $\min(\card{\arelation_2(\aworld)|_{2 \blacktriangleright 2}},2^{s-2}s(s+1)) = \min(\card{T_2},2^{s-2}s(s+1))$
  \end{description}
  or, more schematically, for every $S \in \mathcal{\aset}$, $\min(\card{S},\boundd{S}) = \min(\card{\amap(S)},\boundd{S})$.
\end{description}
After these steps, since $\pair{\amodel'}{\aworld'}$ satisfies
\ref{constr:GML-more-SL-2} and \ref{constr:GML-more-SL-3}, every element $\pair{\aworld_1'}{\aworld_2'} \in \arelation'$ such that $\aworld_1' \in {\arelation'}^*(\aworld)$ has been assigned to either $\arelation_1'$ or $\arelation_2'$. Duplicator then conclude the construction of
$\amodel_1'$ and $\amodel_2'$ by assigning the remaining elements of $\arelation'$ (i.e. the pairs $\pair{\aworld_1'}{\aworld_2'} \in \arelation'$ such that $\aworld_1' \not\in {\arelation'}^*(\aworld)$) to either $\arelation_1'$ or $\arelation_2'$ (for example, it can put all these elements in $\arelation_1'$).
The two models $\amodel_1'$ and $\amodel_2'$ are now defined and they trivially satisfy \ref{constr:GML-more-SL-1}, \ref{constr:GML-more-SL-2} and \ref{constr:GML-more-SL-3} (as they are submodels of $\amodel'$).
Moreover, by construction it is easy to verify that:
\begin{multicols}{2}
\begin{itemize}
\item $\arelation_1'(\aworld')_{=0} = Z_1 + Z_1' + Z_1''$
\item $\arelation_1'(\aworld')_{=1} = O_1 + O_1'$
\item $\arelation_1'(\aworld')_{=2} = T_1$
\item for every $n > 2$, $\arelation_1'(\aworld')_{=n} = \emptyset$
\item $\arelation_2'(\aworld')_{=0} = Z_2 + Z_2' + Z_2''$
\item $\arelation_2'(\aworld')_{=1} = O_2 + O_2'$
\item $\arelation_2'(\aworld')_{=2} = T_2$
\item for every $n > 2$, $\arelation_2'(\aworld')_{=n} = \emptyset$
\end{itemize}
\end{multicols}
\noindent Indeed, we specifically built $\arelation_1'$ and $\arelation_2'$ so that these properties (which we later refer to with \rulelab{($\dagger$)}{GML-m-SL:dagger}) hold.
Now, we end the proof of~\ref{lemmaprop:GML-more-SL-1} by showing that for all $i \in \{1,2\}$,
\begin{description}
\item[\rulelab{zero}{GML-m-SL:zero}] $\min(\card{\arelation_i(\aworld)_{=0}},2^{s-1}) = \min(\card{\arelation_i'(\aworld')_{=0}},2^{s-1})$;
\item[\rulelab{one}{GML-m-SL:one}] $\min(\card{\arelation_i(\aworld)_{=1}},2^{s-1}s) = \min(\card{\arelation_i'(\aworld')_{=1}},2^{s-1}s)$;
\item[\rulelab{two}{GML-m-SL:two}] $\min(\card{\arelation_i(\aworld)_{=2}},2^{s-2}s(s+1)) = \min(\card{\arelation_i'(\aworld')_{=2}},2^{s-2}s(s+1))$.
\end{description}
Indeed, once these three properties are shown we can apply the induction hypothesis to conclude that $\pair{\amodel_1}{\aworld} \gamerel{m,s-1}{\apropset} \pair{\amodel_1'}{\aworld'}$ and $\pair{\amodel_2}{\aworld} \gamerel{m,s-1}{\apropset} \pair{\amodel_2'}{\aworld'}$ and 
 therefore, the play described with the construction above leads to a winning strategy for the duplicator on the game
$\triple{\pair{\amodel}{\aworld}}{\pair{\amodel'}{\aworld'}}{\triple{m}{s}{\apropset}}$, i.e. $\pair{\amodel}{\aworld} \gamerel{m,s}{\apropset} \pair{\amodel'}{\aworld'}$.
The proof of  these three properties is quite easy (each case is similar to the others).
Let $i \in \{1,2\}$.
By using the definitions given during the construction of $\arelation_1'$ and $\arelation_2'$ it holds that
\begin{itemize}
\item $\arelation_i(\aworld)_{=0} = \arelation_i(\aworld)|_{0\blacktriangleright 0} \cup \arelation_i(\aworld)|_{1\blacktriangleright 0} \cup \arelation_i(\aworld)|_{2\blacktriangleright 0}$, and by definition  for all $j,k \in \interval{0}{2}$ such that $j\neq k$ it holds that
$\arelation_i(\aworld)|_{j\blacktriangleright 0} \cap \arelation_i(\aworld)|_{k\blacktriangleright 0} = \emptyset$.

\item $\arelation_i(\aworld)_{=1} = \arelation_i(\aworld)|_{1\blacktriangleright 1} \cup \arelation_i(\aworld)|_{2\blacktriangleright 1}$, and by definition
$\arelation_i(\aworld)|_{1\blacktriangleright 1} \cap \arelation_i(\aworld)|_{2\blacktriangleright 1} = \emptyset$.
\item $\arelation_i(\aworld)|_{=2} = \arelation_i(\aworld)|_{2 \blacktriangleright 2}$.
\end{itemize}
In what follows, we refer to these three properties with \rulelab{($\ddagger$)}{GML-m-SL:ddagger}.
\begin{description}
\item[proof of (\ref{GML-m-SL:zero}).]
By \ref{GML-m-SL:ddagger}, it holds that $\card{\arelation_i(\aworld)_{=0}} = \card{\arelation_i(\aworld)|_{0\blacktriangleright 0}} + \card{\arelation_i(\aworld)|_{1\blacktriangleright 0}} + \card{\arelation_i(\aworld)|_{2\blacktriangleright 0}}$. We divide the proof into two cases. For the first case, suppose
$\card{\arelation_i(\aworld)|_{0\blacktriangleright 0}} < 2^{s-1}$, $\card{\arelation_i(\aworld)|_{1\blacktriangleright 0}} < 2^{s-1}$ and $\card{\arelation_i(\aworld)|_{2\blacktriangleright 0}} < 2^{s-1}$. Then,
\begin{enumerate}
\item\label{GML-m-SL-zero3} $\card{Z_i} = \card{\arelation_i(\aworld)|_{0\blacktriangleright 0}}$ (by (\ref{GML-m-SL:Z1}) or (\ref{GML-m-SL:Z2}), depending on whether $i = 1$ or $i = 2$)
\item\label{GML-m-SL-zero4} $\card{Z_i'} = \card{\arelation_i(\aworld)|_{1\blacktriangleright 0}}$ (by (\ref{GML-m-SL:Z1p})/(\ref{GML-m-SL:Z2p}))
\item\label{GML-m-SL-zero5} $\card{Z_i''} = \card{\arelation_i(\aworld)|_{1\blacktriangleright 0}}$ (by (\ref{GML-m-SL:Z1pp})/(\ref{GML-m-SL:Z2pp}))
\item\label{GML-m-SL-zero6} $\card{\arelation_i'(\aworld')_{=0}} = \card{\arelation_i(\aworld)|_{0\blacktriangleright 0}} + \card{\arelation_i(\aworld)|_{1\blacktriangleright 0}} + \card{\arelation_i(\aworld)|_{1\blacktriangleright 0}}$
(from (\ref{GML-m-SL-zero3}), (\ref{GML-m-SL-zero4}) and (\ref{GML-m-SL-zero5}), by \ref{GML-m-SL:dagger})
\item $\card{\arelation_i'(\aworld')_{=0}} = \card{\arelation_i(\aworld)_{=0}}$ (from \ref{GML-m-SL-zero6}, by \ref{GML-m-SL:ddagger}).
\end{enumerate}
Otherwise, suppose that there is a set among $\arelation_i(\aworld)|_{0\blacktriangleright 0}$, $\arelation_i(\aworld)|_{1\blacktriangleright 0}$ and $\arelation_i(\aworld)|_{2\blacktriangleright 0}$ whose cardinality is at least $2^{s-1}$.
Then from (\ref{GML-m-SL:Z1})/(\ref{GML-m-SL:Z2}), (\ref{GML-m-SL:Z1p})/(\ref{GML-m-SL:Z2p}) or
(\ref{GML-m-SL:Z1pp})/(\ref{GML-m-SL:Z2pp}) (depending on whether $i = 1$ or $i = 2$ and on which set has at least $2^{s-1}$ elements) there is a set among $Z_{i}$, $Z_{i}'$ and $Z_{i}''$ that has cardinality $2^{s-1}$.
Then, by \ref{GML-m-SL:dagger} and \ref{GML-m-SL:ddagger} we have that $\arelation_i(\aworld)_{=0}$ and $\arelation_i'(\aworld')_{=0}$ have both more than $2^{s-1}$ elements.

\item[proof of (\ref{GML-m-SL:one}).]
By \ref{GML-m-SL:ddagger}, it holds that $\card{\arelation_i(\aworld)_{=1}} = \card{\arelation_i(\aworld)|_{1\blacktriangleright 1}} + \card{\arelation_i(\aworld)|_{2\blacktriangleright 1}}$. We divide the proof into two cases. First, suppose
$\card{\arelation_i(\aworld)|_{1\blacktriangleright 1}} < 2^{s-1}s$ and $\card{\arelation_i(\aworld)|_{2\blacktriangleright 1}} < 2^{s-1}s$.
Then,
\begin{enumerate}
\item\label{GML-m-SL-one3} $\card{O_i} = \card{\arelation_i(\aworld)|_{1\blacktriangleright 1}}$ (by (\ref{GML-m-SL:O1}) or (\ref{GML-m-SL:O2}), depending on whether $i = 1$ or $i = 2$)
\item\label{GML-m-SL-one4} $\card{O_i'} = \card{\arelation_i(\aworld)|_{2\blacktriangleright 1}}$ (by (\ref{GML-m-SL:O1p})/(\ref{GML-m-SL:O2p}))
\item\label{GML-m-SL-one6} $\card{\arelation_i'(\aworld')_{=1}} = \card{\arelation_i(\aworld)|_{1\blacktriangleright 1}} + \card{\arelation_i(\aworld)|_{2\blacktriangleright 1}}$
(from (\ref{GML-m-SL-one3}) and (\ref{GML-m-SL-one4}), by \ref{GML-m-SL:dagger})
\item $\card{\arelation_i'(\aworld')_{=1}} = \card{\arelation_i(\aworld)_{=1}}$ (from \ref{GML-m-SL-one6}, by \ref{GML-m-SL:ddagger}).
\end{enumerate}
Otherwise, suppose that there is a set among $\arelation_i(\aworld)|_{1\blacktriangleright 1}$ and $\arelation_i(\aworld)|_{2\blacktriangleright 1}$ whose cardinality is at least $2^{s-1}s$.
Then from (\ref{GML-m-SL:O1})/(\ref{GML-m-SL:O2}) or (\ref{GML-m-SL:O1p})/(\ref{GML-m-SL:O2p}) (depending on whether $i = 1$ or $i = 2$ and on which set has at least $2^{s-1}s$ elements) there is a set among $O_{i}$, $O_{i}'$
that has cardinality $2^{s-1}s$.
Then, by \ref{GML-m-SL:dagger} and \ref{GML-m-SL:ddagger} we have that $\arelation_i(\aworld)_{=1}$ and $\arelation_i'(\aworld')_{=1}$ have both more than $2^{s-1}s$ elements.

\item[proof of (\ref{GML-m-SL:two}).]~
By \ref{GML-m-SL:ddagger}, it holds that $\card{\arelation_i(\aworld)_{=2}} = \card{\arelation_i(\aworld)|_{2\blacktriangleright 2}}$.
Again we divide the proof into two cases. First, suppose
$\card{\arelation_i(\aworld)|_{2\blacktriangleright 2}} < 2^{s-2}s(s+1)$.
Then,
\begin{enumerate}
\item\label{GML-m-SL-two3} $\card{T_i} = \card{\arelation_i(\aworld)|_{2\blacktriangleright 2}}$ (by (\ref{GML-m-SL:T1}) or (\ref{GML-m-SL:T2}), depending on whether $i = 1$ or $i = 2$)
\item\label{GML-m-SL-two6} $\card{\arelation_i'(\aworld')_{=2}} = \card{\arelation_i(\aworld)|_{2\blacktriangleright 2}}$
(from (\ref{GML-m-SL-two3}), by \ref{GML-m-SL:dagger})
\item $\card{\arelation_i'(\aworld')_{=2}} = \card{\arelation_i(\aworld)_{=2}}$ (from \ref{GML-m-SL-two6}, by \ref{GML-m-SL:ddagger}).
\end{enumerate}
Otherwise, suppose that $\card{\arelation_i(\aworld)|_{2\blacktriangleright 2}}$, and hence $\card{\arelation_i(\aworld)_{=2}}$, is at least $2^{s-2}s(s+1)$.
Then,
\begin{enumerate}
\item\label{GML-m-SL-three3} $\card{T_i} \geq 2^{s-2}s(s+1)$ (by (\ref{GML-m-SL:T1})/(\ref{GML-m-SL:T2}))
\item\label{GML-m-SL-three6} $\card{\arelation_i'(\aworld')_{=2}} \geq 2^{s-2}s(s+1)$
(from (\ref{GML-m-SL-three3}), by \ref{GML-m-SL:dagger}).
\end{enumerate}
\end{description}

By relying on the (now proved) validity of \ref{lemmaprop:GML-more-SL-1}, we show the following crucial property.
\begin{nscenter}
$\boxed{\text{\rulelab{(B)}{lemmaprop:GML-more-SL-2}\qquad
$\begin{aligned}[t]&\text{Given a rank } \triple{m}{s}{\apropset}\text{ and two structures }\pair{\amodel = \triple{\worlds}{\arelation}{\avaluation}}{\aworld}\text{ and}\\
&\pair{\amodel' = \triple{\worlds'}{\arelation'}{\avaluation'}}{\aworld'} \text{ satisfying \ref{constr:GML-more-SL-1}, \ref{constr:GML-more-SL-2} and \ref{constr:GML-more-SL-3}, if}\\
&\ \text{\labelitemi\ } \card{\arelation(\aworld)_{=0}} \geq 2^s+1 \text{ and } \card{\arelation'(\aworld')_{=0}} \geq 2^s+1;\\
&\ \text{\labelitemi\ } \card{\arelation(\aworld)_{=1}} = 2 \text{ and } \card{\arelation'(\aworld')_{=1}} = 1;\\
&\ \text{\labelitemi\ } \card{\arelation(\aworld)_{=2}} \geq 2^{s-1}(s+1)(s+2)+1 \text{ and } \card{\arelation'(\aworld')_{=2}} \geq 2^{s-1}(s+1)(s+2)+1\\
&\text{then } \pair{\amodel}{\aworld} \gamerel{m,s}{\apropset} \pair{\amodel'}{\aworld'}
\end{aligned}
$}}$
\end{nscenter}
Notice that~\ref{lemmaprop:GML-more-SL-2} implies the statement of the lemma, as $\amodel,\aworld \models \Gdiamond{=2}\Gdiamond{=1}\true$ whereas $\amodel',\aworld' \not\models \Gdiamond{=2}\Gdiamond{=1}\true$. Indeed, ad absurdum suppose that such an $\modallogicSC$ formula $\aformula$ exists.
Let $m$ be its modal degree, $s$ be its maximal number of imbricated $\separate$ and $\apropset$ be the set of propositional variables
occurring in $\aformula$. Let us consider two pointed forests $\pair{\amodel_1}{\aworld_1}$ and $\pair{\amodel_2}{\aworld_2}$
such that $\amodel_1,\aworld_1 \models \Gdiamond{=2}\Gdiamond{=1}\true$, $\amodel_2,\aworld_2 \not \models \Gdiamond{=2}\Gdiamond{=1}\true$
and satisfying the conditions in~\ref{lemmaprop:GML-more-SL-2}. This would lead to a contradiction, as
$\pair{\amodel_1}{\aworld_1}$ and $\pair{\amodel_2}{\aworld_2}$ are supposed to satisfy  $\aformula$ (or not) equivalently.

The two finite forests of the statement are schematically represented below, with $\pair{\amodel}{\aworld}$ on the left and $\pair{\amodel'}{\aworld'}$ on the right.
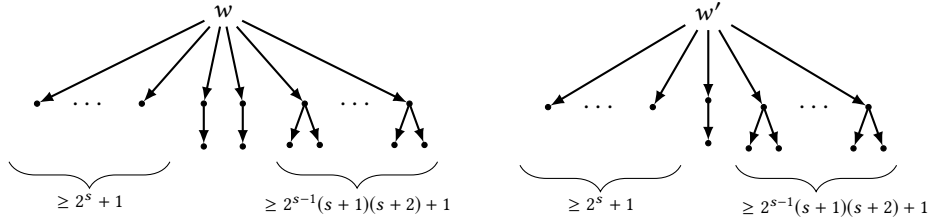
\begin{nscenter}
\begin{tikzpicture}[baseline=-1cm]
\node (w) {$\aworld$};
\node[dot] (w1) [below left = 1cm and 2.2cm of w] {};
\node (h1) [right = 0.3cm of w1] {$\dots$};
\node[dot] (w2) [right = 0.3cm of h1] {};

\node (h2) [below = 0.9cm of w] {};
\node[dot] (w3) [left = 0.1cm of h2] {};
\node[dot] (w4) [right = 0.1cm of h2] {};

\node[dot] (w5) [below right = 1cm and 2.2cm of w] {};
\node (h3) [left = 0.3cm of w5] {$\dots$};
\node[dot] (w6) [left = 0.3cm of h3] {};

\node[dot] (ww1) [below = 0.5cm of w3] {};
\node[dot] (ww2) [below = 0.5cm of w4] {};

\node[dot] (ww3) [below left = 0.5cm and 0.15cm of w5] {};
\node[dot] (ww4) [below right = 0.5cm and 0.15cm of w5] {};

\node[dot] (ww5) [below left = 0.5cm and 0.15cm of w6] {};
\node[dot] (ww6) [below right = 0.5cm and 0.15cm of w6] {};

\draw[pto] (w) -- (w1);
\draw[pto] (w) -- (w2);

\draw[pto] (w) -- (w3);
\draw[pto] (w) -- (w4);

\draw[pto] (w) -- (w5);
\draw[pto] (w) -- (w6);

\draw[pto] (w3) -- (ww1);
\draw[pto] (w4) -- (ww2);

\draw[pto] (w5) -- (ww3);
\draw[pto] (w5) -- (ww4);

\draw[pto] (w6) -- (ww5);
\draw[pto] (w6) -- (ww6);

\node (h4) [below left = 0.5cm and 0.1cm of w1] {};
\node (h5) [below right = 0.5cm and 0.1cm of w2] {};

\node (h8) [below left = 0.5cm and 0.1cm of w6] {};
\node (h9) [below right = 0.5cm and 0.1cm of w5] {};

\draw [decorate,decoration={brace,amplitude=10pt}]
(h5.south east) -- node[below = 0.3cm]
{\scriptsize{$\geq 2^s + 1$}} (h4.south west);

\draw [decorate,decoration={brace,amplitude=10pt}]
(h9.south east) -- node[below = 0.3cm]
{\scriptsize{$\geq 2^{s-1}(s+1)(s+2)+1$}} (h8.south west);
\end{tikzpicture}
\qquad
\begin{tikzpicture}[baseline=-1cm]
\node (w) {$\aworld'$};
\node[dot] (w1) [below left = 1cm and 1.8cm of w] {};
\node (h1) [right = 0.3cm of w1] {$\dots$};
\node[dot] (w2) [right = 0.3cm of h1] {};

\node[dot] (w3) [below = 0.9cm of w] {};

\node[dot] (w5) [below right = 1cm and 1.8cm of w] {};
\node (h3) [left = 0.3cm of w5] {$\dots$};
\node[dot] (w6) [left = 0.3cm of h3] {};

\node[dot] (ww1) [below = 0.5cm of w3] {};

\node[dot] (ww3) [below left = 0.5cm and 0.15cm of w5] {};
\node[dot] (ww4) [below right = 0.5cm and 0.15cm of w5] {};

\node[dot] (ww5) [below left = 0.5cm and 0.15cm of w6] {};
\node[dot] (ww6) [below right = 0.5cm and 0.15cm of w6] {};

\draw[pto] (w) -- (w1);
\draw[pto] (w) -- (w2);

\draw[pto] (w) -- (w3);

\draw[pto] (w) -- (w5);
\draw[pto] (w) -- (w6);

\draw[pto] (w3) -- (ww1);

\draw[pto] (w5) -- (ww3);
\draw[pto] (w5) -- (ww4);

\draw[pto] (w6) -- (ww5);
\draw[pto] (w6) -- (ww6);

\node (h4) [below left = 0.5cm and 0.1cm of w1] {};
\node (h5) [below right = 0.5cm and 0.1cm of w2] {};

\node (h8) [below left = 0.5cm and 0.1cm of w6] {};
\node (h9) [below right = 0.5cm and 0.1cm of w5] {};

\draw [decorate,decoration={brace,amplitude=10pt}]
(h5.south east) -- node[below = 0.3cm]
{\scriptsize{$\geq 2^s + 1$}} (h4.south west);

\draw [decorate,decoration={brace,amplitude=10pt}]
(h9.south east) -- node[below = 0.3cm]
{\scriptsize{\qquad$\geq 2^{s-1}(s+1)(s+2)+1$}} (h8.south west);
\end{tikzpicture}
\end{nscenter}

The proof of \ref{lemmaprop:GML-more-SL-2} is shown by cases on $m$, $s$ and on the moves done by the spoiler.
As in the proof of \ref{lemmaprop:GML-more-SL-1}, if  $m = 0$ then the duplicator has a winning strategy as after $s$ spatial moves the game will be in the state $\pair{\amodel_1}{\aworld}$ and $\pair{\amodel_1'}{\aworld'}$ (notice that $\aworld$ and $\aworld'$ do not change, since $m =0$) w.r.t.\ the rank $\triple{0}{0}{\apropset}$.
From \ref{constr:GML-more-SL-1}, we conclude that the duplicator wins.

Now, suppose $m \geq 1$ and  the spoiler decides to perform a modal move. Notice that, in particular, this case also takes care of the case where $s = 0$ and the spoiler is forced to play a modal move.
Moreover, suppose that the spoiler chooses $\pair{\amodel}{\aworld}$ (the case where it picks $\pair{\amodel'}{\aworld'}$ is analogous).
Then,
suppose that the spoiler chooses a world $\aworld_1 \in \arelation(\aworld)_{=n}$ for some $n \in \set{0,1,2}$.
It is then sufficient for the duplicator to choose $\aworld \in \arelation'(\aworld')_{=n}$ (which is a non-empty set by hypothesis) to guarantee him a victory, as the subtrees rooted in $\aworld_1$ and $\aworld_1'$ are isomorphic.

It remains to show the strategy for the duplicator when the spoiler decides to perform a spatial move (and therefore $s \geq 1$).
The proof distinguishes several cases depending on the structure choosen by the spoiler.
\begin{description}
\item[The spoiler picks $\pair{\amodel}{\aworld}$.] Notice that then the spoiler chooses the structure such that $\card{\arelation(\aworld)_{=1}} = 2$ and the duplicator has to reply in the structure $\pair{\amodel'}{\aworld'}$, where we recall that $\card{\arelation'(\aworld')_{=1}} = 1$. The idea is to make up for this discrepancy by using an element of $\arelation'(\aworld')_{=2}$. Let us see how.

For a moment, consider  the model obtained from $\amodel'$ by removing from $\arelation'$ exactly one pair $\pair{\aworld_1'}{\aworld_2'}$ where $\aworld_1'$ is a world of $\arelation'(\aworld')_{=2}$.
Formally, we are interested in a model $\widehat{\amodel'} = \triple{\worlds'}{\widehat{\arelation'}}{\avaluation'}$ such that $\widehat{\arelation'} = \arelation' \setminus \{\pair{\aworld_1'}{\aworld_2'}\}$ where $\pair{\aworld_1'}{\aworld_2'} \in \arelation'$ and $\aworld_1' \in \arelation'(\aworld')_{=2}$.
If the game was played on $\pair{\amodel}{\aworld}$ and $\pair{\widehat{\amodel'}}{\aworld'}$ w.r.t. $\triple{m}{s}{\apropset}$ then it is clear than the duplicator would have a winning strategy.
Indeed, both $\pair{\amodel}{\aworld}$ and $\pair{\widehat{\amodel'}}{\aworld'}$ satisfy \ref{constr:GML-more-SL-1}, \ref{constr:GML-more-SL-2} and \ref{constr:GML-more-SL-3}. Moreover,
\begin{itemize}
\item $\card{\arelation(\aworld)_{=0}}$ and $\card{\widehat{\arelation'}(\aworld')_{=0}}$ are both at least $2^s$. Notice that by definition $\widehat{\arelation'}(\aworld')_{=0} = \arelation'(\aworld')_{=0}$.
\item $\card{\arelation(\aworld)_{=1}} = 2$ and $\card{\widehat{\arelation'}(\aworld')_{=1}} = 2$. Here, by definition $\widehat{\arelation'}(\aworld')_{=1} = \arelation'(\aworld')_{=1} \cup \{\aworld_1'\}$.
\item $\card{\arelation(\aworld)_{=2}}$ and $\card{\widehat{\arelation'}(\aworld')_{=2}}$ are both at least $2^{s-1}(s+1)(s+2)$. Here, by 
definition $\widehat{\arelation'}(\aworld')_{=2} = \arelation'(\aworld')_{=2} \setminus \{\aworld_1'\}$.
\end{itemize}
These properties allow us to apply~\ref{lemmaprop:GML-more-SL-1} and conclude that $\pair{\amodel}{\aworld} \gamerel{m,s}{\apropset} \pair{\widehat{\amodel'}}{\aworld'}$.
In particular, in this game, if the spoiler picks $\pair{\amodel}{\aworld}$ and chooses $\amodel_1 = \triple{\worlds}{\arelation_1}{\avaluation}$ and $\amodel_2 = \triple{\worlds}{\arelation_2}{\avaluation}$ such that $\amodel_1 + \amodel_2 = \amodel$, then the duplicator can apply the strategy described in \ref{lemmaprop:GML-more-SL-1} in order to construct two structures
$\widehat{\amodel_1'} = \triple{\worlds'}{\widehat{\arelation_1'}}{\avaluation'}$ and $\widehat{\amodel_2'} = \triple{\worlds'}{\widehat{\arelation_2'}}{\avaluation'}$ such that $\widehat{\amodel_1'} + \widehat{\amodel_2'} = \widehat{\amodel'}$ and for every $i \in \{1,2\}$:
\begin{itemize}
\item $\min(\card{\arelation_i(\aworld)_{=0}},2^{s-1}) = \min(\card{\widehat{\arelation_i'}(\aworld')_{=0}},2^{s-1})$;
\item $\min(\card{\arelation_i(\aworld)_{=1}},2^{s-1}s) = \min(\card{\widehat{\arelation_i'}(\aworld')_{=1}},2^{s-1}s)$;
\item $\min(\card{\arelation_i(\aworld)_{=2}},2^{s-2}s(s+1)) = \min(\card{\widehat{\arelation_i'}(\aworld')_{=2}},2^{s-2}s(s+1))$.
\end{itemize}
Notice that these properties, which we later refer to with \rulelab{($\dagger\dagger$)}{GML-m-SL:twodagger} are exactly (\ref{GML-m-SL:zero}), (\ref{GML-m-SL:one}) and (\ref{GML-m-SL:two}) in the proof of \ref{lemmaprop:GML-more-SL-1}.

Let us see how to use these pieces of information to derive a strategy for the duplicator in the original game
$\triple{\pair{\amodel}{\aworld}}{\pair{\amodel'}{\aworld'}}{\triple{m}{s}{\apropset}}$.
As the spoiler chooses $\pair{\amodel}{\aworld}$, it selects $\amodel_1$ and $\amodel_2$ such that $\amodel_1 + \amodel_2 = \amodel$.
Consider the two structures $\widehat{\amodel_1'} = \triple{\worlds'}{\widehat{\arelation_1'}}{\avaluation'}$ and $\widehat{\amodel_2'}= \triple{\worlds'}{\widehat{\arelation_2'}}{\avaluation'}$ choosen by the duplicator following the strategy, discussed above, for the game
$\triple{\pair{\amodel}{\aworld}}{\pair{\widehat{\amodel'}}{\aworld'}}{\triple{m}{s}{\apropset}}$ in the case when the spoiler chooses
$\pair{\amodel}{\aworld}$ and again selects $\amodel_1$ and $\amodel_2$. In particular these structures satisfy \ref{GML-m-SL:twodagger}.
Moreover, the two forests $\widehat{\amodel_1'}$ and $\widehat{\amodel_2'}$ are such that $\widehat{\amodel_1'} + \widehat{\amodel_2'} = \widehat{\amodel}$ and therefore $\widehat{\arelation_1'} \cup \widehat{\arelation_2'} = \widehat{\arelation'} = \arelation' \setminus \{\pair{\aworld_1'}{\aworld_2'}\}$
where $\pair{\aworld_1'}{\aworld_2'} \in \arelation'$ and $\aworld_1' \in \arelation'(\aworld')_{=2}$.
We distinguish two cases.
\begin{itemize}
\item If $\aworld_1' \in \widehat{\arelation_1'}(\aworld')$ then in the original game  $\triple{\pair{\amodel}{\aworld}}{\pair{\amodel'}{\aworld'}}{\triple{m}{s}{\apropset}}$,
the duplicator replies to $\amodel_1$ and $\amodel_2$ with the two forests $\amodel_1' = \triple{\worlds'}{\arelation_1'}{\avaluation'}$ and $\amodel_2' = \triple{\worlds'}{\arelation_2'}{\avaluation'}$ such that $\arelation_1' = \widehat{\arelation_1'}$ and
$\arelation_2' = \widehat{\arelation_2'} \cup \{\pair{\aworld_1'}{\aworld_2'}\}$.
\item Otherwise $\aworld_1' \in \widehat{\arelation_2'}(\aworld')$ and in the game  $\triple{\pair{\amodel}{\aworld}}{\pair{\amodel'}{\aworld'}}{\triple{m}{s}{\apropset}}$
the duplicator replies to $\amodel_1$ and $\amodel_2$ with the two forests $\amodel_1' = \triple{\worlds'}{\arelation_1'}{\avaluation'}$ and $\amodel_2' = \triple{\worlds'}{\arelation_2'}{\avaluation'}$ such that
$\arelation_1' = \widehat{\arelation_1'} \cup \{\pair{\aworld_1'}{\aworld_2'}\}$ and $\arelation_2' = \widehat{\arelation_2'}$.
\end{itemize}
In both cases, as the pair $\pair{\aworld'}{\aworld_1'}$ is in one relation between $\arelation_1'$ and $\arelation_2'$ whereas $\pair{\aworld_1'}{\aworld_2'}$ is in the other relation, the world $\aworld_1'$ effectively behaves like if it was a member of the set $\arelation'(\aworld')_{=1}$ instead of $\arelation'(\aworld')_{=2}$, exactly as in the case of $\widehat{\arelation'}$. In particular, it is easy to see that for $i \in \{1,2\}$:
\begin{nscenter}
  \hfill $\card{\arelation_i'(\aworld')_{=0}} = \card{\widehat{\arelation_i'}(\aworld')_{=0}}$
  \hfill $\card{\arelation_i'(\aworld')_{=1}} = \card{\widehat{\arelation_i'}(\aworld')_{=1}}$
  \hfill $\card{\arelation_i'(\aworld')_{=2}} = \card{\widehat{\arelation_i'}(\aworld')_{=2}}$
\end{nscenter}
Hence, by \ref{GML-m-SL:twodagger} we have that
\begin{itemize}
\item $\min(\card{\arelation_i(\aworld)_{=0}},2^{s-1}) = \min(\card{\arelation_i'(\aworld')_{=0}},2^{s-1})$;
\item $\min(\card{\arelation_i(\aworld)_{=1}},2^{s-1}s) = \min(\card{\arelation_i'(\aworld')_{=1}},2^{s-1}s)$;
\item $\min(\card{\arelation_i(\aworld)_{=2}},2^{s-2}s(s+1)) = \min(\card{\arelation_i'(\aworld')_{=2}},2^{s-2}s(s+1))$.
\end{itemize}
Moreover, $\amodel_1$, $\amodel_2$, $\amodel_1'$ and $\amodel_2'$
all satisfy \ref{constr:GML-more-SL-1}, \ref{constr:GML-more-SL-2} and \ref{constr:GML-more-SL-3} (as they are submodels of $\amodel$ or $\amodel'$), we can apply \ref{lemmaprop:GML-more-SL-1} and conclude that
$\pair{\amodel_1}{\aworld} \gamerel{m,s-1}{\apropset} \pair{\amodel_1'}{\aworld'}$ and $\pair{\amodel_2}{\aworld} \gamerel{m,s-1}{\apropset} \pair{\amodel_2'}{\aworld'}$.
Therefore, the play we just described leads to a winning strategy for the duplicator on the game
$\triple{\pair{\amodel}{\aworld}}{\pair{\amodel'}{\aworld'}}{\triple{m}{s}{\apropset}}$, under the hypothesis that the spoiler chooses $\pair{\amodel}{\aworld}$.

\item[The spoiler picks $\pair{\amodel'}{\aworld'}$.]
Then, the spoiler chooses the structure such that $\card{\arelation'(\aworld')_{=1}} = 1$ and the duplicator has to reply in the structure $\pair{\amodel}{\aworld}$ where $\card{\arelation(\aworld)_{=1}} = 2$.
The proof is very similar to the previous case, but instead of choosing an element of $\arelation'(\aworld')_{=2}$ to make up for the discrepancy between $\card{\arelation(\aworld)_{=1}}$ and $\card{\arelation'(\aworld')_{=1}}$, the duplicator manipulates the additional element in $\arelation(\aworld)_{=1}$ so that it becomes a member of $\arelation_1(\aworld)_{=0}$ or $\arelation_2(\aworld)_{=0}$.
Let us formalise this strategy.

For a moment, consider  the model obtained from $\amodel$ by removing from $\arelation$ exactly one pair $\pair{\aworld_1}{\aworld_2}$ where $\aworld_1$ is a world of $\arelation(\aworld)_{=1}$.
Formally, we are interested in a model $\widehat{\amodel} = \triple{\worlds}{\widehat{\arelation}}{\avaluation}$ such that $\widehat{\arelation} = \arelation \setminus \{\pair{\aworld_1}{\aworld_2}\}$ where $\pair{\aworld_1}{\aworld_2} \in \arelation$ and $\aworld_1 \in \arelation(\aworld)_{=1}$.
If the game was played on $\pair{\widehat{\amodel}}{\aworld}$ and $\pair{\amodel'}{\aworld'}$ w.r.t. $\triple{m}{s}{\apropset}$ then it is clear than the duplicator would have a winning strategy.
Indeed, both $\pair{\widehat{\amodel}}{\aworld}$ and $\pair{\amodel'}{\aworld'}$ satisfy \ref{constr:GML-more-SL-1}, \ref{constr:GML-more-SL-2} and \ref{constr:GML-more-SL-3}. Moreover,
\begin{itemize}
\item $\card{\widehat{\arelation}(\aworld)_{=0}}$ and $\card{\arelation'(\aworld')_{=0}}$ are both at least $2^s$. Here, by definition,
$\widehat{\arelation}(\aworld)_{=0} = \arelation(\aworld)_{=0} \cup \{\aworld_1\}$.
\item $\card{\widehat{\arelation}(\aworld)_{=1}} = 1$ and $\card{\arelation'(\aworld')_{=1}} = 1$. Here, by definition $\widehat{\arelation}(\aworld)_{=1} = \arelation(\aworld)_{=1} \setminus \{\aworld_1\}$.
\item $\card{\widehat{\arelation}(\aworld)_{=2}}$ and $\card{\arelation'(\aworld')_{=2}}$ are both at least $2^{s-1}(s+1)(s+2)$. Here, by definiton $\widehat{\arelation}(\aworld)_{=2} = \arelation(\aworld)_{=2}$.
\end{itemize}
These properties allow us to apply \ref{lemmaprop:GML-more-SL-1} and conclude that $\pair{\widehat{\amodel}}{\aworld} \gamerel{m,s}{\apropset} \pair{\amodel'}{\aworld'}$.
In particular, in this game, if the spoiler picks $\pair{\amodel'}{\aworld'}$ and chooses $\amodel_1' = \triple{\worlds'}{\arelation_1'}{\avaluation'}$ and $\amodel_2' = \triple{\worlds'}{\arelation_2'}{\avaluation'}$ such that $\amodel_1' + \amodel_2' = \amodel'$, then the
duplicator can apply the strategy described in \ref{lemmaprop:GML-more-SL-1}.
Two structures
$\widehat{\amodel_1} = \triple{\worlds}{\widehat{\arelation_1}}{\avaluation}$ and $\widehat{\amodel_2} = \triple{\worlds}{\widehat{\arelation_2}}{\avaluation}$ are constructed such that $\widehat{\amodel_1} + \widehat{\amodel_2} = \widehat{\amodel}$ and for every $i \in \{1,2\}$:
\begin{itemize}
\item $\min(\card{\widehat{\arelation_i}(\aworld)_{=0}},2^{s-1}) = \min(\card{\arelation_i'(\aworld')_{=0}},2^{s-1})$;
\item $\min(\card{\widehat{\arelation_i}(\aworld)_{=1}},2^{s-1}s) = \min(\card{\arelation_i'(\aworld')_{=1}},2^{s-1}s)$;
\item $\min(\card{\widehat{\arelation_i}(\aworld)_{=2}},2^{s-2}s(s+1)) = \min(\card{\arelation_i'(\aworld')_{=2}},2^{s-2}s(s+1))$.
\end{itemize}
Again, notice that these properties, which we later refer to with \desclabel{($\ddagger\ddagger$)}{GML-m-SL:twoddagger}, are exactly (\ref{GML-m-SL:zero}), (\ref{GML-m-SL:one}) and (\ref{GML-m-SL:two}) in the proof of \ref{lemmaprop:GML-more-SL-1}.
Let us see how to use these pieces of information to derive a strategy for the duplicator in the original game
$\triple{\pair{\amodel}{\aworld}}{\pair{\amodel'}{\aworld'}}{\triple{m}{s}{\apropset}}$.
As the spoiler chooses $\pair{\amodel'}{\aworld'}$, it selects $\amodel_1'$ and $\amodel_2'$ such that $\amodel_1' + \amodel_2' = \amodel'$.
Consider the two structures $\widehat{\amodel_1} = \triple{\worlds}{\widehat{\arelation_1}}{\avaluation}$ and $\widehat{\amodel_2}= \triple{\worlds}{\widehat{\arelation_2}}{\avaluation}$ choosen by the duplicator following the strategy, discussed above, for the game
$\triple{\pair{\widehat{\amodel}}{\aworld}}{\pair{\amodel'}{\aworld'}}{\triple{m}{s}{\apropset}}$ in the case when the spoiler chooses
$\pair{\amodel'}{\aworld'}$ and again select $\amodel_1'$ and $\amodel_2'$.
In particular these structures satisfy \ref{GML-m-SL:twoddagger}.
Moreover, the two forests $\widehat{\amodel_1}$ and $\widehat{\amodel_2}$ are such that $\widehat{\amodel_1} + \widehat{\amodel_2} = \widehat{\amodel}$ and therefore $\widehat{\arelation_1} \cup \widehat{\arelation_2} = \widehat{\arelation} = \arelation \setminus \{\pair{\aworld_1}{\aworld_2}\}$
where $\pair{\aworld_1}{\aworld_2} \in \arelation$ and $\aworld_1 \in \arelation(\aworld)_{=1}$.
We distinguish two cases.
\begin{itemize}
\item If $\aworld_1 \in \widehat{\arelation_1}(\aworld)$ then in the original game  $\triple{\pair{\amodel}{\aworld}}{\pair{\amodel'}{\aworld'}}{\triple{m}{s}{\apropset}}$,
the duplicator replies to $\amodel_1'$ and $\amodel_2'$ with the two structures $\amodel_1 = \triple{\worlds}{\arelation_1}{\avaluation}$ and $\amodel_2 = \triple{\worlds}{\arelation_2}{\avaluation}$ such that $\arelation_1 = \widehat{\arelation_1}$ and
$\arelation_2 = \widehat{\arelation_2} \cup \{\pair{\aworld_1}{\aworld_2}\}$.
\item Otherwise $\aworld_1 \in \widehat{\arelation_2}(\aworld)$ and in the game  $\triple{\pair{\amodel}{\aworld}}{\pair{\amodel'}{\aworld'}}{\triple{m}{s}{\apropset}}$
the duplicator replies to $\amodel_1'$ and $\amodel_2'$ with the two structures $\amodel_1 = \triple{\worlds}{\arelation_1}{\avaluation}$ and $\amodel_2 = \triple{\worlds}{\arelation_2}{\avaluation}$ such that
$\arelation_1 = \widehat{\arelation_1} \cup \{\pair{\aworld_1}{\aworld_2}\}$ and $\arelation_2 = \widehat{\arelation_2}$.
\end{itemize}
In both cases, as the pair $\pair{\aworld}{\aworld_1}$ is in one relation between $\arelation_1$ and $\arelation_2'$ whereas $\pair{\aworld_1}{\aworld_2}$ is in the other relation, the world $\aworld_1$ effectively behaves as if it was a member of the set $\arelation(\aworld)_{=0}$ instead of $\arelation(\aworld)_{=1}$, exactly as in the case of $\widehat{\arelation'}$. In particular, it is easy to see that for $i \in \{1,2\}$:
\begin{nscenter}
  \hfill $\card{\arelation_i(\aworld)_{=0}} = \card{\widehat{\arelation_i}(\aworld)_{=0}}$
  \hfill $\card{\arelation_i(\aworld)_{=1}} = \card{\widehat{\arelation_i}(\aworld)_{=1}}$
  \hfill $\card{\arelation_i(\aworld)_{=2}} = \card{\widehat{\arelation_i}(\aworld)_{=2}}$
\end{nscenter}
Hence, by \ref{GML-m-SL:twoddagger} we have
\begin{itemize}
\item $\min(\card{\arelation_i(\aworld)_{=0}},2^{s-1}) = \min(\card{\arelation_i'(\aworld')_{=0}},2^{s-1})$;
\item $\min(\card{\arelation_i(\aworld)_{=1}},2^{s-1}s) = \min(\card{\arelation_i'(\aworld')_{=1}},2^{s-1}s)$;
\item $\min(\card{\arelation_i(\aworld)_{=2}},2^{s-2}s(s+1)) = \min(\card{\arelation_i'(\aworld')_{=2}},2^{s-2}s(s+1))$.
\end{itemize}
Moreover, $\amodel_1$, $\amodel_2$, $\amodel_1'$ and $\amodel_2'$
all satisfy \ref{constr:GML-more-SL-1}, \ref{constr:GML-more-SL-2} and \ref{constr:GML-more-SL-3} (as they are submodels of $\amodel$ or $\amodel'$), we can apply \ref{lemmaprop:GML-more-SL-1} and conclude that
$\pair{\amodel_1}{\aworld} \gamerel{m,s-1}{\apropset} \pair{\amodel_1'}{\aworld'}$ and $\pair{\amodel_2}{\aworld} \gamerel{m,s-1}{\apropset} \pair{\amodel_2'}{\aworld'}$.
Therefore, the play we just described leads to a winning strategy for the duplicator on the game
$\triple{\pair{\amodel}{\aworld}}{\pair{\amodel'}{\aworld'}}{\triple{m}{s}{\apropset}}$, under the hypothesis that the spoiler chooses $\pair{\amodel'}{\aworld'}$.
\end{description}
As we constructed a strategy for the duplicator in both  cases where the spoiler picks $\pair{\amodel}{\aworld}$ and $\pair{\amodel'}{\aworld'}$, we have that $\pair{\amodel}{\aworld} \gamerel{m,s}{\apropset} \pair{\amodel'}{\aworld'}$ and therefore \ref{lemmaprop:GML-more-SL-2} holds. This implies that the class of models satisfying
$\Gdiamond{=2}\Gdiamond{=1}\true$ cannot be characterised by a formula in \modallogicSC.
\end{proof}

\section{Proofs of \Cref{section-other-logics}}

\subsection{Definitions and Proofs of \Cref{section-SAL} (Static Ambient Logic)}
\label{appendix:SAL}

In this part of the appendix,  we provide equisatisfiability preserving translations from \SAL to \modallogicCC, and from \modallogicCC to \SAL.
Since the translations are in polynomial-time and in~\Cref{section-aexppol} we have shown that 
\satproblem{\modallogicCC}  is \aexppol-complete, this entails that the complexity of the satisfiability problem for \SAL is also \aexppol-complete.
In the body of the paper, these results are shown with respect to Kripke-like structures that can be shown isomorphic to the syntactical trees historically used in ambient calculus. Here, we provide the reductions directly on these syntactical trees. Let us start by introducing \SAL.

Let $\aalphabet$ be a countably infinite set of \defstyle{ambient names}. The formulae of \SAL are built from:
\begin{nscenter}
$
  \aformula :=\ \true \ \mid\ \zero\ \mid\ \aname[\aformula]\ \mid\ \aformula \land \aformula\ \mid \lnot\aformula\ \mid\ \aformula \ambientchop \aformula,
$
\end{nscenter}
where $\aname \in \aalphabet$.
$\SAL$ is interpreted on edge-labelled finite trees: syntactical objects equipped with a structural equivalence relation $\equiv$. We denote with $\ambienttrees$ the set of these finite trees. The grammar used to construct these structures, their structural equivalence as well as the satisfaction predicate $\models$ for \SAL are provided in
Figure~\ref{figure:SALsemantics} (the cases for $\land$ and $\lnot$ being omitted).

\begin{figure}[!h]
  \hfill
  \footnotesize
  \begin{minipage}[t]{0.6\linewidth}
  \centering
  \textsc{\vphantom{Structural equivalence}Trees}
  {\small{

  $
  \underbracket[0.7pt]{
  \overbracket[0.7pt]{
  \text{
  \begin{minipage}{0.97\linewidth}
  \begin{nscenter}
  $
  T :=\ \zero \ \mid \ \aname[T] \ \mid \ T\ambientchop T
  $
  \end{nscenter}
  \end{minipage}
  }}}
  $
  }}

  \textsc{\vphantom{Structural equivalence}Semantics}
  {\small{

  $
  \underbracket[0.7pt]{
  \overbracket[0.7pt]{
  \text{
  \begin{minipage}{0.97\linewidth}
  \begin{nscenter}
  $
  \begin{aligned}[t]
  &T \models \true &&&&\text{always holds}\\
  &T \models \zero &&\text{iff}&& T \equiv \zero\\
  &T \models \aname[\aformula] &&\text{iff}&& \exists T'\text{ s.t.\ } T \equiv \aname[T']\text{ and }T' \models \aformula\\
  &T \models \aformula \ambientchop \aformulabis &&\text{iff}&& \exists T_1, T_2\text{ s.t.\ }T \equiv T_1 \ambientchop T_2, T_1 \models \aformula\text{ and }T_2 \models \aformulabis
  \end{aligned}
  $
  \end{nscenter}
  \end{minipage}
  }}}
  $
  }}

  \end{minipage}%
  \hfill
  \begin{minipage}[t]{0.33\linewidth}
  \centering
  \textsc{\vphantom{Trees}Structural equivalence}
  {\small{

  $
  \underbracket[0.7pt]{
  \overbracket[0.7pt]{
  \text{
  \begin{minipage}{0.97\linewidth}
  \begin{itemize}
  \setlength\itemsep{1.81pt}
  \item $T \ambientchop \zero \equiv T$
  \item $T_1 \equiv T_2$ $\implies$ $T_2 \equiv T_1$
  \item $T_1 \equiv T_2$, $T_2 \equiv T_3$ \ $\implies$ \ $T_1 \equiv T_3$
  \item $T_1 \ambientchop T_2 \equiv T_2 \ambientchop T_1$
  \item $(T_1 \ambientchop T_2) \ambientchop T_3 \equiv T_1 \ambientchop (T_2 \ambientchop T_3)$
  \item $T_1 \equiv T_2$ \ $\implies$ \ $T_1 \ambientchop T \equiv T_2 \ambientchop T$
  \item $T_1 \equiv T_2$ \ $\implies$ \ $\aname[T_1] \equiv \aname[T_2]$
  \end{itemize}
  \end{minipage}
  }}}
  $
  }}

  \end{minipage}
  \hfill~
  \caption{Interpretation and semantics of \SAL.}\label{figure:SALsemantics}
  \vspace{-5pt}
  \end{figure}
Obviously  \SAL and \modallogicCC  are strongly related, but how close?
For example,
$\aname[\aformula] \ambientchop \true$ can be seen as a relativised version of
$\Diamond$ of the form $\Diamond(\aname \wedge \aformula)$.
To
formalise this intuition, we borrow the syntax from
\HML~\cite{HennessyM80}
and define the formulae $\HMDiamond{\aname} \aformula \egdef \aname[\aformula] \ambientchop \true$ and
its dual $\HMBox{\aname} \aformula \egdef \lnot \HMDiamond{\aname} \lnot \aformula$.
Below, w.l.o.g. we assume  $\aalphabet=\varprop$ (for the sake of clarity).
%

\subparagraph{From \satproblem{\SAL} to \satproblem{\modallogicCC}.}
This reduction is also quite simple as
 \SAL is essentially interpreted on finite trees where each world
satisfies a single propositional variable (its
 ambient
name).
Let $T \in \ambienttrees$ be a tree built with ambient names from $\apropset {\subseteq_\fin} \varprop$, $\amodel = \triple{\auniverse}{\aaccessrelation}{\apropeval}$ be a finite forest and $\aworld \in \auniverse$. We say that $\pair{\amodel}{\aworld}$ \emph{encodes} $T$ iff:
\begin{enumerate}
\itemsep 0 cm
\item\label{mencodest-1} every $\aworld' \in \arelation^*(\aworld)$ satisfies at most one symbol in $\apropset$;
\item\label{mencodest-2} there is
$\amap$ : $\auniverse \to \ambienttrees$ such that $\amap(\aworld) \equiv T$ and
for all $\aworld' \in \aaccessrelation^*(\aworld)$,  we have
$\amap(\aworld') \equiv \sum_{i \in \interval{1}{K}}  \aname_i[\amap(\aworld_i)]$
where
$\{\aworld_1$, \dots, $\aworld_K\} = \aaccessrelation(\aworld')$
 and $\forall i \in \interval{1}{K}$, $\aworld_i \in \apropeval(\aname_i)$ 
(given  $I = \{i_1,\dots,i_m\}$, $\sum_{i \in I} T_i \egdef T_{i_1} \ambientchop T_{i_2} \ambientchop \dots \ambientchop T_{i_{m}}$).
\end{enumerate}

It is easy to verify that every tree in $\ambienttrees$ has an encoding.
The figure just below depicts a tree $T$ (on the left) and one of its possible encodings as a finite forest (on the right).
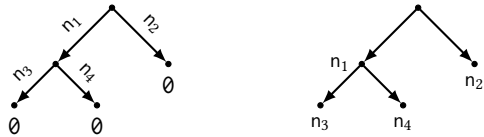
\begin{nscenter}
  \begin{tikzpicture}
  \node[dot] (w) [ ]{ };

  \node[dot] (w1) [below left = 0.7cm and 0.7cm of w] {};
  \node[dot] (w2) [below right  = 0.7cm and 0.7cm of w,label=below:$\zero$] {};

  \node[dot] (w3) [below left = 0.5cm and 0.5cm of w1,label=below:$\zero$] {};
  \node[dot] (w4) [below right  = 0.5cm and 0.5cm of w1,label=below:$\zero$] {};

  \draw[pto] (w) -- node[sloped, anchor=center,above] {\footnotesize$\aname_1$} (w1);
  \draw[pto] (w) -- node[sloped, anchor=center,above] {\footnotesize$\aname_2$} (w2);

  \draw[pto] (w1) -- node[sloped, anchor=center,above] {\footnotesize$\aname_3$} (w3);
  \draw[pto] (w1) -- node[sloped, anchor=center,above] {\footnotesize$\aname_4$} (w4);

 \node[dot] (ww) [right = 4cm of w] { };

  \node[dot] (ww1) [below left = 0.7cm and 0.7cm of ww,label=left:\footnotesize$\aname_1$] {};
  \node[dot] (ww2) [below right  = 0.7cm and 0.7cm of ww,label=below:\footnotesize$\aname_2$] {};

  \node[dot] (ww3) [below left = 0.5cm and 0.5cm of ww1,label=below:\footnotesize$\aname_3$] {};
  \node[dot] (ww4) [below right  = 0.5cm and 0.5cm of ww1,label=below:\footnotesize$\aname_4$] {};

  \draw[pto] (ww) -- (ww1);
  \draw[pto] (ww) -- (ww2);

  \draw[pto] (ww1) -- (ww3);
  \draw[pto] (ww1) -- (ww4);

  \end{tikzpicture} 
\end{nscenter}
\begin{lemma}\label{lemma:tech-lemma-2-1}
Every tree in $\ambienttrees$ has an encoding.
\end{lemma}
%


\begin{proof}
Let $T \in \ambienttrees$.
Let $m$ be the number of ambients in $T$, i.e.\ the number of occurrences of the $\aname[T']$ constructor in $T$.
Let $\auniverse$ be a set of $m+1$ elements.
We fix a total order $<$ with least element $0$ on $\auniverse$.
Then $\aSALmodeltranslationaux{\triple{\auniverse}{\emptyset}{\emptyset}}{T}{0}{<}$ is an encoding of $T$, where
\begin{itemize}
\item $\aSALmodeltranslationaux{\triple{\auniverse}{\aaccessrelation}{\apropeval}}{\zero}{\aworld}{<} = \triple{\auniverse}{\aaccessrelation}{\apropeval}$;

\item $\aSALmodeltranslationaux{\triple{\auniverse}{\aaccessrelation}{\apropeval}}{T_1 \ambientchop T_2}{\aworld}{<} = \aSALmodeltranslationaux{\aSALmodeltranslationaux{\triple{\auniverse}{\aaccessrelation}{\apropeval}}{T_1}{\aworld}{<}}{T_2}{\aworld}{<}$;

\item $\aSALmodeltranslationaux{\triple{\auniverse}{\aaccessrelation}{\apropeval}}{ \aname[T]}{\aworld}{<} = \aSALmodeltranslationaux{\triple{\auniverse}{\aaccessrelation'}{\apropeval'}}{T}{\aworld'}{<}$ where
\begin{itemize}
\item $\aworld' = \min_{<}\{\aworld'' \mid \aworld'' \not \in \pi_1(\aaccessrelation) \cup \pi_2(\aaccessrelation) \cup \{\aworld\}\}$;
\item $\aaccessrelation' = \aaccessrelation \cup \{\pair{\aworld}{\aworld'}\}$;
\item $\apropeval' = \lambda \avarprop .
  \begin{cases}
    \apropeval(\avarprop) \cup \{ \aworld' \} &\text{if}\ \avarprop = \aname\\
    \apropeval(\avarprop) &\text{otherwise}
  \end{cases}$
\end{itemize}
\end{itemize}

It remains to verify that $\aSALmodeltranslationaux{\triple{\auniverse}{\emptyset}{\emptyset}}{T}{0}{<}$ is an encoding of $T$. Condition~\ref{mencodest-1} is obvious, since each ambient name corresponds to 
a different state in $\aSALmodeltranslationaux{\triple{\auniverse}{\emptyset}{\emptyset}}{T}{0}{<}$. For condition~\ref{mencodest-2}, we need to check that there is a map $\amap$  such that $\amap(\aworld) \equiv T$ and for every $\aworld' \in \auniverse$ reachable from $\aworld$ (i.e. $\pair{\aworld}{\aworld'} \in \aaccessrelation^*$) it holds that
$\amap(\aworld') \equiv \sum_{i \in \interval{1}{K}}  \aname_i[\amap(\aworld_i)]$,
where
$\{\aworld_1$, \dots, $\aworld_K\} = \{\aworld'' \mid \pair{\aworld'}{\aworld''} \in \aaccessrelation\}$
is the set of distinct worlds accessible from $\aworld'$,
 and for every $i \in \interval{1}{K}$, $\aworld_i \in \apropeval(\aname_i)$. Take the mapping that assigns $\amap(\aworld)\equiv T$, 
and if $T\equiv \aname[T_1]\ambientchop T_2$, 
$(\aworld,\aworld')\in\arelation$ and $w'\in\avaluation(\aname)$, then $\amap(\aworld')\equiv T_1$. 
One can easily show that $\amap$ validates condition~\ref{mencodest-2}.
\end{proof}

As done in the previous section, we now state two intermediate lemmata that will be helpful to prove the correctness of the forthcoming translation (Lemma~\ref{lemma:correct-SAL-MLCC}).

\begin{lemma}\label{lemma:tech-lemma-2-2}
Let $T \in \ambienttrees$ and $\pair{\amodel}{\aworld}$ be an encoding of $T$. Let $\amap$ be a witness of this encoding.
For every $\aworld'$ accessible from $\aworld$ it holds that $\pair{\amodel}{\aworld'}$ encodes $\amap(\aworld')$.
\end{lemma}



\begin{proof}
It trivially follows from the definition of encoding of a tree in $\ambienttrees$. Moreover, for every world $\aworld'$ accessible from 
$\aworld$, the function $\amap$ is also the witness of the encoding of $\amap(\aworld')$ in $\pair{\amodel}{\aworld'}$.
\end{proof}

\begin{lemma}\label{lemma:tech-lemma-2-3}
Let $T$ be a \SAL-tree and $\pair{\amodel}{\aworld}$ an encoding of $T$.
Then,
\begin{enumerate}
\item for every $T_1$ and $T_2$ such that $T \equiv T_1 \ambientchop T_2$ there are $\amodel_1$ and $\amodel_2$ such that $\amodel = \amodel_1 +_{\aworld} \amodel_2$,
$\pair{\amodel_1}{\aworld}$ is an encoding of $T_1$, and
$\pair{\amodel_2}{\aworld}$ is an encoding of $T_2$.
\item For every  $\amodel_1$ and $\amodel_2$ such that $\amodel = \amodel_1 +_{\aworld} \amodel_2$ there are $T_1$ and $T_2$ such that $T \equiv T_1 \ambientchop T_2$,
$\pair{\amodel_1}{\aworld}$ is an encoding of $T_1$, and
$\pair{\amodel_2}{\aworld}$ is an encoding of $T_2$.
\end{enumerate}
\end{lemma}



\begin{proof}
In the proof of both points, let $\amodel$ (as in the statement) be a model $\triple{\auniverse}{\aaccessrelation}{\apropeval}$. Moreover, let $\amap$ be the witness of the encoding of $T$ in $\pair{\amodel}{\aworld}$.
\begin{enumerate}
\item Suppose $T_1$ and $T_2$ such that $T \equiv T_1 \ambientchop T_2$.
Let $R_{\aworld} = \{\aworld_1,\dots,\aworld_K\} \egdef \{\aworld' \mid \pair{\aworld}{\aworld'} \in \aaccessrelation\}$ be the set of worlds that are accessible from $\aworld$ (notice that this set could be empty).
By definition of $\amap$, we have $\amap(\aworld) \equiv T \equiv \sum_{i \in \interval{1}{K}} \aname_i[\amap(\aworld_i)]$ for some 
$\aname_1,\dots,\aname_k \in \aalphabet$. Notice that if $R_{\aworld}$ is empty then $\sum_{i \in \interval{1}{K}} \aname_i[\amap(\aworld_i)]$ 
is equivalent to the empty tree $\zero$.
Following $\sum_{i \in \interval{1}{K}} \aname_i[\amap(\aworld_i)] \equiv T \equiv T_1 \ambientchop T_2$, we know that
we can partition $R_{\aworld}$ into two sets $R_{\aworld}' = \{\aworld_{i_1},\dots,\aworld_{i_c}\}$ and $R_{\aworld}'' = \{\aworld_{i_{c+1}},\dots,
\aworld_{i_K}\}$ ($c \in \interval{0}{K}$) such that
\begin{itemize}
\item $T_1 \equiv \sum_{j \in \interval{1}{c}} \aname_{i_j}[\amap(\aworld_{i_j})]$;
\item $T_2 \equiv \sum_{j \in \interval{c+1}{K}} \aname_{i_j}[\amap(\aworld_{i_j})]$.
\end{itemize}
By relying on the partitioning of $R_{\aworld}$ into $R_{\aworld}'$ and $R_{\aworld}''$ it is easy to show that we can derive two finite forests $\amodel_1 = \triple{\auniverse}{\aaccessrelation_1}{\apropeval}$ and $\amodel_2 = \triple{\auniverse}{\aaccessrelation_2}{\apropeval}$ such that
\begin{itemize}
\item $\amodel = \amodel_1 +_{\aworld} \amodel_2$;
\item every $\aworld' \in R_{\aworld}'$ is accessible from $\aworld$ in $\aaccessrelation_1$, i.e.\ $\pair{\aworld}{\aworld'} \in \aaccessrelation_1$;
\item every $\aworld'' \in R_{\aworld}''$ is accessible from $\aworld$ in $\aaccessrelation_2$, i.e.\ $\pair{\aworld}{\aworld''} \in \aaccessrelation_2$.
\end{itemize}
Concretely, by defining
$\aaccessrelation_1 \egdef \{\pair{\aworld'}{\aworld''} \in \aaccessrelation \mid \text{ there is } \aworld''' \in R_{\aworld}' \text{ such that } \pair{\aworld'''}{\aworld''} \in \aaccessrelation^* \}$ and $\aaccessrelation_2 \egdef \aaccessrelation \setminus \aaccessrelation_1$,
we obtain $\amodel_1$ and $\amodel_2$ satisfying these properties.
It is now sufficient to consider the two functions $\amap_1$ and $\amap_2$ defined as:
\begin{itemize}
\item $\amap_1(\aworld) = T_1$ and $\amap_2(\aworld) = T_2$
\item for every $\aworld'' \in \auniverse$ s.t.\ $\pair{\aworld'}{\aworld''} \in \aaccessrelation^*$ for some $\aworld' \in R_{\aworld}'$, $\amap_1(\aworld'') = \amap(\aworld'')$ and $\amap_2(\aworld'') = \zero$;
\item for every $\aworld'' \in \auniverse$ s.t.\ $\pair{\aworld'}{\aworld''} \in \aaccessrelation^*$ for some $\aworld' \in R_{\aworld}''$, $\amap_2(\aworld'') = \amap(\aworld'')$ and $\amap_1(\aworld'') = \zero$;
\item for every $\aworld' \in \auniverse$ s.t.\ $\pair{\aworld}{\aworld'} \not\in \aaccessrelation^*$, $\amap_1(\aworld') = \amap_2(\aworld') = \zero$.
\end{itemize}
By definition of the witness function, $\amap_1$ is a witness of the encoding of $T_1$ in $\pair{\amodel_1}{\aworld}$, and
 $\amap_2$ is a witness of the encoding of $T_2$ in $\pair{\amodel_2}{\aworld}$, ending the first part of the proof.
\item
The proof is analogous to the case above.
Suppose $\amodel_1 = \triple{\auniverse}{\aaccessrelation_1}{\apropeval}$ and  $\amodel_2 = \triple{\auniverse}{\aaccessrelation_2}{\apropeval}$
such that $\amodel = \amodel_1 +_{\aworld} \amodel_2$.
Let $R_{\aworld} = \{\aworld_1,\dots,\aworld_K\} \egdef \{\aworld' \mid \pair{\aworld}{\aworld'} \in \aaccessrelation\}$,
$R_{\aworld}' \egdef \{\aworld' \mid \pair{\aworld}{\aworld'} \in \aaccessrelation_1\}$
and $R_{\aworld}'' \egdef \{\aworld' \mid \pair{\aworld}{\aworld'} \in \aaccessrelation_2\}$.
By definition of $\amodel_1$ and $\amodel_2$, the two sets $R_{\aworld}'$ and $R_{\aworld}''$ partition $R_{\aworld}$.
Let then $R_{\aworld}' = \{\aworld_{i_1},\dots,\aworld_{i_c}\}$ and $R_{\aworld}'' = \{\aworld_{i_{c+1}},\dots,\aworld_{i_K}\}$ ($c \in \interval{0}{K}$).
By definition of $\amap$, it holds that
$\amap(\aworld) \equiv T \equiv \sum_{i \in \interval{1}{K}} \aname_i[\amap(\aworld_i)]$ and from the properties of the congruence 
relation $\equiv$ we obtain
\begin{nscenter}
$\sum_{i \in \interval{1}{K}} \aname_i[\amap(\aworld_i)] \equiv \Big(\sum_{j \in \interval{1}{c}} \aname_{i_j}[\amap(\aworld_{i_j})] \Big) \ambientchop 
\Big(\sum_{j \in \interval{{c+1}}{K}} \aname_{i_j}[\amap(\aworld_{i_j})] \Big)$
\end{nscenter}
Let 
$T_1 \equiv \sum_{j \in \interval{1}{c}} \aname_{i_j}[\aworld_{i_j}]$ and
$T_2 \equiv \sum_{j \in \interval{{c+1}}{K}} \aname_{i_j}[\aworld_{i_j}]$.
Trivially, by definition $T_1 \ambientchop T_2 \equiv T$.
Again, it is now sufficient to consider the two functions $\amap_1$ and $\amap_2$ defined as:
\begin{itemize}
\item $\amap_1(\aworld) = T_1$ and $\amap_2(\aworld) = T_2$
\item for every $\aworld'' \in \auniverse$ s.t.\ $\pair{\aworld'}{\aworld''} \in \aaccessrelation^*$ for some $\aworld' \in R_{\aworld}'$, $\amap_1(\aworld'') = \amap(\aworld'')$ and $\amap_2(\aworld'') = \zero$;
\item for every $\aworld'' \in \auniverse$ s.t.\ $\pair{\aworld'}{\aworld''} \in \aaccessrelation^*$ for some $\aworld' \in R_{\aworld}''$, $\amap_2(\aworld'') = \amap(\aworld'')$ and $\amap_1(\aworld'') = \zero$;
\item for every $\aworld' \in \auniverse$ s.t.\ $\pair{\aworld}{\aworld'} \not\in \aaccessrelation^*$, $\amap_1(\aworld') = \amap_2(\aworld') = \zero$.
\end{itemize}
By definition of witness function, it is easy to show that $\amap_1$ is a witness of the encoding of $T_1$ in $\pair{\amodel_1}{\aworld}$, and
 $\amap_2$ is a witness of the encoding of $T_2$ in $\pair{\amodel_2}{\aworld}$. \qedhere
\end{enumerate}
\end{proof}

Given a formula $\aformula$ of \SAL, we define its translation $\atranslation(\aformula)$ in \modallogicCC.
$\atranslation$ is homomorphic for Boolean connectives and $\true$, and otherwise it is inductively defined as follows:
\begin{nscenter}
\hfill
  $\atranslation(\zero)$ $\egdef$ $\Box \false$;
\qquad\qquad
  $\atranslation(\aformula \ambientchop \aformulabis)$ $\egdef$ $\atranslation(\aformula) \chopop \atranslation(\aformulabis)$;
\qquad\qquad
  $\atranslation( \aname[\aformula] )$ $\egdef$ $\Diamond(\aname \land \atranslation(\aformula)) \land \lnot (\Diamond \true \chopop \Diamond \true)$.
\hfill\,
\end{nscenter}
We prove that this translation is correct.

\begin{lemma}\label{lemma:correct-SAL-MLCC}
If $\pair{\amodel}{\aworld}$ encodes $T \in \ambienttrees$, for every
$\aformula$ be in \SAL,
$T \models \aformula$ iff $\amodel,\aworld \models \atranslation(\aformula)$.
\end{lemma}
%
We are now ready to tackle the proof of  \Cref{lemma:correct-SAL-MLCC}. Thanks to the previous three results, the proof can be achieved with an easy structural induction.

\begin{proof}[Proof of \Cref{lemma:correct-SAL-MLCC}]
Let $\amodel$ be defined as $\triple{\auniverse}{\aaccessrelation}{\apropeval}$
and $\amap$ be the witness of the encoding of $T$ in $\pair{\amodel}{\aworld}$.
The proof is by structural induction on $\aformula$, as done for Lemma~\ref{lemma:correct-MLCC-SAL} (again, the 
cases for $\land$ and $\lnot$ are omitted, see proof of Lemma~\ref{lemma:height-at-most-n}).
\begin{description}
\item[Base case: $\aformula = \true$.] Trivially $T \models \true$ and $\amodel,\aworld \models \true$.
\item[Base case: $\aformula = \zero$.]~
\begin{itemize}
\item $T \models \zero$
\item if and only if  $T \equiv \zero$ (by definition of $\models$)
\item if and only if $\amap(\aworld) \equiv \zero$ (by definition of $\amap$)
\item if and only if $\aworld \not \in \pi_1(\aaccessrelation)$ (by definition of $\amap$)
\item if and only if $\amodel, {\aworld} \models \Box \bottom$ (by definition of $\models$ for $\Box \bottom$)
\item if and only if $\amodel, {\aworld} \models \atranslation(\zero)$ (by definition of $\atranslation$).
\end{itemize}
\item[Induction case: $\aformula = \aname{[\aformulabis]}$.]
For the left to right direction, suppose $T \models \aname[\aformulabis]$. Then,
\begin{enumerate}
\item\label{pc2:p2} there is $T'$ such that $T \equiv \aname[T']$ and $T' \models \aformulabis$ (by definition of $\models$ and hypothesis $T \models \aname[\aformulabis]$)
\item\label{pc2:p3} $\amap(\aworld) \equiv \aname[T']$ and  there is $\aworld' \in \auniverse$ such that
$\{\aworld'\} =  \aaccessrelation(\aworld)$,
$\amap(\aworld') \equiv T'$ and $\aworld' \in \apropeval(\aname)$ (from (\ref{pc2:p2}), by definition of $\amap$)
\item\label{pc2:p4} $\pair{\amodel}{\aworld'}$ encodes $T'$ (from (\ref{pc2:p3}), by Lemma~\ref{lemma:tech-lemma-2-2})
\item\label{pc2:p5} $\amodel, {\aworld'} \models \atranslation(\aformulabis)$
(from (\ref{pc2:p2}) and (\ref{pc2:p4}), by the induction hypothesis)
\item\label{pc2:p6} $\amodel, {\aworld'} \models \aname$ (from $\aworld' \in \apropeval(\aname)$ (see \ref{pc2:p3}), by definition of $\models$)
\item\label{pc2:p7} $\amodel, {\aworld'} \models \aname \land \atranslation(\aformulabis)$ (from (\ref{pc2:p5}) and (\ref{pc2:p6}), by definition of $\models$)
\item\label{pc2:p8} $\amodel, {\aworld} \models \Diamond(\aname \land \atranslation(\aformulabis))$ (from (\ref{pc2:p7}) and $(\aworld,\aworld') \in \aaccessrelation$ (see \ref{pc2:p3}), by def.\ of
$\models$)
\item\label{pc2:p9} $\amodel, {\aworld} \models \lnot (\Diamond \true \chopop \Diamond \true)$ (from
$\{\aworld'\} =  \aaccessrelation(\aworld)$ (see \ref{pc2:p3}), by def.\ of
$\models$)
\item\label{pc2:p10} $\amodel, {\aworld} \models \Diamond(\aname \land \atranslation(\aformulabis)) \land \lnot (\Diamond \true \chopop \Diamond \true)$
(from (\ref{pc2:p8}) and (\ref{pc2:p9}), by def.\ of
$\models$)
\item\label{pc2:p11} $\amodel, {\aworld} \models \atranslation(\aname[\aformula])$
(from (\ref{pc2:p10}), by definition of $\atranslation$).
\end{enumerate}
For the right to left direction, suppose $\amodel, {\aworld} \models \atranslation(\aname[\aformula])$. Then,
\begin{enumerate}
\item\label{pc2:a1} $\amodel, {\aworld} \models \Diamond(\aname \land \atranslation(\aformulabis)) \land \lnot (\Diamond \true \chopop \Diamond \true)$
(by def.\ of $\atranslation$ and hyp.\ $\amodel, {\aworld} \models \atranslation(\aname[\aformula])$)
\item\label{pc2:a2} $\card{\aaccessrelation(\aworld)}$ is at most $1$
(from $\amodel, {\aworld} \not\models \Diamond \true \chopop \Diamond \true$ (\ref{pc2:a1}), by def.\ of $\models$)
\item\label{pc2:a3} $\pair{\aworld}{\aworld'} \in \aaccessrelation$ and  $\amodel, {\aworld'} \models \aname \land \atranslation(\aformulabis)$ for some $\aworld' \in \auniverse$
(from (\ref{pc2:a1}), by def.\ of $\models$)
\item\label{pc2:a4} $\aworld' \in \apropeval(\aname)$ (from (\ref{pc2:a3}), by def.\ of $\models$)
\item\label{pc2:a5} $\amodel, {\aworld'} \models \atranslation(\aformulabis)$ (from (\ref{pc2:a3}), by def.\ of $\models$)
\item\label{pc2:a6} $\pair{\amodel}{\aworld'}$ encodes $\amap(\aworld')$ (by Lemma~\ref{lemma:tech-lemma-2-2}, since $\pair{\amodel}{\aworld}$ encodes $T$)
\item\label{pc2:a7} $\amap(\aworld') \models \aformulabis$ (from (\ref{pc2:a5}) and (\ref{pc2:a6}), by the induction hypothesis)
\item\label{pc2:a8} $T \equiv \amap(\aworld) \equiv \aname[\amap(\aworld')]$ (from (\ref{pc2:a2}), (\ref{pc2:a3}), (\ref{pc2:a4}) and (\ref{pc2:a6}), by definition of $\amap$)
\item\label{pc2:a9} there is $T'$ (concretely, $\amap(\aworld')$) such that $T \equiv \aname[T']$ and $T' \models \aformulabis$ (from (\ref{pc2:a7}) and (\ref{pc2:a8}))
\item\label{pc2:a10} $T \models \aname[\aformulabis]$ (from (\ref{pc2:a9}) by definition of $\models$).
\end{enumerate}
\item[Induction case: $\aformula = \aformulabis \ambientchop \aformulater$.] For the left to right direction, suppose $T \models \aformulabis \ambientchop \aformulater$. Then,
\begin{enumerate}
\item\label{pc2:b1} there are $T_1$ and $T_2$ such that $T \equiv T_1 \ambientchop T_2$, $T_1 \models \aformulabis$ and
$T_2 \models \aformulater$ (by definition of $\models$)
\item\label{pc2:b2} there are $\amodel_1$ and $\amodel_2$ such that
$\amodel = \amodel_1  +_{\aworld} \amodel_2$,
$\pair{\amodel_1}{\aworld}$ encodes $T_1$, and
$\pair{\amodel_2}{\aworld}$ encodes $T_2$ (from (\ref{pc2:b1}) and $\pair{\amodel}{\aworld}$ encodes $T$, by Lemma~\ref{lemma:tech-lemma-2-3}.1)
\item\label{pc2:b3} $\amodel_1,\aworld \models \atranslation(\aformulabis)$ and $\amodel_2,\aworld \models \atranslation(\aformulater)$ (from (\ref{pc2:b1}) and (\ref{pc2:b2}), by the induction hypothesis)
\item\label{pc2:b4} $\amodel \models \atranslation(\aformulabis) \chopop \atranslation(\aformulater)$ (from (\ref{pc2:b2}) and (\ref{pc2:b3}), by definition of $\models$)
\item\label{pc2:b5} $\amodel \models \atranslation(\aformulabis \ambientchop \aformulater)$ (from (\ref{pc2:b4}), by definition of $\atranslation$)
\end{enumerate}
For the right to left direction, suppose $\amodel,\aworld \models \atranslation(\aformulabis \ambientchop \aformulater)$. Then,
\begin{enumerate}
\item\label{pc2:c1} there are $\amodel_1$ and $\amodel_2$ such that $\amodel = \amodel_1 +_{\aworld} \amodel_2$, $\amodel_1, \aworld \models \atranslation(\aformulabis)$, and $\amodel_2, \aworld \models \atranslation(\aformulater)$ (by definition of $\atranslation$ and $\models$)
\item\label{pc2:c2} there are $T_1$ and $T_2$ such that $T \equiv T_1 \ambientchop T_2$, $\pair{\amodel_1}{\aworld}$ encodes $T_1$,
and $\pair{\amodel_2}{\aworld}$ encodes $T_2$ (from (\ref{pc2:c1}) and $\pair{\amodel}{\aworld}$ encodes $T$, by Lemma~\ref{lemma:tech-lemma-2-3}.2)
\item\label{pc2:c3} $T_1 \models \aformulabis$ and $T_2 \models \aformulater$ (from (\ref{pc2:c1}) and (\ref{pc2:c2}), by the induction hypothesis)
\item\label{pc2:c4} $T \models \aformulabis \ambientchop \aformulater$ (from (\ref{pc2:c2}) and (\ref{pc2:c3}), by definition of $\models$)
\qedhere
\end{enumerate}
\end{description}
\end{proof}

So, we can complete the reduction.
%
\begin{theorem}\label{theorem:SAL-equisat-MLCC}
Let $\aformula$ be in \SAL  built over  $\apropset \subseteq_{\fin} \varprop$ and
$\anvarprop \not\in \apropset$.
$\aformula$ is satisfiable if and only if $\atranslation(\aformula) \land \bigwedge_{i \in \interval{1}{\fsize{\aformula}}} \Box^{i} \bigvee_{\aname \in \apropset \cup \{\anvarprop\}}
\big( \aname \land \bigwedge_{\anamebis \in (\apropset \cup \{\anvarprop\}) \setminus \{\aname\}} \lnot \anamebis \big)$
 is satisfiable.
\end{theorem}
 \begin{proof}
Suppose $\aformula$ satisfiable. Then, there is $T$ such that $T \models \aformula$.
In general, it could be that $T$ contains ambient names that do not appear in $\aformula$. However, we can assume that there is only one name in $T$ that does not appear in $\aformula$ and that name is $\avarprop$ (as in the statement of this theorem).
Indeed, this assumption relies on the following property of static ambient logic (see \cite{CCG03}, Lemma~8).
\begin{nscenter}
\begin{minipage}{0.92\linewidth}
Let $\avarprop, \avarpropbis$ be two ambient names not appearing in $\aformula$.
Then $T \models \aformula$ iff $T[\avarprop \gets \avarpropbis] \models \aformula$, where $T[\avarprop \gets \avarpropbis]$ is the tree obtained from $T$ by replacing every occurrence of $\avarprop$ with $\avarpropbis$.
\end{minipage}
\end{nscenter}
Let $\pair{\amodel}{\aworld}$ be a pointed forest, where $\amodel = \triple{\auniverse}{\aaccessrelation}{\apropeval}$,
encoding of $T$ (it exists by Lemma~\ref{lemma:tech-lemma-2-1}).
 By Lemma~\ref{lemma:correct-SAL-MLCC} we have $\amodel, \aworld \models \atranslation(\aformula)$.
Let us recall the properties of the encoding of $T$ by a model $\pair{\amodel}{\aworld}$:
\begin{enumerate}
\item every world in $\auniverse$ satisfies at most one propositional symbol in $\apropset$;
\item there is a function $\amap$ from $\auniverse$ to $\ambienttrees$ such that $\amap(\aworld) \equiv T$ and for every
$\aworld' \in \aaccessrelation^*(\aworld)$, we have
$\amap(\aworld') \equiv \sum_{i \in \interval{1}{K}}  \aname_i[\amap(\aworld_i)]$
where
$\{\aworld_1$, \dots, $\aworld_K\} = \aaccessrelation(\aworld')$
 and \fbox{for all $i \in \interval{1}{K}$, $\aworld_i \in \apropeval(\aname_i)$}.
\end{enumerate}
The first property together with the highlighted part of the second property imply that every world reachable in at least one step from $\aworld$ satisfies exactly one propositional symbol of $\apropset$. Then trivially $\amodel,\aworld \models \bigwedge_{i \in \interval{1}{\fsize{\aformula}}} \Box^{i} \bigvee_{\aname \in \apropset \cup \{\avarprop\}} \big( \aname \land \bigwedge_{\anamebis \in (\apropset  \cup \{\avarprop\})\setminus \{\aname\}} \lnot \anamebis \big)$.

Conversely, suppose $\aformulabis = \atranslation(\aformula) \land \bigwedge_{i \in \interval{1}{\fsize{\aformula}}} \Box^{i}
 \bigvee_{\aname \in \apropset  \cup \{\avarprop\}} \big( \aname \land \bigwedge_{\anamebis \in (\apropset  \cup \{\avarprop\})\setminus \{\aname\}}
\lnot \anamebis \big)$ satisfiable.
To prove the result it is sufficient to show that there is a pair $\pair{\amodel}{\aworld}$ encoding a tree $T$ that satisfies $\aformulabis$. Indeed, if this is the case then by $\amodel,\aworld \models \atranslation(\aformula)$ we obtain $T \models \aformula$ by Lemma~\ref{lemma:correct-SAL-MLCC}.
As $\aformulabis$ is satisfiable, we know that
there is a forest  $\amodel = \triple{\auniverse}{\aaccessrelation}{\apropeval}$ and a world $\aworld \in \auniverse$ such that $\amodel,\aworld \models \aformulabis$.
It is important to notice that, as in \Cref{th:mlcc-sal-equisat}, we can get rid of all the parts beyond $\md{\varphi}$, so we can ensure that as $\amodel,\aworld\models\psi$, then it is a encoding of some $T$, and therefore, $T\models\aformula$.
\end{proof}

 \subparagraph{From \satproblem{\modallogicCC} to \satproblem{\SAL}.}
 As explained in Section~\ref{section-SAL}, to obtain a polynomial-time reduction from \satproblem{\modallogicCC} to
 \satproblem{\SAL}, we have to understand how to encode a finite set of propositional symbols.
 It is crucial to deal with two issues: we need to avoid an exponential blow up in the representation, and we have to maintain information about the children of a node. We solve both issues by representing a propositional symbol $\avarprop$ as a particular ambient, and copying enough times the ambient encoding $\avarprop$. 
 Let $\apropset \subseteq_{\fin} \varprop$ and $n \in \Nat^{>0}$, where $\Nat^{>0}$ denotes the set of positive natural numbers.
 Let $\amodel = \triple{\auniverse}{\aaccessrelation}{\apropeval}$ be a finite forest and $\aworld \in \auniverse$.
 Let $\ambientchild$ and $\ambientprop$ be two ambient names not in $\apropset$.
 The ambient name $\ambientchild$ encodes the relation $\aaccessrelation$ whereas $\ambientprop$ can be seen as a
 \emph{container} for propositional variables holding on the current world.
 We say that $T \in \ambienttrees$ is an \emph{encoding} of $\pair{\amodel}{\aworld}$ with respect to $\apropset$ and $n$ iff
 \begin{enumerate}
 \item every ambient name in $T$ is from $\apropset \cup \{\ambientchild,\ambientprop\}$;
 \item there is a function $\amap$ from $\auniverse$ to $\ambienttrees$ s.t.\ $\amap(\aworld) \equiv T$ and for every $\aworld' \in \aaccessrelation^*(\aworld)$ there is $m \geq n$ s.t.\
 \begin{nscenter}
 $\amap(\aworld') \equiv \displaystyle \Big(\sum_{i \in \interval{1}{m}} \ambientprop[ \sum_{\mathclap{\substack{\avarprop \in \apropset\\ \aworld' \in \apropeval(\avarprop)}}} \avarprop[\zero] ]\Big) \ \ambientchop \
 \sum_{\mathclap{\qquad\aworld'' \in \aaccessrelation(\aworld')}} \ambientchild[\amap(\aworld'')]$
 \end{nscenter}
 We recall that given  $I = \{i_1,\dots,i_m\}$, $\sum_{i \in I} T_i \egdef T_{i_1} \ambientchop T_{i_2} \ambientchop \dots \ambientchop T_{i_{m}}$.
 \end{enumerate}

 The figure below shows on the right a possible encoding of the model on the left.
 \begin{figure}[ht]
  \begin{tikzpicture}
  \node[dot] (w) [ label=above:$w$, label=right:\scriptsize{$\{\avarprop_1,\ldots,\avarprop_l\}$}]{ };

  \node[dot] (w1) [below left = 0.7cm and 0.7cm of w,label=below:$w_1$] {};
  \node [below   =  0.7cm of w] {$\ldots$};
  \node[dot] (w3) [below right = 0.7cm and 0.7cm of w,label=below:$w_k$] {};

  \draw[pto] (w) -- (w1);
  \draw[pto] (w) -- (w3);

  \node[dot] (r) [right= 5cm of w, label=above:$\amap(w)$]{ };
  \node[dot] (p3) [below left = 1.4cm and 1cm of r,label=below:$\amap(w_k)$] {};

  \node[dot] (pL) [below right = 0.8cm and 0.4cm of r = 1.5cm of p3] {};
  \node[dot] (oL) [below left = 0.5cm and 0.5 cm of pL, label=below:$\zero$] {};
  \node[dot] (oLL) [below right = 0.5cm and 0.5 cm of pL, label=below:$\zero$] {};

  \node[dot] (p1) [right = 2.8cm of pL] {};
  \node[dot] (o1) [below left = 0.5cm and 0.5 cm of p1, label=below:$\zero$] {};
  \node[dot] (o11) [below right = 0.5cm and 0.5 cm of p1, label=below:$\zero$] {};

  \node (aa) [below right = 0cm and 0.55cm of pL] {\footnotesize{$\dots$ \ $m$ \ $\dots$}};
  \node (aabis) [below = -5pt of aa] {\footnotesize{times}};

  \node[dot] (p4) [left = 1cm of p3, label=below:{$\amap(w_1)$}] {};

  \node [left = 0.2cm of o11] {$\ldots$};
  \node [left = 0.2cm of oLL] {$\ldots$};
  \node [left = 0.1cm of p3] {$\ldots$};

  \node (g1) [below = 0.15cm of o11] {};
  \node (gL) [below = 0.15cm of oL] {};

   \draw[pto] (r) -- node[sloped, anchor=center,above] {\scriptsize{$\ambientprop$}} (p1);
   \draw[pto] (r) -- node[sloped, anchor=center,below] {\scriptsize{$\ambientprop$}}  (pL);
   \draw[pto] (p1) -- node[sloped, anchor=center,above] {\footnotesize{$\avarprop_1$}} (o1);
   \draw[pto] (pL) -- node[sloped, anchor=center,above] {\footnotesize{$\avarprop_1$}} (oL);
   \draw[pto] (p1) -- node[sloped, anchor=center,above] {\footnotesize{$\avarprop_l$}} (o11);
   \draw[pto] (pL) -- node[sloped, anchor=center,above] {\footnotesize{$\avarprop_l$}} (oLL);
   \draw[pto] (r) -- node[sloped, anchor=center,below] {\scriptsize{$\ambientchild$}} (p3);
   \draw[pto] (r) -- node[sloped, anchor=center,above] {\scriptsize{$\ambientchild$}} (p4);

  \end{tikzpicture} 
\end{figure}

 It is easy to verify that $\pair{\amodel}{\aworld}$ always admits such an encoding.

 We start by stating three intermediate results about the encoding of a finite forest in a model of static ambient logic. These lemmata will be fundamental to show the correctness of the translation in~\Cref{lemma:correct-MLCC-SAL}.
 The first lemma below shows that such an encoding always exists.
 In what follows, we call $\amap$ (as in the definition of the encoding) the \emph{witness of the encoding} of $\pair{\amodel}{\aworld}$ in $T$.

 \begin{lemma}\label{lemma:tech-lemma-1-1}
 Let $\amodel$ be a finite forest and $\aworld$ be one of its worlds.
 Let $\apropset \subseteq_{\fin} \varprop$ and $n \in \Nat^{>0}$. There is a tree $T \in \ambienttrees$ that encodes $\pair{\amodel}{\aworld}$ w.r.t.\ $\apropset$ and $n$.
 \end{lemma}
 %

\begin{proof}
Let $\amodel = \triple{\auniverse}{\aaccessrelation}{\apropeval}$ be a model. By following directly the properties of 
the witness function, we define the tree $T$ as $\aCCmodeltranslation{\triple{\auniverse}{\aaccessrelation}{\apropeval}}{\aworld}{n}{\apropset}$ where
\begin{nscenter}
$\aCCmodeltranslation{\triple{\auniverse}{\aaccessrelation}{\apropeval}}{\aworld}{n}{\apropset} = \displaystyle
\Big(\sum_{i \in \interval{1}{n}} \ambientprop[ \sum_{\mathclap{\substack{\avarprop \in \apropset\\ \aworld \in \apropeval(\avarprop)}}} \avarprop[\zero] ]\Big) \ \ambientchop \
\sum_{\mathclap{\substack{\aworld' \in \auniverse\\ \pair{\aworld}{\aworld'} \in \aaccessrelation}}} \ambientchild[ \aCCmodeltranslation{\triple{\auniverse}{\aaccessrelation}{\apropeval}}{\aworld'}{n}{\apropset}]
$
\end{nscenter}
As $\amodel$ is a finite forest, for every $\aworld \in \auniverse$ and $n \in \Nat$, the computation of $\aCCmodeltranslation{\triple{\auniverse}{\aaccessrelation}{\apropeval}}{\aworld}{n}{\apropset}$ terminates.
Let $\amap(\aworld) \egdef \aCCmodeltranslation{\triple{\auniverse}{\aaccessrelation}{\apropeval}}{\aworld}{n}{\apropset}$. Trivially, $\amap$ witnesses that $T$ is an encoding of $\pair{\amodel}{\aworld}$ w.r.t.\ $\apropset$ and $n$.
\end{proof}



 The second lemma can be seen as a semantical counterpart of the modality $\Diamond$.

\begin{lemma}\label{lemma:tech-lemma-1-2}
Let $\amodel$ be a finite forest and $\aworld$ be one of its worlds. Let  $\apropset \subseteq_{\fin} \varprop$ and $n \in \Nat^{>0}$.
Let $T \in \ambienttrees$ be an encoding of $\pair{\amodel}{\aworld}$ with respect to $\apropset$ and $n$. Then,
\begin{enumerate}
\item For every $n' \leq n$, $T$ is also an encoding of $\pair{\amodel}{\aworld}$ with respect to $\apropset$ and $n'$.
\item Let $\amap$ be a witness of this encoding.
For every $\aworld'$ accessible from $\aworld$ it holds that $\amap(\aworld')$ is an encoding of $\pair{\amodel}{\aworld'}$ with respect to $\apropset$ and $n$.
\end{enumerate}
\end{lemma}

\begin{proof}
Both properties trivially follow from the definition of encoding. Moreover, for (2) notice that for every world $\aworld'$ accessible from $\aworld$ the function $\amap$ is also the witness that $\amap(\aworld')$ is an encoding of $\pair{\amodel}{\aworld'}$ with respect to $\apropset$ and $n$.
\end{proof}

 The third lemma can be seen as the semantical counterpart of the modality $\chopop$.

\begin{lemma}\label{lemma:tech-lemma-1-3}
Let $\amodel$ be a finite forest and $\aworld$ be one of its worlds.
Let $\apropset \subseteq_{\fin} \varprop$ and $n \in \Nat^{>0}$.
Let $T \in \ambienttrees$ be an encoding of $\pair{\amodel}{\aworld}$ with respect to $\apropset$ and $n$.
Let $n_1 , n_2 \in \Nat$ such that $n = n_1 + n_2$.
Then,
\begin{enumerate}
  \itemsep 0cm
\item For all $\amodel_1$ and $\amodel_2$ such that $\amodel = \amodel_1 +_{\aworld} \amodel_2$ there are $T_1$ and $T_2$ such that $T \equiv T_1 \ambientchop T_2$,
 $T_1$ is an encoding of $\pair{\amodel_1}{\aworld}$ with respect to $\apropset$ and $n_1$, and $T_2$ is an encoding of
$\pair{\amodel_2}{\aworld}$ with respect to $\apropset$ and $n_2$.
\item For all $T_1$ and $T_2$ such that
\begin{nscenter}
$T \equiv T_1 \ambientchop T_2 \ambientchop \sum_{i \in \interval{1}{n}} \ambientprop[
\sum_{\avarprop \in \apropset,\ \aworld \in \apropeval(\avarprop)} \avarprop[\zero] ],$
\end{nscenter}
there are $\amodel_1$ and $\amodel_2$ such that $\amodel = \amodel_1 +_{\aworld} \amodel_2$ and
\begin{itemize}
\item $T_1 \ambientchop \sum_{i \in \interval{1}{n_1}} \ambientprop[
\sum_{\avarprop \in \apropset,\ \aworld \in \apropeval(\avarprop)} \avarprop[\zero]]$ is an encoding of $\pair{\amodel_1}{\aworld}$ w.r.t.\ $\apropset$ and $n_1$;
\item $T_2 \ambientchop \sum_{i \in \interval{1}{n_2}} \ambientprop[
\sum_{\avarprop \in \apropset,\ \aworld \in \apropeval(\avarprop)} \avarprop[\zero]]$ is an encoding of $\pair{\amodel_2}{\aworld}$ w.r.t.\ $\apropset$ and $n_2$.
\end{itemize}
\end{enumerate}
\end{lemma}

\begin{proof}
In the proof of both points, let $\amodel$ (as in the statement) be $\triple{\auniverse}{\aaccessrelation}{\apropeval}$. 
Moreover, let $\amap$ be the witness of the encoding of $\pair{\amodel}{\aworld}$ in $T$.
\begin{enumerate}
\item
Suppose $\amodel_1 = \triple{\auniverse}{\aaccessrelation_1}{\apropeval}$ and  $\amodel_2 = \triple{\auniverse}{\aaccessrelation_2}{\apropeval}$
such that $\amodel = \amodel_1 +_{\aworld} \amodel_2$.
Let $W_{\aworld} = \{\aworld_1,\dots,\aworld_K\} \egdef \{\aworld' \mid \pair{\aworld}{\aworld'} \in \aaccessrelation\}$,
$W_{\aworld}' \egdef \{\aworld' \mid \pair{\aworld}{\aworld'} \in \aaccessrelation_1\}$
and $R_{\aworld}'' \egdef \{\aworld' \mid \pair{\aworld}{\aworld'} \in \aaccessrelation_2\}$ be the set of worlds accessible from $\aworld$ by considering respectively $\aaccessrelation$, $\aaccessrelation_1$ and $\aaccessrelation_2$ as accessibility relations.
By definition of $\amodel_1$ and $\amodel_2$, the two sets $W_{\aworld}'$ and $W_{\aworld}''$ partition $W_{\aworld}$.
Then, let $W_{\aworld}' = \{\aworld_{i_1},\dots,\aworld_{i_c}\}$ and $W_{\aworld}'' = \{\aworld_{i_{c+1}},\dots,\aworld_{i_K}\}$ ($c \in \interval{0}{K}$).
By definition of $\amap$, it holds that
\begin{nscenter}
$\displaystyle\amap(\aworld) \equiv T \equiv \Big(\sum_{i \in \interval{1}{m}} \ambientprop[ \sum_{\mathclap{\substack{\avarprop \in \apropset\\ \aworld \in \apropeval(\avarprop)}}} \avarprop[\zero] ]\Big) \ \ambientchop \
\sum_{\mathclap{i \in \interval{1}{K}}} \ambientchild[\amap(\aworld_i)]$.
\end{nscenter}
where $m \geq n$. As $m \geq n$, there are $m_1$ and $m_2$ such that $m = m_1 + m_2$, $m_1 \geq n_1$ and $m_2 \geq n_2$.
From the properties of the congruence relation $\equiv$ we can show that $T$ is equivalent to $T_1 \ambientchop T_2$, where
\begin{nscenter}
$\displaystyle T_1 \egdef \big(\sum_{\mathclap{i \in \interval{1}{m_1}}} \ambientprop[ \sum_{\mathclap{\substack{\avarprop \in \apropset\\ \aworld \in \apropeval(\avarprop)}}} \avarprop[\zero] ]\big)
\ \ambientchop \
\sum_{\mathclap{j \in \interval{1}{c}}} \ambientchild[\amap(\aworld_{i_j})]$;

$\displaystyle T_2 \egdef \big(\sum_{\mathclap{i \in \interval{1}{m_2}}} \ambientprop[ \sum_{\mathclap{\substack{\avarprop \in \apropset\\ \aworld \in \apropeval(\avarprop)}}} \avarprop[\zero] ]\big)
\ \ambientchop \
\sum_{\mathclap{i \in \interval{c+1}{K}}} \ambientchild[\amap(\aworld_{i_j})]$.
\end{nscenter}
By definition, $T_1 \ambientchop T_2 \equiv T$.
We now consider the two functions $\amap_1$ and $\amap_2$ defined as:
\begin{itemize}
\item $\amap_1(\aworld) = T_1$ and $\amap_2(\aworld) = T_2$
\item for every $\aworld'' \in \auniverse$ s.t.\ $\pair{\aworld'}{\aworld''} \in \aaccessrelation^*$ for some $\aworld' \in W_{\aworld}'$, $\amap_1(\aworld'') = \amap(\aworld'')$ and $\amap_2(\aworld'') = \zero$;
\item for every $\aworld'' \in \auniverse$ s.t.\ $\pair{\aworld'}{\aworld''} \in \aaccessrelation^*$ for some $\aworld' \in W_{\aworld}''$, $\amap_2(\aworld'') = \amap(\aworld'')$ and $\amap_1(\aworld'') = \zero$;
\item for every $\aworld' \in \auniverse$ s.t.\ $\pair{\aworld}{\aworld'} \not\in \aaccessrelation^*$, $\amap_1(\aworld') = \amap_2(\aworld') = \zero$.
\end{itemize}
By definition of the witness function and recalling that $m_1 \geq n_1$ and $m_2 \geq n_2$, it is easy to show that $\amap_1$ witnesses that $T_1$ is an encoding of $\pair{\amodel_1}{\aworld}$ w.r.t.\ $\apropset$ and $n_1$, whereas
 $\amap_2$ witnesses that $T_2$ is an encoding of $\pair{\amodel_2}{\aworld}$ w.r.t.\ $\apropset$ and $n_2$.
\item Suppose now $T_1$ and $T_2$ such that
\begin{nscenter}
$T \equiv T_1 \ambientchop T_2 \ambientchop \sum_{i \in \interval{1}{n}} \ambientprop[
\sum_{\avarprop \in \apropset,\ \aworld \in \apropeval(\avarprop)} \avarprop[\zero] ]$.
\end{nscenter}
By recalling that $n = n_1 + n_2$,
from the properties of the congruence relation $\equiv$, we can then show that $T$ is equivalent to
\begin{nscenter}
$(\dag) \ \
\displaystyle
\Big(T_1 \ambientchop \sum_{\mathclap{i \in \interval{1}{n_1}}} \ambientprop[ \sum_{\mathclap{\substack{\avarprop \in \apropset\\ \aworld \in \apropeval(\avarprop)}}} \avarprop[\zero] ] \Big)
\ \ambientchop\
\Big( T_2 \ambientchop \sum_{\mathclap{i \in \interval{1}{n_2}}} \ambientprop[ \sum_{\mathclap{\substack{\avarprop \in \apropset\\ \aworld \in \apropeval(\avarprop)}}} \avarprop[\zero] ] \Big)$
\end{nscenter}
Then, for $j \in \{1,2\}$ let 
\begin{nscenter}
  $T_j' \egdef T_j \ambientchop \sum_{i \in \interval{1}{n_j}} \ambientprop[ \sum_{\avarprop \in \apropset,\ \aworld \in \apropeval(\avarprop)} \avarprop[\zero] ]$
\end{nscenter}
so that $T \equiv T_1' \ambientchop T_2'$.
In order to conclude the proof, we have to show that it is possible to partition 
$\aaccessrelation$ into $\aaccessrelation_1$ and $\aaccessrelation_2$ so that
$\amodel_1 = \triple{\auniverse}{\aaccessrelation_1}{\apropeval}$, $\amodel_2 = \triple{\auniverse}{\aaccessrelation_2}{\apropeval}$,
$\amodel = \amodel_1 +_{\aworld} \amodel_2$ and
\begin{itemize}
\item $T_1'$ is an encoding of $\pair{\amodel_1}{\aworld}$ w.r.t.\ $\apropset$ and $n_1$;
\item $T_2'$ is an encoding of $\pair{\amodel_2}{\aworld}$ w.r.t.\ $\apropset$ and $n_2$.
\end{itemize}
We consider the accessibility relation $\aaccessrelation$.
Let $W_{\aworld} = \{\aworld_1,\dots,\aworld_K\} \egdef \{\aworld' \mid \pair{\aworld}{\aworld'} \in \aaccessrelation\}$ be the set of worlds that are accessible from $\aworld$ (notice that this set could be empty).
As $T$ is an encoding of $\pair{\amodel}{\aworld}$, we have the following equivalence:
\begin{nscenter}
$\displaystyle\amap(\aworld) \equiv T \equiv \Big(\sum_{i \in \interval{1}{m}} \ambientprop[ \sum_{\mathclap{\substack{\avarprop \in \apropset\\ \aworld \in \apropeval(\avarprop)}}} \avarprop[\zero] ]\Big) \ \ambientchop \
\sum_{\mathclap{i \in \interval{1}{K}}} \ambientchild[\amap(\aworld_i)]$,
\end{nscenter}
with $m \geq n$.
Notice that if $W_{\aworld}$ is empty then we have that $\sum_{i \in \interval{1}{K}} \ambientchild[\amap(\aworld_i)]$ is equivalent to the empty tree $\zero$.
Following the equivalence between $T$ and $(\dag)$, we know that
we can partition $W_{\aworld}$ into two sets $W_{\aworld}' = \{\aworld_{i_1},\dots,\aworld_{i_c}\}$ and $W_{\aworld}'' = \{\aworld_{i_{c+1}},\dots,\aworld_{i_K}\}$ ($c \in \interval{0}{K}$)
so that, for some $m_1, m_2 \in \Nat$ such that $m = m_1 + m_2$, $m_1 \geq n_1$ and $m_2 \geq n_2$
we have
\begin{itemize}
\item $T_1' \equiv \big(\sum_{j \in \interval{1}{c}} \ambientchild[\amap(\aworld_{i_j})] \big) \ambientchop \sum_{i \in \interval{1}{m_1}} \ambientprop[ \sum_{\avarprop \in \apropset,\ \aworld \in \apropeval(\avarprop)} \avarprop[\zero] ]$;
\item $T_2 \equiv \big( \sum_{j \in \interval{c+1}{K}} \ambientchild[\amap(\aworld_{i_j})] \big) \ambientchop \sum_{i \in \interval{1}{m_2}} \ambientprop[ \sum_{\avarprop \in \apropset,\ \aworld \in \apropeval(\avarprop)} \avarprop[\zero] ]$.
\end{itemize}
By relying on the partitioning of $W_{\aworld}$ into $W_{\aworld}'$ and $W_{\aworld}''$ it is easy to show that we can derive two finite forests $\amodel_1 = \triple{\auniverse}{\aaccessrelation_1}{\apropeval}$ and $\amodel_2 = \triple{\auniverse}{\aaccessrelation_2}{\apropeval}$ such that
\begin{itemize}
\item $\amodel = \amodel_1 +_{\aworld} \amodel_2$;
\item every $\aworld' \in W_{\aworld}'$ is accessible from $\aworld$ in $\aaccessrelation_1$, i.e.\ $\pair{\aworld}{\aworld'} \in \aaccessrelation_1$;
\item every $\aworld'' \in W_{\aworld}''$ is accessible from $\aworld$ in $\aaccessrelation_2$, i.e.\ $\pair{\aworld}{\aworld''} \in \aaccessrelation_2$.
\end{itemize}
By defining
$\aaccessrelation_1 \egdef \{\pair{\aworld'}{\aworld''} \in \aaccessrelation \mid \text{ there is } \aworld''' \in W_{\aworld}' \text{ such that } \pair{\aworld'''}{\aworld''} \in \aaccessrelation^* \}$ and $\aaccessrelation_2 \egdef \aaccessrelation \setminus \aaccessrelation_1$
we obtain $\amodel_1$ and $\amodel_2$ satisfying these properties.
It is now sufficient to consider the two functions $\amap_1$ and $\amap_2$ defined as:
\begin{itemize}
\item $\amap_1(\aworld) = T_1'$ and $\amap_2(\aworld) = T_2'$
\item for every $\aworld'' \in \auniverse$ s.t.\ $\pair{\aworld'}{\aworld''} \in \aaccessrelation^*$ for some $\aworld' \in W_{\aworld}'$, $\amap_1(\aworld'') = \amap(\aworld'')$ and $\amap_2(\aworld'') = \zero$;
\item for every $\aworld'' \in \auniverse$ s.t.\ $\pair{\aworld'}{\aworld''} \in \aaccessrelation^*$ for some $\aworld' \in W_{\aworld}''$, $\amap_2(\aworld'') = \amap(\aworld'')$ and $\amap_1(\aworld'') = \zero$;
\item for every $\aworld' \in \auniverse$ s.t.\ $\pair{\aworld}{\aworld'} \not\in \aaccessrelation^*$, $\amap_1(\aworld') = \amap_2(\aworld') = \zero$.
\end{itemize}
By definition of the witness function, $\amap_1$ witnesses the encoding of $T_1'$ in $\pair{\amodel_1}{\aworld}$, and
 $\amap_2$ witnesses the encoding of $T_2'$ in $\pair{\amodel_2}{\aworld}$.
  \qedhere
\end{enumerate}
\end{proof}

 In the figure just above, we present a  model for \modallogicCC (on the left), and one possible encoding (on the right), via some $\amap$ and w.r.t. $n$.
 We define the translation of
 $\aformula$, written $\atranslation(\aformula)$,
 into \SAL. It is homomorphic for Boolean connectives and $\true$, $\atranslation(\avarprop)$ $\egdef$ $\HMDiamond{\ambientprop}\HMDiamond{\avarprop}\true$ and otherwise it is inductively defined:
 \begin{nscenter}
   $\atranslation(\Diamond \aformula)$ $\egdef$ $\HMDiamond{\ambientchild}\atranslation(\aformula)$; \\
 $\atranslation(\aformula \chopop \aformulabis) \egdef \big(\atranslation(\aformula) \land \HMDiamond{\ambientprop}_{\geq \fsize{\aformula}} \true \big)\ \ambientchop \
 \big( \atranslation(\aformulabis) \land \HMDiamond{\ambientprop}_{\geq \fsize{\aformulabis}} \true \big)$,
 \end{nscenter}
 where $\HMDiamond{n}_{\geq k} \aformula$ is
 the graded modality
 defined as $\true$ for $k = 0$, otherwise $(\HMDiamond{n} \aformula) \ \ambientchop \ \HMDiamond{n}_{\geq k-1} \aformula$.
 In the translation of $\chopop$, the model of \SAL has to be split in such a way that both subtrees contain enough $\ambientprop$ ambients to correctly answer to the formula $\HMDiamond{\ambientprop}\HMDiamond{\avarprop}\true$.
 It is easy to see that the size of $\atranslation(\aformula)$ is quadratic in $\fsize{\aformula}$.

 \begin{lemma}\label{lemma:correct-MLCC-SAL}
 Let $\amodel$ be a finite forest and $\aworld$ be one of its worlds.
 Let $\apropset \subseteq_{\fin} \varprop$ and $n \in \Nat^{>0}$.
 Let $T$ be an encoding of $\pair{\amodel}{\aworld}$ w.r.t $\apropset$ and $n$.
 For every formula $\aformula$ built over $\apropset$ with  $\fsize{\aformula} \leq n$,
 we have $\amodel,\aworld \models \aformula$  iff
 $T \models \atranslation(\aformula)$.
 \end{lemma}


\begin{proof}
Let $\amodel$ be a model $\triple{\auniverse}{\aaccessrelation}{\apropeval}$.
and $\amap$ be the witness that $T$ encodes $\pair{\amodel}{\aworld}$ with respect to $\apropset$ and $n$.
The proof is by structural induction on $\aformula$ and it is quite straightforward (cases for $\land$ and $\lnot$ omitted,
see the proof of Lemma~\ref{lemma:height-at-most-n}).
\begin{description}
\item[Base case: $\aformula = \avarprop$.]~
\begin{itemize}
\item $\amodel,\aworld \models \avarprop$
\item if and only if $\aworld \in \apropeval(\avarprop)$ (by definition of $\models$)
\item if and only if there are $T_1,T_2 \in \ambienttrees$ such that $\amap(\aworld) \equiv T \equiv \ambientprop[\avarprop[\zero] \ambientchop T_1] \ambientchop T_2$ (by def.\ of $\amap$)
\item if and only if $T \models \HMDiamond{\ambientprop}\HMDiamond{\avarprop}\true$ (by definition of $\models$)
\item if and only if $T \models \atranslation(\avarprop)$ (by definition of $\atranslation$).
\end{itemize}
\item[Induction case: $\aformula = \Diamond \aformulabis$.] For the left to right direction, suppose $\amodel,\aworld \models \Diamond \aformulabis$. Then,
\begin{enumerate}
\item\label{pc1:a1} there is $\aworld' \in \auniverse$ s.t.\ $\pair{\aworld}{\aworld'} \in \aaccessrelation$ and $\amodel,\aworld' \models \aformulabis$ (by def. of $\models$ and hyp. $\amodel,\aworld \models \Diamond \aformulabis$)
\item\label{pc1:a2} $\amap(\aworld')$ is an encoding of $\pair{\amodel}{\aworld'}$ w.r.t.\ $\apropset$ and $n$ (from $\pair{\aworld}{\aworld'} \in \aaccessrelation$ (see \ref{pc1:a1}), by Lemma~\ref{lemma:tech-lemma-1-2}.2)
\item\label{pc1:a3} $\amap(\aworld') \models \atranslation(\aformulabis)$ (from (\ref{pc1:a1}) and (\ref{pc1:a2}),
by the induction hypothesis)
\item\label{pc1:a4} there is $T' \in \ambienttrees$ s.t.\ $\amap(\aworld) \equiv T \equiv \ambientchild[\amap(\aworld')] \ambientchop T'$ (by def. of $\amap$, as $T$ encodes $\pair{\amodel}{\aworld}$)
\item\label{pc1:a5} $T \models \HMDiamond{\ambientchild}\atranslation(\aformulabis)$ (from (\ref{pc1:a3}) and (\ref{pc1:a4}), by definition of $\models$)
\item\label{pc1:a6} $T \models \atranslation(\Diamond \aformulabis)$ (from (\ref{pc1:a5}), by definition of $\atranslation$).
\end{enumerate}
For the right to left direction, suppose $T \models \atranslation(\Diamond \aformulabis)$. Then,
\begin{enumerate}
\item\label{pc1:b1} $T \models \HMDiamond{\ambientchild}\atranslation(\aformulabis)$ (by definition of $\atranslation$)
\item\label{pc1:b2} $T \equiv \ambientchild[T_1] \ambientchop T_2$ and $T_1 \models \atranslation(\aformulabis)$ for some $T_1,T_2 \in \ambienttrees$ (from (\ref{pc1:b1}), by def. of $\models$)
\item\label{pc1:b3} there is $\aworld' \in \auniverse$ s.t.\ $\pair{\aworld}{\aworld'} \in \aaccessrelation$ and $\amap(\aworld') \equiv T_1$ (from (\ref{pc1:b2}) and $\amap(\aworld) \equiv T$, by def. of $\amap$)
\item\label{pc1:b4} $T_1$ is an encoding of $\pair{\amodel}{\aworld'}$ w.r.t.\ $\apropset$ and $n$ (from (\ref{pc1:b3}), by Lemma~\ref{lemma:tech-lemma-1-2}.2)
\item\label{pc1:b5} $\amodel,\aworld' \models \aformulabis$ (from (\ref{pc1:b2}) and (\ref{pc1:b4}), by
the induction hypothesis)
\item\label{pc1:b6} $\amodel,\aworld \models \Diamond(\aformulabis)$ (from $\pair{\aworld}{\aworld'} \in \aaccessrelation$ (see \ref{pc1:b3}) and (\ref{pc1:b5}), by definition of $\models$).
\end{enumerate}
\item[Induction case: $\aformula = \aformulabis \chopop \aformulater$]
For the left to right direction, suppose $\amodel,\aworld \models \aformulabis \chopop \aformulater$. Then,
\begin{enumerate}
\item\label{pc1:c1} $\amodel_1,\aworld \models \aformulabis$ and $\amodel_2,\aworld \models \aformulater$ for some $\amodel_1$ and $\amodel_2$ such that $\amodel = \amodel_1 +_{\aworld} \amodel_2$
(by def of $\models$ and hyp. $\amodel,\aworld \models \aformulabis \chopop \aformulater$)
\item\label{pc1:c2} There are $n_1, n_2 \in \Nat$ s.t.\ $n_1 + n_2 = n$, $n_1 \geq \fsize{\aformulabis}$ and $n_2 \geq \fsize{\aformulater}$ (as $n \geq \fsize{\aformula} = \fsize{\aformulabis} + \fsize{\aformulater} + 1$ by hypothesis)
\item\label{pc1:c3} there are $T_1$ and $T_2$ such that $T \equiv T_1 \ambientchop T_2$,
$T_1$ is an encoding of $\pair{\amodel_1}{\aworld}$ with respect to $\apropset$ and $n_1$, and $T_2$ is an encoding of
$\pair{\amodel_2}{\aworld}$ with respect to $\apropset$ and $n_2$ (from (\ref{pc1:c1}), (\ref{pc1:c2}) and since $T$ is an encoding of $\pair{\amodel}{\aworld}$, from Lemma~\ref{lemma:tech-lemma-1-3}.1)
\item\label{pc1:c4}
$T_1 \models \atranslation(\aformulabis)$ and $T_2 \models \atranslation(\aformulater)$ (from (\ref{pc1:c1}) and (\ref{pc1:c3}),
by the induction hypothesis)
\item\label{pc1:c5}
$T_1 \models \HMDiamond{\ambientprop}_{\geq \fsize{\aformulabis}} \true$ and
$T_2 \models \HMDiamond{\ambientprop}_{\geq \fsize{\aformulater}} \true$ (from (\ref{pc1:c3}), by the definition of witness of an encoding, recalling that $n_1 \geq \fsize{\aformulabis}$ and $n_2 \geq \fsize{\aformulater}$)
\item\label{pc1:c6} $T \models \atranslation(\aformulabis \chopop \aformulater)$ (from $T \equiv T_1 \ambientchop T_2$ (see \ref{pc1:c3}), (\ref{pc1:c4}) and (\ref{pc1:c5}), by def. of $\models$ and $\atranslation$).
\end{enumerate}
For the right to left direction, suppose $T \models \atranslation(\aformulabis \chopop \aformulater)$.
\begin{enumerate}
\item\label{pc1:d1} There are two trees $T_1$ and $T_2$ such that
$T \equiv T_1 \ambientchop T_2$, \ $T_1 \models \atranslation(\aformulabis) \land \HMDiamond{\ambientprop}_{\geq \fsize{\aformulabis}} \true$ and
$T_2 \models \atranslation(\aformulater) \land \HMDiamond{\ambientprop}_{\geq \fsize{\aformulater}} \true$ (by definition of $\atranslation$ and $\models$)
\item\label{pc1:d2}
$\amap(\aworld) \equiv T \equiv \displaystyle \Big(\sum_{i \in \interval{1}{m}} \ambientprop[ \sum_{\mathclap{\substack{\avarprop \in \apropset\\ \aworld \in \apropeval(\avarprop)}}} \avarprop[\zero] ]\Big) \ \ambientchop \
\sum_{\mathclap{i \in \interval{1}{k}}} \ambientchild[\amap(\aworld_i)]$ for some $m \geq n$\\ (by hypothesis $T$ encodes $\pair{\amodel}{\aworld}$)
\item\label{pc1:d3}
there are $n_1, n_2 \in \Nat$, $T_1'$ and $T_2'$ so that $n = n_1 + n_2$, $n_1 \geq \fsize{\aformulabis}$, $n_2 \geq \fsize{\aformulater}$ and
\begin{nscenter}
$T_1 \equiv T_1' \mid \displaystyle \sum_{i \in \interval{1}{n_1}} \ambientprop[ \sum_{\mathclap{\substack{\avarprop \in \apropset\\ \aworld \in \apropeval(\avarprop)}}} \avarprop[\zero] ]$
\qquad
$T_2 \equiv T_2' \mid \displaystyle \sum_{i \in \interval{1}{n_2}} \ambientprop[ \sum_{\mathclap{\substack{\avarprop \in \apropset\\ \aworld \in \apropeval(\avarprop)}}} \avarprop[\zero] ]$
\end{nscenter}
(from (\ref{pc1:d1}) and (\ref{pc1:d2}) as otherwise $T_1 \not\models \HMDiamond{\ambientprop}_{\geq \fsize{\aformulabis}} \true$ or $T_2 \not\models \HMDiamond{\ambientprop}_{\geq \fsize{\aformulater}} \true$)
\item\label{pc1:d4}
$T \equiv T_1' \mid T_2' \mid \displaystyle \sum_{i \in \interval{1}{n}} \ambientprop[ \sum_{\mathclap{\substack{\avarprop \in \apropset\\ \aworld \in \apropeval(\avarprop)}}} \avarprop[\zero] ]$

(from $T \equiv T_1 \ambientchop T_2$ (see \ref{pc1:d1}) and (\ref{pc1:d3}) by the definition of $\equiv$)
\item\label{pc1:d5}
there are $\amodel_1$ and $\amodel_2$ s.t.\ $\amodel = \amodel_1 +_{\aworld} \amodel_2$, $T_1$ is an encoding of $\pair{\amodel_1}{\aworld}$ w.r.t.\ $\apropset$ and $n_1$, and $T_2$ is an encoding of $\pair{\amodel_1}{\aworld}$ w.r.t.\ $\apropset$ and $n_2$
(from (\ref{pc1:d3}) and (\ref{pc1:d4}), by Lemma~\ref{lemma:tech-lemma-1-3}.2)
\item\label{pc1:d6} $\amodel_1,\aworld \models \aformulabis$ and $\amodel_2,\aworld \models \aformulater$ (from (\ref{pc1:d1}) and \ref{pc1:d5}, by the induction hypothesis)
\item\label{pc1:d7} $\amodel,\aworld \models \aformulabis \chopop \aformulater$ (from $\amodel = \amodel_1 +_{\aworld} \amodel_2$ (see \ref{pc1:d5}) and (\ref{pc1:d6}), by definition of $\models$) \qedhere
\end{enumerate}
\end{description}
\end{proof}

 The subset of $\ambienttrees$  encoding pointed forests can be properly approximated, which completes our reduction.

 \begin{theorem}
 \label{th:mlcc-sal-equisat}
 Let
 $\aformula$ be  in \modallogicCC  built over $\apropset$.
 $\aformula$ is satisfiable iff  $\aformulabis$ below is satisfiable:
 \begin{nscenter}
 $
 \psi \egdef
 \atranslation(\aformula)
 \land\displaystyle \bigwedge_{\mathclap{i \in \interval{0}{\fsize{\aformula}}}} \HMBox{\ambientchild}^i
 \Big(\HMDiamond{\ambientprop}_{\geq \fsize{\aformula}}\true
 \land
 \bigwedge_{\mathclap{\avarprop \in \apropset}}
 \big(
 \HMDiamond{\ambientprop}\HMDiamond{\avarprop}\true \implies \HMBox{\ambientprop}\HMDiamond{\avarprop} \true \big)
 \land
 \HMBox{\ambientprop}\sum_{\avarprop \in \apropset} (\avarprop[\zero] \lor \zero) \Big).
 $
 \end{nscenter}
 \end{theorem}
 Now, we are  ready to provide the proof of the correctness of the
reduction from \satproblem{\modallogicCC} to  \satproblem{\SAL}.

\begin{proof}
Let $\aformula$ be in \modallogicCC built over propositional variables in $\apropset\subseteq_{\fin} \varprop$.
For the left to the right direction, suppose that $\aformula$ is satisfiable. There exist a finite forest $\amodel$ and a world
$\aworld$  such that $\amodel,\aworld\models\aformula$. Let $T$ be an encoding of $(\amodel,\aworld)$ via $\amap$, with respect to $\apropset$ and $\size(\aformula)$.
Given a tree $T$ congruent to $\aname[T'] \ambientchop T''$, by an \defstyle{$\aname$-successor of $T$}, we mean a tree $T'$.

    First, by \Cref{lemma:correct-MLCC-SAL} we know that $T\models\atranslation(\aformula)$. For the second conjunct, suppose $\amap(\aworld)$ has at least one $\ambientchild$-successor, otherwise it becomes trivially true. Take some
child $\amap(\aworld')$ reachable from $\amap(\aworld)$ in an arbitrary number of $\ambientchild$ steps.
So, by the definition of $\amap$, there is  at least $\size(\aformula)$ $\ambientprop$-successors.
On the other hand, suppose $\HMDiamond{\ambientprop}\HMDiamond{\avarprop}\true$ is true at $\amap(\aworld')$, for $\avarprop\in\apropset$. Again,
since $\amap(\aworld')$ is an encoding of $(\amodel,\aworld')$, for each $\ambientprop$-successor of $\amap(\aworld')$,
    there exists a $\avarprop$-successor. Finally, to check $\HMBox{\ambientprop}\sum_{\avarprop \in \apropset} (\avarprop[\zero] \lor \zero)$, notice that each $\ambientprop$-successor of $\amap(\aworld')$ is either the ambient $\zero$ (in case the valuation of $\aworld'$ is the empty set), or there are successors via some $\avarprop \in \apropset$, and these successors are the ambient $\zero$.

    For the other direction, suppose $T\models\aformulabis$, for some $T\in\ambienttrees$.
    Let
    $\apropset'=\apropset\cup\{\ambientchild,\ambientprop\}$ and  $T|_{\apropset'}$ be the tree obtained from $T$ by
    replacing with $\zero$ every occurrence of $\aname[T']$ s.t.\ $n \not \in \apropset'$.
    One can show that $T|_{\apropset'}\models \atranslation(\aformula)$.
\cut{
Let
    $\apropset'=\apropset\cup\{\ambientchild,\ambientprop\}$, by \Cref{lemma:ambient-restricted},
    $T|_{\apropset'}\models \atranslation(\aformula)$.
}
    Let us extend adequately the notion of modal degree to formulae in \SAL, for example by counting
    the maximal number of imbricated formulae of the form $\aname[\cdot]$.
    Notice that a property similar to \Cref{lemma:height-at-most-n} (Appendix~\ref{appendix-preliminary-proofs})
    holds for \SAL, so we can remove all the
    parts of the model which are not reachable beyond $\md{\aformulabis}$ steps.
    Hence, w.l.o.g., we can assume that $T$ such that  $T\models\aformulabis$ has tree depth at most $\md{\aformulabis}$  with
    $\md{\atranslation(\aformula)} \leq \md{\aformulabis} \leq \md{\aformula} + 2$.
    As seen earlier, $T|_{\apropset'} \models \atranslation(\aformula)$.
    What about the satisfaction
    of
    $$
   \aformulabis' \egdef \displaystyle \bigwedge_{\mathclap{i \in \interval{0}{\fsize{\aformula}}}} \HMBox{\ambientchild}^i
\Big(\HMDiamond{\ambientprop}_{\geq \fsize{\aformula}}\true
\land
\bigwedge_{\mathclap{\avarprop \in \apropset}}
\big(
\HMDiamond{\ambientprop}\HMDiamond{\avarprop}\true \implies \HMBox{\ambientprop}\HMDiamond{\avarprop} \true \big)
\land
\HMBox{\ambientprop}\sum_{\avarprop \in \apropset} (\avarprop[\zero] \lor \zero) \Big)?
$$
    It is easy to show that  $T|_{\apropset'} \models \aformulabis'$, as transforming $T$ to  $T|_{\apropset'}$ does
    not remove any edge labelled by a name in $\set{\ambientprop} \cup \apropset$, which is the set of names that may
    occur in $\aformulabis'$ with a ``$\HMDiamond{\cdot}$ polarity''.
    Note that  $T|_{\apropset'}$ is almost the encoding of some pointed forest, except that there may exist
    a subtree reachable from the root with a path $\ambientchild^d$  with $d = \md{\aformula}$
    that does not satisfy the conditions for being part of an encoding.
    Thanks to the satisfaction of $\aformulabis'$, we know that the subtree is congruent to
    a tree of the form
    $$T' \ambientchop \ambientprop[T_1 \ambientchop \cdots \ambientchop T_n] \ambientchop \cdots \ambientchop
      \ambientprop[T_1 \ambientchop \cdots \ambientchop T_n]
    $$
     where
    $\apropset = \set{\avarprop_1, \ldots, \avarprop_n}$, and each $T_i$ is either $\zero$ or $\avarprop_i[\zero]$.
    Moreover, $T'$ is not congruent to a tree of the form $\ambientprop[T''] \ambientchop T^{\star}$.
    In  $T|_{\apropset'}$, we replace that occurrence of the subtree by
    $\ambientprop[T_1 \ambientchop \cdots \ambientchop T_n] \ambientchop \cdots \ambientchop
      \ambientprop[T_1 \ambientchop \cdots \ambientchop T_n]$.
    By performing all the necessary replacements, we obtain a tree $T'$ that is
     the encoding of some pointed forest
    $(\amodel,\aworld)$ with respect to $\size(\aformula)$ and $\apropset'$.
    Note also that $\aformulabis'$ is satisfied by $T'$ because we took the precaution to keep the subtrees
    of the form    $\ambientprop[T_1 \ambientchop \cdots \ambientchop T_n] \ambientchop \cdots \ambientchop
      \ambientprop[T_1 \ambientchop \cdots \ambientchop T_n]$.
    Similarly, one can show by structural induction that $T' \models \atranslation(\aformula)$, using essentially
    that in the formula tree of $\atranslation(\aformula)$, there is no branch with strictly more than $d+1$ $\HMDiamond{\ambientchild}$ nodes
    and the truncations to define $T'$ preserve the number of $\ambientprop$-successors.
    By Lemma~\ref{lemma:correct-MLCC-SAL}, we conclude that  $\model,\aworld \models \aformula$.
\end{proof}

















\subsection{Proofs of \Cref{section:MSL} (Modal Separation Logic)}
Let $\amodel = \triple{\worlds}{\arelation}{\avaluation}$ be a finite forest and $\aworld \in \worlds$.
Let $\amodel' = \triple{\worlds'}{\arelation'}{\avaluation'}$ and $\aworld' \in \worlds'$ be a model of $\MSL{}{\separate,\Diamondminus}$.
Given $n \in \Nat$ and $r \in \AP$, we say that $\pair{\amodel'}{\aworld'}$ is an $\pair{n}{r}$-\emph{encoding}
of $\pair{\amodel}{\aworld}$ if and only if
there is a bijection $\amap : (\arelation|_{\aworld}^{\leq n})^*(\aworld) \to ({\arelation'}^{-1}|_{\aworld'}^{\leq n})^*(\aworld')$ such that
\begin{enumerate}
  \item $\amap(\aworld) = \aworld'$ and for every $\pair{\aworld_1}{\aworld_2} \in \arelation|_{\aworld}^{\leq n}$,
  $(\amap(\aworld_2),\amap(\aworld_1)) \in {\arelation'}^{-1}|_{\aworld'}^{\leq n}$;
  \item for every $\avarprop \in \AP \setminus \{r\}$ and $\aworld_1 \in (\arelation|_{\aworld}^{\leq n})^*(\aworld)$, $\aworld_1 \in \avaluation(\avarprop)$ $\iff$ $\amap(\aworld_1) \in \avaluation'(\avarprop)$;
  \item $\avaluation'(r) \cap (\arelation|_{\aworld}^{\leq n})^*(\aworld) = \{\aworld\}$.
\end{enumerate}
Recall that $(\arelation|_{\aworld}^{\leq n})^*(\aworld)$ corresponds to the set of worlds appearing in $\arelation|_{\aworld}^{\leq n}$.
Notice that then, in particular $\amap$ describes a tree isomorphism between the
trees defined from $\arelation|_{\aworld}^{\leq n}$ and ${\arelation'}^{-1}|_{\aworld'}^{\leq n}$.
Moreover, for every $n \geq 2$, if $\pair{\amodel'}{\aworld'}$ is a $n$-encoding
of $\pair{\amodel}{\aworld}$ then $\pair{\amodel'}{\aworld'}$ is also a $(n-1)$-encoding
of $\pair{\amodel}{\aworld}$.

\begin{lemma}
\label{A:lemma:MSL-induction}
Let $\aformula$ in $\modallogicSC$. Let $\amodel = \triple{\worlds}{\arelation}{\avaluation}$ be a finite forest and $\aworld \in \worlds$.
Let $\amodel' = \triple{\worlds'}{\arelation'}{\avaluation'}$ and $\aworld' \in \worlds'$ be a model of $\MSL{}{\separate,\Diamondminus}$ such that
$\pair{\amodel'}{\aworld'}$ is a $\pair{n}{r}$-encoding
of $\pair{\amodel}{\aworld}$, for some $n \geq \md{\aformula}$ and $r \in \AP$ not appearing in $\aformula$.
Then,
\begin{nscenter}
  $\amodel, \aworld \models \aformula \iff \amodel',\aworld' \models \aformula[\Diamond \gets \Diamond^{-1}]$.
\end{nscenter}
\end{lemma}

\begin{proof}
  The result is proven with a rather straightforward structural induction on $\aformula$, by using the property of $\amap$, the bijection witnessing that
  $\pair{\amodel'}{\aworld'}$ is an $\pair{n}{r}$-encoding
  of $\pair{\amodel}{\aworld}$.
  The base case for atomic propositions, as well as the inductive cases for Boolean connectives are trivial.
  For the inductive cases $\Diamond \aformulabis$ and $\aformula_1 \separate \aformula_2$, we have
  \begin{description}
    \item[($\aformula = \Diamond \aformulabis$).] ($\Rightarrow$) If $\amodel,\aworld \models \Diamond \aformulabis$ then there is $\aworld_1 \in \arelation(\aworld)$
    such that $\amodel,\aworld_1 \models \aformulabis$.
    It is easy to see that $\pair{\amodel'}{\amap(\aworld_1)}$ is a $\pair{n-1}{r}$-encoding
    of $\pair{\amodel}{\aworld_1}$.
    By the induction hypothesis, $\amodel',\amap(\aworld_1) \models \aformulabis[\Diamond \gets \Diamond^{-1}]$. Moreover, by definition of $\amap$, $(\amap(\aworld_1),\aworld') \in {\arelation'}^{-1}$.
    Hence, $\amodel',\aworld' \models \Diamond^{-1}\aformulabis[\Diamond \gets \Diamond^{-1}]$. The other direction is analogous.
    \item[($\aformula = \aformula_1 \separate \aformula_2$)]
    ($\Rightarrow$) If $\amodel,\aworld \models \aformula_1 \separate \aformula_2$
    then there are $\amodel_1 = \triple{\worlds}{\arelation_1}{\avaluation}$ and $\amodel_2 = \triple{\worlds}{\arelation_2}{\avaluation}$ s.t. $\amodel = \amodel_1 + \amodel_2$, $\amodel_1,\aworld \models \aformula_1$ and $\amodel_2,\aworld \models \aformula_2$.
    We partition $\arelation'$ into $\arelation_1'$ and $\arelation_2'$ (hence, $\amodel' = \triple{\worlds'}{\arelation_1'}{\avaluation'} + \triple{\worlds'}{\arelation_2'}{\avaluation'}$) so that
    \begin{itemize}
      \item for every $\pair{\aworld_1}{\aworld_2} \in \arelation_1|_{\aworld}^{\leq n}$ ${\arelation_1'}^{-1}$, $\pair{\aworld_2}{\aworld_1} \in \arelation_1'$;
      \item for every $\pair{\aworld_1}{\aworld_2} \in \arelation_2|_{\aworld}^{\leq n}$ ${\arelation_1'}^{-1}$, $\pair{\aworld_2}{\aworld_1} \in \arelation_2'$.
    \end{itemize}
    From the first property of $\amap$, this partitioning can always be done, and moreover $\pair{\triple{\worlds'}{\arelation_1'}{\avaluation'}}{\aworld'}$ can be shown to be a $\pair{n}{r}$-encoding of $\pair{\amodel_1}{\aworld}$, whereas
    $\pair{\triple{\worlds'}{\arelation_2'}{\avaluation'}}{\aworld'}$ is a $\pair{n}{r}$-encoding of $\pair{\amodel_2}{\aworld}$.
    By the induction hypothesis, $\triple{\worlds'}{\arelation_1'}{\avaluation'},\aworld' \models \aformula_1[\Diamond \gets \Diamond^{-1}]$
    and $\triple{\worlds'}{\arelation_2'}{\avaluation'},\aworld' \models \aformula_2[\Diamond \gets \Diamond^{-1}]$.
    Thus, $\amodel',\aworld' \models (\aformula_1 \separate \aformula_2)[\Diamond \gets \Diamond^{-1}]$.
    The other direction is analogous.
    \qedhere
  \end{description}
\end{proof}

\begin{lemma}
  Let $\aformula$ in $\modallogicSC$. Let $\mathtt{locacycl}$ be the $\MSL{}{\separate,\Diamondminus}$ formula $r \land\!\bigwedge_{i \in \interval{1}{\md{\aformula}}}(\Box^{-1})^i \lnot r$, where $r$ is an atomic proposition no appearing in $\aformula$.
  $\aformula$ is satisfiable w.r.t.\ $\modallogicSC$ if and only if $\aformula[\Diamond \gets \Diamond^{-1}] \land \mathtt{locacycl}$ is satisfiable w.r.t.\ $\MSL{}{\separate,\Diamondminus}$.
\end{lemma}

\begin{proof}
  \begin{description}
    \item[($\Rightarrow$):] Let $\aformula$ be satisfiable and suppose $\amodel = \triple{\worlds}{\arelation}{\avaluation}$ be a finite forest and $\aworld \in \worlds$ s.t.\ $\pair{\amodel}{\aworld} \models \aformula$.
    W.l.o.g. assume $\worlds \subseteq_{\fin} \Nat$.
    Let us consider the $\MSL{}{\separate,\Diamondminus}$ model
    $\amodel' = \triple{\Nat}{\arelation^{-1}}{\avaluation'}$, where $\avaluation'(r) \egdef \{\aworld\}$ whereas for every $\avarprop \in \AP \setminus \{r\}$ $\avaluation'(\avarprop) = \avaluation(\avarprop)$.

    It is straightforward to show that $\pair{\amodel'}{\aworld}$ is a $\pair{\md{\aformula}}{r}$-encoding of $\pair{\amodel}{\aworld}$.
    Since $\amodel$ is acyclic, so is $\amodel'$ and from the definition of $\avaluation'$ we conclude that $\amodel,\aworld \models \mathtt{locacycl}$.
    By Lemma~\ref{A:lemma:MSL-induction}, $\amodel,\aworld \models \aformula[\Diamond \gets \Diamond^{-1}]$.

    \item[($\Leftarrow$):]
      Let $\aformula[\Diamond \gets \Diamond^{-1}] \land \mathtt{locacycl}$ be satisfiable. Let $\amodel' = \triple{\worlds'}{\arelation'}{\avaluation'}$ be a model of $\MSL{}{\separate,\Diamondminus}$ and $\aworld' \in \worlds'$ such that
      $\amodel',\aworld' \models \aformula[\Diamond \gets \Diamond^{-1}] \land \mathtt{locacycl}$.
      Let us consider the Kripke-like structure
      $\amodel = \triple{\worlds}{\arelation}{\avaluation}$ such that
      \begin{itemize}
        \item $\arelation = {\arelation'}^{-1}|_{\aworld'}^{\leq \md{\aformula}}$;
        \item $\worlds =  {\arelation}^{*}(\aworld')$, i.e. the set of worlds appearing in ${\arelation'}^{-1}|_{\aworld'}^{\leq \md{\aformula}}$;
        \item for every $\avarprop \in \AP$, $\avaluation(\avarprop) = \avaluation'(\avarprop) \cap \worlds$.
      \end{itemize}
      By $\pair{\amodel'}{\aworld'} \models \mathtt{locacycl}$,
      we can show that $\arelation$ is acyclic. Hence, $\amodel$ is a finite forest.
      By definition, $\pair{\amodel'}{\aworld'}$ is a $\pair{\md{\aformula}}{r}$-encoding of $\pair{\amodel}{\aworld'}$.
      Thus, from $\pair{\amodel'}{\aworld'} \models \aformula[\Diamond \gets \Diamond^{-1}]$ and by Lemma~\ref{A:lemma:MSL-induction}, we conclude that
      $\pair{\amodel}{\aworld'} \models \aformula$. \qedhere
  \end{description}
\end{proof}

\fi

\end{document}